\pdfoutput=1
\documentclass[11pt,a4paper]{article}

\usepackage[T1]{fontenc}
\usepackage[utf8]{inputenc}
\usepackage{lmodern}
\usepackage[margin=1in]{geometry}
\usepackage{microtype}
\usepackage{amsmath,amssymb,amsthm,mathtools}
\usepackage[dvipsnames]{xcolor}
\usepackage{tikz}
\usetikzlibrary{arrows.meta,positioning,shapes.geometric,calc,fit,backgrounds,svg.path}
\makeatletter
\newcommand\FAicon[1]{\begingroup\edef\FA@scale{\strip@pt\dimexpr 1em/640\relax}%
  \tikz[baseline=-0.8em,scale=\FA@scale,yscale=-1]\fill svg {#1};\endgroup}
\makeatother
\DeclareRobustCommand\faCircleXmark{\FAicon{M320 576C461.4 576 576 461.4 576 320C576 178.6 461.4 64 320 64C178.6 64 64 178.6 64 320C64 461.4 178.6 576 320 576zM231 231C240.4 221.6 255.6 221.6 264.9 231L319.9 286L374.9 231C384.3 221.6 399.5 221.6 408.8 231C418.1 240.4 418.2 255.6 408.8 264.9L353.8 319.9L408.8 374.9C418.2 384.3 418.2 399.5 408.8 408.8C399.4 418.1 384.2 418.2 374.9 408.8L319.9 353.8L264.9 408.8C255.5 418.2 240.3 418.2 231 408.8C221.7 399.4 221.6 384.2 231 374.9L286 319.9L231 264.9C221.6 255.5 221.6 240.3 231 231z}}
\DeclareRobustCommand\faCopyright{\FAicon{M320 576C461.4 576 576 461.4 576 320C576 178.6 461.4 64 320 64C178.6 64 64 178.6 64 320C64 461.4 178.6 576 320 576zM269.1 370.9C297.2 399 342.8 399 370.9 370.9C380.3 361.5 395.5 361.5 404.8 370.9C414.1 380.3 414.2 395.5 404.8 404.8C357.9 451.7 282 451.7 235.1 404.8C188.2 357.9 188.2 282 235.1 235.1C282 188.2 357.9 188.2 404.8 235.1C414.2 244.5 414.2 259.7 404.8 269C395.4 278.3 380.2 278.4 370.9 269C342.8 240.9 297.2 240.9 269.1 269C241 297.1 241 342.7 269.1 370.8z}}
\DeclareRobustCommand\faCreativeCommonsZero{\FAicon{M319.6 72C461.4 72 568 182.1 568 320C568 467.1 449.5 568 319.6 568C185.6 568 72 458.5 72 320C72 187.1 176.7 72 319.6 72zM320.4 116.7C202.2 116.7 116.7 214.6 116.7 320C116.7 429.8 207.9 522.8 320.4 522.8C423.6 522.8 523.2 441.7 523.2 320C523.3 206.2 433 116.7 320.4 116.7zM320 177.2C238.1 177.2 217.5 254.5 217.5 320C217.5 385.5 238.1 462.8 320 462.8C401.9 462.8 422.5 385.5 422.5 320C422.5 254.5 401.9 177.2 320 177.2zM320 231.1C323.3 231.1 326.4 231.6 329.2 232.3C335.1 237.4 338 244.4 332.3 254.2L277.8 354.4C276.1 341.7 275.9 329.3 275.9 320C275.9 291.2 277.9 231.1 320 231.1zM360.8 277.3C363.7 292.7 364.1 308.7 364.1 320C364.1 348.9 362.1 408.9 320 408.9C306.5 408.9 287.4 401.2 299.9 382.5L360.8 277.3z}}
\DeclareRobustCommand\faRegistered{\FAicon{M320 576C461.4 576 576 461.4 576 320C576 178.6 461.4 64 320 64C178.6 64 64 178.6 64 320C64 461.4 178.6 576 320 576zM264 208L344 208C383.8 208 416 240.2 416 280C416 308.9 399 333.8 374.4 345.3L404.6 395.6C411.4 407 407.7 421.7 396.4 428.5C385.1 435.3 370.3 431.6 363.5 420.3L322.5 352L288.1 352L288.1 408C288.1 421.3 277.4 432 264.1 432C250.8 432 240.1 421.3 240.1 408L240.1 232C240.1 218.7 250.8 208 264.1 208zM336 304L344 304C357.3 304 368 293.3 368 280C368 266.7 357.3 256 344 256L288 256L288 304L336 304z}}
\usepackage{booktabs,tabularx,array,multirow,longtable}
\usepackage{pifont}
\usepackage{enumitem}
\usepackage[nottoc,notlot,notlof]{tocbibind}
\usepackage[hidelinks]{hyperref}
\usepackage{xurl}
\usepackage[capitalise,nameinlink]{cleveref}
\usepackage{caption}
\usepackage[section]{placeins}
\newcolumntype{L}[1]{>{\raggedright\arraybackslash}p{#1}}
\newcolumntype{Y}{>{\raggedright\arraybackslash}X}

\definecolor{wBlack}{HTML}{000000}
\definecolor{wOrange}{HTML}{E69F00}
\definecolor{wSky}{HTML}{56B4E9}
\definecolor{wGreen}{HTML}{009E73}
\definecolor{wYellow}{HTML}{F0E442}
\definecolor{wBlue}{HTML}{0072B2}
\definecolor{wVermillion}{HTML}{D55E00}
\definecolor{wPurple}{HTML}{CC79A7}

\newcommand{\Collapse}{{\faCircleXmark}{\,\,\textsc{Certified}}}
\newcommand{\Cond}{{\faCopyright}{\,\,\textsc{Conditional}}}
\newcommand{\Open}{{\faCreativeCommonsZero}{\,\,\textsc{Open}}}
\newcommand{\Survive}{{\faRegistered}{\,\,\textsc{Outside reach}}}
\newcommand{\vC}{{\faCircleXmark}}
\newcommand{\vK}{{\faCopyright}}
\newcommand{\vO}{{\faCreativeCommonsZero}}
\newcommand{\vS}{{\faRegistered}}

\theoremstyle{plain}
\newtheorem{theorem}{Theorem}[section]
\newtheorem{proposition}[theorem]{Proposition}
\newtheorem{lemma}[theorem]{Lemma}
\newtheorem{corollary}[theorem]{Corollary}
\theoremstyle{definition}
\newtheorem{definition}[theorem]{Definition}

\newtheorem{question}[theorem]{Open Question}
\theoremstyle{remark}
\newtheorem{remark}[theorem]{Remark}
\crefname{hypothesis}{Hypothesis}{Hypotheses}
\crefname{question}{Open Question}{Open Questions}

\newcommand{\Z}{\mathbb{Z}}
\newcommand{\R}{\mathbb{R}}
\newcommand{\F}{\mathbb{F}}
\newcommand{\Zq}{\mathbb{Z}_q}
\newcommand{\Lat}{\mathcal{L}}
\newcommand{\poly}{\mathrm{poly}}
\newcommand{\negl}{\mathrm{negl}}
\newcommand{\BQP}{\mathsf{BQP}}
\newcommand{\LWE}{\mathsf{LWE}}
\newcommand{\RLWE}{\mathsf{RLWE}}
\newcommand{\MLWE}{\mathsf{MLWE}}
\newcommand{\SIS}{\mathsf{SIS}}
\newcommand{\MSIS}{\mathsf{MSIS}}
\newcommand{\DCP}{\mathsf{DCP}}
\newcommand{\EDCP}{\mathsf{EDCP}}
\newcommand{\SVP}{\mathsf{SVP}}
\newcommand{\uSVP}{\mathsf{uSVP}}
\newcommand{\GapSVP}{\mathsf{GapSVP}}
\newcommand{\SIVP}{\mathsf{SIVP}}
\newcommand{\BDD}{\mathsf{BDD}}
\newcommand{\CVP}{\mathsf{CVP}}
\newcommand{\HReg}{\mathsf{H}_{\mathrm{Reg}}}
\newcommand{\HEDCP}{\mathsf{H}_{\mathrm{EDCP}}}
\newcommand{\HHS}{\mathsf{H}_{\mathrm{HS}}}
\newcommand{\Hstar}{\mathsf{H}^{\star}}
\newcommand{\gu}{\gamma_{\mathrm{u}}}
\newcommand{\gs}{\gamma^{\star}}
\newcommand{\ket}[1]{\lvert #1\rangle}
\newcommand{\bra}[1]{\langle #1\rvert}
\newcommand{\Otil}{\widetilde{O}}

\newcommand{\Q}{\mathbb{Q}}
\newcommand{\SurviveDag}{\Survive\textsuperscript{\,\dag}}
\newcommand{\bR}{\textsuperscript{R}}
\newcommand{\bH}{\textsuperscript{H}}

\hypersetup{pdftitle={Beyond LWE: Post-Quantum Assumptions under a Counterfactual Break of the Dihedral Coset Problem},pdfauthor={Ramses Fernandez}}
\usepackage{pgfplots}
\pgfplotsset{compat=1.18}

\title{\textbf{Beyond LWE: Post-Quantum Assumptions under a\\ Counterfactual Break of the Dihedral Coset Problem}}
\author{Ramses Fernandez\textsuperscript{\href{https://orcid.org/0000-0002-8959-636X}{[0000-0002-8959-636X]}}\\[0.3em]%\thanks{The author declares an affiliation with AegisQ Technologies as an employee.}\\[0.3em]
\small AegisQ Technologies\\
\small\texttt{rfernandez@aegisq.tech}}
\date{\today}

\begin{document}
\maketitle

\begin{abstract}
A refuted 2026 claim of an efficient quantum algorithm for the Dihedral Coset Problem (DCP) motivates this survey of the consequences of a genuine breakthrough. We distinguish a noiseless cyclic coset-state solver, a solver with a prescribed failure exponent, and an extrapolated-state solver. Regev's rare-corruption observation already makes the first hypothesis sufficient for lattice algorithms at some polynomial approximation factors; the witnessing algorithm determines the quantitative bounds. We compose unstructured reductions with explicit sample, state-preparation and success conditions, and compare their sufficient guarantees with more than thirty lattice and non-lattice assumptions. We restate the structured state conversion of Wen and Zheng and audit its forward parameter certificate; the printed bound supplies no certificate at the ML-KEM ring and modulus. We use an explicit affine embedding for heuristic gap estimates, explain the lossiness connection from decision NTRU to search Ring-LWE, and distinguish ordinary from module uniqueness. Chen's withdrawn 2024 quantum algorithm for LWE and a proposed repair of its post-processing would, if correct, reach only a many-sample, small-noise LWE regime that the failure-exponent hypothesis already covers for exponents below $1/2$, and no deployed parameter set. Tables separate rigorous consequences, heuristic comparisons and unverified applications, and distinguish solving a proof assumption from attacking a scheme. Classical structural attacks on Classic McEliece, which do not involve DCP but put the concrete security of the scheme under revision, are summarised here and analysed in detail in a companion paper. The survey closes with open questions and a scenario-planning checklist that requires a complete efficient reduction.
\end{abstract}


\tableofcontents

% ======================================================================
\section{Introduction}\label{sec:intro}
% ======================================================================

\subsection{Post-quantum cryptography and the place of LWE}\label{sec:intro-pqc}

In 1994 Shor showed that a quantum computer factors integers and computes discrete logarithms in polynomial time \cite{Sho97}. His algorithms solve instances of the \emph{abelian hidden subgroup problem}, and they undermine the factoring and discrete-logarithm assumptions behind RSA, finite-field Diffie--Hellman and elliptic-curve cryptography. Post-quantum cryptography is the response, public-key primitives based on problems for which no efficient quantum algorithm is known. NIST began a public standardisation process in 2016. Its first outputs, published in August 2024, are FIPS 203 (ML-KEM, derived from CRYSTALS-Kyber), FIPS 204 (ML-DSA, derived from CRYSTALS-Dilithium) and FIPS 205 (SLH-DSA, derived from SPHINCS$^+$) \cite{FIPS203,FIPS204,FIPS205,Kyber18,ducas2017,bernstein2019}. A fourth standard, FN-DSA (derived from Falcon \cite{Falcon20}), is in preparation as FIPS 206. In March 2025 NIST selected the code-based KEM HQC as a second key-encapsulation standard \cite{NISTIR8545,HQC}, and in May 2026 it advanced nine candidates to the third round of a separate process for additional signatures \cite{NIST26sig}, one of which, the lattice-isomorphism signature Hawk, has since been withdrawn by its submission team (\cref{sec:lip}) \cite{NIST26r3}.

Three of the four schemes selected in NIST's main process, the three published standards and FN-DSA, are built on lattices, and two of them, ML-KEM and ML-DSA, on the \emph{Learning with Errors} (LWE) problem in its module form. LWE, introduced by Regev in 2005 \cite{Reg09}, owes this position to an unusual combination of properties. It enjoys a worst-case to average-case reduction: solving random LWE instances is at least as hard as approximating standard lattice problems in the \emph{worst} case, quantumly \cite{Reg09} and, for some parameters, classically \cite{Pei09,BLPRS13}. It is versatile, supporting public-key encryption, identity- and attribute-based encryption, and fully homomorphic encryption \cite{GPV08,BGV12,Pei16}. And its algebraically structured variants, Ring-LWE and Module-LWE \cite{LPR13,LS15}, are efficient enough to deploy at scale.

The quantum hardness of LWE is a conjecture supported by reductions and by the study of attacks in specified parameter regimes. Quantum lattice-reduction and sieving algorithms give important benchmarks: the cited heuristic classical sieve has running time $2^{0.292d+o(d)}$ in dimension $d$ \cite{BDGL16}, and the cited quantum random-walk sieve has running time $2^{0.2570d+o(d)}$ in its quantum-memory model \cite{CL21}. These are algorithmic estimates, not lower bounds on all attacks. A different route is provided by reductions to dihedral problems. Regev reduced approximate unique-SVP to the \emph{dihedral coset problem} with a failure parameter \cite{Reg04}; Brakerski, Kirshanova, Stehlé and Wen connected LWE to an extrapolated form of that problem under parameter-dependent quantum reductions \cite{BKSW18}. This survey asks which cryptographic consequences those reductions justify for a hypothetical efficient solver, and at which parameters.

The cited dihedral algorithms remain subexponential: Kuperberg's sieves take time $2^{O(\sqrt{\log N})}$, while Regev's polynomial-space variant takes $2^{O(\sqrt{\log N\log\log N})}$ \cite{Kup05,Reg04b,Kup13}. In Regev's lattice reduction, $N=(2M)^n$ with $M=2^{4n}$, hence $\log N=\Theta(n^2)$ \cite{Reg04}. Substituting this modulus into the first running-time bound gives $2^{O(n)}$, which does not establish a polynomial-time lattice algorithm or a concrete advantage over classical attacks. Sample consumption, bad-register tolerance and state-preparation accuracy must also match the reduction. \Cref{sec:prelim} explains these requirements.

\subsection{Previous surveys and research context}\label{sec:intro-alarms}

The relevant surveys approach this subject from complementary directions. Childs and van Dam \cite{CvD10} explain algebraic quantum algorithms; Peikert \cite{Pei16} develops lattice assumptions, reductions and constructions; Biasse, Bonnetain, Kirshanova, Schrottenloher and Song \cite{BBKSS23} examine quantum cryptanalysis and its costs. Chen and Sun \cite{CS24} give a dedicated dihedral overview with measurement limitations, while Dutto, Mercuri, Murru and Romano \cite{DMM25} survey the HSP from mathematical and cryptographic perspectives. \Cref{tab:prior-surveys} compares these works with the question studied here.

\begin{table}[htbp]
\centering\small
\renewcommand{\arraystretch}{1.13}
\begin{tabularx}{\textwidth}{@{}L{3.1cm}Y Y@{}}
\toprule
\textbf{Prior survey} & \textbf{Relevant scope and overlap} & \textbf{Focus of the present survey}\\
\midrule
Childs--van Dam (2010) \cite{CvD10} & Algebraic quantum algorithms, including hidden subgroups and hidden shifts; explains the algorithmic framework behind dihedral methods. & Fix state-input solver hypotheses and track their cryptographic consequences through quantitative reductions.\\
Peikert (2016) \cite{Pei16} & SIS, LWE and ring variants, their worst-case foundations and cryptographic applications; includes the NTRU lossiness connection. & Distinguish an algorithm for a proof assumption, loss of a worst-case hardness premise, and an explicit scheme attack.\\
Biasse et al.\ (2023) \cite{BBKSS23} & Quantum attacks and cost models for classical and post-quantum assumptions; overlaps on lattices, codes, isogenies and hidden shifts. & Ask what follows if specified dihedral solvers become efficient, accounting for samples, preparation errors and instance promises.\\
Chen--Sun (2024) \cite{CS24} & Dihedral HSP and DCP, Fourier sampling, alternative measurements and bounds on success from coset samples. & Use the solver interface as a premise and audit the reductions needed to transfer its capability to cryptographic instances.\\
Dutto et al.\ (2025) \cite{DMM25} & Abelian and non-abelian HSP, Fourier and black-box methods, and links to lattices, graph isomorphism and code equivalence. & Compare a selected range of assumptions under separate cyclic, robust and extrapolated-state hypotheses.\\
\bottomrule
\end{tabularx}
\caption{Comparison by organising question and analytical output. The scope descriptions concern the cited versions; they are not claims that these works omit every individual connection used here.}
\label{tab:prior-surveys}
\end{table}

\paragraph{Positioning and coverage.} Our contribution is the combined parameter and assumption audit: compatible solver domains, available classical samples and quantum states, preparation and success losses, and the promises required by concrete embeddings. This builds on the mechanisms explained in the cited literature; it is not a priority claim for the underlying reductions. \Cref{sec:intro-contrib} identifies their antecedents and the quantitative derivations made explicit here. The classification is selective: it includes assumptions with a direct dihedral connection, lattice assumptions reached by the composed oracles, and prominent alternatives that illustrate the limits of those routes. Inclusion outside a known route is not evidence of a separation. The recent structured connection of \cite{WZ26} is treated as a primary reduction to audit, rather than as another survey. Concrete attack estimation has its own literature \cite{APS15,LvdPdW12,SL26}, discussed in \cref{sec:prelim-algorithms}; our heuristic comparisons do not replace it.

\paragraph{Historical context.} The following claims and cryptanalytic developments motivate the need to distinguish algorithmic statements, parameter-qualified consequences and deployed attacks.

\paragraph{Structured lattices, 2014--2017.} Quantum algorithms for computing unit groups, class groups and principal ideal generators in number fields of arbitrary degree \cite{EHKS14,BS16}, combined with a classical method for recovering a \emph{short} generator from an arbitrary one in cyclotomic fields \cite{CDPR16}, broke schemes whose security relied on the hardness of finding short generators of principal ideals. Cramer, Ducas and Wesolowski extended the approach to arbitrary ideals, solving the shortest vector problem in ideal lattices of cyclotomic fields in quantum polynomial time within an approximation factor $\exp(\Otil(\sqrt n))$ \cite{CDW17}, and Pellet-Mary, Hanrot and Stehlé obtained trade-offs using a preprocessing phase that depends only on the field \cite{PHS19}. These attacks exploit algebraic structure, but at subexponential approximation factors, far above the polynomial factors on which Ring-LWE and Module-LWE are based. They exhibit algorithms exploiting particular algebraic structures; their stated approximation factors do not themselves give attacks on the deployed Ring-LWE or Module-LWE parameters considered here.

\paragraph{Eldar and Hallgren, 2021--2022.} Eldar and Hallgren described an efficient quantum algorithm for the bounded distance decoding problem on a class of integer lattices with a subexponential approximation factor, using a quantum state on lattices as an approximate eigenvector to self-reduce the instance \cite{EH22}. Ducas and van Woerden observed that known classical, and even deterministic, polynomial-time algorithms, LLL in particular, already solve the instances covered by the claim \cite{DvW21note}. This comparison supplies no improvement over the cited classical algorithms at those parameters, and no attack on the cryptographic instances considered here.

\paragraph{Chen, 2024, and a proposed repair.} In April 2024 Chen claimed polynomial-time quantum algorithms for $\GapSVP$ and $\SIVP$ on all $n$-dimensional lattices within approximation factors $\widetilde\Omega(n^{4.5})$, obtained through a quantum algorithm for LWE with certain polynomial modulus-to-noise ratios. On April 18 the author reported a bug in Step 9 of the algorithm, found by Hongxun Wu and, independently, by Thomas Vidick, and stated that he did not know how to fix it; the claim of a polynomial-time quantum algorithm for LWE with polynomial modulus-to-noise ratios was withdrawn \cite{Che24}. In September 2025 Zhang proposed an exact replacement for Steps 8 and 9 \cite{Zha25}. Its current version states the consequence for LWE only conditionally, on assumptions that include side information about a run-dependent offset, supplied by the state-preparation oracle, and in October 2025 Apon published a rebuttal of its first and third versions \cite{Apo25}. \Cref{sec:chen-zhang} places both claims in the map of this survey.

\paragraph{Classical breaks of post-quantum candidates, 2022.} Two candidates fell to \emph{classical} cryptanalysis. SIDH and its KEM SIKE, an isogeny-based scheme in the fourth round of the NIST process, were broken by torsion-point attacks that recover the secret isogeny in polynomial time \cite{CD23,MMPPW23,Rob23}. The multivariate signature Rainbow, a finalist, was broken by Beullens with a key-recovery attack that takes on average 53 hours on a laptop for the security-level-1 parameters \cite{Beu22}. Neither break used a quantum computer.

\paragraph{Structural attacks on McEliece, 2024--2026.} Starting with a subexponential distinguisher in 2024 \cite{Ran25}, classical structural attacks on binary Goppa codes led in 2026 to a quasi-polynomial distinguisher \cite{GIJS26} and to conditional key-recovery estimates below the classical-work references discussed by their authors \cite{Saa26cost,Wei26}; the lowest are about $2^{89}$ to $2^{128}$ per run of the underlying sparse linear algebra, from bit operations with free memory to memory-charged costs, depending on the parameter set and the cost model, under heuristic assumptions tested only on small keys \cite{Wei26}. Key-recovery challenges with brute-force labels up to $254$ were solved along the way \cite{BLT26,Hem26,Saa26a,Wei26,Saa26b}. The German BSI advised against new deployments of the scheme \cite{BSI26}. None of these attacks is quantum; \cref{sec:mceliece-status} summarises them, and a companion paper \cite{FV26mceliece} analyses them in detail, separating the proved, experimental and heuristic claims.

These episodes illustrate three recurring issues. Several high-profile claims received rapid scrutiny, including counterexamples and formal developments. Algebraic structure can enable attacks that generic estimates miss. Diversifying assumptions can reduce reliance on one problem family, provided the combined construction has an appropriate security argument; diversity alone is not a security guarantee.

\subsection{The Simon episode}\label{sec:intro-simon}

On August 6, 2026, Daniel R.\ Simon made public a preliminary draft, received by the ePrint archive on August 3, claiming a polynomial-time quantum algorithm for the Dihedral Coset Problem \cite{Sim26}. According to the abstract of its current revision, the method follows Regev's reduction of the dihedral subgroup problem to modular subset sum but erases sample bits without a subset-sum oracle, tolerates faulty samples at rates up to $1/O(\log n)$, and would yield quantum algorithms for approximate SVP with factor $\sqrt n\,\mathrm{polylog}(n)$ and for LWE with matching parameters. The claim carried unusual weight. The DCP is the problem to which Regev reduced the unique Shortest Vector Problem in 2002 \cite{Reg04}, and it is equivalent, in an extrapolated form, to LWE \cite{BKSW18}. The claim therefore prompted questions about ML-KEM and ML-DSA \cite{FIPS203,FIPS204}; consequences for their parameters require the sample, distribution and solver-domain checks developed below. It is a small irony of the history that, according to contemporaneous reporting \cite{PQC26news}, the author of the claim is the same Simon whose 1994 algorithm for hidden periods over $\F_2^n$ \cite{Sim97} inspired Shor's factoring algorithm \cite{Sho97}.

The claim did not survive scrutiny. Within days, Daniel Apon identified a gap in the argument for the third lemma of the preprint and stated publicly that the lemma is false, a one-page counterexample was posted, and the revised drafts through August 17 did not repair it in the cited analyses \cite{Sim26,PQC26news}. On August 15, Gupte, Ragavan and Zhandry published a no-go result, whose main bound is formalised in Lean 4 for a formal model of the algorithm, showing that the procedure keeps too little information about the Fourier labels after its third step to recover the hidden shift with non-negligible advantage; either the algorithm is incorrect or its expected running time is not polynomial \cite{GRZ26,GRZLean}. Guo and Yang then gave rigorous statements and proofs of the four lemmas of \cite{Sim26}, and isolated the hypothesis they all require, that a certain partition be fixed independently of the measured string, which the algorithm, as specified, does not satisfy \cite{GY26}. An independent Lean 4 audit reached the same conclusion: one lemma admits a corrected proof, but the adaptive-selection step on which the polynomial running time depends is not repaired \cite{RoyessAudit}. 

The author has not withdrawn the claim. The revision dated October 7, 2026, posted on October 8 as the last of six revisions recorded by the ePrint archive, adds a note arguing that, because the groups the algorithm keeps are post-selected on all-zero Hadamard outcomes, the analysis of \cite{GRZ26}, which traces out the selection bits, mispredicts the algorithm's output; the note concedes that the fourth lemma fails if the selection bits of those groups are traced out \cite{Sim26}. In our analysis the note does not remove the obstruction: the bound of \cite{GRZLean} holds for every Boolean-output channel applied to the state that the algorithm keeps after its partial subset-sum measurement, with the Fourier labels reduced to their digests; its one-attempt form counts restarts as an outcome; and post-selection that succeeds with non-negligible probability can raise an exponentially small advantage only by a factor of order the inverse of its acceptance probability, which leaves it negligible (\cref{sec:simon-response}). \Cref{sec:simon} describes the strategy, the three analyses and the response in detail.

\subsection{Why a counterfactual survey}\label{sec:intro-why}

The episode motivates an explicit audit of the consequences of a specified quantum breakthrough. Refuting that algorithm does not establish DCP hardness. The cited cyclic sieves remain subexponential \cite{Kup05,Reg04b,Kup13}, and relaxed variants have distinct algorithms: Bai, Jangir, Kirshanova, Ngo and Youmans give a quasi-polynomial EDCP algorithm over power-of-two moduli \cite{BJKNY25}, and Boucher, Fouque and Shen give a sieve for a cyclotomic generalisation \cite{BFS26}. Their time and sample bounds, moduli and amplitude requirements must all be matched to a lattice reduction. Improving sample efficiency would address one obstruction, but would not by itself establish polynomial running time or applicability to deployed distributions.

The purpose of this counterfactual survey is to ask what a polynomial-time dihedral algorithm would imply through specified reductions. The cyclic hypotheses are qualitatively equivalent after existentially quantifying the failure exponent, by Regev's rare-corruption observation (\cref{lem:hs-rare}); a prescribed small exponent remains a distinct quantitative requirement. The extrapolated-state hypothesis supplies a separate route whose relation to the cyclic ones is unresolved. The classification records sufficient guarantees and their limitations, rather than predicting a breakthrough or giving operational migration advice.

\begin{enumerate}
  \item \textbf{Lattice consequences extend beyond LWE.} Regev's DCP reduction supplies a worst-case ordinary-uSVP oracle \cite{Reg04}. Combining it with the BDD connection and quantum step gives the SIVP and SIS guarantees stated in \cref{sec:reach}, using mechanisms already present in \cite{LM09,Reg09,SSTX09}. These consequences depend on approximation factors, norm bounds and input promises. Solving a SIS proof assumption removes that hardness premise from the security argument; it does not by itself yield a forgery (\cref{sec:sis}).
  \item \textbf{Some non-lattice actions have a direct route.} Free and transitive cyclic group actions reduce to noiseless DCP when their group structure, public inputs and coherent evaluation meet the full conditions of \cref{thm:reach}\ref{it:hs}. Isogeny-based actions motivate this route \cite{CJS14,CLMPR18,BKV19}; commutativity alone does not verify its implementation requirements (\cref{sec:isogenies}).
  \item \textbf{Parameters decide applicability.} We compare sufficient reduction bounds with instance parameters and identify where those comparisons give no certificate (\cref{sec:reach,sec:parameters}). This is a limit of the analysed routes, not a security margin.
  \item \textbf{Sample availability matters.} The unstructured compositions have explicit sample requirements. The structured forward certificate is audited in \cref{sec:wz-audit}; the polynomial-modulus adapter and conversions for deployed distributions still require the checks in \cref{q:wz}.
\end{enumerate}

\subsection{Summary of results}\label{sec:intro-summary}

We formalise the counterfactual in three variants, noiseless cyclic coset states ($\HHS$), Regev's DCP with failure parameter ($\HReg(f)$), and the extrapolated DCP (EDCP) of Brakerski, Kirshanova, Stehlé and Wen with $\ell$ states ($\HEDCP[\ell]$), and compose the published reductions into a reach theorem with explicit thresholds, deriving the quantitative compositions with their stated premises (\cref{sec:hypothesis}). The following points summarise the classification.

\paragraph{Quantitative hypotheses.} We distinguish quantitative counterfactuals for
dihedral coset algorithms. A polynomial-time solver
for noiseless cyclic coset states already yields
polynomial-factor lattice algorithms, because it
tolerates sufficiently rare corruptions. A prescribed
failure exponent makes the approximation factors
quantitative, while extrapolated-state solvers
provide a separate route with explicit sample
requirements. The resulting guarantees are compared
with cryptographic instances only after accounting
for sample availability, distributions, lattice
promises and finite-dimensional bounds.

\paragraph{Polynomial-factor consequences.} Under $\HReg(f)$, the compositions give worst-case $\GapSVP$ at $\Otil(n^{1+2f})$, $\SIVP$ at $O(n^{3/2+2f})$, and SIS algorithms above the stated sample-dependent bounds. Under $\HHS$ these conclusions hold for some solver-dependent $f$ (\cref{cor:hs-lattice}). Solving a proof assumption does not by itself produce a forgery. For $\HEDCP[\ell]$, the restricted ball composition has an explicit admissible sample set and preparation budget; together with the information bound $\ell=\Omega(n\log q)$, its sufficient noise scale is at most $1/\widetilde\Omega(n^2)$. Its composed SIVP and SIS guarantees are about $n^3$ and $n^{5/2}$ up to logarithmic factors, without bounding every possible algorithm under the hypothesis.

\paragraph{Limits of the classification.} For the few-sample Module-LWE instances considered here,
the unstructured compositions do not furnish a
deployed-parameter attack certificate. The printed structured forward certificate is audited in
\cref{sec:wz-audit} and gives no certificate at the ML-KEM
ring and modulus; the polynomial-modulus adapter and
conversions for deployed distributions remain unverified here. The embedding estimates
also require a finite-dimensional reduction guarantee
and a justified gap estimate. These are limits of the
present classification, not a proof of security or a
measured security margin. Standard-noise LWE is not certified by the unstructured compositions presented here. For NTRU, rotations rule out ordinary uniqueness above factor one in the power-of-two cyclotomic case and above $\sqrt3$ for NTRU Prime; decision lossiness and module uniqueness require separate analysis. Small-norm SIS and the Lattice Isomorphism Problem likewise lack a certificate from these compositions.

\paragraph{Outside lattices.} Binary code-based encryption, hash-based and symmetric-key signatures, multivariate signatures and non-commutative isogeny problems have no applicable attack certificate from the dihedral reductions analysed here. Classical structural attacks have, however, put the concrete security of Classic McEliece under revision since 2024, while HQC, also outside the reach of the counterfactual, is not affected by them according to the BSI \cite{BSI26} (\cref{sec:mceliece-status}).

\paragraph{The 2026 algorithm.} We also explain the strategy of the 2026 preprint, which replaces the subset-sum step of Regev's reduction by a bit-erasure procedure that uses only the top third of the Fourier labels, and the three published analyses showing why it cannot work as specified (\cref{sec:simon}). In our analysis, the author's response of October 7, 2026, which relies on post-selection, does not change that conclusion: the published bound covers every processing of the state that the algorithm keeps, with the Fourier labels reduced to their digests, and post-selection that succeeds with non-negligible probability raises the distinguishing advantage by at most a factor equal to the inverse of its acceptance probability (\cref{sec:simon-response}).

\paragraph{Earlier claims.} Chen's withdrawn 2024 algorithm for LWE \cite{Che24} and Zhang's conditional repair of its post-processing \cite{Zha25} are not dihedral solvers. In our analysis they target a single LWE regime, many samples and $q\in\widetilde\Omega((\alpha q)^4m^2)$, which lies inside the reach of $\HReg(f)$ for every $f<1/2$ (for prime $q$ and $m\geq2n\lg q$); a correct version of Simon's claim would therefore have subsumed Chen's conclusion for LWE in that range. The regime excludes standard-noise LWE and the deployed parameter sets, and no known reduction derives any of the three hypotheses from a solver for it (\cref{sec:chen-zhang}).

In total we classify more than thirty assumptions and schemes, including the NIST standards and the third-round candidates of the additional-signature process, estimate heuristic gaps for deployed parameter sets, and compare the efficiency of the schemes outside the reach.

\subsection{Contributions and organisation}\label{sec:intro-contrib}

\begin{itemize}[leftmargin=*,itemsep=2pt]
  \item \Cref{sec:prelim} introduces the background used in the classification: lattices and Gaussian measures, the main lattice problems and their known complexity, lattice algorithms and concrete security, LWE and SIS with their worst-case reductions, the structured variants (Ring-, Module-, Middle-Product and Order-LWE, LWR, CLWE, succinct LWE) and NTRU, and the hidden-subgroup family (Kuperberg's sieve, DCP, EDCP and hidden shifts), with a map of their relations (\cref{fig:dcp-map}).
  \item \Cref{sec:simon} explains the strategy of \cite{Sim26} within Regev's subset-sum template, and why it fails: the information-theoretic barrier of \cite{GRZ26}, the partition hypothesis of \cite{GY26}, and the lemma-level audit of \cite{RoyessAudit}, together with the author's response of October 2026 and why, in our analysis, it does not remove the barrier; it also places Chen's withdrawn 2024 algorithm for LWE and Zhang's proposed repair in the same map (\cref{sec:chen-zhang}, \cref{fig:claims-map,tab:claims}).
  \item \Cref{sec:hypothesis} formalises the counterfactual as explicit hypotheses $\HHS$, $\HReg(f)$ and $\HEDCP[\ell]$, with their sources, explains Regev's reductions of 2002 and 2005 in enough detail to follow the composition, and proves a \emph{reach theorem} (\cref{thm:reach}) with explicit thresholds: a sufficient $\gu=O(n^{1/2+2f})$ for unique-SVP, $\gamma_{\mathrm g}=\Otil(n^{1+2f})$ for $\GapSVP$, $\gs=O(n^{3/2+2f})$ for $\SIVP$ under $\HReg(f)$ and $\Otil(\ell n^{2})$ under $\HEDCP[\ell]$, and norm bounds above which SIS instances are solved. The quantitative compositions are derived explicitly as propositions (\cref{sec:proofs}), using a version of Regev's quantum step for randomised and average-case decoders (\cref{lem:qstep-imperfect}) and an explicit admissible range for the number of samples used in one ball-reduction invocation of \cite[Thm. 3]{BKSW18} (\cref{lem:qary-count}); an information-theoretic bound shows that the stated fixed-width, inverse-polynomial-success EDCP solver needs $\ell=\Omega(n\log q)$ states (\cref{prop:ell-lower}), which constrains the stated fixed-width solver and the restricted ball composition. The section also analyses a sample-count obstruction for the specified few-sample constructions (\cref{prop:samples}) and illustrates it numerically.
  \item \Cref{sec:wz-audit} uses the full version of \cite{WZ26} to restate the exact conversion of structured states (\cref{lem:wz-state-map}) and to test the printed forward parameter certificate at the ML-KEM and ML-DSA public-key parameters: at the ML-KEM ring and modulus the certificate admits no feasible width (\cref{prop:wz-certificate,tab:wz-audit}), and the polynomial-modulus adapter still requires a complete search, width, success and sample audit.
  \item \Cref{sec:collapse} records sufficient consequences and unresolved checks hypothesis by hypothesis, including the applications to lattice standards and FrodoKEM.
  \item \Cref{sec:lattice-survivors} examines the lattice assumptions that are \emph{not} reached by the composition: SIS at small norm bounds, where it separates solving the assumption of a proof from forging, unstructured CVP-based schemes, NTRU, with its ring-dependent ordinary-uniqueness obstruction (\cref{lem:ntru-sym}), and the Lattice Isomorphism Problem; it estimates heuristic gaps for ML-KEM keys and ciphertexts, FrodoKEM, ML-DSA and FHE-style parameters, together with the reach of the SIS route on the lattice problems behind deployed signatures, and summarises its conclusions in \cref{tab:survivors}.
  \item \Cref{sec:non-lattice} surveys the non-lattice alternatives, codes, hash-based and symmetric-key signatures, multivariate schemes, isogenies and equivalence-problem group actions, describing representative constructions and relevant quantum attacks, summarises the 2024--2026 structural cryptanalysis of Classic McEliece, analysed in detail in a companion paper \cite{FV26mceliece}, and compares the key and signature sizes of the schemes outside the reach.
  \item \Cref{sec:frontier} discusses the research frontier of non-commutative hidden-subgroup problems, including candidates outside the analysed dihedral routes and the limited evidence supplied by measurement lower bounds.
  \item \Cref{sec:open} collects open questions with their background and priorities, scenario-planning consequences for hybrid designs, and a scenario-planning checklist for crypto-agility with worked examples, and concludes.
  \item \Cref{app:notation}, which closes this introduction, collects the notation.
\end{itemize}

\paragraph{Contribution and scope.}
Our contribution is the explicit accounting of parameters and assumptions
in the compositions presented here: sample availability, preparation
error, decoder success and the admissible range of the unstructured
ball reduction. We derive the stated SIVP and SIS guarantees with these
conditions and separate them from heuristic instance estimates.
The quantum connection from $\BDD_{1/\gamma}$ to
$\SIVP_{O(n\gamma)}$ is already recorded in \cite[Sec. 1.2]{LM09},
using Regev's quantum step; neither this connection nor its asymptotic
loss is claimed as new. Our direct SIVP proof makes its constants and
decoder implementation explicit. The SIS compositions use the same
quantum step, with the average-case route following the method
of \cite{SSTX09}. The rare-corruption observation is already present
in Regev's work, the NTRU rotation argument is elementary, and the
information bound is credited to the reviewer.
For the structured route, we restate the state conversion
of \cite{WZ26} and audit its forward width certificate: the printed
bound supplies no certificate at the ML-KEM ring and modulus, while
the polynomial-modulus adapter still needs the checks specified in
\cref{q:wz}. We make no priority claim for the underlying reduction
mechanisms, for these elementary observations, or for an exhaustive
classification of all possible dihedral consequences.

\subsection{How to read this survey}\label{sec:intro-read}

The survey introduces the notation and reduction statements needed to follow its classification, and outlines the arguments behind the central connections. Full proofs, version-specific premises and algorithmic details should be checked in the primary sources. Readers with different aims can follow different paths.
\begin{itemize}
  \item \emph{Migration planners and standards engineers} can read \cref{sec:intro}, the conclusion tables (\cref{tab:collapse,tab:survivors,tab:nonlattice}), \cref{sec:parameters,sec:efficiency} and the scenario-planning checklist of \cref{sec:open}.
  \item \emph{Cryptographers interested in the reductions} can read \cref{sec:prelim,sec:hypothesis,sec:wz-audit} and then the conclusions of \cref{sec:collapse,sec:lattice-survivors}.
  \item \emph{Quantum-algorithms researchers} can read \cref{sec:prelim-hsp,sec:simon,sec:frontier} and the open questions.
\end{itemize}

%Conclusions are given separately for each hypothesis. \Collapse{} denotes a proved sufficient consequence at parameters satisfying all stated conditions; it does not automatically denote an attack on a named scheme. \Cond{} denotes a parameter-dependent or heuristic sufficient comparison, or an application whose verification is incomplete. \Open{} denotes an unresolved connection; \Survive{} means that the analysis supplies no attack certificate by the reductions considered, not a security proof or an impossibility result. The dagger in \SurviveDag{} records an independent revision of concrete security. Basis marks are \textsuperscript{R} for rigorous guarantees, \textsuperscript{H} for heuristic instance comparisons, \textsuperscript{O} for a mathematical open question, and \textsuperscript{U} for a parameter application unverified in this survey. In every lattice comparison the finite-dimensional bound and the instance promise must be checked. A conclusion for a SIS proof assumption does not imply a forgery (\cref{sec:sis}).

%\paragraph{A note on sources.} Several of the works surveyed are 2026 preprints. We cite them as preprints, describe them only to the extent that their public versions (abstracts, specifications or formal developments) allow, and say so when a statement rests on secondary reporting. Paragraphs marked as \emph{informal intuition} are our own explanations and are not taken from the cited works.
% ==================================================================
\subsection{Notation}\label{app:notation}
% ==================================================================

\begingroup\small
\renewcommand{\arraystretch}{1.18}
\begin{longtable}{@{}>{\raggedright\arraybackslash}p{3.8cm}>{\raggedright\arraybackslash}p{11.0cm}@{}}
\toprule
\textbf{Symbol} & \textbf{Meaning}\\
\midrule
\endfirsthead
\toprule
\textbf{Symbol} & \textbf{Meaning}\\
\midrule
\endhead
\bottomrule
\endfoot
\multicolumn{2}{@{}l}{\emph{Lattices}}\\
$\Lat$, $\Lat^*$ & a lattice and its dual\\
$\lambda_i(\Lat)$, $\lambda_1^\infty$ & $i$-th successive minimum (Euclidean); minimum in the $\ell_\infty$ norm\\
$\mathrm{gh}(\Lat)$ & Gaussian heuristic $\sqrt{d/(2\pi e)}\det(\Lat)^{1/d}$\\
$\rho_s$, $D_{\Lat+c,s}$ & Gaussian function $e^{-\pi\|x\|^2/s^2}$; discrete Gaussian on a coset\\
$\eta_\epsilon(\Lat)$ & smoothing parameter\\
$\Lambda_q^\perp(A)$, $\Lambda_q(A^{\top})$ & $q$-ary lattices $\{z:Az\equiv0\}$ and $\{y\equiv A^{\top}s\}$\\
$\SVP_\gamma$, $\uSVP_\gamma$, $\GapSVP_\gamma$, $\SIVP_\gamma$, $\BDD_\delta$, $\CVP_\gamma$ & lattice problems of \cref{def:lattice-problems}; $\CVP_{\Lat,d}$ denotes decoding within distance $d$\\
gap & ratio $\lambda_2/\lambda_1$ of a unique-SVP instance; $g$, $g_A$ its heuristic value (for NTRU keys, the density ratio $\mathrm{gh}(\Lat_h)/\|(f,g)\|$, which is not a unique-SVP gap; \cref{lem:ntru-sym})\\
$\Lat_h$ & NTRU lattice $\{(u,v)\in R^2:uh\equiv v\pmod q\}$\\
\midrule
\multicolumn{2}{@{}l}{\emph{LWE, SIS and variants}}\\
$n$, $m$, $q$ & secret dimension, number of (scalar) samples, modulus\\
$\alpha$, $\sigma$, $\sigma_s$, $\sigma_e$ & noise rate (Gaussian parameter $\alpha q$); standard deviations of secret and error\\
$\LWE_{n,q,\alpha,m}$, $\RLWE$, $\MLWE$ & Learning with Errors and its ring and module versions\\
$\SIS_{n,m,q,\beta}$, $\MSIS$ & Short Integer Solution with absolute norm bound $\beta$; module version\\
$R$, $R_q$, $k$, $\ell'$ & ring $\Z[x]/(x^d+1)$ (degree $256$ in the standards), $R/qR$, module ranks (ML-DSA's $(k,\ell')$ is $(k,\ell)$ in \cite{FIPS204})\\
$\Gamma_{f,\mathcal A}(d)$, $c_0$ & sufficient finite-dimensional ordinary-uSVP bound from \cite[Thm. 1.1]{Reg04}, $O(d^{1/2+2f})$ asymptotically; coefficient of the separate numerical model\\
\midrule
\multicolumn{2}{@{}l}{\emph{Hidden structure}}\\
$D_N$ & dihedral group $\Z_N\rtimes\Z_2$\\
$\DCP_N^f$ & dihedral coset problem with failure parameter $f>0$ (\cref{def:dcp}); $\DCP_N^\infty$ is the noiseless problem, and $f\to0^+$ the limit of increasing tolerance\\
$\EDCP^\ell_{n,N,\rho}$, $\mathrm{G}\text{-}\EDCP^\ell_{n,N,r}$ & extrapolated DCP with $\ell$ states and amplitude $\rho$; Gaussian version of width $r$ (\cref{def:edcp})\\
$\ell$ & number of EDCP states; $\ell=\Omega(n\log N)$ at bounded width and inverse-polynomial success for every secret and arbitrary offsets (\cref{prop:ell-lower})\\
IP-M-EDCP & structured EDCP of \cite{WZ26} (\cref{rem:wz-map})\\
$\kappa$, $c$, $r_0$ & statistical parameter of \cite{BKSW18}, fixed here to $\lceil\log^2n\rceil$, and constant $6\sqrt{2\pi e}$ of \cite[Thm. 3]{BKSW18}; minimal width in $\HEDCP[\ell]$\\
$\mathcal M_q(n)$, $m_{\max}$ & numbers of samples admissible in the proof of \cite[Thm. 3]{BKSW18}, conditions (B1)--(B3); their maximum, $\Theta(n)$ for $n^2\leq q\leq\poly(n)$\\
\midrule
\multicolumn{2}{@{}l}{\emph{Hypotheses and thresholds}}\\
$\HReg(f)$, $\HEDCP[\ell]$, $\HHS$, $\HHS^{\mathrm{ab}}$ & hypotheses of \cref{sec:hyp-def}\\
$\Hstar$ & conjunction $\HReg(f)\wedge\HEDCP$; no strict separation claimed\\
$\gu(d)$ & sufficient bound $\Gamma_{f,\mathcal A}(d)$ in theorem statements; $c_0d^{1/2+2f}$ is the numerical comparison model\\
$\gamma_{\mathrm g}(d)=\Otil(d^{1+2f})$ & $\GapSVP$ factor reached under $\HReg(f)$\\
$\gs_{\mathrm{Reg}}=O(n^{3/2+2f})$, $\gs_{\mathrm E}=\Otil(\ell n^{2})$ & $\SIVP$ factors reached under $\HReg(f)$ and $\HEDCP[\ell]$\\
$g_{\mathrm E}(d)=2d\,\gs_{\mathrm E}(d,\ell)$ & gap above which $\HEDCP[\ell]$ decodes $d$-dimensional lattices in the worst case (\cref{cor:edcp-sivp})\\
$\alpha_{\mathrm{Reg}}(m)$, $\alpha_{\mathrm E}(n,m,\ell)$, $\alpha^*_{\mathrm E}(n,\ell)$ & sufficient LWE noise bounds of the two compositions, and the best admissible value $\alpha_{\mathrm E}(n,m_{\max},\ell)=1/\Otil(\ell n)$ (\cref{sec:reach-thm})\\
$\beta_{\mathrm{Reg}}$, $\beta_{\mathrm E}$, $\beta_{\mathrm{col}}$, $\beta_{\mathrm{void}}$ & sufficient SIS solution bounds and bounds making the stated worst-case hardness premise unavailable, with each route's parameter and distribution conditions\\
\midrule
\multicolumn{2}{@{}l}{\emph{Conclusions (given per hypothesis; compact glyphs \vC, \vK, \vO, \vS{} in multi-column tables)}}\\
\Collapse & reached by a composition of published reductions, or by a proposition of \cref{sec:proofs}, at the relevant parameters\\
\Cond & parameter-dependent sufficient comparison, heuristic estimate or unverified application\\
\Open & an unresolved connection, without an established consequence here\\
\Survive & no certificate from the reductions analysed; not a security or impossibility claim\\
\SurviveDag & outside reach, with concrete security under active revision for unrelated reasons\\
\textsuperscript{R}, \textsuperscript{H}, \textsuperscript{O}, \textsuperscript{U} & rigorous, heuristic, mathematical open question, unverified parameter application\\
(A), (B), (C) & explicit attack; proof assumption solvable; guarantee no longer informative (\cref{sec:sis})\\
\end{longtable}
\endgroup
\addtocounter{table}{-1}

% ======================================================================
\section{Preliminaries}\label{sec:prelim}
% ======================================================================

This section collects the background used in the rest of the survey. It is longer than a list of definitions because the reach of the counterfactual depends on the exact form of several classical results, the approximation factors in worst-case reductions, the number of samples a reduction consumes, the error tolerance of a quantum algorithm, and these details are easy to lose in a summary.

\subsection{Lattices and Gaussian measures}\label{sec:prelim-lattices}

A \emph{lattice} is a discrete additive subgroup $\Lat\subset\R^d$. Every lattice is generated by a basis $B=(b_1,\dots,b_k)$ of linearly independent vectors, $\Lat=\Lat(B)=\{Bz:z\in\Z^k\}$; the integer $k$ is the \emph{rank}, and the lattice is \emph{full-rank} if $k=d$. The \emph{determinant} or covolume is $\det(\Lat)=\sqrt{\det(B^{\top}B)}$, independent of the choice of basis. The $i$-th \emph{successive minimum} $\lambda_i(\Lat)$ is the smallest $r$ such that the closed ball of radius $r$ contains $i$ linearly independent lattice vectors; in particular $\lambda_1(\Lat)$ is the length of a shortest nonzero vector. We use the Euclidean norm throughout unless a superscript says otherwise, as in $\lambda_1^\infty$. The \emph{dual lattice} is $\Lat^*=\{x\in\mathrm{span}(\Lat):\langle x,v\rangle\in\Z\ \text{for all}\ v\in\Lat\}$, with $\det(\Lat^*)=1/\det(\Lat)$.

Minkowski's theorem bounds $\lambda_1(\Lat)\leq\sqrt d\,\det(\Lat)^{1/d}$ for a full-rank lattice. For ``random'' lattices the \emph{Gaussian heuristic} predicts the more precise value
\[
  \lambda_1(\Lat)\approx\mathrm{gh}(\Lat):=\sqrt{\frac{d}{2\pi e}}\,\det(\Lat)^{1/d},
\]
which is the radius of a ball of volume $\det(\Lat)$, up to lower-order terms. The minima of a lattice and of its dual are linked by Banaszczyk's transference theorem: $1\leq\lambda_1(\Lat)\,\lambda_d(\Lat^*)\leq d$ \cite{Ban93}.

\paragraph{Gaussian measures.} For $s>0$ let $\rho_s(x)=\exp(-\pi\|x\|^2/s^2)$. The \emph{discrete Gaussian} distribution $D_{\Lat+c,s}$ on a coset $\Lat+c$ assigns to each point $x$ the probability $\rho_s(x)/\rho_s(\Lat+c)$. Micciancio and Regev introduced the \emph{smoothing parameter}
\[
  \eta_\epsilon(\Lat)=\min\{s>0:\rho_{1/s}(\Lat^*\setminus\{0\})\leq\epsilon\},
\]
the smallest width at which a Gaussian, reduced modulo the lattice, is $\epsilon$-close to uniform \cite{MR07}. Two facts about it are used repeatedly. First, it is bounded by the last minimum: $\eta_\epsilon(\Lat)\leq\sqrt{\ln(2d(1+1/\epsilon))/\pi}\cdot\lambda_d(\Lat)$ \cite{MR07}. Second, Gaussian tail and anti-concentration estimates control the lengths and independence of samples; the width and error conditions used here are stated in \cref{lem:tails} \cite{MR07,Reg09}. Gaussian measures are the technical heart of every worst-case to average-case reduction discussed below.

\paragraph{$q$-ary lattices.} For a modulus $q$ and a matrix $A\in\Zq^{n\times m}$ define
\begin{gather*}
  \Lambda_q^\perp(A)=\{z\in\Z^m:Az\equiv0\pmod q\},\\
  \Lambda_q(A^{\top})=\{y\in\Z^m:y\equiv A^{\top}s\pmod q\ \text{for some}\ s\in\Z^n\}.
\end{gather*}
Both contain $q\Z^m$, so all their successive minima are at most $q$. They are dual to each other up to scaling, $\Lambda_q^\perp(A)=q\cdot\Lambda_q(A^{\top})^*$, and when $A$ has full rank modulo a prime $q$ their determinants are $q^n$ and $q^{m-n}$ respectively. LWE lives on lattices of the form $\Lambda_q(A^{\top})$ and SIS on lattices of the form $\Lambda_q^\perp(A)$; when we write LWE with an $m\times n$ matrix $A$, the corresponding lattice is $\Lambda_q(A)=\{y\equiv As\}$.

\subsection{Lattice problems and their complexity}\label{sec:prelim-problems}

We use the following standard problems, parameterised by an approximation factor $\gamma=\gamma(d)\geq1$.

\begin{definition}[Lattice problems]\label{def:lattice-problems}
Given a basis of $\Lat\subset\R^d$:
\begin{itemize}
  \item $\SVP_\gamma$: find a nonzero $v\in\Lat$ with $\|v\|\leq\gamma\lambda_1(\Lat)$.
  \item $\GapSVP_\gamma$: decide whether $\lambda_1(\Lat)\leq1$ or $\lambda_1(\Lat)>\gamma$.
  \item $\SIVP_\gamma$: find $d$ linearly independent vectors of $\Lat$ of length at most $\gamma\lambda_d(\Lat)$.
  \item $\uSVP_\gamma$: under the promise $\lambda_2(\Lat)>\gamma\lambda_1(\Lat)$, find a shortest nonzero vector.
  \item $\BDD_\delta$: given $t$ with $\mathrm{dist}(t,\Lat)\leq\delta\lambda_1(\Lat)$, find the closest lattice vector.
  \item $\CVP_\gamma$: given $t$, find $v\in\Lat$ with $\|t-v\|\leq\gamma\,\mathrm{dist}(t,\Lat)$.
\end{itemize}
\end{definition}

The \emph{gap} of a $\uSVP$ instance is the ratio $\lambda_2/\lambda_1$; a $\uSVP_\gamma$ algorithm solves every instance whose gap exceeds $\gamma$. The problems are closely related. Lyubashevsky and Micciancio proved that $\uSVP$, $\BDD$ and $\GapSVP$ are equivalent up to a factor $\sqrt{n/\log n}$; in particular $\BDD_{1/(2\gamma)}$ reduces to $\uSVP_\gamma$ and $\uSVP_\gamma$ reduces to $\BDD_{1/\gamma}$, a connection tight up to a small constant, and $\GapSVP$ reduces to $\BDD$ with a loss of $\sqrt{n/\log n}$ \cite{LM09}. Their paper used this to show that the Ajtai--Dwork cryptosystem \cite{AD97} and Regev's 2004 cryptosystem \cite{Reg04c}, previously based only on $\uSVP$, can be based on worst-case $\GapSVP$ with factors $O(n^{2.5})$ and $O(n^2)$ respectively.

The approximation factor determines how hard these problems are. Exact $\SVP$ is $\mathsf{NP}$-hard under randomised reductions \cite{Ajt98}. At the other end, $\GapSVP_\gamma$ lies in $\mathsf{NP}\cap\mathsf{coAM}$ for $\gamma=\sqrt{d/\log d}$ \cite{GG00} and in $\mathsf{NP}\cap\mathsf{coNP}$ for $\gamma=\sqrt d$ \cite{AR05}, so it is unlikely to be $\mathsf{NP}$-hard at those factors. The worst-case foundations used below involve polynomial approximation factors. The cited general basis-reduction algorithms supply much larger approximation guarantees in polynomial time \cite{LLL82,SE94}; their guarantees are not lower bounds against other algorithms or statements about every structured lattice. The counterfactual asks which polynomial factors become attainable through the specified reductions.

\subsection{Algorithms and concrete security}\label{sec:prelim-algorithms}

Concrete security estimates for lattice schemes rest on a small number of algorithmic families.

\paragraph{Basis reduction.} The LLL algorithm finds, in polynomial time, a basis whose first vector satisfies $\|b_1\|\leq2^{(d-1)/2}\lambda_1(\Lat)$ \cite{LLL82}. The BKZ algorithm of Schnorr and Euchner generalises it by calling an exact $\SVP$ solver on projected sublattices of dimension $\beta$, the \emph{block size} \cite{SE94}. A commonly used heuristic estimates $\|b_1\|\approx\delta_\beta^{\,d}\det(\Lat)^{1/d}$, where the root Hermite factor $\delta_\beta$ decreases slowly with $\beta$; the standard estimate is $\delta_\beta\approx\big(\tfrac{\beta}{2\pi e}(\pi\beta)^{1/\beta}\big)^{1/(2(\beta-1))}$ for large $\beta$ \cite{APS15}. The cost of BKZ is dominated by the cost of $\SVP$ in dimension $\beta$, and the usual block-reduction cost model associates polynomial approximation factors with large blocks and exponential cost. This is a description of that model, not an algorithm-independent lower bound.

\paragraph{Exact SVP.} Enumeration solves $\SVP$ in time $2^{O(d\log d)}$ and polynomial memory. Sieving algorithms trade memory for time: the cited heuristic classical sieve runs in time $2^{0.292d+o(d)}$ \cite{BDGL16}, and the cited quantum random-walk variant in time $2^{0.2570d+o(d)}$ \cite{CL21}, improving on earlier quantum sieves \cite{LMvdP13} and on the exponent $0.2653$ of Laarhoven's quantum algorithm, as reported in \cite{CL21}. Among provable algorithms, Aggarwal, Dadush, Regev and Stephens-Davidowitz solve $\SVP$ in time and space $2^{d+o(d)}$ by discrete Gaussian sampling \cite{ADRS15}. The General Sieve Kernel of Albrecht et al.\ is an implementation framework for sieving used in public lattice challenges \cite{ADHKPS19}.

\paragraph{Estimating schemes.} An LWE instance is attacked either as a unique-SVP instance in an embedding lattice (the \emph{primal} attack, which succeeds when BKZ with block size $\beta$ isolates the planted short vector), or through short vectors in the dual lattice that distinguish LWE samples from uniform (the \emph{dual} attack), or by combinatorial and algebraic methods. Albrecht, Player and Scott surveyed these attacks and their costs and introduced the estimator methodology used in concrete lattice-security analysis \cite{APS15}. Two broader surveys complement it: Laarhoven, van de Pol and de Weger give a tutorial on hard lattice problems, basis reduction, enumeration and sieving and their use in security estimates \cite{LvdPdW12}, and Shen and Li survey the current state of classical lattice algorithms \cite{SL26}. Throughout the survey, ``concrete security'' refers to estimates of this kind; the counterfactual supplies additional sufficient algorithmic guarantees. The estimates quoted for existing attacks do not bound the cost of all algorithms that a future solver might enable.

\subsection{Learning with Errors}\label{sec:prelim-lwe}

\begin{definition}[LWE \cite{Reg09}]
Let $n,q,m$ be integers and $\chi$ an error distribution on $\Z$, typically a discrete Gaussian of parameter $\alpha q$ for a \emph{noise rate} $\alpha\in(0,1)$. Search-$\LWE_{n,q,\chi,m}$: given $(A,b=As+e\bmod q)$ with $A\leftarrow\Zq^{m\times n}$ uniform, $s\in\Zq^n$ and $e\leftarrow\chi^m$, find $s$. Decision-$\LWE$: distinguish $(A,b)$ from uniform.
\end{definition}

Under the conditions of the normal-form reduction, one can transfer hardness from a uniform secret to a secret drawn from the error distribution: Applebaum, Cash, Peikert and Sahai showed that LWE with a uniform secret reduces to LWE with a secret drawn from $\chi$, the \emph{normal form} \cite{ACPS09}. The concrete LWE-based schemes considered here use short secrets, but their precise bounded distributions and the sample cost of conversions must still be checked.

\paragraph{LWE as a lattice problem.} The pair $(A,b)$ gives a closest-vector instance on the $q$-ary lattice $\Lambda_q(A)=\{y\in\Z^m:y\equiv As\pmod q\}$, with a planted lattice point at distance $\|e\|$ from a chosen lift of $b$. It meets a BDD promise only if that distance is sufficiently small relative to the lattice minimum. For short-secret LWE, an affine embedding contains the vector $(s,e,\tau)$; this gives an ordinary-uSVP instance only when that vector is shortest and sufficiently separated from independent vectors. A Gaussian-heuristic comparison is
\begin{equation}\label{eq:gap}
  g_{\LWE}\;\approx\;\frac{\mathrm{gh}(\text{embedding lattice})}{\|(s,e,\tau)\|}.
\end{equation}
\Cref{prop:samples} specifies the embedding, its determinant and the additional heuristic assumptions. The number of samples affects this comparison.

\paragraph{Regev's cryptosystem.} The simplest LWE-based encryption scheme illustrates how LWE is used. The public key is an LWE instance $(A,b)$ with $m=O(n\log q)$ samples. To encrypt a bit $\mu$, choose a random subset $S\subseteq\{1,\dots,m\}$ and output $(\sum_{i\in S}a_i,\ \sum_{i\in S}b_i+\mu\lfloor q/2\rfloor)$, where $a_i$ are the rows of $A$. Decryption subtracts $\langle\sum_{i\in S}a_i,s\rangle$ and checks whether the result is closer to $0$ or to $q/2$; correctness holds because the accumulated error $\sum_{i\in S}e_i$ is small. Semantic security follows from decision-LWE together with the leftover hash lemma \cite{Reg09}.

\paragraph{Worst-case hardness.} Regev's main theorem states that, for $\alpha q>2\sqrt n$, an efficient algorithm for $\LWE_{n,q,\alpha}$ yields an efficient \emph{quantum} algorithm for $\SIVP$ and $\GapSVP$ within $\Otil(n/\alpha)$ in the worst case \cite{Reg09}. The proof is iterative. Given samples from a discrete Gaussian $D_{\Lat,r}$ on an arbitrary lattice $\Lat$, the LWE oracle is used to solve $\BDD$ on the dual lattice $\Lat^*$ within distance about $\alpha q/r$ (the \emph{classical step}); a quantum procedure then turns this $\BDD$ solver into samples from a narrower Gaussian $D_{\Lat,r'}$ with $r'<r$ (the \emph{quantum step}). Iterating from a very wide Gaussian, which is easy to sample, down to width about $\sqrt n\,\eta_\epsilon(\Lat)/\alpha$ yields short lattice vectors, which solve $\SIVP$. The quantum step is the only place where quantum computation is needed. Peikert gave a classical reduction from $\GapSVP$ to LWE with an exponential modulus $q\geq2^{n/2}$ \cite{Pei09}, and Brakerski, Langlois, Peikert, Regev and Stehlé showed that LWE in dimension $n$ with modulus $q$ reduces to LWE in dimension about $n\log q$ with a polynomial modulus, which yields classical hardness for polynomial moduli at the cost of a larger dimension \cite{BLPRS13}.

\paragraph{Search and decision.} Search and decision LWE are connected by reductions with specific modulus, distribution and sample requirements. Regev's reduction works for prime polynomial moduli \cite{Reg09}; Micciancio and Mol gave sample-preserving reductions for a wide class of moduli \cite{MM11}. Sudo, Tezuka, Hara and Yoshida gave a \emph{quantum} search-to-decision reduction, based on a quantum Goldreich--Levin theorem over $\Zq$ for prime $q$, that invokes the decision distinguisher only twice, whereas classical reductions require polynomially many invocations \cite{STHY23}.

\paragraph{Algorithms for special parameters.} LWE is easy when the error is very small. Arora and Ge linearise the system by observing that each error coordinate is a root of a known low-degree polynomial, which solves LWE in subexponential time when the error magnitude is $o(\sqrt n)$ and in polynomial time when it is constant, given enough samples \cite{AG11}. Micciancio and Peikert showed, conversely, a hardness reduction for LWE with binary errors when the number of samples is limited \cite{MP13}. Both results illustrate that the number of samples is a security parameter of LWE in its own right.

\subsection{The Short Integer Solution problem}\label{sec:prelim-sis}

\begin{definition}[SIS \cite{ajtai1996}]
$\SIS_{n,m,q,\beta}$: given uniform $A\in\Zq^{n\times m}$, find nonzero $z\in\Z^m$ with $Az\equiv0\pmod q$ and $\|z\|\leq\beta$, i.e.\ a short nonzero vector of $\Lambda_q^\perp(A)$.
\end{definition}

A solution exists whenever $\beta\geq\sqrt m\,q^{n/m}$, by Minkowski's theorem applied to $\Lambda_q^\perp(A)$. Ajtai's 1996 paper, which founded lattice-based cryptography, showed that solving SIS on average is at least as hard as approximating lattice problems within polynomial factors in the worst case, and deduced that the function $f_A(x)=Ax\bmod q$ on short inputs $x$ is one-way and, for suitable parameters, collision-resistant \cite{ajtai1996}. Micciancio and Regev tightened the reduction using Gaussian measures, basing SIS on $\SIVP$ within $\Otil(n)$ factors \cite{MR07}, and Micciancio and Peikert extended it to moduli as small as $\beta\cdot n^{\delta}$ for any constant $\delta>0$ \cite{MP13}. In the form we use, for $q\geq\beta\cdot\omega(\sqrt{n\log n})$, $\SIS_{n,m,q,\beta}$ is at least as hard as $\SIVP_{\beta\cdot\Otil(\sqrt n)}$ in the worst case \cite{MR07,GPV08,MP13}.

SIS is the basis of lattice-based hash functions and signatures. Gentry, Peikert and Vaikuntanathan showed how to generate $A$ together with a short basis of $\Lambda_q^\perp(A)$ (a \emph{trapdoor}), how to use it to sample short preimages of $f_A$ from a discrete Gaussian without leaking the trapdoor, and how this yields hash-and-sign signatures in the random-oracle model and identity-based encryption from the dual of Regev's scheme \cite{GPV08}. Micciancio and Peikert simplified trapdoor generation with structured ``gadget'' matrices \cite{MP12}. Lyubashevsky gave signatures without trapdoors in the Fiat--Shamir paradigm, using rejection sampling to make the distribution of signatures independent of the secret key \cite{Lyu12}; this ``Fiat--Shamir with aborts'' design underlies ML-DSA. The duality $\Lambda_q^\perp(A)=q\cdot\Lambda_q(A^{\top})^*$ connects the two problems: short vectors of $\Lambda_q^\perp(A)$ are exactly what the dual attack uses to distinguish LWE samples from uniform, so suitable nontrivial SIS solutions can distinguish LWE when the induced error inner product has a detectable bias. A norm bound alone must be interpreted together with the error distribution and distinguishing advantage.

\subsection{Algebraically structured variants}\label{sec:prelim-structured}

Plain LWE needs public matrices of $\Theta(n^2\log q)$ bits. Structured variants replace the integers by a ring of algebraic integers, so that one ring element encodes a whole structured block of the matrix.

\paragraph{Ring-LWE.} Let $R=\Z[x]/(x^d+1)$ with $d$ a power of two (the ring of integers of a power-of-two cyclotomic field) and $R_q=R/qR$. A Ring-LWE sample is $(a,b=a\cdot s+e)\in R_q^2$ with uniform $a$, a secret $s\in R_q$ and a short error $e\in R$. Multiplication by $a$ is a negacyclic $d\times d$ matrix, so one ring sample is $d$ structured scalar LWE samples. Lyubashevsky, Peikert and Regev introduced the problem for general number fields, proved a quantum reduction from worst-case approximate $\SVP$ on ideal lattices to search Ring-LWE, and a search-to-decision reduction for cyclotomic fields with suitable moduli \cite{LPR13}. Peikert, Regev and Stephens-Davidowitz later gave a direct worst-case hardness reduction for decision Ring-LWE over any number field and modulus, with the error requirements of their theorem \cite{PRS17}.

\paragraph{Module-LWE and Module-SIS.} Module-LWE uses vectors in $R_q^k$: samples are $(a,b=\langle a,s\rangle+e)$ with $a,s\in R_q^k$. It interpolates between Ring-LWE ($k=1$) and LWE ($d=1$). Brakerski, Gentry and Vaikuntanathan introduced it as ``general LWE'' \cite{BGV12}, and Langlois and Stehlé proved worst-case to average-case reductions for Module-LWE and Module-SIS from problems on module lattices \cite{LS15}. Albrecht and Deo showed that Module-LWE of rank $k$ with modulus $q$ reduces to Ring-LWE with modulus $q^k$ \cite{AD17}. ML-KEM and ML-DSA use Module-LWE over $R=\Z[x]/(x^{256}+1)$ with small rank $k\in\{2,3,4\}$ (ML-KEM) or rank pairs up to $(8,7)$ (ML-DSA) \cite{FIPS203,FIPS204}.

\paragraph{Middle-Product, Order and general structured LWE.} Roşca, Sakzad, Stehlé and Steinfeld introduced Middle-Product LWE (MP-LWE), in which the product of polynomials is replaced by the ``middle'' coefficients of the product, so that the problem is defined independently of any particular number field; they proved that MP-LWE is at least as hard as Polynomial-LWE for every polynomial in an exponentially large family, so hardness of any source family transfers under the corresponding parameter map \cite{RSSS17}. Njah Nchiwo and Pellet-Mary improved the reduction from RLWE to MP-LWE \cite{NP26}, and Pellet-Mary and Xia gave a unified reduction, parameterised by field elements, that recovers the earlier ones as special cases and applies to number fields satisfying their polynomial-boundedness conditions \cite{PX26}. Bolboceanu, Brakerski, Perlman and Sharma introduced Order-LWE, a version over orders of number fields, and used it to study Ring-LWE with entropic secrets \cite{BBPS19}. Peikert and Pepin gave a unified framework for these algebraically structured variants with simpler and tighter reductions among them \cite{PP24}. The direction of all these reductions matters below: they show that the variants are \emph{at least as hard as} Ring-LWE, not that they are easier.

\paragraph{Short distributions.} Deployed schemes use secrets and errors with small bounded coefficients rather than Gaussians. Boudgoust, Jeudy, Roux-Langlois and Wen proved that Module-LWE with a uniform $\eta$-bounded secret, for any $1\leq\eta\ll q$, and Gaussian error is at least as hard as standard Module-LWE when the module rank is at least logarithmic in the ring degree, in both search and decision versions; they also proved that search Module-LWE with uniform $\eta$-bounded error is at least as hard as the standard problem when the number of samples is close to the module rank \cite{BJRW23}.

\paragraph{Learning with Rounding.} Banerjee, Peikert and Rosen replaced the random error by deterministic rounding: an LWR sample is $(a,\lfloor(p/q)\langle a,s\rangle\rceil\bmod p)$ for a smaller modulus $p$ \cite{BPR12}. They showed that LWR is at least as hard as LWE when $q/p$ exceeds the LWE error by a large enough factor, and used it to build pseudorandom functions. An LWR instance is again a $\BDD$ instance, with an error of size about $q/p$.

\paragraph{Assumptions for advanced primitives.} Wee introduced the $\ell$-succinct LWE assumption, a falsifiable assumption implied by evasive LWE, to build attribute-based encryption and laconic function evaluation for circuits with sublinear ciphertext overhead \cite{Wee24}. Abram, Malavolta and Roy introduced the decomposed LWE assumption, which can be interpreted as postulating the circular security of a natural lattice-based encryption scheme, used it for succinct randomised encodings and attribute-based encryption for RAM programs, and showed that it is implied by succinct LWE \cite{AMR25}. Abram, Arnon, Cini, Lou, Malavolta and Roy proved the converse, so that the two assumptions are equivalent under appropriate parameters \cite{AACLMR26}. Both are strengthenings of LWE, in the sense that their hardness implies that of LWE. For an algorithm the relevant direction is the converse: an LWE solver breaks succinct LWE directly, because its instances contain a standard LWE instance with a uniform matrix, and breaks decomposed LWE only through the reduction of \cite{AACLMR26}, at the parameters for which it is proved (\cref{sec:collapse-variants}).

\paragraph{Continuous LWE.} Bruna, Regev, Song and Tang introduced Continuous LWE (CLWE), in which samples are Gaussian vectors $y\in\R^n$ labelled by $\gamma\langle y,w\rangle+e\bmod1$ for a hidden unit vector $w$; they proved its hardness from worst-case lattice problems by a quantum reduction and deduced the hardness of learning mixtures of Gaussians \cite{BRST21}. Gupte, Vafa and Vaikuntanathan gave a classical reduction from LWE to CLWE \cite{GVV22}.

\subsection{NTRU}\label{sec:prelim-ntru}

In the NTRU cryptosystem of Hoffstein, Pipher and Silverman the secret key is a pair of polynomials $f,g$ with small coefficients in a ring $R_q=\Z_q[x]/(\Phi)$, and the public key is $h=g/f$ \cite{HPS98}. Encryption adds a small multiple of $h$ to an encoding of the message; decryption multiplies by $f$, which turns the ciphertext into a short element that can be recovered exactly and reduced modulo a small modulus. Key recovery asks for a short vector of the rank-two module lattice $\Lat_h:=\{(u,v)\in R^2:uh\equiv v\pmod q\}$, which contains $(f,g)$; for coefficient size $\sigma$, the ratio $\mathrm{gh}(\Lat_h)/\|(f,g)\|=\Theta(\sqrt q/\sigma)$ measures how much shorter the key is than a typical lattice vector. Because $\Lat_h$ is an $R$-module, the rotations $x^i(f,g)$ are as short, or almost as short, as the key itself, giving the ring-dependent uniqueness bounds of \cref{lem:ntru-sym} in \cref{sec:ntru}. The \emph{decision NTRU} problem asks to distinguish $h$ from a uniform element of $R_q$.

Three results frame its relation to LWE. Stehlé and Steinfeld gave an NTRU variant with Gaussian keys of width about $\sqrt q\cdot\poly(n)$, for which the public key is statistically close to uniform, so that ciphertexts are Ring-LWE samples and security reduces to Ring-LWE \cite{SS11}. Pellet-Mary and Stehlé observed that decision NTRU reduces to search Ring-LWE, which gives only an upper bound on its hardness, and provided reduction-based lower bounds: worst-case approximate $\SVP$ on ideal lattices reduces to an average-case search variant of NTRU, and another average-case search variant reduces to decision NTRU \cite{PS21}. The relation $hf-g=0$ is a homogeneous short relation,
not a sample from the standard Ring-LWE distribution:
$h$ depends on the short pair $(f,g)$. The connection from
decision NTRU to search Ring-LWE instead uses a lossiness
argument; see \cite[Sec. 4.4.4]{Pei16}. To compose that
connection with a dihedral solver, one must match the
Ring-LWE oracle's sample count, secret and error
distributions, and success probability. For certain large-modulus NTRU parameter families, independent attacks are known: Albrecht, Bai and Ducas showed that the norm map to a subfield reduces key recovery to a lower-dimensional lattice when $q$ is large, breaking the ``overstretched'' parameters used in some multilinear-map and homomorphic-encryption proposals \cite{ABD16}, and Kirchner and Fouque showed that lattice reduction itself performs much better than expected on such parameters because of a dense sublattice \cite{KF17}. NTRU Prime replaces the cyclotomic ring by a prime-degree field with a large Galois group, to avoid algebraic structure that attacks might exploit \cite{BCLV17}.

\subsection{Hidden subgroups, dihedral cosets and hidden shifts}\label{sec:prelim-hsp}

\paragraph{The hidden subgroup problem.} In the hidden subgroup problem (HSP) for a finite group $G$, one is given oracle access to a function $f:G\to S$ that is constant on the left cosets of an unknown subgroup $H$ and takes distinct values on distinct cosets; the task is to find $H$. The \emph{standard method} prepares $|G|^{-1/2}\sum_{g}\ket{g}\ket{f(g)}$ and measures the second register, leaving a uniformly random \emph{coset state} $|H|^{-1/2}\sum_{h\in H}\ket{gh}$. For abelian $G$, measuring the quantum Fourier transform of coset states yields elements of the annihilator of $H$, and polynomially many of them determine $H$; this is Simon's algorithm over $\F_2^n$ \cite{Sim97} and the core of Shor's algorithms over cyclic groups \cite{Sho97}. For non-abelian groups the Fourier transform is indexed by irreducible representations, and the question becomes which measurements on coset states are informative. Measuring only the name of the representation (\emph{weak} Fourier sampling) suffices for normal subgroups \cite{HRT03}; measuring also inside the representation (\emph{strong} Fourier sampling) helps for some groups and provably fails for others (\cref{sec:frontier}). Ettinger, H{\o}yer and Knill showed that the \emph{query} complexity of the HSP is polynomial for every finite group: polynomially many coset states always contain enough information to identify $H$, although extracting it may take exponential time \cite{EHK04}.

\paragraph{The dihedral group.} The dihedral group $D_N=\Z_N\rtimes\Z_2$ is the simplest non-abelian group of cryptographic interest. Its interesting subgroups are generated by a reflection $(d,1)$, and coset sampling for the subgroup $\{0,(d,1)\}$ yields states $\tfrac{1}{\sqrt2}(\ket{0,x}+\ket{1,x+d})$ with uniformly random $x$. Ettinger and H{\o}yer showed that $O(\log N)$ such states determine $d$ information-theoretically, but their procedure for extracting it takes exponential time \cite{EH00}.

\paragraph{Kuperberg's sieve.} Applying the Fourier transform over $\Z_N$ to the second register of a coset state and measuring it yields a uniformly random classical label $y\in\Z_N$ and the qubit
\[
  \ket{\psi_y}=\tfrac{1}{\sqrt2}\big(\ket{0}+\omega^{yd}\ket{1}\big),\qquad\omega=e^{2\pi i/N}.
\]
For $N$ a power of two, a qubit with label $y=N/2$ is $\tfrac{1}{\sqrt2}(\ket0+(-1)^d\ket1)$ and reveals the parity of $d$ by a Hadamard measurement. Kuperberg's idea is to manufacture such labels by combining qubits: measuring the parity of two qubits $\ket{\psi_{y_1}}\otimes\ket{\psi_{y_2}}$ leaves, with probability $1/2$ each, a qubit equivalent to $\ket{\psi_{y_1+y_2}}$ or $\ket{\psi_{y_1-y_2}}$, with known label. By repeatedly combining pairs whose labels agree on many low-order bits, the algorithm zeroes those bits in stages, reaching label $N/2$ after $O(\sqrt{\log N})$ stages; the parity of $d$ follows, and the remaining bits are found recursively. The algorithm runs in time, query complexity and quantum space $2^{O(\sqrt{\log N})}$ \cite{Kup05}. Regev gave a variant with polynomial space and running time $2^{O(\sqrt{\log N\log\log N})}$ \cite{Reg04b}, and Kuperberg a second algorithm, the collimation sieve, with $2^{O(\sqrt{\log N})}$ quantum time, $O(\log N)$ quantum space and $2^{O(\sqrt{\log N})}$ classical space, which performs better when that classical memory can be accessed quantumly \cite{Kup13}. All three algorithms consume a subexponential number of coset states, which is one obstacle to composing them with reductions that supply only polynomially many states.

\paragraph{DCP with failure parameter.} Regev's reduction from lattices produces dihedral coset states that are only \emph{mostly} correct. He therefore defined a version of the problem that tolerates a fraction of corrupted inputs.

\begin{definition}[DCP with failure parameter {\cite[Def. 2.1]{Reg04}}]\label{def:dcp}
Let $N\geq2$ be an integer and $f>0$, and use binary logarithms in this definition. The input to $\DCP_N^f$ is a list of $\poly(\log N)$ registers. Each is, with probability at least $1-1/(\log N)^f$, in the state $\tfrac{1}{\sqrt2}(\ket{0,x}+\ket{1,(x+d)\bmod N})$ on $1+\lceil\log N\rceil$ qubits, where $x\in\Z_N$ is arbitrary and $d\in\Z_N$ is fixed; otherwise it is in an arbitrary state $\ket{b,x}$ (a \emph{bad} register). The goal is to output $d$ with probability $\poly(1/\log N)$ in time $\poly(\log N)$.
\end{definition}

A smaller $f$ means more bad registers, hence a stronger algorithm; for $f=0$ every register could be bad, which is why $f>0$ is required, and we write $f\to0^+$ for the limit of increasing tolerance. The noiseless problem, written $\DCP_N^\infty$ ($f=\infty$), is the plain DCP. The main theorem of \cite{Reg04} states that a solution to $\DCP^f$ yields a quantum algorithm for $\uSVP_{\gamma}$ with $\gamma=\Theta(n^{1/2+2f})$; \cref{sec:regev02} explains how. The same paper shows that an algorithm solving a $1/\poly(\log N)$ fraction of legal subset-sum instances with modulus $N$ yields a solution to $\DCP^f$ with $f=1$ \cite[Thm. 1.3]{Reg04}, so that average-case subset sum would give a quantum worst-case to average-case reduction for $\Theta(n^{2.5})$-unique-SVP \cite[Cor. 1.5]{Reg04}. This connection between DCP and subset sum is the template that the 2026 claim tried to complete (\cref{sec:simon}).

\paragraph{Extrapolated DCP.} Brakerski, Kirshanova, Stehlé and Wen generalised the coset states to superpositions over more than two multiples of the secret.

\begin{definition}[Extrapolated DCP {\cite[Def. 4]{BKSW18}}]\label{def:edcp}
For integers $n,N$ and an amplitude function $\rho:\Z\to\R_{\geq0}$, $\EDCP^\ell_{n,N,\rho}$ asks, given $\ell$ states $\sum_{j\in\mathrm{supp}(\rho)}\rho(j)\ket{j}\ket{(x_k+j\cdot s)\bmod N}$ (normalisation omitted), $k=1,\dots,\ell$, with arbitrary $x_k\in\Z_N^n$ and a fixed secret $s\in\Z_N^n$, to find $s$. The two main instantiations are \emph{uniform} EDCP, with $\rho$ uniform on $\{0,\dots,M-1\}$, and \emph{Gaussian} EDCP, $\mathrm{G}\text{-}\EDCP_{n,N,r}$, with $\rho(j)\propto\exp(-\pi j^2/r^2)$. For $n=1$ and $\rho$ the indicator of $\{0,1\}$ this is $\DCP_N$ without bad registers.
\end{definition}

They proved that LWE and EDCP are equivalent under quantum polynomial-time reductions, the extrapolation growing as the LWE noise rate decreases \cite{BKSW18}. Quantitatively, from $m=\Omega(\kappa)$ samples of $\LWE_{n,q,\alpha}$ they produce $\ell$ states of $\mathrm{G}\text{-}\EDCP_{n,q,r}$ for every $r<1/\big(c\,\sqrt m\,\kappa\,\ell\,\alpha\,q^{n/m}\big)$, with $c=6\sqrt{2\pi e}$ and $\kappa$ a statistical parameter, using a reduction based on intersections of balls \cite[Thm. 3]{BKSW18}; a simpler reduction based on separating cubes, in the spirit of Regev's 2002 reduction, gives $r<1/(32\,m\kappa\alpha\ell q^{n/m})$ \cite[Thm. 2]{BKSW18}. The proof of Theorem 3 needs a lower bound on the minimum of $\Lambda_q(A)$ and a radius condition \cite[Lem. 6, App. B.1]{BKSW18} that together restrict $m$ to a range excluding both $m$ close to $n$ and $m=\Theta(n\log q)$ (\cref{sec:reach-thm}). The factor $\ell$ appears because the $\ell$ states are produced by repetitions of one procedure on the same samples and must all be valid; \cref{sec:proofs} shows that it cannot be treated as a constant. Conversely, for widths $r=\Omega(\sqrt\kappa)$, $\ell=\Omega(\kappa)$ states of $\mathrm{G}\text{-}\EDCP_{n,N,r}$ yield $\ell$ samples of $\LWE_{n,N,\alpha}$ with $\alpha=1/r$ \cite[Thm. 4]{BKSW18}. The authors observe that, for polynomial $M$, a DCP algorithm yields an algorithm for uniform EDCP, giving consequences for LWE within the corresponding parameter promises; they also note a connection with the generalised hidden shift problem of Childs and van Dam, discussed below. Wen and Zheng extended the equivalence to Module-LWE over $\Z[x]/(x^d+1)$, $d$ a power of two, with constant module rank, polynomial modulus and $O(\log q)$ module samples, via a structured problem, IP-M-EDCP, that is equivalent to Module-LWE and reduces to plain EDCP \cite{WZ26}; \cref{sec:wz-audit} recovers the exact state map and checks its parameter certificate; the first ordinary EDCP output has modulus $q^d+1$, and no deployed attack follows from the audited bounds. Bai, Jangir, Kirshanova, Ngo and Youmans gave a quasi-polynomial algorithm for EDCP over power-of-two moduli; it requires quasi-polynomially many EDCP samples, while the reductions from LWE provide only polynomially many, and it does not apply to the moduli of deployed schemes \cite{BJKNY25}. Boucher, Fouque and Shen introduced the Cyclotomic Coset Problem, a generalisation of DCP that keeps a hidden-subgroup structure, and solved it with a quantum sieve in time and sample complexity $2^{O_p(\log n\log q)}$ using polynomial quantum space; this is quasi-polynomial for polynomial moduli, but the available reductions again generate only limited approximate states, so the result does not give a quasi-polynomial algorithm for standard LWE \cite{BFS26}.

\paragraph{Related problems.} Chen, Liu and Zhandry gave polynomial-time quantum algorithms for variants of average-case lattice problems via a \emph{filtering} technique: for SIS in the infinity norm with very wide matrices, polynomial modulus and bound $q/2-O(1)$, and for EDCP with suitable amplitude functions, both through LWE with quantum auxiliary states \cite{CLZ22}. Kothari, O'Donnell and Wu gave classical, deterministic algorithms for all the SIS$^\infty$ variants treated there, eliminating the exponential quantum speedup \cite{KOW26}. Cartor, Manganiello and Youmans generalised the $\BDD$-to-EDCP direction to finite abelian groups with translation-invariant metrics, covering Hamming-, rank- and Lee-metric decoding \cite{CMY26}; we return to the consequences for codes in \cref{sec:codes}.

\paragraph{Hidden shifts.}
\begin{definition}[Hidden shift]
Given injective $f_0,f_1:G\to S$ on an abelian group $G$ with $f_1(x)=f_0(x+s)$ for all $x$, find $s$.
\end{definition}

The abelian hidden shift problem is an instance of the HSP in the generalised dihedral group $G\rtimes\Z_2$. For $G=\Z_N$, coset sampling reduces it to DCP: superposing over $x$ and a control bit and measuring the function register yields the states of \cref{def:dcp}, with no bad registers (\cref{lem:hs-dcp}). The converse is not known: an algorithm with oracle access to $f_0,f_1$ may use them in ways that a supply of coset states does not allow. Kuperberg's algorithm applies to every finite abelian group \cite{Kup05}. Childs and van Dam studied a \emph{generalised} hidden shift problem, in which an injective function $f:\{0,\dots,M-1\}\times\Z_N\to S$ satisfies $f(b+1,x+s)=f(b,x)$; for $M=2$ it is the dihedral HSP, for $M=N$ it is an abelian problem solvable by Shor's methods, and for $M=N^{\varepsilon}$ with fixed $\varepsilon>0$ they gave a polynomial-time quantum algorithm that applies the pretty good measurement to $\max\{3,\log(1/\varepsilon)\}$ registers at once and reduces the resulting state-distinguishing task to integer programming in constant dimension \cite{CvD07}. The problem is close in spirit to EDCP with a large extrapolation, which is the connection noted in \cite{BKSW18}. Childs and van Dam give a broad account of algebraic quantum algorithms \cite{CvD10}; \cref{tab:prior-surveys} explains how our organising question differs.

\paragraph{A map of the relations.} \Cref{fig:dcp-map} summarises how the problems of this section are related. The left column contains the lattice problems, the middle column the dihedral coset problems that the reductions produce, and the right column the hidden-subgroup and hidden-shift problems from which the dihedral coset problem comes. Read from left to right, the figure shows why an efficient algorithm for a dihedral coset problem would reach lattices; read upwards in the right column, it displays the route for cyclic isogeny actions satisfying the full evaluation and public-input conditions of \cref{sec:isogenies}. The tags anticipate the three hypotheses of \cref{sec:hyp-def}, each of which assumes that the tagged problem is easy.

\begin{figure}[htbp]
\centering
\begin{tikzpicture}[
  font=\footnotesize,
  every node/.style={inner sep=2pt},
  box/.style={draw, rounded corners=2pt, align=center, text width=34mm, minimum height=10.5mm, line width=0.6pt},
  lat/.style={box, fill=wBlue!9, draw=wBlue!85!black},
  dcp/.style={box, fill=wOrange!13, draw=wOrange!85!black},
  hsp/.style={box, fill=wGreen!9, draw=wGreen!75!black},
  app/.style={box, fill=black!4, draw=black!60, dashed},
  red/.style={-{Stealth[length=2.3mm]}, line width=0.7pt},
  eqv/.style={{Stealth[length=2.3mm]}-{Stealth[length=2.3mm]}, line width=0.7pt},
  lab/.style={font=\scriptsize, fill=white, align=center, inner sep=1pt}
]
\newcommand{\hyptag}[1]{\\[1pt]{\scriptsize\color{wVermillion!85!black}#1}}
% column headings
\node[font=\footnotesize\bfseries, align=center, text width=34mm, anchor=south] at (0,4.95)    {Lattice problems};
\node[font=\footnotesize\bfseries, align=center, text width=34mm, anchor=south] at (5.4,4.95)  {Dihedral coset\\problems};
\node[font=\footnotesize\bfseries, align=center, text width=34mm, anchor=south] at (10.8,4.95) {Hidden subgroups\\and shifts};
% left column
\node[lat] (gap)  at (0,4.2)   {$\GapSVP_\gamma$, $\SIVP_\gamma$\\(worst case)};
\node[lat] (usvp) at (0,2.4)   {$\uSVP_\gamma$\\(worst case)};
\node[lat] (lwe)  at (0,0.6)   {$\LWE_{n,q,\alpha}$};
\node[lat] (mlwe) at (0,-1.2)  {Module-LWE};
% middle column
\node[dcp] (ss)   at (5.4,4.2)  {Subset sum\\(average case)};
\node[dcp] (dcpf) at (5.4,2.4)  {$\DCP_N^f$ (with bad registers)\hyptag{$\HReg(f)$}};
\node[dcp] (edcp) at (5.4,0.6)  {$\EDCP_{n,N,\rho}$ (over $\Z_N^n$)\hyptag{$\HEDCP$}};
\node[dcp] (ipm)  at (5.4,-1.2) {IP-M-EDCP (structured)};
% right column
\node[hsp] (hsp)   at (10.8,4.2)  {HSP on finite $G$\\(easy for abelian $G$)};
\node[hsp] (dhsp)  at (10.8,2.4)  {Dihedral HSP\\on $D_N=\Z_N\rtimes\Z_2$};
\node[hsp] (dcp)   at (10.8,0.6)  {$\DCP_N$, perfect states\\(noiseless coset states)\hyptag{$\HHS$}};
\node[app] (csidh) at (10.8,-1.45) {Cyclic isogeny actions\ with efficient evaluation};
% reductions
\draw[red] (usvp) -- node[lab, above=1pt] {\cite{Reg04}\\$\Theta(n^{1/2+2f})$} (dcpf);
\draw[red] (dcpf) -- node[lab] {$f=1$ \cite{Reg04}} (ss);
\draw[red] (gap.west) to[out=205, in=155, looseness=0.55] node[lab] {\cite{Reg09}} (lwe.west);
\draw[eqv] (lwe) -- node[lab, above=1pt] {\cite{BKSW18}} (edcp);
\draw[eqv] (mlwe) -- node[lab, above=1pt] {\cite{WZ26}} (ipm);
\draw[red] (ipm) -- node[lab] {\cite{WZ26}} (edcp);
\draw[red] (dcp.north west) -- node[lab, sloped, above=1pt] {zero noise} (dcpf.south east);
\draw[red] (dhsp) -- node[lab] {coset sampling,\\\cref{lem:hs-dcp}} (dcp);
\draw[red] (dhsp) -- node[lab] {special case} (hsp);
\draw[red] (csidh) -- node[lab] {cyclic class group\\\cite{CJS14,BKV19}} (dcp);
\end{tikzpicture}
\caption{Relations between lattice problems, dihedral coset problems, and hidden-subgroup and hidden-shift problems. An arrow $A\to B$ means that $A$ reduces to $B$, so that an efficient quantum algorithm for $B$ yields one for $A$; a double arrow denotes an equivalence. Labels give the source of each step, with the approximation factor $\gamma$ of the first reduction and the failure parameter $f$ of \cref{def:dcp}. The zero-noise arrow uses a solver for a prescribed noisy problem on noiseless inputs. In the converse algorithmic direction, \cref{lem:hs-rare} permits some solver-dependent failure exponent; it does not identify any fixed-$f$ problems. The structured edges record the parameter-qualified published connection; the first EDCP modulus is $q^d+1$, and the printed forward bound supplies no ML-KEM certificate (\cref{sec:wz-audit}). Tags in red identify the corresponding solver hypotheses of \cref{sec:hyp-def}. The $\HEDCP$ tag covers only its stated polynomial-modulus Gaussian domain, not every EDCP state in the diagram; in particular it does not include the first large-modulus structured output. The cited algorithms for noiseless cyclic DCP include Kuperberg's $2^{O(\sqrt{\log N})}$ sieves and Regev's polynomial-space variant with an additional $\sqrt{\log\log N}$ in the exponent \cite{Kup05,Reg04b,Kup13}. General abelian actions require the corresponding group and evaluation conditions \cite{CJS14}.}
\label{fig:dcp-map}
\end{figure}
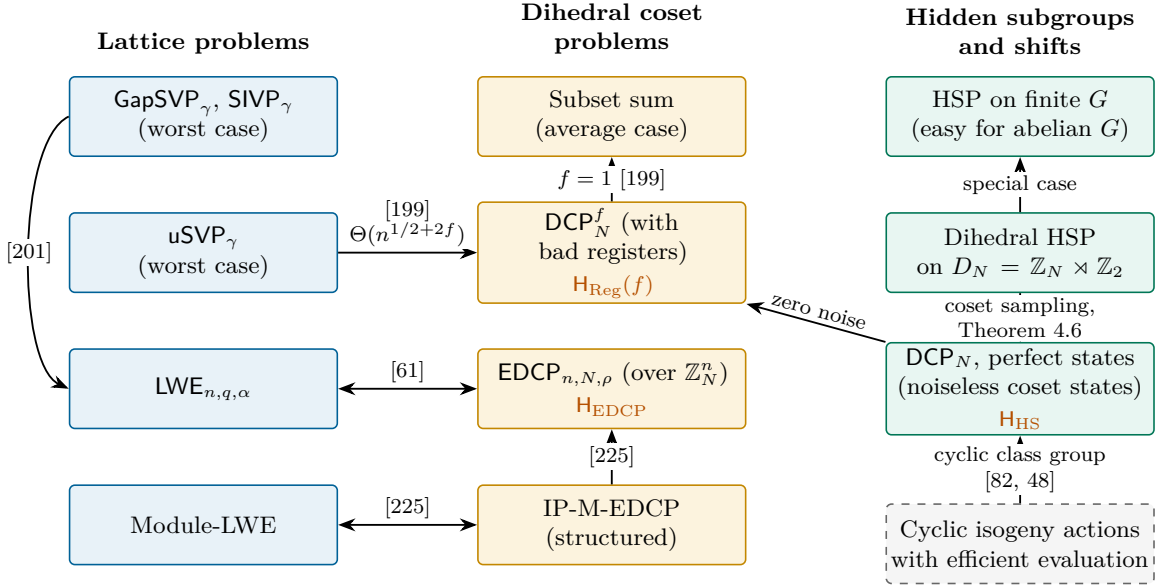

% ======================================================================
\section{The claimed algorithms: strategy and status}\label{sec:simon}
% ======================================================================

This section describes the algorithm of \cite{Sim26} and the reasons it does not solve DCP, and, in \cref{sec:chen-zhang}, Chen's earlier claim for LWE and a proposed repair of it. We rely on the preprint itself, in its revision dated October 7, 2026, which we have read in full \cite{Sim26}; on the refutation of Gupte, Ragavan and Zhandry and its Lean development \cite{GRZ26,GRZLean}; on the rigorous restatement of the lemmas by Guo and Yang \cite{GY26}; and on the independent Lean audit \cite{RoyessAudit}. The three analyses predate the October revision: the Lean development of \cite{GRZLean} models the draft dated August 11, Guo and Yang's paper was received on August 17, and the audit, as consulted, does not discuss the October note. Contemporaneous public commentary is cited through \cite{PQC26news,Quentir26} and is used for dates and for the public objections that preceded the formal analyses. Where the sources describe the same step at different levels of detail, we say which source we follow. We do not reconstruct details that none of the sources states; the paragraph of informal intuition in \cref{sec:simon-failure}, the assessment in \cref{sec:simon-response} and our other readings of the draft are marked as ours. Statements about the two Lean developments refer to their public repositories as consulted on October 2, 2026, whose descriptions we re-read on October 18; we were not able to record commit identifiers, and the repositories may have changed since. A Lean development verifies the stated theorems about a formal model of the algorithm; whether the model matches every detail of a draft that was revised several times is argued in natural language by the authors, and we report it as such.

\subsection{The template: Regev's reduction from DCP to subset sum}\label{sec:simon-template}

Kuperberg's and Regev's algorithms start from the same observation \cite{Kup05,Reg04}, recalled in \cref{sec:prelim-hsp}. Applying the quantum Fourier transform over $\Z_N$ to the second register of a coset state $\tfrac{1}{\sqrt2}(\ket{0,x}+\ket{1,x+d})$ and measuring it yields a uniformly random classical \emph{Fourier label} $y\in\Z_N$ together with the qubit $\ket{\psi_y}=\tfrac{1}{\sqrt2}(\ket{0}+e^{2\pi i\,yd/N}\ket{1})$. The labels are known; the phases carry the information about $d$. A qubit whose label is $N/2$ reveals the parity of $d$ directly, and the difficulty is to produce such a label from random ones.

The two algorithms differ in how they combine qubits. Kuperberg combines them two at a time: a parity measurement on $\ket{\psi_{y_1}}\otimes\ket{\psi_{y_2}}$ produces a qubit with label $y_1\pm y_2$, and pairing qubits whose labels agree on many low-order bits zeroes those bits in stages. The price is that each stage discards most of the qubits, so that $2^{O(\sqrt{\log N})}$ samples are needed \cite{Kup05}. Regev instead combines many qubits at once. Taking the tensor product of $Q$ qubits gives the state
\[
  2^{-Q/2}\sum_{b\in\{0,1\}^Q}e^{2\pi i\,d\,(\sum_i b_iy_i)/N}\,\ket{b},
\]
in which the phase of each selection vector $b$ depends on $d$ only through the subset sum $\sum_ib_iy_i$. If one could coherently map every $b$ with a given subset sum to a single register value, the selection bits would disappear and the phases would encode $d$ in a form that a Fourier measurement can read. Regev proved that an algorithm solving a $1/\poly(\log N)$ fraction of legal subset-sum instances with modulus $N$ suffices for this, and obtains a solution to DCP with failure parameter $f=1$ \cite[Thm. 1.3]{Reg04}. Schematically, the template prepares a superposition over selection vectors, partially measures the subset sum, and then removes the selection bits coherently before a bit of $d$ can be read; Gupte, Ragavan and Zhandry call this last part the \emph{uncomputation stage} \cite{GRZ26}. In Regev's reduction, the subset-sum solver is what makes the uncomputation possible. Average-case modular subset sum with these parameters is not known to be easy, the subset-sum instances in question are themselves closely related to lattice problems \cite{LO85,LPS10}, so the template had not, before 2026, produced an efficient DCP algorithm.

\subsection{Simon's strategy}\label{sec:simon-strategy}

According to its abstract, the algorithm ``is based on Regev's polynomial-time reduction of the Dihedral Subgroup Problem (DSP) to the modular subset sum problem'', but ``uses a different technique to erase sample bits without use of a subset sum oracle'' \cite{Sim26}. The intended contribution is thus a replacement, inside Regev's template, of the subset-sum solver in the uncomputation stage. We follow the step-by-step description of Guo and Yang \cite{GY26}, which covers the same operations as the seven-step specification formalised in \cite{GRZLean} and the detailed description in Section 2.2 of the October revision \cite{Sim26}, divided slightly differently (see below). The algorithm works with $N=2^n$ and the parameters
\begin{gather*}
  Q=kn^{c+1}\ \text{samples}\ (k>c),\qquad m=c\log n\ \text{bits per group}\ (c\geq12),\\
  G=Q/m\ \text{groups},\qquad a=n/\log n,\qquad L=\log n.
\end{gather*}
With $c\geq12$ the algorithm consumes at least $kn^{13}$ coset states per output bit: polynomially many, as the reductions require, but a large polynomial. It proceeds as follows.
\begin{enumerate}[label=\textbf{Step \arabic*.}]
  \item \emph{Fourier sampling.} Apply the Fourier transform to the $Q$ samples, producing quantum Fourier labels \cite{GY26,GRZLean}.
  \item \emph{Subset-sum measurement.} Measure a subset-sum quantity, obtaining a classical outcome $z'$ and retaining a quantum register $(b,h)$ \cite{GRZLean}; in \cite{GY26} this step produces the record entries $Y$ and $z'$.
  \item \emph{Grouping.} Partition $\{1,\dots,Q\}$ into $G$ groups of $m$ indices, compute the subset sum of each group and keep its top $L$ bits \cite{GY26}; the Lean specification calls the results the ``high group words'' \cite{GRZLean}. The Lean specification models this step through coordinatewise \emph{digests} of the Fourier labels that retain only their $n-\ell$ most significant bits, with $\ell=\lfloor 2n/3\rfloor$ \cite{GRZLean}.
  \item \emph{Hadamard measurement.} Apply a Hadamard transform to all $Q$ selection bits and measure, obtaining a classical string $\varphi'$ divided into blocks $\varphi'_j$, one per group \cite{GY26}.
  \item \emph{Adaptive partition.} A group whose block $\varphi'_j$ is all-zero contributes no sign. Let $A$ consist of the first $a$ such groups and $B$ of the rest \cite{GY26}; the remaining words are kept in the order of $A$ \cite{GRZLean}.
  \item \emph{Erasure.} Perform the ``certified reversible erasure'' of the selection information, including carry-flag detection, and the associated measurements \cite{GRZLean}. This is the bit-erasure technique of the abstract.
  \item \emph{Output.} Measure the resulting register $h^*$ and output one bit of the hidden shift \cite{GY26,GRZLean}.
\end{enumerate}
Steps 3--6 play the role of the uncomputation stage of the template: instead of a subset-sum solver, they use the high-order part of the group sums and the Hadamard outcomes to decide which selection information to erase. In the draft's own numbering, which \cite{GRZLean} follows, step 4 is our Step 4 together with the choice of $A$, and step 5 measures the group words of the groups in $B$ and keeps those of $A$ in order; its step 6 adds the group words of $A$ modulo $n$ into a sum $s^*$, uses it to erase the last of them, and proceeds only if a carry flag is zero and the Hadamard-transformed low-order bits of $s^*$ are all zero, the latter an event of expected probability $2/n$; and its step 7 computes and measures $h'=h\oplus h^*$, which erases $h$, before $h^*$ is measured in the Hadamard basis to give the last bit $d_n$ of the hidden shift \cite{Sim26}. The procedure is repeated polynomially many times for each bit, and the remaining bits are obtained by the recursion of Regev's template.

\paragraph{The four lemmas.} The correctness argument of \cite{Sim26} rests on four lemmas, of which only Lemma 2 is proved in full in the draft; Lemmas 1, 3 and 4 are proved by sketches. In the October revision, and in the readings of \cite{GY26,RoyessAudit}:
\begin{itemize}
  \item \emph{Lemma 1} asserts that, if $k>c$, then with constant probability the measured Hadamard string has at least $n/\log n$ all-zero groups, so that Step 5 succeeds \cite{Sim26}. Guo and Yang state it as follows: if $Z$ is the number of groups with $\varphi'_j=0$, then $\Pr[Z\geq a]\to1$ under the hypothesis $k>c$ \cite{GY26}.
  \item \emph{Lemma 2} asserts that the contribution of the groups in $B$ is the same for two values of $z^*$ that differ only in their most significant bit $h^*$. It is an exact statement, with a complete proof in \cite{Sim26}, as \cite{GY26} confirm.
  \item \emph{Lemma 3} asserts that after the measurements no value $z^*$ contributes an amplitude $\alpha_{z^*}$ of exponentially large magnitude, except with probability $O(2^{-n})$ over the measured record, and a corollary bounds by $\sqrt n$, except with probability $O(1/n)$, the amplitude collected by each value of the top $\log n$ bits of $z^*$ \cite{Sim26}. The October sketch computes the expectation of $\alpha_{z^*}^2$ under the actual distribution of the Hadamard string by rewriting it as an expectation under the uniform distribution, which is how it keeps the two distributions apart. Guo and Yang prove, assuming the partition hypothesis discussed below, that $\mathbb E\big[\sum_{z^*}\alpha_{z^*}^2\big]\leq1$, hence $\Pr\big[\max_{z^*}|\alpha_{z^*}|>\theta\big]\leq\theta^{-2}$ for every threshold $\theta>0$ \cite{GY26}.
  \item \emph{Lemma 4} asserts, for $c\geq12$, that the amplitudes collected by two values of the top $\log n$ bits of $z^*$ that differ only in $h^*$ differ by an expected multiplicative factor $1\pm O(n^{-((c-1)/2-1)})$, with probability at least $1-1/n$, so that the bias survives amplification \cite{Sim26}. Guo and Yang prove instead an additive bound on the difference between the two amplitudes that determine the output bit, of order $n^{-((c-1)/2-1)}$ times a factor depending on the record, plus a residual term, and observe that the closing argument needs no more than this \cite{GY26}.
\end{itemize}
The closing argument of the draft combines Lemma 4 with the corollary of Lemma 3: the differences between paired amplitudes add up to at most $1/n$ except with probability $1/n$, so that measuring $h^*$ in the Hadamard basis returns $d_n$ with probability ``significantly greater than 1/2'' \cite{Sim26}.

\paragraph{Claimed consequences.} The theorem of the draft asserts that there is a constant $c'$ such that, from samples that are correct with probability at least $1-1/(c'\log n)$, a polynomial-time quantum algorithm computes the last bit of the hidden shift \cite{Sim26}. Its abstract states that, combined with the reductions of \cite{Reg04,BKSW18}, it yields polynomial-time quantum algorithms for SVP with approximation factor $\sqrt n\,\mathrm{polylog}(n)$ and for LWE with $\alpha=\sqrt n\,\mathrm{polylog}(n)$, in the draft's notation; its final corollary states the same with $O(\cdot)$, and for these parameters the draft cites personal communications from S. Ragavan and D. Micciancio \cite{Sim26}. In the language of \cref{sec:hypothesis}, it is a claim of the type $\HReg(f)$, which also implies $\HHS$; \cref{rem:simon-claim} relates the claimed tolerance to the failure exponent. The draft's route to LWE through \cite{BKSW18} needs more than a cyclic solver: \cite{BKSW18} connects LWE to extrapolated coset states over $\Z_N^n$ with polynomial $N$, not to cyclic coset states over $\Z_{2^n}$; no general implication from a cyclic solver to an extrapolated one is known (\cref{sec:hyp-def}). Under $\HReg(f)$, LWE is reached instead through the unique-SVP embedding of \cref{thm:reach}\,\ref{it:lwe}, with its sample and noise conditions. The draft does not mention any deployed parameter set \cite{Sim26}.

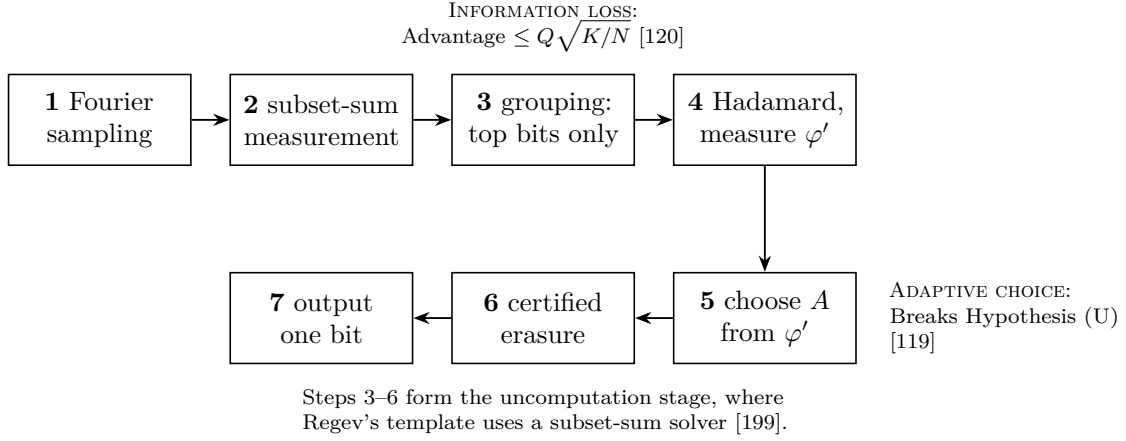
\begin{figure}[tbp]
\centering
\begin{tikzpicture}[
  >=Stealth, font=\small,
  st/.style={draw=wBlack, line width=0.6pt, align=center, inner sep=3pt, minimum height=12mm, text width=22mm, fill=white},
  bad/.style={draw=wBlack, line width=0.6pt, align=center, inner sep=3pt, minimum height=12mm, text width=22mm, fill=white},
  warn/.style={draw=wBlack, line width=0.6pt, align=center, inner sep=3pt, minimum height=12mm, text width=22mm, fill=white},
  ar/.style={->, line width=0.7pt}
]
\node[st] (s1) {\textbf{1} Fourier\\ sampling};
\node[st, right=5mm of s1] (s2) {\textbf{2} subset-sum\\ measurement};
\node[bad, right=5mm of s2] (s3) {\textbf{3} grouping:\\ top bits only};
\node[st, right=5mm of s3] (s4) {\textbf{4} Hadamard,\\ measure $\varphi'$};
\node[warn, below=14mm of s4] (s5) {\textbf{5} choose $A$\\ from $\varphi'$};
\node[st, left=5mm of s5] (s6) {\textbf{6} certified\\ erasure};
\node[st, left=5mm of s6] (s7) {\textbf{7} output\\ one bit};
\draw[ar] (s1) -- (s2); \draw[ar] (s2) -- (s3); \draw[ar] (s3) -- (s4); \draw[ar] (s4) -- (s5);
\draw[ar] (s5) -- (s6); \draw[ar] (s6) -- (s7);
\node[above=1.5mm of s3, align=center, font=\scriptsize] {\textsc{Information loss}:\\ Advantage $\leq Q\sqrt{K/N}$ \cite{GRZ26}};
\node[right=3mm of s5, align=left, font=\scriptsize, text width=30mm] {\textsc{Adaptive choice}:\\ Breaks Hypothesis (U) \cite{GY26}};
\node[below=2mm of s6, align=center, font=\scriptsize, text width=80mm] {Steps 3--6 form the uncomputation stage, where Regev's template uses a subset-sum solver \cite{Reg04}.};
\end{tikzpicture}
\caption{The seven steps of the algorithm of \cite{Sim26}, as described in \cite{GY26,GRZLean}, and the two places where the analyses locate the failure. The draft, and \cite{GRZLean} after it, numbers our Step 4 and the choice of $A$ as step 4, and the measurement of the words of $B$ as step 5.}
\label{fig:simon}
\end{figure}

\subsection{Why the strategy does not work}\label{sec:simon-failure}

The three cited analyses challenge the claim by different routes. The first proves a limitation for its formal observation model; the second and third identify gaps in the correctness argument. Relating the formal model to the evolving draft remains a separate step, as noted above. The author's response to the first analysis is assessed in \cref{sec:simon-response}.

\paragraph{An information-theoretic barrier.} Gupte, Ragavan and Zhandry prove a limitation that cannot be removed by repairing the analysis while retaining their formal observation model \cite{GRZ26}. Their argument concerns the output bit for two secrets $d_0,d_1$ whose least significant bits differ: an algorithm that recovers $d$ bit by bit must, in particular, distinguish such a pair. The key observation is structural. Because Step 3 passes the Fourier labels to the rest of the computation only through coordinatewise digests $H_i:\Z_N\to[K]$, functions applied separately to each label $y_i$, with at most $K$ possible values, the final observation depends on the labels only through the digests $H_i(y_i)$. The main theorem of the Lean development, stated for a formal model of the algorithm's observation (the seven-step specification of \cite{GRZLean}), asserts that for every $n$, every sample count $Q$, every digest size $1\leq K\leq N=2^n$, every coordinatewise digest family $H=(H_1,\dots,H_Q)$, every such pair $d_0,d_1$ and every Boolean-output quantum channel applied to the algorithm's observation, the digests, the Step 2 outcome $z'$ and the retained register $(b,h)$, the output distributions satisfy
\[
  \Delta\big(P_{d_0},P_{d_1}\big)\;\leq\;Q\sqrt{K/N},
\]
where $\Delta$ is the statistical (trace) distance \cite{GRZLean}; the Lean development proves it by a hybrid argument, and the bound grows linearly with the number of samples. For the algorithm of \cite{Sim26}, the digests keep $n-\ell$ bits with $\ell=\lfloor2n/3\rfloor$, so $K=2^{n-\ell}$ and the bound is $Q\cdot2^{-\ell/2}\approx Q\cdot2^{-n/3}$, negligible for every polynomial $Q$. Within that observation model, the bound allows arbitrary post-processing: retries, several invocations, adaptive stopping, majority votes and randomised decisions. Their refutation corollary turns it into the dichotomy quoted in \cref{sec:intro}: for every sufficiently large $n$, every pair with opposite least significant bits and every prefix-stable controller, either one of the two executions has expected sample count at least $2^{\lfloor n/9\rfloor}$, or both terminate almost surely and their output laws are within $\poly(n)\cdot2^{-\lfloor n/9\rfloor}$ \cite{GRZLean}. Thus the formalised algorithm cannot have both the claimed efficiency and correctness. In the words of the abstract, the method operates ``using only the most-significant third of the classical Fourier labels'', whereas any successful DCP algorithm in Regev's template would require ``extensive use of the classical Fourier labels in the uncomputation stage'' \cite{GRZ26}. The barrier is not specific to Simon's choices of parameters: it applies to algorithms satisfying the same coordinatewise-digest observation model \cite{GRZ26,Quentir26}.

\paragraph{Informal intuition.} The following example illustrates one parity-extraction method; it is not a characterisation of all quantum measurements. Write $d=2d'+d_{\mathrm{lsb}}$. A combined phase qubit with label $Y$ has phase $2\pi Y(2d'+d_{\mathrm{lsb}})/N$. Taking $Y=N/2$ eliminates the contribution of $d'$ while preserving the parity-dependent sign. Producing this label requires control of its low-order bits. This explains why those bits matter to this method, but does not itself prove that other methods fail. The formal limitation is the bound of \cite{GRZLean}, with its coordinatewise-digest and observation-model premises.

\paragraph{The partition hypothesis.} Guo and Yang give rigorous statements and complete proofs of the lemmas of \cite{Sim26}, in some cases stronger than the originals: Lemma 1 holds with probability tending to one, provided $k>c$; Lemma 3 holds at every threshold without the well-behavedness condition of the original; and Lemma 4 holds in an additive form, which is all the closing argument uses \cite{GY26}. According to Guo and Yang, however, the lemmas require the following hypothesis:
\begin{quote}
\emph{Hypothesis (U).} The $A/B$ partition is fixed and does not vary with $\varphi'$; with the rest of the record held fixed, $\varphi'_B$ ranges over all of $\{0,1\}^{Q_B}$ \cite{GY26}.
\end{quote}
In Step 5, $A$ is defined as the first $a$ groups whose block $\varphi'_j$ is all-zero, which makes $A$ a function of the measured string $\varphi'$: which groups end up in $A$ depends on which blocks happened to be measured as zero. The algorithm therefore does not supply Hypothesis (U), and, as Guo and Yang conclude, establishing the four lemmas ``does not by itself establish the correctness of the algorithm'' \cite{GY26}.

\paragraph{Lemma-level audit.} The independent Lean audit examines the published proofs one by one, and distinguishes carefully between what it formalises and what it argues in natural language \cite{RoyessAudit}.
\begin{itemize}
  \item \emph{Lemma 1.} The published proof uses a simultaneous swap of selection bits and low parts of the samples that is meant to preserve the fibres of the measured subset sum. The audit shows that the swap changes the contribution by $B(\varphi_j-\varphi_i)(H_i-H_j)$, which need not vanish modulo the measured modulus, and gives a counterexample with $n=8$ in which the contribution changes from $1$ to $5$ modulo $128$. The constant-probability claim nevertheless has a corrected proof, formalised in Lean for $n\geq1024$, $k=24$, $c=12$ (in the audit's notation), using collision bounds and Chebyshev's inequality.
  \item \emph{Lemma 3.} After $A$ is chosen adaptively, states with identical $B$-components become perfectly correlated rather than pairwise independent, contrary to the variance argument. An analysis in two copies, stated in natural language and not formalised, indicates that the first-zero decoder makes almost every Fourier mode diagonal. For one natural repair, the decoder without a guard, the correct and wrong output masses are $1/2+o(1)$ each; for the draft's own post-selection in its step 6, subject to one guard estimate still to be completed, the calculation predicts masses $1/(2L)+o(1/n)$ each and acceptance $1/L+o(1/n)$, with $L=n/2$. In both cases the output bit carries no bias. The audit also notes that conditioning on acceptance introduces a factor equal to the inverse of the acceptance probability, which can spoil the claimed exponential bounds. It concludes that Lemma 3 is ``not repaired in the sense needed by the paper's polynomial-time algorithm''.
  \item \emph{Lemma 4.} A printed exponent contains an extra $+n$, which is repairable, but the multiplicative amplitude bound remains unproved: independent coordinates can become correlated after conditioning, and nearly equal counts of signed terms give no multiplicative control when cancellation is possible. The audit records Lemma 4 as unproved, not refuted; its clean-oracle calculation, also in natural language and pending the same guard estimate, contradicts the closeness of branch amplitudes that the proposed decoder needs.
  \item \emph{Implementation.} Every decoder the audit analyses meets a coherent scale of $2^{n/2}$, and none of the realisations considered has polynomial cost.
\end{itemize}

\paragraph{The revision record.} The preprint was received by the ePrint archive on August 3, 2026, made public on August 6, and revised six times, last on October 8; the PDF of that revision is dated October 7, 2026, and is still labelled a preliminary draft \cite{Sim26,PQC26news}. According to the first note on the ePrint record, the proof of Lemma 3 in the first draft had conflated pairwise independence under the uniform and the actual distributions of the Hadamard string, and the revised proof keeps them apart; a separate note, added after \cite{GRZ26} appeared, said that the refutation was being evaluated. According to the ePrint record, the October revision adds a note addressing the claims of \cite{GRZ26}, discussed in \cref{sec:simon-response}, and clarifies the proof of Lemma 3 in response to comments by Sam Jaques \cite{Sim26}. Both revisions concern the proof of Lemma 3, whose October sketch still rests on pairwise independence over the Hadamard string \cite[\S2.2]{Sim26}. The audit finds that pairwise independence still fails once $A$ has been chosen adaptively \cite{RoyessAudit}, Guo and Yang locate the same dependence in Hypothesis (U) \cite{GY26}, and the barrier of \cite{GRZ26} does not depend on any lemma, so revising those lemmas alone cannot overcome it within the formal observation model.

\begin{table}[htbp]
\centering\small
\renewcommand{\arraystretch}{1.2}
\begin{tabularx}{\textwidth}{@{}L{3.0cm}Y L{3.4cm}@{}}
\toprule
\textbf{Ingredient} & \textbf{Finding} & \textbf{Source}\\
\midrule
Algorithm as specified & Fails: advantage $\leq Q\sqrt{K/N}\approx Q\,2^{-n/3}$ for post-processing within the formal observation model; exponential samples or negligible advantage. & \cite{GRZ26,GRZLean} (Lean)\\
Lemma 1 & Repaired: published proof invalid (the October sketch keeps the swap argument \cite{Sim26}); the constant-probability claim has a corrected proof, and holds with probability tending to one. & \cite{RoyessAudit} (Lean), \cite{GY26}\\
Lemma 2 & Holds: exact, with complete proof. & \cite{GY26}\\
Lemma 3 & Conditional: holds under Hypothesis (U); not repaired for the algorithm as specified; the October sketch still varies the Hadamard string freely with the rest of the record fixed. & \cite{GY26,RoyessAudit}; \cref{sec:simon-response}\\
Lemma 4 & Conditional: additive form holds under (U); multiplicative form unproved; the author concedes that it fails when the selection bits of $A$ are traced out. & \cite{GY26,RoyessAudit}, \cite[\S3]{Sim26}\\
Partition rule (Step 5) & Fails: $A$ depends on $\varphi'$, so Hypothesis (U) is not supplied. & \cite{GY26}\\
Polynomial implementation & Unknown: analysed decoders have coherent scale $2^{n/2}$. & \cite{RoyessAudit}\\
Response of October 7 & No escape: post-selection accepted with probability $\Omega(1/n)$ stays within the bound, which covers every channel on the digest-restricted state, restarts included; the adaptive step is asserted, not proved. & \cite[\S3]{Sim26}; our analysis (\cref{sec:simon-response})\\
\bottomrule
\end{tabularx}
\caption{Status of the ingredients of \cite{Sim26} according to the published analyses and, for the revision of October 7, 2026, our analysis.}
\label{tab:simon}
\end{table}

\subsection{The author's response of October 2026}\label{sec:simon-response}

The revision dated October 7 adds a section, ``A Note on [6]'', in which the author answers the refutation of Gupte, Ragavan and Zhandry, reference [6] of the draft \cite[\S3]{Sim26}. The note restates the result of \cite{GRZ26} as the statement that the density matrix after Step 2 is exponentially close to the one obtained if all the selection qubits $b_i$ are traced out, which would destroy the relation between the phases of the two values of $h$ that encodes $d_n$. It then makes three claims.
\begin{itemize}
  \item \emph{A concession.} Lemmas 1--3 are ``completely compatible'' with \cite{GRZ26}; Lemmas 2 and 3 concern the part of the state indexed by $B$, which could be traced out without affecting the algorithm. ``Rather, it's Lemma 4, the claim that the $\phi_A$ state portions are roughly evenly distributed among $z^*$ values, that fails once the bits in $\phi_A$ are traced out'', and this ``would certainly cause the algorithm to fail''.
  \item \emph{Post-selection.} Because $A$ is chosen so that the Hadamard outcomes on its selection bits are all zero, the values of $h$ and $h^*$ computed when those bits are traced out may ``diverge non-negligibly from the actual case, despite the proof in [6]''. According to the note, the closeness proved in \cite{GRZ26} is an average over the Fourier labels, and an all-zero Hadamard outcome changes that average.
  \item \emph{Three variations.} If all the selection bits are measured as zero after their Hadamard transform, the algorithm works, but this happens with exponentially small probability. If a set $A$ of $n/\log n$ groups is fixed in advance and the other selection bits are measured, the algorithm works in the also exponentially unlikely case that the Hadamard outcomes on $A$ are all zero, whereas tracing out the bits of $A$ predicts failure; according to the note, it works equally well if the other bits are not measured, provided the outcomes on $A$ are all zero ``and would form exactly the set A chosen by the algorithm''. For the algorithm as specified, ``its output will be identical to that of the previous variation in which $A$ is pre-selected, conditioned on the Hadamard transforms all coming out as hoped, the only difference is the much larger success probability due to post-selection''.
\end{itemize}
The note reconciles this with \cite{GRZ26} by observing that, whatever set $A$ is post-selected, ``the probability that that particular set produced all zeroes when Hadamard-transformed was a priori exponentially small, and thus nominally consistent with the result in [6]'', and summarises that post-selection lets the required deviation ``hold with high probability overall, rather than an exponentially small one''. The note concludes that tracing out all the selection bits ``will thus produce an incorrect prediction of the algorithm's output''. The note does not discuss the statement of the main theorem of \cite{GRZ26}, a bound for every channel applied to the digest-restricted state, or its Lean development, and it names no step of the draft that their formal model leaves out. It accepts the closeness of density matrices as it restates it, and maintains only that tracing out the selection bits mispredicts the output of the post-selected algorithm. The algorithm it defends has the seven steps modelled in \cite{GRZLean} (\cref{sec:simon-strategy}).

\paragraph{Assessment.} We find in the note no argument that removes the barrier, for four reasons.
\begin{enumerate}
  \item \emph{The bound already covers post-selection.} The main theorem of \cite{GRZLean} bounds the trace distance between the states of the two secrets formed by the digests of the Fourier labels, the Step 2 outcome $z'$ and the retained register $(b,h)$; its proof dephases the selection bits one at a time and shows that the fully dephased state does not depend on the secret. This is the closeness that the note restates, with one difference that matters: it concerns the digest-restricted state, not the state that keeps the full labels, from which $d_n$ can in principle be read, as a subset-sum solver does in Regev's template. Since no channel increases trace distance, the bound applies to the output law of every Boolean-output channel applied to that state, with no restriction on efficiency; Hadamard transforms, measurements, adaptive choices and post-selection are all such channels. The exact procedure is not literally of this form, because the draft's step 3 keeps the top bits of the exact group sums, which depend on the low-order label bits through carries. The refutation corollary replaces these words by words rebuilt from the digests, which differ only on a set of labels of probability at most polynomial times $2^{\ell-n}$; it applies the main theorem to the rebuilt procedure, grouping, Hadamard transforms, choice of $A$ with a restart if too few groups qualify, the two post-selections of the draft's step 6 and the final measurements, and transfers the bound to the exact procedure. In this one-attempt form a restart is an outcome in its own right, and the total-variation distance between the laws of the outcomes restart, output 0 and output 1 is at most $\poly(n)\cdot2^{-\lfloor n/3\rfloor}$ \cite{GRZLean}.
  \item \emph{Conditioning loses only the acceptance probability.} Let $\varepsilon$ bound the total-variation distance between the two outcome laws, restarts included; let $p_0,p_1$ be the acceptance probabilities for the two secrets, $q_0,q_1$ the probabilities of output 1 given acceptance, and $\Delta_1=p_0q_0-p_1q_1$, $\Delta_0=p_0(1-q_0)-p_1(1-q_1)$. Then $p_0(q_0-q_1)=(1-q_1)\Delta_1-q_1\Delta_0$ and $|\Delta_0|,|\Delta_1|\leq\varepsilon$, so $p_0|q_0-q_1|\leq\varepsilon$: given acceptance, the distinguishing advantage is at most $\varepsilon/p_0$. The bound concerns the acceptance event of the whole algorithm, the union over all candidate sets $A$; that each particular set is a priori exponentially unlikely does not enter. Read as conditional probabilities, the draft's claims put the acceptance probability of a run at $\Omega(1/n)$, constant probability for Lemma 1 and for the carry flag, and expected probability $2/n$ for the all-zero outcome of step 6, so the conditional advantage remains at most $O(n)\cdot\poly(n)\cdot2^{-\lfloor n/3\rfloor}$, still negligible. The same computation answers the remark that the closeness proved in \cite{GRZ26} is an average over the Fourier labels: the algorithm's success probability is the same average, and an all-zero Hadamard outcome is one of the accepted events it covers. Whatever the acceptance probability $p$, an advantage $\delta$ given acceptance requires $p\leq\varepsilon/\delta$, which is exponentially small for $\delta\geq1/\poly(n)$; the third point then applies. The audit applies the same elementary rule to the tail bounds of Lemma 3, where conditioning on acceptance mass $\rho$ costs a factor $1/\rho$, notes that a concrete lower bound on the acceptance probability of step 6 still needs a new hypothesis or repair, and, in a two-copy calculation stated in natural language and pending one guard estimate, predicts for that step an acceptance probability $1/L+o(1/n)$ with $L=n/2$ and equal correct and wrong output masses: acceptance of order $1/n$ and no bias \cite{RoyessAudit}.
  \item \emph{Amplification needs exponentially rare events.} Post-selection can turn a small distance into a large one only on events of exponentially small probability, as in the three variations of the note. For the algorithm itself, repeating a run until such an event occurs costs exponentially many samples in expectation, which is the first alternative of the refutation corollary \cite{GRZLean}; for the note's variations, which the corollary does not model, the same conclusion follows from the main theorem and the second point. The note relies on the opposite move when it passes from its second variation to its third: it invokes the closeness proved in \cite{GRZ26} for measuring or not measuring the selection bits of $B$, a bound of order $Q\cdot2^{-n/3}$, and then conditions on all-zero outcomes on $A$, an event it calls exponentially unlikely. By the same accounting the closeness gives no control after such conditioning; this is the reasoning that the note applies to $A$ but not to $B$.
  \item \emph{The decisive step is the adaptive choice.} The note's final claim, that choosing $A$ after the Hadamard measurement gives the same output as fixing it in advance and conditioning on the same outcome, is asserted without proof. Since $A$ is a function of the measured string, the two procedures leave the same state for each complete outcome, so the claim amounts to the correctness of the pre-selected procedure when one also conditions on $A$ being the set that the algorithm would choose. That condition constrains the outcomes on $B$; it is the conditioning excluded by Hypothesis (U) \cite{GY26} and the source of the correlations found by the audit \cite{RoyessAudit}. The note derives that correctness from its second variation, which does not condition on $B$, and from the proof of its Section 2, which, according to \cite{GY26}, needs a fixed partition. Within the formal model, the claim cannot hold together with the claimed correctness of the pre-selected variation: together they would give a conditional advantage bounded away from zero at acceptance probability $\Omega(1/n)$, hence an unconditional advantage of order $1/n$, far above the bound.
\end{enumerate}
The clarified sketch of Lemma 3 does not address the adaptive choice either. It varies the Hadamard string with the rest of the record fixed and treats each of its bits as free, ``assuming uniform $D$'', although the same string determines which groups form $A$. Lemma 1 is still proved by the swap argument for which the audit gives a counterexample, and Lemma 4 is still stated in multiplicative form with a sketch \cite{Sim26,RoyessAudit}. The revision therefore leaves the findings of \cref{tab:simon} unchanged; the table records the October sketches and adds a row for the note. It does make one point explicit: by the author's own account, correctness rests on Lemma 4 in the regime in which the selection bits of $A$ are not traced out, which is where the adaptive choice of $A$ enters.

\subsection{What the failure teaches}\label{sec:simon-lessons}

The three analyses and the response discussed in \cref{sec:simon-response} support four lessons, each of which bears on the counterfactual studied in the rest of this survey.
\begin{enumerate}
  \item \emph{Uncomputation requires a justified algorithm.} Regev's template reduces DCP to removing the selection bits after a partial measurement of a subset sum \cite{Reg04}. Simon's strategy attacks exactly that step, and the barrier of \cite{GRZ26} constrains algorithms in the stated observation model: a digest that keeps a third of each label leaves a distinguishing advantage of at most $Q\cdot2^{-n/3}$, and more generally the bound $Q\sqrt{K/N}$ is non-negligible for polynomial $Q$ only if each coordinatewise digest retains all but $O(\log n)$ bits of its label. Within that coordinatewise-digest model, an algorithm realising $\HReg(f)$ would therefore have to make, in the words of \cite{GRZ26}, ``extensive use of the classical Fourier labels'' in the uncomputation stage.
  \item \emph{Adaptive measurements need adaptive analyses.} The gap identified in \cite{GY26} is a classical pitfall: choosing a subset of registers using measurement outcomes can invalidate independence assumptions; the cited analyses identify this failure in the proposed partition rule. The audit's correlation counterexamples make the same point at the level of individual lemmas \cite{RoyessAudit}.
  \item \emph{Formal developments make the verified scope explicit.} Both negative analyses come with Lean developments, which verify the stated bounds for formal models of the algorithm and, in the audit, part of the lemma-level arguments \cite{GRZLean,RoyessAudit}; the decisive one appeared nine days after the preprint became public \cite{PQC26news}. The formal bounds support the refutation of the specified model; the correspondence between model and draft remains an informal step requiring its own justification. This is also why the refutation does not make the counterfactual moot: it disposes of a template with limited label information, not of DCP, as \cref{sec:intro} and \cref{rem:kuperberg} explain.
  \item \emph{Post-selection does not create information.} When a distinguishing advantage is bounded for every channel acting on the retained registers, post-selection that succeeds with non-negligible probability can increase it by at most the inverse of the acceptance probability. Larger gains require accepting exponentially rare events and paying for them with exponentially many samples. The same accounting applies to any proposed DCP algorithm that combines a lossy observation with adaptive post-selection (\cref{sec:simon-response}).
\end{enumerate}

% ======================================================================
\subsection{Chen's LWE algorithm and a proposed repair}\label{sec:chen-zhang}

Two earlier claims bear on the same question from another direction: instead of solving a dihedral problem, they target LWE directly. We describe them from the preprints themselves, Chen's paper in its revision of April 19, 2024, which carries the author's note of April 18 \cite{Che24} and Zhang's in its eighth arXiv version, of May 14, 2026 \cite{Zha25}, and from the public rebuttal of the latter \cite{Apo25}. The comparisons with the reach theorem at the end of this subsection are our analysis.

\paragraph{Chen's claim.} In April 2024 Chen announced a polynomial-time quantum algorithm for LWE with certain polynomial modulus-noise ratios \cite{Che24}, following earlier work on quantum algorithms for LWE-type problems given quantum samples \cite{CLZ22,CHLLT25}. As stated, the algorithm solves LWE with $m\geq\Omega(n\log q)$ samples and noise $D_{\Z,\beta}$, $\beta=\alpha q$, whenever $q\in\widetilde\Omega(\beta^4m^2)$, in time $\poly(m,\log q,\beta)$. Combined with Regev's reduction \cite{Reg09} at $q\in\widetilde O(n^4)$ and $\alpha\in\widetilde O(n^{-3.5})$, it would give $\SIVP$ and $\GapSVP$ for all $n$-dimensional lattices within factors of order $n^{4.5}$, up to logarithmic factors \cite[Thm. 1.6, Cor. 1.7]{Che24}. The algorithm introduces Gaussian functions with complex variances, exploiting the ``Karst wave'' in the discrete Fourier transform of complex Gaussians, and windowed quantum Fourier transforms with complex Gaussian windows. In nine steps it prepares a superposition over a lattice under a complex Gaussian window, alternates quantum Fourier transforms with a second window, measures low-order bits, extracts a state with purely imaginary Gaussian amplitudes, and in Step 9 extracts a linear equation over the secret; repetitions give a linear system, solved by Gaussian elimination. On April 18, 2024 the author reported that Step 9 contains a bug, found by Hongxun Wu and, independently, Thomas Vidick, that he did not know how to fix, and stated that the claim of a polynomial-time quantum algorithm for LWE with polynomial modulus-noise ratios no longer holds; the paper remains online, with a clarification of Step 8, as a source of techniques \cite{Che24}. The details of the bug are given in Section 3.5.9 of the updated paper; we do not reconstruct them.

\paragraph{Zhang's proposed repair.} Zhang proposed to replace Chen's Steps 8 and 9 by a single exact post-processing routine \cite{Zha25}. In the eighth version, this Step 9$^\dagger$ measures the residue $\tau=X_1\bmod D^2$ of the first coordinate, applies a diagonal quadratic phase that cancels the chirp, using a run-local gauge value, and applies a quantum Fourier transform over $\Z_M$. Writing $b=(-1,2p_1s,2p_1e)\in\Z^{1+n+m}$ for the planted vector of Chen's lattice formulation of LWE and $Q=(c+1)\|b\|^2$, the paper proves, under conditions AC1--AC5 described below (AC1--AC4 for the uniformity statement, given resonance), that accepted outcomes $u$ satisfy $\langle b,u\rangle\equiv0\pmod Q$ with probability $1-\negl(n)$ and are uniform on the hyperplane $\{v\in\Z_Q^{1+n+m}:\langle b,v\rangle\equiv0\pmod Q\}$, so that $O(n+m)$ of them determine $b$, hence $s$ and $e$ \cite[Lem. 2.2, Lem. 3.3, Thm. 3.4]{Zha25}. The consequence for LWE is stated conditionally: assuming Chen's reduction from LWE to a lattice problem and its unique-shortest-vector promise, if Chen's front end can be compiled in a re-tuned configuration satisfying AC1--AC4 and the Gaussian-envelope hypotheses of the paper's Proposition 2.6, which imply AC5, the routine solves Chen's chosen-secret instance and, through Chen's classical reductions, LWE in the same parameter regime \cite[Cor. 3.5]{Zha25}. AC1 concerns the state after Chen's Step 7, AC2 requires $Q$ to be odd, AC3 keeps Chen's scaling, and AC5 is a spectral concentration condition, which the paper proves for Gaussian envelopes under additional hypotheses. AC4 assumes that the preparation oracle ``also returns, as classical side information,'' the residue $v_1^*(E)\bmod Q$ of the run-dependent affine offset of the front end, or an equivalent gauge value. The paper does not assume that this value is derivable from the coordinate registers, calls it an additional oracle promise for a black-box oracle, and offers as the alternative a compilation that requires $E$ to be ``available coherently during the preparation'' \cite[Sections 1 and 2, Prop. 2.3]{Zha25}. The paper does not show that a re-tuned front end supplies AC4 or meets the hypotheses of Proposition 2.6, and it gives no cost for supplying AC4.

\paragraph{The published rebuttal.} Apon's rebuttal \cite{Apo25}, first posted on October 18, 2025 and revised on October 28, raises two objections, both about the first steps of Zhang's routine. First, it points out that in the first version, of September 15, 2025, the coherent shift of Step 9$^\dagger$ adds to a copy of the coordinate registers a shift along the planted vector, so that computing it presupposes the answer. Second, it argues that in the third version, of October 1, the latest when the rebuttal was written, the idea behind the routine, a ``pair-shift difference'' that subtracts a shifted copy to cancel the unknown offset, relies on a reversible ``coordinate evaluator'' meant to reproduce Chen's pipeline up to Step 8.f, although that pipeline measures a register in its first step, and it concludes that there is no evidence that an efficient evaluator of this kind exists. The revision of October 28 replies to a note of October 24 in which, according to the rebuttal, Zhang argued that deferred measurement resolves this point; the reply states that the deferred circuit needs two strongly correlated measurements to return the same value, which happens with negligible probability. Its author believes that these issues refute the paper's claims. The eighth version, which the rebuttal predates, no longer contains the pair-shift step. It receives one copy of the state after Chen's Step 7, does not assume that two copies with the same offset can be obtained, and takes the residue of the run-dependent offset, which depends on Chen's measurement outcomes, from the oracle as side information (AC4). In our reading, what the pair-shift step was meant to do, remove the dependence on the unknown offset, now rests on that assumption, which, as noted above, the paper does not show how to meet. Neither the paper nor the rebuttal has been peer reviewed.

\begin{table}[htbp]
\centering\small
\renewcommand{\arraystretch}{1.2}
\begin{tabularx}{\textwidth}{@{}L{2.3cm}Y Y L{3.0cm}@{}}
\toprule
\textbf{Claim} & \textbf{What was claimed} & \textbf{Where it would enter the map} & \textbf{Status}\\
\midrule
Chen 2024 \cite{Che24} & Quantum polynomial-time search-LWE for $m\geq\Omega(n\log q)$ and $q\in\widetilde\Omega((\alpha q)^4m^2)$; hence $\SIVP$ and $\GapSVP$ within factors of order $n^{4.5}$ via \cite{Reg09}. & A many-sample, small-noise LWE regime inside the reach of $\HReg(f)$ for $f<1/2$; large-width $\mathrm G\text{-}\EDCP$ via \cite[Thm. 4]{BKSW18}; no known route to the three hypotheses; excludes ML-KEM, FrodoKEM and standard-noise LWE. & Withdrawn: bug in Step 9, found by H. Wu and, independently, T. Vidick; claim withdrawn on April 18, 2024.\\
Zhang 2025--26 \cite{Zha25} & An exact replacement for Chen's Steps 8--9 whose accepted outcomes, under AC1--AC5, are uniform on the dual hyperplane of the planted vector; LWE in Chen's regime if the front end meets the conditions. & The same regime as Chen's claim, if the conditions held; no new consequence. AC4 requires the preparation oracle to supply a residue of a run-dependent offset, which the paper does not assume derivable from the coordinate registers. & Conditional: version 8 of May 2026, not peer reviewed; rebuttal of versions 1 and 3 \cite{Apo25}.\\
Simon 2026 \cite{Sim26} & Quantum polynomial-time DCP tolerating faulty samples at rate $1/O(\log n)$; SVP at $\sqrt n\,\mathrm{polylog}(n)$ and LWE through \cite{Reg04,BKSW18}. & A claim of type $\HReg(f)$ for every constant $f>0$ (\cref{rem:simon-claim}); a correct version would subsume Chen's conclusion for LWE (\cref{thm:reach}\,\ref{it:lwe}). & Refuted: \cref{sec:simon-failure,sec:simon-response}.\\
\bottomrule
\end{tabularx}
\caption{Three claimed polynomial-time quantum algorithms bearing on LWE, their place in the counterfactual map and their status.}
\label{tab:claims}
\end{table}

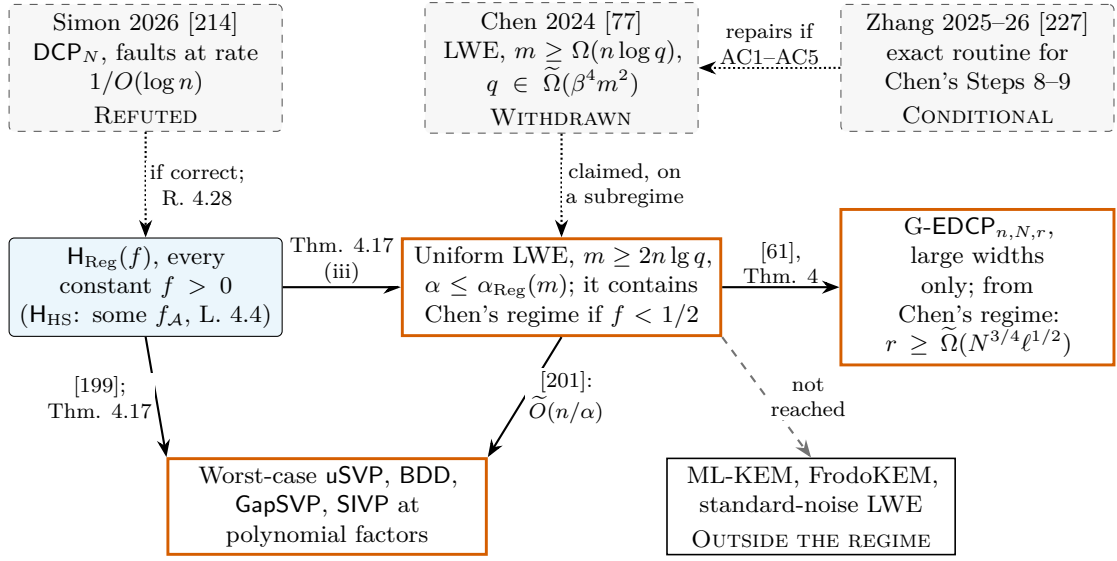
\begin{figure}[htbp]
\centering
\begin{tikzpicture}[font=\footnotesize,
  claimnode/.style={draw=black!65,dashed,fill=black!3,rounded corners=2pt,align=center,
    text width=34mm,minimum height=14mm,inner sep=3pt},
  hypnode/.style={draw=wBlack,fill=wSky!12,rounded corners=2pt,align=center,
    text width=34mm,minimum height=12mm,inner sep=3pt},
  resnode/.style={draw=wVermillion,line width=1.1pt,align=center,
    text width=34mm,minimum height=12mm,inner sep=3pt},
  outnode/.style={draw=wBlack,line width=0.6pt,align=center,
    text width=36mm,minimum height=12mm,inner sep=3pt},
  lab/.style={font=\scriptsize,fill=white,inner sep=1.2pt,align=center},
  arr/.style={->,>=Stealth,line width=.8pt},
  carr/.style={->,>=Stealth,line width=.8pt,densely dotted}]
 \node[claimnode] (sim) at (0,5.8)
   {Simon 2026 \cite{Sim26}\\$\DCP_N$, faults at rate\\$1/O(\log n)$\\[2pt]\textsc{Refuted}};
 \node[claimnode] (che) at (5.5,5.8)
   {Chen 2024 \cite{Che24}\\LWE, $m\geq\Omega(n\log q)$,\\$q\in\widetilde\Omega(\beta^4m^2)$\\[2pt]\textsc{Withdrawn}};
 \node[claimnode] (zha) at (11,5.8)
   {Zhang 2025--26 \cite{Zha25}\\exact routine for\\Chen's Steps 8--9\\[2pt]\textsc{Conditional}};
 \node[hypnode] (reg) at (0,2.9)
   {$\HReg(f)$, every\\constant $f>0$\\($\HHS$: some $f_{\mathcal A}$, L. \ref{lem:hs-rare})};
 \node[resnode,text width=40mm] (lwe) at (5.5,2.9)
   {Uniform LWE, $m\geq2n\lg q$,\\$\alpha\leq\alpha_{\mathrm{Reg}}(m)$; it contains\\Chen's regime if $f<1/2$};
 \node[resnode] (edcp) at (11,2.9)
   {$\mathrm G\text{-}\EDCP_{n,N,r}$, large widths\\only; from Chen's regime:\\$r\geq\widetilde\Omega(N^{3/4}\ell^{1/2})$};
 \node[resnode,text width=40mm] (wc) at (2.4,0)
   {Worst-case $\uSVP$, $\BDD$,\\$\GapSVP$, $\SIVP$ at\\polynomial factors};
 \node[outnode] (dep) at (8.8,0)
   {ML-KEM, FrodoKEM,\\standard-noise LWE\\[2pt]\textsc{Outside the regime}};
 \draw[carr] (sim) -- node[lab,right]{if correct;\\R. \ref{rem:simon-claim}} (reg);
 \draw[carr] (che) -- node[lab,right]{claimed, on\\a subregime} (lwe);
 \draw[carr] (zha) -- node[lab,above]{repairs if\\AC1--AC5} (che);
 \draw[arr] (reg) -- node[lab,above]{Thm. \ref{thm:reach}\\\ref{it:lwe}} (lwe);
 \draw[arr] (reg.south) -- node[lab,left]{\cite{Reg04};\\Thm. \ref{thm:reach}} (wc.north west);
 \draw[arr] (lwe.south) -- node[lab,right]{\cite{Reg09}:\\$\Otil(n/\alpha)$} (wc.north east);
 \draw[arr] (lwe) -- node[lab,above]{\cite{BKSW18},\\Thm. 4} (edcp);
 \draw[arr,dashed,draw=black!55] (lwe.south east) -- node[lab,right]{not\\reached} (dep.north);
\end{tikzpicture}
\caption{Where the three claimed algorithms would enter the counterfactual map, drawn in the solver direction of \cref{fig:reach}: an arrow from $A$ to $B$ means that an algorithm for $A$ yields one for $B$. Dotted arrows are claims, and the one from Chen's claim ends at the node that contains its regime; solid arrows are proved under the cited premises; the dashed arrow marks parameters outside the claimed regime. The inclusion of Chen's regime, the width bound for $\mathrm G\text{-}\EDCP$ and the comparison with deployed parameters are our analysis (\cref{sec:chen-zhang}).}
\label{fig:claims-map}
\end{figure}

\paragraph{Where the claims enter the counterfactual.} \Cref{fig:claims-map,tab:claims} summarise the following comparison.
\begin{enumerate}
  \item \emph{A direct LWE solver, not a dihedral one.} Chen's statement, and Zhang's conditional version of it, concern LWE in a single regime: many samples, $m\geq\Omega(n\log q)$, and noise that is small relative to the modulus, $\alpha=\beta/q\leq1/\widetilde\Omega(\beta^3m^2)$. No known reduction derives $\HHS$, $\HReg(f)$ or $\HEDCP[\ell]$ from a solver for this regime: no reduction of cyclic DCP to LWE is known, and the only known general route from an LWE solver to extrapolated states is the converse reduction \cite[Thm. 4]{BKSW18}, which turns $\ell$ states of $\mathrm G\text{-}\EDCP_{n,N,r}$ into $\ell$ samples of LWE with $\alpha=1/r$. With Chen's condition, $N\in\widetilde\Omega((N/r)^4\ell^2)$, and $\ell\geq\Omega(n\log N)$ states, this covers only widths $r\geq\widetilde\Omega(N^{3/4}\ell^{1/2})$, up to the noise distribution produced by that reduction, whereas $\HEDCP[\ell]$ requires every $r\geq r_0$.
  \item \emph{Inside the reach of $\HReg(f)$.} \Cref{thm:reach}\,\ref{it:lwe} solves search-LWE with $m\geq2n\lg q$ samples, prime $q$ and noise $D_{\Z^m,\alpha q}$ whenever $\alpha\leq\alpha_{\mathrm{Reg}}(m)=1/(8\sqrt m\,\gu(m+1))$, that is, $\alpha\leq1/O(m^{1+2f})$. Chen's regime gives $\alpha\leq1/\widetilde\Omega(\beta^3m^2)$, so for noise parameters $\beta\geq1$ and every $f<1/2$, each instance of that regime with $m\geq2n\lg q$ and prime $q$ lies inside the reach once $n$ is large enough, for all but a $q^{-n}$ fraction of $A$. At the parameters of the lattice corollary ($\beta\approx\sqrt n$, $\alpha\approx n^{-3.5}$, $m=\Theta(n\log q)$) the same holds for every $f<5/4$. The corollary's factor of order $n^{4.5}$ is weaker than the $\HReg(f)$ guarantees, $\SIVP$ at $O(n^{3/2+2f})$ and $\GapSVP$ at $\Otil(n^{1+2f})$, for $f<3/2$ and $f<7/4$ respectively. Under $\HEDCP[\ell]$, with $q\leq\poly(n)$ and the samples reduced to $m_{\max}(n,q)$, the certified scale $\alpha^*_{\mathrm E}(n,\ell)=1/\Otil(\ell n)$ exceeds the corollary's $\alpha\approx n^{-3.5}$ whenever $\ell\leq n^{5/2}/\mathrm{polylog}(n)$. Under $\HHS$ the comparison depends on the witness through $f_{\mathcal A}$ (\cref{lem:hs-rare}).
  \item \emph{Subsumed by a correct version of Simon's claim.} Read as a claim of type $\HReg(f)$ (\cref{sec:simon-strategy}), the tolerance claimed in \cite{Sim26} exceeds the one that $\HReg(f)$ grants for every constant $f>0$ (\cref{rem:simon-claim}). A solver with that tolerance would therefore satisfy $\HReg(f)$ for every constant $f$, in particular for $f<1/2$, and by the previous point would solve every instance of Chen's regime with $m\geq2n\lg q$, prime $q$ and $\beta\geq1$. In this precise sense Simon's claim was the more far-reaching one: it would have implied Chen's conclusion for LWE in that range, and stronger worst-case factors; no converse implication is known.
  \item \emph{No deployed parameter set.} Chen's condition and sample count fail for the deployed LWE-based KEMs. For ML-KEM and FrodoKEM, with $\beta=\sqrt{2\pi}\,\sigma$ for the error standard deviation $\sigma$, the ratio $\beta^4m^2/q$ ranges from about $4\times10^3$ (FrodoKEM-1344) to $3\times10^4$ (FrodoKEM-640) before the logarithmic factors hidden in $\widetilde\Omega$, and each key offers $m=n$ samples rather than $\Omega(n\log q)$. Standard-noise LWE, with $\alpha=1/\Otil(\sqrt n)$, is excluded as well, since the regime requires $\alpha\leq1/\widetilde\Omega(m^2)$ when $\beta\geq1$. A correct version of the withdrawn claim would thus have certified only LWE instances that \cref{tab:collapse} already lists as reached under $\HReg(f)$ for $f<1/2$ (many samples, small noise), and none of the parameter sets of \cref{tab:standards}.
  \item \emph{A repaired step is not a repaired algorithm.} Zhang's routine addresses Chen's post-processing. Its consequence for LWE remains conditional on the front end and on side information about a run-dependent offset that the paper does not show how to obtain from LWE samples. Even if AC1--AC5 were established, the result would solve LWE only in the regime of the first point and would add no consequence to the classification. We therefore record the Chen--Zhang route as a withdrawn claim with a conditional repair, not as a hypothesis of this survey.
\end{enumerate}

% ======================================================================
\section{The counterfactual hypotheses and their reach}\label{sec:hypothesis}\label{sec:reach}
% ======================================================================

This section turns ``DCP is easy'' into three precise hypotheses, which differ in their access model, robustness to faulty inputs and extrapolation requirements, and derives their consequences. We compose published reductions while retaining the parameters and promises needed at each step. The propositions of \cref{sec:proofs} make the resulting bounds, decoder success and preparation error explicit; the underlying mechanisms are credited there. The information-theoretic bound of \cref{prop:ell-lower}, due to an anonymous reviewer, constrains the number of extrapolated states. To make the compositions checkable, \cref{sec:regev02} explains Regev's reduction from lattices to DCP in detail.

\subsection{Three hypotheses}\label{sec:hyp-def}

``DCP is easy'' is ambiguous in four ways that change the consequences: the tolerance to bad registers, the group ($\Z_N$ with $N$ exponential, versus $\Z_N^n$ with $N$ polynomial), the robustness to non-uniform amplitudes, and the access model (coset states, or oracle access to the functions that define them). We therefore state three hypotheses. All three are \emph{state-input} hypotheses: the algorithm receives quantum states and nothing else.

\begin{enumerate}[label=\textbf{Hypothesis \thesection.\arabic*}, ref=\thesection.\arabic*,
  leftmargin=1.5em, labelindent=0pt, labelsep=0.4em, itemindent=*, align=left, itemsep=3pt, topsep=4pt]
  \item\label[hypothesis]{hyp:reg} ($\HReg(f)$). Fix a constant $f>0$. There is a quantum algorithm solving
$\DCP_N^f$ in time $\poly(\log N)$, with inverse-polynomial
success probability, for every $N$.
  \item\label[hypothesis]{hyp:edcp} ($\HEDCP[\ell]$). Let $\ell=\ell(n,N)$ be bounded by a polynomial in $n$. There are a constant $r_0>0$ and a quantum polynomial-time algorithm that, for every polynomially bounded modulus $N$ and every $r\geq r_0$, solves $\mathrm{G}\text{-}\EDCP^\ell_{n,N,r}$ from $\ell$ states, for every secret and arbitrary offsets, with success probability at least $1/\poly(n)$. We write $\HEDCP$ when the sample complexity is left unspecified.
  \item\label[hypothesis]{hyp:hs} ($\HHS$). There is a quantum algorithm that, given polynomially many noiseless dihedral coset states over $\Z_N$ (the input of $\DCP_N^\infty$, \cref{def:dcp}), recovers the hidden shift in time $\poly(\log N)$ with inverse-polynomial success probability, for every $N$, every secret and arbitrary offsets. The generalised form $\HHS^{\mathrm{ab}}$ asserts the same for finite abelian groups given by generators and relations.
\end{enumerate}
\addtocounter{theorem}{3}% the three hypotheses keep the numbers 4.1--4.3 of the theorem counter

\begin{lemma}\label{lem:hs-rare}
Let $L=\max\{1,\log_2N\}$, using the binary logarithm in the corruption probability of \cref{def:dcp}. Suppose a noiseless
cyclic coset-state algorithm $\mathcal A$ uses at most
$Q(L)\leq C L^a$ registers and succeeds with probability at
least $p(L)\geq c L^{-b}$, for every secret and arbitrary
offsets. Under the corruption model of $\DCP_N^f$, its success
is at least
\[
 p(L)-Q(L)L^{-f}.
\]
Consequently, $\mathcal A$ establishes $\HReg(f)$ for every
fixed $f>a+b$, with sufficiently large input size. In
particular,
\[
 \HHS\quad\Longleftrightarrow\quad
 \bigl(\exists\,f>0\text{ constant}:\HReg(f)\bigr).
\]
\end{lemma}

\begin{proof}
A register corrupted with probability at most $L^{-f}$ has
trace distance at most $L^{-f}$ from its ideal coset state.
The tensor-product hybrid bound gives total distance at most
$Q(L)L^{-f}$. Applying $\mathcal A$ and measuring its output
cannot increase this distance. If $f>a+b$, then
$CL^{a-f}\leq cL^{-b}/2$ for sufficiently large $L$, preserving
inverse-polynomial success. The finitely many smaller input
sizes can be handled separately. The reverse implication
uses inputs with no bad registers. This is the
rare-corruption observation underlying
\cite[Cor. 1.2]{Reg04}.
\end{proof}

\paragraph{Sources and justification.} The hypotheses isolate the input promises of the reductions, rather than claiming minimal solver assumptions. Regev's Definition 2.1 and Theorem 1.1 give the probabilistic bad-register model and the fixed-$f$ lattice consequence \cite{Reg04}. His Corollary 1.2 already records the qualitative consequence of a noiseless solver; \cref{lem:hs-rare} makes its sample and success dependence explicit. The unstructured hypothesis $\HEDCP[\ell]$ is used with the Gaussian-state reductions of \cite{BKSW18}. We expose $\ell$ because the forward reduction must prepare all states correctly, while \cref{prop:ell-lower} gives $\ell=\Omega(n\log N)$ at fixed width. The structured state map of \cite{WZ26} and its parameter audit appear in \cref{sec:wz-audit}. Its large-modulus output is not an additional clause hidden inside $\HEDCP$; the exact positive composition there uses $\HHS^{\mathrm{ab}}$ (or $\HHS$ for rank one), subject to the forward preparation and state budget. The cyclic and abelian hidden-shift hypotheses use the coset-state access model of \cite{Kup05,CJS14}; \cref{lem:hs-dcp} gives the required state preparation from efficient injective functions.

\paragraph{What the hypotheses mean.} The parameter $f$ in $\HReg(f)$ is prescribed: each register is bad with probability at most $(\log N)^{-f}$ in the model of \cref{def:dcp}. It is not an unrestricted adversarial entangled-noise model. The vector-state hypothesis $\HEDCP[\ell]$ concerns polynomial moduli and Gaussian widths at least $r_0$. The cyclic hypothesis $\HHS$ concerns noiseless states; \cref{lem:hs-rare} explains why sufficiently rare corruptions preserve its inverse-polynomial success. Known subexponential cyclic sieves \cite{Kup05,Reg04b,Kup13} and the quasi-polynomial EDCP algorithm \cite{BJKNY25} do not furnish these polynomial-time, polynomial-sample guarantees in the required regimes.

\paragraph{Relations between the hypotheses.}
For each fixed $f>0$, $\HReg(f)$ implies $\HHS$.
Conversely, \cref{lem:hs-rare} shows that $\HHS$ implies
$\HReg(f_{\mathcal A})$ for some constant $f_{\mathcal A}$
depending on the sample complexity and success probability
of its witnessing algorithm. Thus the existential form of
$\HReg$ and $\HHS$ are qualitatively equivalent. A specified
small value of $f$ remains a stronger quantitative
requirement: the lemma gives no such value independently
of the solver. Also, $\HReg(f_1)$ implies $\HReg(f_2)$
whenever $f_2\geq f_1$, because the permitted corruption
probability decreases. We establish no additional
implication involving $\HEDCP[\ell]$. Its relation to the
cyclic hypotheses remains a separate question. We retain $\Hstar:=\HReg(f)\wedge\HEDCP$ as notation for the conjunction, without asserting a strict separation from $\HHS$. The next lemmas concern smooth extrapolation and coset-state preparation, respectively.

\begin{lemma}\label{lem:dcp-edcp}
Let $\rho:\Z\to\R_{\geq0}$ have finite support $J$, with the index $j\in J$ stored in binary on $\poly(n)$ qubits and membership in $J$ decidable in classical polynomial time, and set $\varepsilon:=\max_k\big|\rho(2k+1)/\rho(2k)-1\big|$ over the $k$ with $\rho(2k)\rho(2k+1)>0$. Suppose a quantum algorithm $\mathcal A$ recovers $s\in\Z_N^n$ from $\ell$ states of the form $a\ket{0}\ket{y_k}+b\ket{1}\ket{y_k+s}$ with arbitrary $y_k$ and $|b/a-1|\leq\varepsilon$, succeeding with probability $p$. Then there is a quantum algorithm that solves $\EDCP^{\ell'}_{n,N,\rho}$ with $\ell'=O(\ell/\pi_\rho)$ samples and success probability $p-2^{-\Omega(\ell)}$, where $\pi_\rho:=\Pr_{j\sim\rho^2}[\,j\text{ and }j\oplus1\text{ both in }\mathrm{supp}(\rho)\,]$.
\end{lemma}

\begin{proof}
The algorithm needs no access to $\rho$ beyond the given states and the membership test for $J$: it only measures a function of the index register, never prepares amplitudes, and it is efficient when $\pi_\rho\geq1/\poly(n)$. Write each extrapolated state as $\sum_{k}\big(\rho(2k)\ket{2k}\ket{x+2ks}+\rho(2k+1)\ket{2k+1}\ket{x+(2k+1)s}\big)$ and measure the observable $\lfloor j/2\rfloor$. On outcome $k$ the state collapses, up to normalisation, to $\rho(2k)\ket{2k}\ket{y}+\rho(2k+1)\ket{2k+1}\ket{y+s}$ with $y=x+2ks$, which after relabelling the first register as a qubit is a two-point state with $b/a=\rho(2k+1)/\rho(2k)$. The outcome $k$ is known, but $y$ is not; since $\mathcal A$ accepts arbitrary $y_k$, this does not matter. A sample is useful when both $2k$ and $2k+1$ lie in the support, which happens with probability $\pi_\rho$; by a Chernoff bound, $O(\ell/\pi_\rho)$ samples yield $\ell$ useful ones except with probability $2^{-\Omega(\ell)}$. Run $\mathcal A$ on them.
\end{proof}

For Gaussian $\rho$ of width $r$, first truncate to $|j|\leq tr$: by an elementary estimate of the Gaussian tail, each state moves by trace distance at most $e^{-\pi t^2/2}$ for $t\geq1$. On the truncated support, $\rho(2k+1)/\rho(2k)=e^{-\pi(4k+1)/r^2}$ with $|4k+1|\leq2tr+1$, so $\varepsilon\leq e\pi(2tr+1)/r^2\leq3e\pi t/r$ when $r\geq7t$ (using $|e^{-x}-1|\leq e|x|$ for $|x|\leq1$), and $\pi_\rho\geq1/2$. Robust two-point solvers with tolerance $\varepsilon$ therefore handle $\mathrm{G}\text{-}\EDCP$ of width $r=\Omega(t/\varepsilon)$, at a cost of $\ell e^{-\pi t^2/2}$ in success probability; this is the quantitative form of the observation in \cite{BKSW18} that a DCP algorithm yields an EDCP algorithm. The lemma is stated over $\Z_N^n$. $\HReg(f)$ concerns the cyclic group $\Z_N$, and we do not know an unconditional translation between two-point solvers over $\Z_N$ and over $\Z_N^n$ (\cref{q:translate}); this is why $\HEDCP$ is stated separately. (Regev's own embedding of $\Z^n$ into $\Z_N$, recalled in \cref{sec:regev02}, works for differences with small coordinates, which is the situation of the two-point problem but not of an arbitrary secret $s\in\Z_N^n$.)

\begin{lemma}\label{lem:hs-dcp}
Let $f_0,f_1:\Z_N\to S$ be injective, efficiently computable in superposition, with $f_1(x)=f_0(x+s)$. Then $\poly(\log N)$ registers of $\DCP_N^{\infty}$ with secret $-s$ can be prepared with $O(1)$ evaluations each. In particular, under $\HHS$ the shift $s$ can be recovered in quantum polynomial time from oracle access to $f_0,f_1$.
\end{lemma}

\begin{proof}
Prepare $\frac{1}{\sqrt{2N}}\sum_{b\in\{0,1\}}\sum_{x\in\Z_N}\ket{b}\ket{x}\ket{f_b(x)}$ and measure the last register, obtaining some $z\in S$. By injectivity there are exactly two preimages: $x_0=f_0^{-1}(z)$ with $b=0$, and $x_1=f_1^{-1}(z)=x_0-s$ with $b=1$. The post-measurement state is $\tfrac{1}{\sqrt2}(\ket{0,x_0}+\ket{1,x_0-s})$, a dihedral coset state with uniformly random $x_0$ and secret $-s$, and no bad registers.
\end{proof}

\begin{remark}\label{rem:kuperberg}
Kuperberg's algorithm uses $2^{O(\sqrt{\log N})}$ registers and time \cite{Kup05}. Regev's reduction produces instances with $N=2^{\Theta(n^2)}$ (\cref{sec:regev02}) and needs an algorithm that tolerates bad registers, while the reductions of \cite{BKSW18,WZ26} provide only polynomially many samples. These resource and modulus requirements prevent these direct compositions from establishing the claimed improvements for the LWE instances considered here. In particular, the quasi-polynomial EDCP algorithm of \cite{BJKNY25} requires quasi-polynomially many samples and power-of-two moduli; no matching preparation from ML-KEM is supplied by this comparison. A mismatch with its original modulus $q=3329$ alone would not rule out every possible reduction. The cyclotomic sieve of \cite{BFS26} runs in time $2^{O_p(\log n\log q)}$, quasi-polynomial for polynomial $q$, but the available reductions again generate only limited approximate states. For hidden-shift instances with perfect states and no sample bottleneck, such as CSIDH, Kuperberg-type algorithms \emph{do} apply, at subexponential cost \cite{CJS14,BS20,Pei20}.
\end{remark}

\subsection{Regev's reduction from unique-SVP to DCP}\label{sec:regev02}

Because $\HReg(f)$ is defined by what Regev's reduction needs, we summarise how that reduction works, following \cite{Reg04}. It has two stages.

\paragraph{From unique-SVP to the two-point problem.} The intermediate problem is a vector version of DCP. The input to the \emph{two-point problem} with failure parameter $f$ consists of $\poly(n\log M)$ registers; each is, with probability at least $1-1/(n\log(2M))^f$, in a state $\tfrac{1}{\sqrt2}(\ket{0,a}+\ket{1,a'})$ with $a,a'\in\{0,\dots,M-1\}^n$ and a fixed difference $a'-a$, and otherwise in an arbitrary state; the goal is to output the difference. Given a lattice with a unique shortest vector $u$, the reduction creates a superposition over many lattice points and then collapses it, by a measurement, onto just two points whose difference is $u$. The collapse uses a random partition of space into translated cubes whose side length is proportional to $\|u\|$. Because $u$ is unique up to sign and every lattice vector not parallel to it is much longer, a cube can contain at most a pair of points differing by $u$; measuring the cube index leaves the superposition of those two points. Registers become \emph{bad} when the cube boundary separates the pair, which happens with small but non-zero probability, and this is the origin of the failure parameter. Coordinates are written in an LLL-reduced basis, which keeps them in a range of size $M$.

\paragraph{From the two-point problem to DCP.} The vector difference is encoded into a single integer by the map
\[
  (a_1,\dots,a_n)\longmapsto a_1+a_2\cdot2M+\dots+a_n(2M)^{n-1},
\]
which embeds $\{0,\dots,M-1\}^n$ into $\Z_N$ with $N=(2M)^n$; the spacing $2M$ prevents carries from mixing coordinates, so the cyclic difference determines the vector difference \cite[Lem. 3.2]{Reg04}. With $M=2^{4n}$, chosen to accommodate the size of the LLL coefficients, the modulus is $N=2^{n(4n+1)}$, so $\log N=\Theta(n^2)$.

\paragraph{Parameters.} Balancing the probability that a pair is split against the required failure rate gives the approximation factor. With cubes the reduction handles $\Theta(n^{1+2f})$-unique-SVP; replacing cubes by balls improves this to $\Theta(n^{1/2+2f})$ \cite{Reg04}. Combined with the subset-sum connection of \cite[Thm. 1.3]{Reg04}, which yields DCP with failure parameter $f=1$, this gives a quantum worst-case to average-case reduction for $\Theta(n^{2.5})$-unique-SVP \cite[Cor. 1.5]{Reg04}. Two features of the reduction drive the rest of this section: the oracle must tolerate bad registers, and the reduction makes no use of algebraic structure, so it also applies to ideal and module lattices that meet the ordinary uniqueness promise.

\subsection{Technical lemmas}\label{sec:technical}

We use the following facts, stated with the parameters needed below. Throughout, $q$ is prime. Fix $f>0$ and a witnessing algorithm $\mathcal A$. Write $\Gamma_{f,\mathcal A}(d)$ for a sufficient finite-dimensional ordinary-uSVP promise bound for the reduction of \cite[Thm. 1.1]{Reg04}; asymptotically it is $O(d^{1/2+2f})$. In theorem statements below, $\gu(d)$ denotes such a valid bound, enlarged if necessary to be at least one. For numerical comparisons only, we use the model $c_0d^{1/2+2f}$. Its constant, the starting dimension and the running-time polynomial may depend on $f$ and $\mathcal A$; no uniform limit as $f\to0^+$ is asserted.

\begin{lemma}[Random $q$-ary lattices {\cite[Lem. 5.3]{GPV08}}]\label{lem:qary}
Let $m\geq2n\lg q$. For all but a $q^{-n}$ fraction of $A\in\Zq^{m\times n}$, $\lambda_1(\Lambda_q(A))\geq\lambda_1^\infty(\Lambda_q(A))\geq q/4$. For all but a negligible fraction of $A\in\Zq^{n\times m}$ and any function $\omega(\sqrt{\log m})$, the smoothing parameter satisfies $\eta_\epsilon(\Lambda_q^\perp(A))\leq\omega(\sqrt{\log m})$ for some negligible $\epsilon$.
\end{lemma}

\begin{lemma}\label{lem:qary-count}
Let $m\geq n$ satisfy $q^{(m-n)/m}\geq2\sqrt{2\pi e}$, and let $A\in\Zq^{m\times n}$ be uniform. Then $\lambda_1(\Lambda_q(A))\geq\min\{q,\sqrt m\,q^{(m-n)/m}/(2\sqrt{2\pi e})\}$ except with probability $(3/4)^m$, and except with probability $2^{-m+1}$ if moreover $q^{(m-n)/m}\geq\sqrt{2\pi e}\,m$.
\end{lemma}

\begin{proof}
Let $R$ be the minimum in the statement. A nonzero $y\in\Z^m$ with $\|y\|<R\leq q$ is not in $q\Z^m$, so it lies in $\Lambda_q(A)$ only if $As\equiv y$ for some $s\neq0$, which for each fixed $s$ has probability $q^{-m}$; hence $\Pr[\lambda_1<R]\leq q^{n-m}N_m(R)$, where $N_m(R)$ counts the integer points of norm less than $R$. Unit cubes centred at those points lie in the ball of radius $R+\sqrt m/2$, whose volume is at most $\big((R+\sqrt m/2)\sqrt{2\pi e/m}\big)^m/\sqrt{\pi m}$. With $R\leq\sqrt m\,q^{(m-n)/m}/(2\sqrt{2\pi e})$ this gives $\Pr[\lambda_1<R]\leq\big(\tfrac12+\sqrt{2\pi e}/(2q^{(m-n)/m})\big)^m/\sqrt{\pi m}$, which is at most $(3/4)^m$ when $q^{(m-n)/m}\geq2\sqrt{2\pi e}$, and at most $\big(\tfrac12(1+1/m)\big)^m/\sqrt{\pi m}\leq e\,2^{-m}/\sqrt{\pi m}\leq2^{-m+1}$ when $q^{(m-n)/m}\geq\sqrt{2\pi e}\,m$.
\end{proof}

This is Lemma 6 of \cite{BKSW18} with an explicit range of validity; as stated there, for all $m\geq n$, the lemma fails at $m=n$, where $\Lambda_q(A)=\Z^n$ for invertible $A$. The first statement of \cref{lem:qary} says that a random $q$-ary lattice has no unusually short vectors once there are enough rows: a vector $As\bmod q$ with all coordinates below $q/4$ is unlikely to exist for $s\neq0$ when $m$ is large compared with $n\log q$. The second says that $\Lambda_q^\perp(A)$ is ``smooth'' at a width barely above constant, a property used in the distributional analysis of Gaussian preimage sampling \cite{GPV08}. The smoothing bound alone does not supply a trapdoor or an efficient sampler.

\begin{lemma}[Gaussian facts \cite{Ban93,MR07,Reg09}]\label{lem:tails}
Let $\Lambda\subset\R^m$ be a full-rank lattice.
\begin{enumerate}[label=(\alph*)]
  \item For every $s>0$, a sample of $D_{\Lambda,s}$ has norm greater than $s\sqrt m$ with probability less than $2^{-m}$ \cite[Lem. 1.5]{Ban93}; the same holds for a continuous Gaussian vector of parameter $s$ in $\R^m$.
  \item If $s\geq\sqrt2\,\eta_\epsilon(\Lambda)$ with $\epsilon\leq1/10$, a sample of $D_{\Lambda,s}$ lies outside any fixed hyperplane with probability at least $1/10$ \cite[Lem. 3.15]{Reg09}.
  \item $\eta_{2^{-m}}(\Lambda)\leq\sqrt m/\lambda_1(\Lambda^*)$, and $\eta_\epsilon(\Z)\leq\sqrt{\ln(2+2/\epsilon)/\pi}$ \cite[Lems. 3.2 and 3.3]{MR07}; in particular $\eta_{2^{-m}}(\Z)\leq\sqrt{m/2}$ for $m\geq2$.
  \item $\lambda_m(\Lambda)\leq\sqrt{2m}\,\eta_\epsilon(\Lambda)$ for every $\epsilon\leq1/10$.
\end{enumerate}
\end{lemma}

\begin{proof}[Proof of (d)]
Sample from $D_{\Lambda,s}$ with $s=\sqrt2\,\eta_\epsilon(\Lambda)$. By (a) each sample has norm at most $s\sqrt m$ except with probability $2^{-m}$, and by (b) it lies in any fixed hyperplane with probability at most $9/10$; hence $O(m)$ samples contain $m$ linearly independent vectors of norm at most $s\sqrt m$ with positive probability, so $\lambda_m(\Lambda)\leq s\sqrt m$.
\end{proof}

\begin{lemma}[$\BDD$, $\uSVP$ and $\GapSVP$ \cite{LM09}]\label{lem:lm09}
For every $\gamma\geq1$: $\BDD_{1/(2\gamma)}$ reduces to $\uSVP_\gamma$; and $\GapSVP_{\gamma'}$ with $\gamma'=O(\gamma\sqrt{d/\log d})$ reduces to $\uSVP_\gamma$ (via $\BDD$). All reductions are classical, polynomial-time, and dimension-preserving up to one.
\end{lemma}

The first reduction is Kannan's embedding: a $\BDD$ target $t$ close to a lattice point $v$ becomes the unusually short vector $(t-v,\mu)$ of the lattice generated by the basis and $(t,\mu)$ for a suitable $\mu$, whose gap is about $1/\delta$. The second is a randomised self-reduction in the style of Goldreich and Goldwasser: perturb a lattice point by a random short vector and ask the $\BDD$ oracle to undo the perturbation, which succeeds reliably only when the lattice has no short vectors \cite{GG00,LM09}. Decoding can also be obtained from short \emph{dual} vectors, with a larger loss.

\begin{lemma}\label{lem:bdd-sivp}
For every $\gamma\geq1$, $\BDD_{1/(2n\gamma)}$ on $n$-dimensional lattices reduces in classical polynomial time to $\SIVP_\gamma$ on their duals.
\end{lemma}

\begin{proof}
Let $t$ satisfy $\mathrm{dist}(t,\Lat)\leq\lambda_1(\Lat)/(2n\gamma)$. The $\SIVP_\gamma$ oracle applied to $\Lat^*$ returns linearly independent $s_1,\dots,s_n\in\Lat^*$ of length at most $\gamma\lambda_n(\Lat^*)$, which can be converted into a basis $T$ of $\Lat^*$ whose Gram--Schmidt vectors satisfy $\|\tilde t_i\|\leq\|\tilde s_i\|\leq\gamma\lambda_n(\Lat^*)$ \cite[Lem. 7.1]{MG02}. The dual basis $B=T^{-\top}$ of $\Lat$, taken in reverse order, has Gram--Schmidt norms $1/\|\tilde t_n\|,\dots,1/\|\tilde t_1\|$, so Babai's nearest-plane algorithm with this basis returns the closest lattice point to every target within distance $\tfrac12\min_i\|\tilde b_i\|\geq1/(2\gamma\lambda_n(\Lat^*))$ \cite{Bab86}. By Banaszczyk's transference theorem $\lambda_1(\Lat)\lambda_n(\Lat^*)\leq n$ \cite{Ban93}, this radius is at least $\lambda_1(\Lat)/(2n\gamma)$.
\end{proof}

\begin{lemma}[Regev's worst-case reduction {\cite[Thm. 3.1]{Reg09}}]\label{lem:regev09}
Let $\alpha\in(0,1)$ and $q$ satisfy $\alpha q>2\sqrt n$. An efficient algorithm that solves search-$\LWE_{n,q,\alpha'}$ for every $\alpha'\leq\alpha$, from polynomially many samples, yields an efficient quantum algorithm for $\SIVP_{\gamma}$ and $\GapSVP_{\gamma}$ on all $n$-dimensional lattices, with $\gamma=\Otil(n/\alpha)$.
\end{lemma}

The structure of the proof, an iteration that alternates a classical $\BDD$ step using the LWE oracle with a quantum step that narrows a discrete Gaussian, is described in \cref{sec:prelim-lwe}. For our purposes the lemma is a black box with one important property: the reduction chooses the LWE parameters $(n,q,\alpha,m)$ itself, so it may request the polynomially many samples allowed by its LWE-oracle interface. This is different from a fixed public-key sample budget. Regev states the theorem for errors drawn from a rounded continuous Gaussian: inside the reduction the oracle is called on samples whose continuous Gaussian error is rounded to the nearest integer. Replacing that deterministic rounding by the randomised rounding of \cref{lem:round} below produces discrete Gaussian errors at the cost of an additive term $\omega(\sqrt{\log n})$, in quadrature, in their parameter, which is negligible against $\alpha q>2\sqrt n$; the lemma therefore also applies to algorithms for errors $D_{\Z,\alpha'q}$, the convention of \cite{BKSW18}. This observation is standard \cite{Pei10}.

The quantum step of that iteration is useful on its own, because it converts a dual-lattice decoder satisfying the stated radius, success and coherent-implementation requirements into a Gaussian sampler for the primal lattice.

\begin{lemma}[Regev's quantum step {\cite[Lem. 3.14]{Reg09}}]\label{lem:qstep}
There is an efficient quantum algorithm that, given an $n$-dimensional lattice $\Lat$, a number $d<\lambda_1(\Lat^*)/2$ and an oracle that solves $\CVP_{\Lat^*,d}$, that is, finds the closest point of $\Lat^*$ to any target within distance $d$ of $\Lat^*$, outputs a sample from $D_{\Lat,\sqrt n/(\sqrt2\,d)}$, up to statistical distance $2^{-\Omega(n)}$.
\end{lemma}

The algorithm has four steps \cite{Reg09}: (1) prepare the Gaussian superposition $\sum_{y}\rho_{d/\sqrt n}(y)\ket{y}$ over a fine grid, truncated to $\|y\|<d$; (2) compute $t=y\bmod\mathcal P(\Lat^*)$, the reduction modulo a fundamental parallelepiped of $\Lat^*$, in a second register; (3) use the oracle once, on the second register, to find the lattice point closest to $t$, hence $y$, and erase the first register, possible because $d<\lambda_1(\Lat^*)/2$ makes $y$ a function of $t$; (4) apply the quantum Fourier transform to the remaining register, which maps a periodic Gaussian state on $\Lat^*$ to a Gaussian superposition on $\Lat$ of reciprocal width, and measure. In Regev's reduction the decoder is built from the LWE oracle and from wider Gaussian samples; under $\HReg(f)$ a worst-case decoder is available directly (\cref{lem:lm09}), which is what makes \cref{prop:sivp-reg,prop:sis-reg} below sharper than a detour through LWE. Two technicalities arise when the lemma is used with the decoders available to us: they are randomised and may fail, and in one application they succeed only on average over the target. The next lemma handles both.

\begin{lemma}\label{lem:qstep-imperfect}
Let $\Lat\subset\R^n$ be a full-rank lattice, $0<d<\lambda_1(\Lat^*)/2$, and let $\mathcal D$ be a quantum algorithm that, on input $t\in\R^n$, outputs a point of $\Lat^*$. Let $y$ be a continuous Gaussian vector of parameter $d/\sqrt{2n}$ conditioned on $\|y\|<d$, let $t:=y\bmod\mathcal P(\Lat^*)$, and let
\[
  \bar p:=\Pr\big[\mathcal D(t)=t-y\big],
\]
the probability being over $y$ and the internal randomness of $\mathcal D$; note that $t-y$ is the point of $\Lat^*$ closest to $t$. Run the algorithm of \cref{lem:qstep} with step (3) implemented by running $\mathcal D$ reversibly, using its output to erase the first register, and uncomputing $\mathcal D$, and let $w:=\sqrt n/(\sqrt2\,d)$. Then:
\begin{enumerate}[label=(\alph*)]
  \item the output distribution is within statistical distance $\sqrt{2(1-\bar p)}+2^{-\Omega(n)}$ of $D_{\Lat,w}$;
  \item for every set $G$ of outputs with $\delta_G:=1-D_{\Lat,w}(G)+2^{-\Omega(n)}\leq\bar p^{\,2}$, the output lies in $G$ with probability at least $(\bar p-\sqrt{\delta_G})^2$.
\end{enumerate}
In particular, if $\mathcal D$ errs with probability at most $\varepsilon$ on every target within distance $d$ of $\Lat^*$, the output is within statistical distance $\sqrt{2\varepsilon}+2^{-\Omega(n)}$ of $D_{\Lat,w}$.
\end{lemma}

\begin{proof}
After step (2) the state is $\sum_y c_y\ket{y}\ket{t(y)}$, where $|c_y|^2\propto\rho_{d/\sqrt n}(y)^2=\rho_{d/\sqrt{2n}}(y)$ on the ball $\|y\|<d$, which is the distribution of $y$ in the statement; since $d<\lambda_1(\Lat^*)/2$, the map $y\mapsto t(y)$ is injective on the ball. Let $A$ be a unitary implementation of $\mathcal D$ that leaves its input register unchanged, $A\ket{t}\ket{0}=\ket{t}\sum_x\ket{x}\ket{\gamma_{t,x}}$, where $\|\gamma_{t,x}\|^2$ is the probability that $\mathcal D$ outputs $x$ on input $t$, and let $S$ be the permutation $\ket{y}\ket{t}\ket{x}\mapsto\ket{y-t+x}\ket{t}\ket{x}$. Step (3) applies $A^\dagger SA$ (Bennett's compute--use--uncompute \cite{Ben73}); with an exact oracle it maps $\ket{y}\ket{t(y)}\ket{0}$ to $\ket{0}\ket{t(y)}\ket{0}$. For each $y$, writing $t=t(y)$,
\[
  \bra{0}\bra{t}\bra{0}\,A^\dagger SA\,\ket{y}\ket{t}\ket{0}
  =\sum_{x,x'}\langle 0|y-t+x\rangle\,\langle x'|x\rangle\,\langle\gamma_{t,x'}|\gamma_{t,x}\rangle
  =\|\gamma_{t,t-y}\|^2=\Pr[\mathcal D(t)=t-y],
\]
and terms with distinct $y$ are orthogonal because the second register is never changed. Hence the actual state $\ket\Psi$ and the ideal state $\ket{\Psi_0}$ after step (3) satisfy $\langle\Psi_0|\Psi\rangle=\sum_y|c_y|^2\Pr[\mathcal D(t(y))=t(y)-y]=\bar p$, and the same holds after the unitary step (4). For (a): the outcomes of one measurement on two pure states are within statistical distance at most their trace distance $\sqrt{1-|\langle\Psi_0|\Psi\rangle|^2}\leq\sqrt{2(1-\bar p)}$, and the ideal outcome is $2^{-\Omega(n)}$-close to $D_{\Lat,w}$ by \cref{lem:qstep}. For (b): let $P_G$ project onto the outcomes in $G$; then
\[
  \|P_G\Psi\|\;\geq\;|\langle\Psi_0|P_G|\Psi\rangle|\;\geq\;|\langle\Psi_0|\Psi\rangle|-|\langle\Psi_0|(1-P_G)|\Psi\rangle|\;\geq\;\bar p-\|(1-P_G)\Psi_0\|\;\geq\;\bar p-\sqrt{\delta_G},
\]
since $\|(1-P_G)\Psi_0\|^2$ is the probability that the ideal procedure outputs a point outside $G$. The last claim is (a) with $\bar p\geq1-\varepsilon$, because every $t(y)$ lies within distance $\|y\|<d$ of $\Lat^*$.
\end{proof}

\begin{lemma}[Randomised rounding]\label{lem:round}
Let $\varepsilon\in(0,1/2]$, $r\geq\eta_\varepsilon(\Z)$ and $s_c>0$. Let $e$ be a continuous Gaussian variable of parameter $s_c$ and, given $e$, let $y\leftarrow D_{\Z,r,e}$ be a discrete Gaussian sample over $\Z$ centred at $e$. Then $y$ is within statistical distance $2\varepsilon$ of $D_{\Z,\sqrt{s_c^2+r^2}}$. Consequently, for $a\in\Z$, rounding $a+e$ randomly in this way yields $a+y$.
\end{lemma}

\begin{proof}
For $r\geq\eta_\varepsilon(\Z)$ and every $c\in\R$, Poisson summation gives $\rho_r(\Z+c)\in[1-\varepsilon,1+\varepsilon]\cdot r$ \cite{MR07}. Hence, with $s:=\sqrt{s_c^2+r^2}$ and the convolution identity $\int_\R\rho_{s_c}(e)\rho_r(y-e)\,de=(s_cr/s)\,\rho_s(y)$,
\[
  \Pr[y]=\int_{\R}\frac{\rho_{s_c}(e)}{s_c}\cdot\frac{\rho_r(y-e)}{\rho_r(\Z-e)}\,de\;\in\;\Big[\frac1{1+\varepsilon},\frac1{1-\varepsilon}\Big]\cdot\frac{\rho_s(y)}{s}.
\]
Since $s\geq r\geq\eta_\varepsilon(\Z)$, also $\rho_s(\Z)\in[1-\varepsilon,1+\varepsilon]\cdot s$, so the ratio of $\Pr[y]$ to $D_{\Z,s}(y)=\rho_s(y)/\rho_s(\Z)$ lies in $[(1-\varepsilon)/(1+\varepsilon),(1+\varepsilon)/(1-\varepsilon)]$, and the statistical distance is at most $\varepsilon/(1-\varepsilon)\leq2\varepsilon$. The last statement holds because $D_{\Z,r,a+e}=a+D_{\Z,r,e}$ for $a\in\Z$. This is the one-dimensional, continuous case of the convolution theorem of \cite{Pei10}.
\end{proof}

\subsection{The reach theorem}\label{sec:reach-thm}

The thresholds that govern the consequences are
\begin{gather*}
  \gamma_{\mathrm g}(d):=O\big(\gu(d)\sqrt{d/\log d}\big)=\Otil(d^{1+2f}),\qquad
  \alpha_{\mathrm{Reg}}(m):=\frac{1}{8\sqrt m\;\gu(m+1)},\\
  \gs_{\mathrm{Reg}}(n):=\Big(1+\frac1n\Big)\sqrt2\,n\,\gu(n+1)=O\big(n^{3/2+2f}\big),\\
  \beta_{\mathrm{Reg}}(m):=4\sqrt2\,m\,\gu(m+1)=O\big(m^{3/2+2f}\big),\qquad
  g_{\mathrm E}(d):=2d\,\gs_{\mathrm E}(d,\ell),\\
  \alpha_{\mathrm E}(n,m,\ell):=\frac{1}{c\,r_0\,\kappa\,\ell\sqrt m\,q^{n/m}}\quad(m\in\mathcal M_q(n)),\qquad
  \alpha^*_{\mathrm E}(n,\ell):=\alpha_{\mathrm E}(n,m_{\max},\ell),\\
  \beta_{\mathrm E}(n,m,\ell):=\max\Big\{\frac{\sqrt m}{\sqrt2\,\alpha^*_{\mathrm E}(n,\ell)},\;4\sqrt2\,m\Big\},
\end{gather*}
where $c=6\sqrt{2\pi e}$ is the constant of \cite[Thm. 3]{BKSW18} and $r_0$ the constant of $\HEDCP[\ell]$. The set $\mathcal M_q(n)$ contains the numbers of samples for which the proof of \cite[Thm. 3]{BKSW18} applies:
\begin{enumerate}[label=(B\arabic*)]
  \item $m=\Omega(\kappa)$;
  \item $q^{(m-n)/m}\geq2\sqrt{2\pi e}$, so that the lower bound on $\lambda_1(\Lambda_q(A))$ of \cite[Lem. 6]{BKSW18} holds except with probability $(3/4)^m$ (\cref{lem:qary-count});
  \item $\sqrt m\,q^{(m-n)/m}/(3\sqrt{2\pi e})\leq q$, the radius condition of \cite[App. B.1]{BKSW18}, which keeps the balls of the reduction below the minimum of $\Lambda_q(A)$, itself at most $q$;
\end{enumerate}
and $m_{\max}=m_{\max}(n,q):=\max\mathcal M_q(n)$. Condition (B2) is a mild lower bound on $m$, $m\geq n\ln q/\ln\big(q/(2\sqrt{2\pi e})\big)$. Condition (B3) is an upper bound on $m$: in the regime $m=\Theta(n\log q)$ it fails, since $\sqrt m\,q^{-n/m}$ grows with $n$, so the width bound of \cite[Thm. 3]{BKSW18} cannot be used there. For ML-KEM's modulus $q=3329$ and $n=768$, for instance, $\mathcal M_q(n)$ consists of roughly $1.35n\leq m\leq5n$. The parameter $\kappa$ is the statistical parameter of \cite{BKSW18}: Brakerski, Kirshanova, Stehlé and Wen treat $n$, $q$ and $m$ as functions of $\kappa$ and use it to truncate Gaussian amplitudes to $|j|\leq\sqrt\kappa\,r$ and to bound failure probabilities by $2^{-\Omega(\kappa)}$. We fix $\kappa:=\lceil\log^2n\rceil$, which makes those probabilities negligible; soft-$O$ notation hides powers of $\log n$ and $\log q$, including $\kappa$. Finally, $\gs_{\mathrm E}(n,\ell)$ denotes the approximation factor $\Otil(n/\alpha)$ that \cref{lem:regev09} yields at $\alpha=\alpha^*_{\mathrm E}(n,\ell)/2$ (\cref{cor:edcp-sivp}), and $g_{\mathrm E}(d)$ is the gap above which the worst-case route of $\HEDCP[\ell]$ decodes $d$-dimensional lattices. For $n\leq q\leq\poly(n)$ (\cref{cor:edcp-ceiling}), $\mathcal M_q(n)\neq\emptyset$ and $m_{\max}=\Theta(n)$, so that $\alpha^*_{\mathrm E}=1/\Otil(\ell n)$, $\gs_{\mathrm E}=\Otil(\ell n^2)$, $g_{\mathrm E}(d)=\Otil(\ell d^3)$ and, for $m=\Theta(n\log q)$, $\beta_{\mathrm E}=\Otil(\ell n^{3/2})$.

\begin{theorem}\label{thm:reach}\label{prop:reach}
Let $f>0$ and let $\ell$ be as in $\HEDCP[\ell]$.
\begin{enumerate}[label=(\roman*)]
  \item\label{it:usvp} Under $\HReg(f)$, $\uSVP_{\gamma}$ is in $\BQP$ for every $\gamma\geq\gu(d)$, in the worst case over all $d$-dimensional lattices, including ideal and module lattices.
  \item\label{it:bdd} Under $\HReg(f)$, $\BDD_\delta$ is in $\BQP$ for $\delta\leq1/(2\gu(d+1))$, and $\GapSVP_{\gamma_{\mathrm g}(d)}$ is in $\BQP$, in the worst case.
  \item\label{it:lwe} Under $\HReg(f)$, search-$\LWE_{n,q,\alpha,m}$ with $m\geq 2n\lg q$ and $\alpha\leq\alpha_{\mathrm{Reg}}(m)$ is in $\BQP$ for all but a $q^{-n}$ fraction of $A$; for $m=\Theta(n\log q)$ this covers every $\alpha\leq1/\Otil(n^{1+2f})$. More generally, every LWE-type instance (unstructured, ring, module, middle-product or rounding) whose embedding lattice is a $\gu$-unique-SVP instance, with the planted vector as its shortest vector, is solved. Decision-LWE is then solved directly.
  \item\label{it:sivp} Under $\HReg(f)$, $\SIVP_{\gs_{\mathrm{Reg}}(n)}$ is in $\BQP$ in the worst case (\cref{prop:sivp-reg}).
  \item\label{it:sis} Under $\HReg(f)$, $\SIS_{n,m,q,\beta}$ with $m\geq2n\lg q$ is in $\BQP$ whenever $\beta\geq\beta_{\mathrm{Reg}}(m)$, for all but a $q^{-n}$ fraction of $A$ (\cref{prop:sis-reg}). For $\beta\cdot\Otil(\sqrt n)\geq\gs_{\mathrm{Reg}}(n)$, that is for $\beta\geq\Otil(n^{1+2f})$, the cited worst-case hardness premise of \cite{MR07,MP13} becomes unavailable, under the reduction's parameter conditions.
  \item\label{it:edcp} Under $\HEDCP[\ell]$, let $q\leq\poly(n)$. Search-$\LWE_{n,q,\alpha,m}$ with $m\in\mathcal M_q(n)$ is solved with probability $1/\poly(n)$ whenever $\alpha<\alpha_{\mathrm E}(n,m,\ell)$, with largest certified bound $\alpha<\alpha^*_{\mathrm E}(n,\ell)=1/\Otil(\ell n)$ (\cref{prop:edcp-lwe}); $\SIVP_{\gs_{\mathrm E}}$ and $\GapSVP_{\gs_{\mathrm E}}$ are in $\BQP$, and so is $\BDD$ on $d$-dimensional lattices with gap at least $g_{\mathrm E}(d)$, in the worst case (\cref{cor:edcp-sivp}); and, if $\mathcal M_q(n)\neq\emptyset$, $\SIS_{n,m,q,\beta}$ with uniform $A$, $m\geq2n\lg q$ and $q>\sqrt2\,\beta_{\mathrm E}(n,m,\ell)$ is solved with probability $1/\poly(n)$ for every $\beta>\beta_{\mathrm E}(n,m,\ell)$ (\cref{prop:sis-edcp}). Since $\ell=\Omega(n\log q)$ for every algorithm at width $r_0$ (\cref{prop:ell-lower}), the restricted ball composition certifies noise scales at most $\alpha\approx n^{-2}$ and yields factors $\gamma\approx n^3$ and bounds $\beta\approx n^{5/2}$, up to polylogarithmic factors (\cref{cor:edcp-ceiling}).
  \item\label{it:hs} Under $\HHS$, the vectorization problem for every free and transitive action $\star$ of a cyclic group $\langle\mathfrak g\rangle\cong\Z_N$ on a set $X$ is solved, provided $N$ and $\mathfrak g$ are known and $x\mapsto\mathfrak g^x\star E$ is computable in quantum polynomial time for arbitrary $x\in\Z_N$. All required evaluation inputs must be public or efficiently recoverable, and the evaluation must use consistent encodings and polynomial total cost. Under $\HHS^{\mathrm{ab}}$ the same holds for free and transitive abelian group actions satisfying these conditions.
\end{enumerate}
\end{theorem}

\begin{proof}
\ref{it:usvp} is \cite[Thm. 1.1]{Reg04} applied to the algorithm postulated by $\HReg(f)$: the reduction of \cref{sec:regev02} generates $\DCP_N$ registers that are bad with probability at most $1/(\log N)^f$, which is the tolerance assumed, and it makes no use of algebraic structure in the input lattice. Keeping the shortest of the candidates returned by repeated runs amplifies the success probability.

\ref{it:bdd}: compose \ref{it:usvp} with \cref{lem:lm09}; the embedding used for $\BDD$ raises the dimension by one.

\ref{it:lwe}: let $b=As+e$ with $e\leftarrow D_{\Z^m,\alpha q}$. By \cref{lem:qary}, $\lambda_1(\Lambda_q(A))\geq q/4$ for all but a $q^{-n}$ fraction of $A$, and by \cref{lem:tails}(a), $\mathrm{dist}(b,\Lambda_q(A))\leq\|e\|\leq\alpha q\sqrt m$ except with probability $2^{-m}$. If $\alpha\leq\alpha_{\mathrm{Reg}}(m)$, then $\alpha q\sqrt m\leq(q/4)/(2\gu(m+1))\leq\lambda_1(\Lambda_q(A))/(2\gu(m+1))$, so $b$ is a $\BDD_{1/(2\gu(m+1))}$ instance, solved by \ref{it:bdd}. A candidate lattice point $v$ is accepted if $\|b-v\|\leq\alpha q\sqrt m$; since $2\alpha q\sqrt m<q/4$, only the closest point passes, which allows amplification. The closest point is $As\bmod q$, from which $s$ follows by linear algebra, since $A$ has full column rank except with probability at most $q^{n-m}$. For $m=\Theta(n\log q)$ and sufficiently large dimension, the bound $\gu(m+1)=O(m^{1/2+2f})$ gives the sufficient noise scale $1/\Otil(n^{1+2f})$, with constants depending on the fixed reduction and witness. For decision-LWE, a uniform $b$ lies within distance $\alpha q\sqrt m$ of $\Lambda_q(A)$ only with negligible probability, by a counting argument, so no candidate passes. The general statement is \ref{it:usvp} applied to the Kannan embedding of the instance; it requires a bound on the gap, which for structured instances must be verified separately since \cref{lem:qary} concerns uniform $A$.

\ref{it:sivp}, \ref{it:sis} and \ref{it:edcp} are proved in \cref{sec:proofs}: \ref{it:sivp} is \cref{prop:sivp-reg}; the first claim of \ref{it:sis} is \cref{prop:sis-reg}, and the second holds because \cite{MR07,MP13} base $\SIS_\beta$ on $\SIVP_{\beta\cdot\Otil(\sqrt n)}$; \ref{it:edcp} combines \cref{prop:edcp-lwe,cor:edcp-sivp,prop:sis-edcp,prop:ell-lower,cor:edcp-ceiling}.

\ref{it:hs}: the functions $f_0(x)=\mathfrak g^x\star E$ and $f_1(x)=\mathfrak g^x\star E'$ with $E'=\mathfrak g^{a}\star E$ are injective (freeness) and satisfy $f_1(x)=f_0(x+a)$; apply \cref{lem:hs-dcp}. For $\HHS^{\mathrm{ab}}$, decompose the group into cyclic factors and apply the generalised hypothesis to the resulting abelian hidden shift.
\end{proof}

\begin{corollary}
\label{cor:hs-lattice}
Under $\HHS$, there is a constant $f_{\mathcal A}>0$ such
that all conclusions of \cref{thm:reach} under
$\HReg(f_{\mathcal A})$ hold, with their stated sample,
distribution and promise conditions. In particular,
unique-SVP, GapSVP and SIVP admit quantum polynomial-time
algorithms at some polynomial approximation factors.
The SIS and LWE conclusions hold at the corresponding
solver-dependent bounds.
\end{corollary}

\begin{proof}
Apply \cref{lem:hs-rare}, then the fixed-parameter statements
of \cref{thm:reach}.
\end{proof}

\subsection{Proofs of the compositions}\label{sec:proofs}

We give explicit derivations of items \ref{it:sivp}--\ref{it:edcp}
of \cref{thm:reach} from the results recalled in \cref{sec:technical}.
The reduction mechanisms have antecedents: in particular,
\cite[Sec. 1.2]{LM09} already records the quantum reduction
$\SIVP_{O(n\gamma)}\leq_Q\BDD_{1/\gamma}$, and the average-case
decoder route to SIS follows \cite{SSTX09}. The propositions below
specify the resulting bounds under our solver hypotheses, with the
dimension shifts, matrix promises, Gaussian conventions, preparation
precision and success probabilities needed for composition.
\Cref{lem:qstep-imperfect} supplies the coherent implementation and
error analysis for the decoders used here. The information-theoretic bound of \cref{prop:ell-lower} is due to an
anonymous reviewer.

\begin{proposition}[$\SIVP$ under $\HReg(f)$]\label{prop:sivp-reg}
Assume $\HReg(f)$. There is a quantum polynomial-time algorithm that, given a basis of a full-rank lattice $\Lat\subset\R^n$, outputs $n$ linearly independent vectors of $\Lat$ of length at most $\gs_{\mathrm{Reg}}(n)\,\lambda_n(\Lat)$, except with probability $2^{-\Omega(n)}$.
\end{proposition}

\begin{proof}
\emph{A reliable decoder.} By \cref{thm:reach}\ref{it:bdd} there is a quantum algorithm $\mathcal B$ that solves $\BDD_{1/(2\gu(n+1))}$ on $n$-dimensional lattices with success probability at least $1/p(n)$, for a polynomial $p$. For $d>0$, let $\mathcal B_d$ run $\mathcal B$ on its target $t$ (for the lattice $\Lat^*$) $n\,p(n)$ times and return the first candidate $x\in\Lat^*$ with $\|t-x\|\leq d$, or $0$ if there is none. If $d\leq\lambda_1(\Lat^*)/(2\gu(n+1))$ and $\mathrm{dist}(t,\Lat^*)\leq d$, each run returns the closest point with probability at least $1/p(n)$, and, since $d<\lambda_1(\Lat^*)/2$, the closest point is the only point of $\Lat^*$ within distance $d$; so $\mathcal B_d$ errs with probability at most $(1-1/p(n))^{np(n)}\leq e^{-n}$.

\emph{Sampling.} For such $d$, \cref{lem:qstep-imperfect} with $\mathcal D=\mathcal B_d$ yields samples within statistical distance $2^{-\Omega(n)}$ of $D_{\Lat,w}$, $w=\sqrt n/(\sqrt2\,d)$.

\emph{Guessing $d$.} The LLL algorithm \cite{LLL82} applied to $\Lat^*$ returns $v\in\Lat^*\setminus\{0\}$ with $\lambda_1(\Lat^*)\leq\|v\|\leq2^{(n-1)/2}\lambda_1(\Lat^*)$. For $i=0,1,\dots,n^2$ let $d_i:=\|v\|(1+1/n)^{-i}/(2\gu(n+1))$. Since $d_0\geq\lambda_1(\Lat^*)/(2\gu(n+1))>d_{n^2}$ and consecutive values differ by the factor $1+1/n$, some index $i^*$ satisfies
\[
  \frac{\lambda_1(\Lat^*)}{(1+1/n)\,2\gu(n+1)}<d_{i^*}\leq\frac{\lambda_1(\Lat^*)}{2\gu(n+1)}.
\]
For each $i$, draw $40n$ samples with $d=d_i$, discard those that are zero or not in $\Lat$ (membership is a linear-algebra test), and pool the rest. Output $n$ linearly independent vectors of the pool, chosen greedily by increasing length; by the exchange property of linear independence, their maximum length is the smallest possible among linearly independent $n$-subsets of the pool.

\emph{Analysis at $i^*$.} Let $w^*:=\sqrt n/(\sqrt2\,d_{i^*})$. A nonzero dual vector has a nonzero integer inner product with one of $n$ linearly independent lattice vectors of length at most $\lambda_n(\Lat)$, so $\lambda_1(\Lat^*)\lambda_n(\Lat)\geq1$ and
\[
  w^*\sqrt n\;<\;\frac{(1+1/n)\sqrt2\,n\,\gu(n+1)}{\lambda_1(\Lat^*)}\;\leq\;\gs_{\mathrm{Reg}}(n)\,\lambda_n(\Lat).
\]
On the other hand $w^*\geq\sqrt{2n}\,\gu(n+1)/\lambda_1(\Lat^*)\geq\sqrt2\,\eta_{2^{-n}}(\Lat)$ by \cref{lem:tails}(c). By \cref{lem:tails}(a,b), a sample of $D_{\Lat,w^*}$ has length at most $w^*\sqrt n$ except with probability $2^{-n}$, and lies outside any fixed hyperplane with probability at least $1/10$. While the samples drawn so far span a proper subspace, each new sample leaves it with probability at least $1/10-2^{-\Omega(n)}\geq1/20$; by a Chernoff bound, $40n$ samples contain $n$ linearly independent ones except with probability $2^{-\Omega(n)}$, and all of them have length at most $w^*\sqrt n$ except with probability $40n\,2^{-n}$. The greedy choice then has maximum length at most $\gs_{\mathrm{Reg}}(n)\lambda_n(\Lat)$.
\end{proof}

\begin{proposition}[SIS under $\HReg(f)$]\label{prop:sis-reg}
Assume $\HReg(f)$, let $q$ be prime and $m\geq2n\lg q$. There is a quantum polynomial-time algorithm that, given $A\in\Zq^{n\times m}$ with $\lambda_1^\infty(\Lambda_q(A^{\top}))\geq q/4$, which holds for all but a $q^{-n}$ fraction of $A$ (\cref{lem:qary}), outputs a nonzero $z\in\Lambda_q^\perp(A)$ with $\|z\|\leq\beta_{\mathrm{Reg}}(m)$, except with probability $2^{-\Omega(m)}$.
\end{proposition}

\begin{proof}
Let $\Lat:=\Lambda_q^\perp(A)$, of rank $m$; its dual is $\Lat^*=q^{-1}\Lambda_q(A^{\top})$, so $\lambda_1(\Lat^*)\geq\lambda_1^\infty(\Lat^*)\geq1/4$. Set $d:=1/(8\gu(m+1))$, so that $d\leq\lambda_1(\Lat^*)/(2\gu(m+1))$ and $d<\lambda_1(\Lat^*)/2$. With the decoder $\mathcal B_d$ of the previous proof, now for the $m$-dimensional lattice $\Lat^*$, \cref{lem:qstep-imperfect} yields samples within statistical distance $2^{-\Omega(m)}$ of $D_{\Lat,w}$ with $w=\sqrt m/(\sqrt2\,d)=4\sqrt{2m}\,\gu(m+1)$. By \cref{lem:tails}(a), a sample has length at most $w\sqrt m=\beta_{\mathrm{Reg}}(m)$ except with probability $2^{-m}$. It is zero with probability $1/\rho_w(\Lat)$, and by the Poisson summation formula
\[
  \rho_w(\Lat)=\det(\Lat^*)\,w^m\rho_{1/w}(\Lat^*)\;\geq\;\frac{w^m}{\det(\Lat)}\;\geq\;\frac{w^m}{q^n},
\]
because every term of $\rho_{1/w}(\Lat^*)$ is positive and the term at $0$ equals $1$, and $\det(\Lat)=[\Z^m:\Lat]\leq q^n$. Since $w\geq4$ and $q^n\leq2^{m/2}$, a sample is zero with probability at most $2^{-3m/2}$. Output the first sample that is nonzero and has length at most $\beta_{\mathrm{Reg}}(m)$.
\end{proof}

\begin{proposition}[LWE under $\HEDCP{[\ell]}$]\label{prop:edcp-lwe}
Assume $\HEDCP[\ell]$, let $q\leq\poly(n)$ be prime and $m\in\mathcal M_q(n)$. If $\alpha<\alpha_{\mathrm E}(n,m,\ell)$, there is a quantum polynomial-time algorithm that, given $m$ samples of $\LWE_{n,q,\alpha}$ with a uniform matrix, any secret and errors $D_{\Z,\alpha q}$, recovers the secret with probability at least $1/\poly(n)$. Let $p_{\mathcal A}(n)$ be an inverse-polynomial lower bound on the EDCP solver's success at the parameters used.

Given $m'\geq2n\lg q$ independent further samples
with the same secret, and
$\alpha<1/(8\sqrt{m'})$, the recovery procedure can
be repeated and checked against these further
samples. For all but a
\[
 \delta_{\rm mat}
   =q^{-n}+q^{n-m'}+(3/4)^m
\]
fraction of pairs of matrices, $T$ independent
repetitions recover the secret with probability at
least
\[
 1-2^{-m}-2^{-m'}
   -\exp\!\bigl(-T p_{\mathcal A}(n)/8\bigr).
\]
The probability here is over the two error vectors
and the internal randomness, with the matrices fixed.
Over all matrices as well, subtract
$\delta_{\rm mat}$ from the displayed lower bound.
\end{proposition}

\begin{proof}
Since $\alpha<\alpha_{\mathrm E}$, fix $r\in[r_0,r^*)$ with $r^*=1/(c\sqrt m\,\kappa\ell\alpha q^{n/m})>r_0$. Conditions (B1)--(B3) supply the dimension, first-minimum and radius premises used in the ball construction of \cite[Thm. 3 and App. B.1]{BKSW18}. The good preparation matrix event also requires full column rank. The unrounded bound in \cref{lem:qary-count} is $(3/4)^m/\sqrt{\pi m}$; rank failure costs at most $q^{n-m}$. By (B2), $q^{n-m}\leq[2\sqrt{2\pi e}]^{-m}$, and their sum is at most $(3/4)^m$. The error-norm condition $\|e\|\leq\sqrt m\,\alpha q$ fails with probability at most $2^{-m}$. For such fixed good samples, the ideal reduction prepares $\ell$ states with the input secret, and all are valid with probability at least $(1-1/\ell)^\ell\geq1/4$ for $\ell\geq2$. The asymptotic lower bound in \cref{prop:ell-lower} permits this restriction; bounded exceptional dimensions can be handled separately.

\paragraph{Preparation precision.}
Fix the input samples satisfying the matrix and
error-norm conditions. In the ideal preparation,
the branch on which all $\ell$ states are valid
has probability at least $1/4$. Hence the ideal
composed procedure recovers the secret with
probability at least $p_{\mathcal A}(n)/4$.

Choose the ball-state preparation error per use
at most $p_{\mathcal A}(n)/(16\ell)$, and the total
Gaussian-truncation and circuit-implementation error
at most $p_{\mathcal A}(n)/16$. The hybrid bound and
contractivity under the complete measurement
instrument give distance at most
$p_{\mathcal A}(n)/8$ between the actual and ideal
complete procedures. The actual recovery
probability is therefore at least
$p_{\mathcal A}(n)/8$.

This comparison includes the measurement outcomes
and bad branches; it does not condition on a rare
postselection event. The arbitrary inverse-polynomial
accuracy in the ball-preparation lemma permits the
required choice. Since $\ell$ and
$p_{\mathcal A}^{-1}$ are polynomially bounded,
a sufficiently large fixed precision exponent
suffices. With $\kappa=\lceil\log^2 n\rceil$,
the Gaussian truncation error is eventually smaller
than the allotted inverse-polynomial budget.

If a preparation circuit uses $K$ approximate ball-state primitives in total, allocate $p_{\mathcal A}/(16K)$ to each; here $K$ is polynomially bounded. This is the same hybrid budget. The available arbitrary fixed inverse-polynomial accuracy \cite[App. B.1, Lem. 14]{BKSW18} suffices; no uniform polynomial dependence on $\log(1/\varepsilon)$ for arbitrarily small $\varepsilon$ is assumed.

For every fixed inverse-polynomial precision target,
the reduction runs in time polynomial in the total
input length and in $n,m,\log q,\ell,\kappa$.
The polynomial degree may depend on the fixed
precision exponent and on the postulated solver.
Our choice $\kappa=\lceil\log^2n\rceil$ controls
statistical truncation errors; it is not a bound
on the length of the input.

\paragraph{Verification and amplification.}
Reserve independent samples $(A',b')$ and accept a candidate $\hat s$ when the centred residual $b'-A'\hat s\bmod q$ has norm at most $\sqrt{m'}\alpha q$. Exclude the $q^{-n}+q^{n-m'}$ fraction for which the bound $\lambda_1(\Lambda_q(A'))\geq q/4$ or full column rank fails. For the remaining matrices and bounded errors, a wrong secret differs by a nonzero lattice vector and has residual norm at least $q/4-\sqrt{m'}\alpha q>\sqrt{m'}\alpha q$. Only the correct secret passes.

The error vectors are sampled once. The probability
that either exceeds its stated norm bound is at most
$2^{-m}+2^{-m'}$. Condition on both bounds and on
the good matrices. The residual test rejects every
incorrect candidate and accepts the secret.
Each repetition uses fresh internal randomness and
has conditional success at least
$p_{\mathcal A}(n)/8$. Thus the probability that
all $T$ repetitions fail is at most
$(1-p_{\mathcal A}(n)/8)^T$, bounded by the
displayed exponential. Finally remove the
conditioning by the union bound.

Taking $T=\lceil8(n+2)\ln2/p_{\mathcal A}(n)\rceil$ bounds internal failure by $2^{-n-2}$. Together with the fixed-error tails this gives $1-2^{-\Omega(n)}$ success on good matrices. Removing the matrix conditioning subtracts $\delta_{\rm mat}$ as stated.
\end{proof}

The factor $\ell$ in $\alpha_{\mathrm E}$ is easy to overlook, because it is absent from the qualitative statement ``LWE and EDCP are equivalent''. It cannot be absorbed into a constant: the next proposition gives a lower bound of order $n\log N$ at fixed width and inverse-polynomial success probability.

\begin{proposition}\label{prop:ell-lower}
Let $n,N,\ell\geq1$ and $r>0$. Any algorithm, of any running time, that receives $\ell$ states of $\mathrm{G}\text{-}\EDCP_{n,N,r}$ with a uniformly random secret $s\in\Z_N^n$ and independent uniformly random offsets outputs $s$ with probability at most
\[
  \frac{(1+r)^{2\ell}}{N^n}.
\]
Consequently, success probability $p$ requires $\ell\geq(n\log N-\log(1/p))/(2\log(1+r))$. In particular, an algorithm as in $\HEDCP[\ell]$, run at width $r_0$ with success probability $p\geq1/\poly(n)$, requires $\ell\geq(1-o(1))\,n\log N/(2\log(1+r_0))=\Omega(n\log N)$.
\end{proposition}

\begin{proof}
A $\mathrm{G}\text{-}\EDCP$ state is $\ket{\psi_{x,s}}=\sum_{j\in\Z}a_j\ket{j}\ket{x+js}$ with $a_j=\rho_r(j)/\sqrt{\rho_{r/\sqrt2}(\Z)}>0$ and $\sum_ja_j^2=1$. Let $S:=\sum_ja_j=\rho_r(\Z)/\sqrt{\rho_{r/\sqrt2}(\Z)}$. Since $\rho_r$ is decreasing on $[0,\infty)$, $\rho_r(\Z)\leq1+2\int_0^\infty\rho_r=1+r$, and $\rho_{r/\sqrt2}(\Z)\geq1$, so $S\leq1+r$. For vectors $v_j$ and positive weights $w_j$, the Cauchy--Schwarz inequality $|\sum_j\langle u,v_j\rangle|^2\leq(\sum_jw_j)\sum_j|\langle u,v_j\rangle|^2/w_j$ gives the operator inequality $(\sum_jv_j)(\sum_jv_j)^{\dagger}\leq(\sum_jw_j)\sum_jw_j^{-1}v_jv_j^{\dagger}$. With $v_j=a_j\ket{j}\ket{x+js}$ and $w_j=a_j$ it yields $\ket{\psi_{x,s}}\bra{\psi_{x,s}}\leq S\sum_ja_j\ket{j}\bra{j}\otimes\ket{x+js}\bra{x+js}$. Averaging over the uniform offset $x$, and using $\sum_x\ket{x+js}\bra{x+js}=I$ for every $j$ and $s$, the state $\sigma_s$ of one register satisfies
\[
  \sigma_s\;\leq\;S^2\tau,\qquad\tau:=\Big(\sum_j\frac{a_j}{S}\ket{j}\bra{j}\Big)\otimes N^{-n}I,
\]
where $\tau$ is a density operator that does not depend on $s$. Hence $\sigma_s^{\otimes\ell}\leq S^{2\ell}\tau^{\otimes\ell}$, and for every measurement $\{M_s\}$ the probability of a correct guess is $N^{-n}\sum_s\mathrm{Tr}(M_s\sigma_s^{\otimes\ell})\leq S^{2\ell}N^{-n}\sum_s\mathrm{Tr}(M_s\tau^{\otimes\ell})=S^{2\ell}N^{-n}$. The consequences follow by taking logarithms; an algorithm as in $\HEDCP[\ell]$ succeeds with probability $p\geq1/\poly(n)$ for every secret and arbitrary offsets, hence for uniform ones, at width $r_0$ in particular. This argument was suggested by an anonymous reviewer; it replaces a truncation argument of an earlier draft and removes a factor $\log\log n$. Its form, dominating the averaged state by a secret-independent state, is standard, and we do not claim priority for the bound.
\end{proof}

\begin{corollary}\label{cor:edcp-ceiling}
Assume $\HEDCP[\ell]$ and let $q\leq\poly(n)$. For every $m\geq n$ satisfying the radius condition (B3),
\[
  \alpha_{\mathrm E}(n,m,\ell)\;\leq\;\frac{1}{2\,r_0\,\kappa\,\ell\,m}\;\leq\;\frac{(1+o(1))\log(1+r_0)}{r_0\,\kappa\,n^2\log q}\;=\;\frac{1}{\widetilde\Omega(n^2)}.
\]
If moreover $q\geq n$ and $n$ is large, then $\mathcal M_q(n)\neq\emptyset$, $m_{\max}\in\mathcal M_q(n)$, the first bound is attained up to a factor $1+o(1)$ at $m=m_{\max}$, so that $\alpha^*_{\mathrm E}(n,\ell)=(1+o(1))/(2r_0\kappa\ell m_{\max})$ is the largest admissible value, and
\[
  m_{\max}=\frac{(1+o(1))\,2n\ln q}{\ln m_{\max}-2\ln(3\sqrt{2\pi e})}=\Theta\Big(\frac{n\log q}{\log n}\Big),
\]
which is $\Theta(n)$ for $q=n^{\Theta(1)}$.
\end{corollary}

\begin{proof}
Condition (B3) is equivalent to $q^{n/m}\geq\sqrt m/(3\sqrt{2\pi e})$, so $\sqrt m\,q^{n/m}\geq m/(3\sqrt{2\pi e})$ and $\alpha_{\mathrm E}(n,m,\ell)\leq3\sqrt{2\pi e}/(c\,r_0\kappa\ell m)=1/(2r_0\kappa\ell m)$, since $c=6\sqrt{2\pi e}$; the second inequality follows from $m\geq n$ and \cref{prop:ell-lower}. The consequences for $\SIVP$ and SIS follow because the thresholds of \cref{cor:edcp-sivp,prop:sis-edcp} are $\Otil(n/\alpha)$ and $\sqrt m/\alpha$. For the second statement: the function $m\mapsto\sqrt m\,q^{-n/m}$ is increasing, so (B3) holds exactly for $m$ up to the largest integer $m_3$ at which it holds; for $q\geq n$ it holds at $m=n$, so $m_3\geq n$. At $m_3$ the inequality $q^{n/m}\geq\sqrt m/(3\sqrt{2\pi e})$ holds with near-equality, by the maximality of $m_3$, which gives the attainment and, after taking logarithms, the formula for $m_3$. Condition (B2) holds at $m_3$, because $q^{(m_3-n)/m_3}=q\cdot q^{-n/m_3}\geq(1-o(1))\,3\sqrt{2\pi e}\,q/\sqrt{m_3}\geq2\sqrt{2\pi e}$ when $q\geq n$, and (B1) holds because $m_3\geq n\geq\kappa$; hence $m_{\max}=m_3$. Since $m\mapsto\sqrt m\,q^{n/m}$ is decreasing for $m<2n\ln q$, which contains $m_{\max}$ for large $n$, $\alpha_{\mathrm E}$ is largest at $m_{\max}$ among admissible $m$. The bound persists if the number of states may depend on the width: the admissible noise rate is governed by the product $\ell r$, and $\ell r\geq r(n\log q-O(\log n))/(2\log(1+r))$ is increasing in $r$, hence smallest at $r=r_0$.
\end{proof}

The sufficient noise bound derived for the ball
reduction with the admissibility conditions
(B1)--(B3) is at most
$1/(2r_0\kappa\ell m)$. Together with
\cref{prop:ell-lower}, this makes the largest noise
scale certified by this calculation at most
$1/\widetilde\Omega(n^2)$. The SIVP and SIS
parameters reported below are the guarantees
obtained by composing this calculation. This is
not a lower bound on every algorithm available
under $\HEDCP[\ell]$.

Condition (B3) constrains the number of samples
used in one invocation of this ball construction,
not the total number of samples an adversary may
possess. Additional samples may be reserved for
verification or a suitable subset may be chosen.
The cube construction of \cite[Thm. 2]{BKSW18}
has its own parameter conditions and is not
excluded by (B3).
The separate cube construction has sufficient width bound $r<1/(32m\kappa\ell\alpha q^{n/m})$ \cite[Thm. 2]{BKSW18}; displaying it does not establish its other premises for a concrete instance.

\begin{corollary}[Worst-case consequences of $\HEDCP{[\ell]}$]\label{cor:edcp-sivp}
Assume $\HEDCP[\ell]$. Then $\SIVP_{\gs_{\mathrm E}}$ and $\GapSVP_{\gs_{\mathrm E}}$ are in $\BQP$ with $\gs_{\mathrm E}=\Otil(\ell n^2)$, and $\BDD_{1/g_{\mathrm E}(d)}$ is in $\BQP$ on $d$-dimensional lattices, in the worst case.
\end{corollary}

\begin{proof}
Let $q$ be a prime with $n^2\leq q\leq\poly(n)$, chosen below, $m:=m_{\max}(n,q)\in\mathcal M_q(n)$ and $\alpha:=\alpha^*_{\mathrm E}(n,\ell)/2=1/\Otil(\ell n)$. By \cref{prop:edcp-lwe}, with $m'=\lceil2n\lg q\rceil$ further samples (and $\alpha<1/(8\sqrt{m'})$, which holds for large $n$ because $\kappa\to\infty$), $\LWE_{n,q,\alpha'}$ is solved with probability $1-2^{-\Omega(n)}$ from $m+m'$ samples, for every $\alpha'\leq\alpha$ and all but a negligible fraction of matrices. Choose $q>2\sqrt n/\alpha$, which is possible with $q$ polynomial. \Cref{lem:regev09}, with its remark on discrete Gaussian errors (the factor $2$ in $\alpha$ absorbs the $1+o(1)$ loss), gives $\SIVP$ and $\GapSVP$ within $\Otil(n/\alpha)=\Otil(\ell n^2)$. The statement on $\BDD$ follows from \cref{lem:bdd-sivp}, applied in dimension $d$.
\end{proof}

\begin{proposition}[SIS under $\HEDCP{[\ell]}$]\label{prop:sis-edcp}
Assume $\HEDCP[\ell]$. Let $q\leq\poly(n)$ be a prime with $\mathcal M_q(n)\neq\emptyset$, let $2n\lg q\leq m\leq\poly(n)$, and suppose $\beta_{\mathrm E}(n,m,\ell)<q/\sqrt2$. For every $\beta>\beta_{\mathrm E}(n,m,\ell)$ there is a quantum polynomial-time algorithm that solves $\SIS_{n,m,q,\beta}$ with probability at least $1/\poly(n)$ over a uniform $A\in\Zq^{n\times m}$ and its internal randomness.
\end{proposition}

\begin{proof}
Let $\beta':=\min\{\beta,q/\sqrt2\}>\beta_{\mathrm E}(n,m,\ell)$ and $\alpha:=\sqrt m/(\sqrt2\,\beta')$. Then $\alpha<\min\{\alpha^*_{\mathrm E}(n,\ell),1/(8\sqrt m)\}$ by the definition of $\beta_{\mathrm E}$, and $\alpha q\geq\sqrt m$.

Let $m_0:=m_{\max}(n,q)$, which is smaller than $2n\lg q\leq m$, and let $\mathcal S$ be the algorithm of \cref{prop:edcp-lwe} for $m_0$ samples of $\LWE_{n,q,\alpha}$, without amplification; let $p_{\mathcal S}\geq1/\poly(n)$ be its success probability over a uniform matrix, a uniform secret, errors $D_{\Z^{m_0},\alpha q}$ and its randomness.

Fix $A$ with $\lambda_1^\infty(\Lambda_q(A^{\top}))\geq q/4$, which excludes a $q^{-n}$ fraction (\cref{lem:qary}). Let $\Lat:=\Lambda_q^\perp(A)$, so that $\Lat^*=q^{-1}\Lambda_q(A^{\top})\supseteq\Z^m$ and $\lambda_1(\Lat^*)\geq1/4$, and set $d:=\sqrt m\,\alpha<1/8\leq\lambda_1(\Lat^*)/2$. Define a decoder $\mathcal D$ for $\Lat^*$: on input $t\in\R^m$, (1) pick a uniform $\tilde s\in\Zq^n$ and set $t':=t+q^{-1}A^{\top}\tilde s$; (2) round each coordinate of $qt'$ randomly, $b_i\leftarrow D_{\Z,\alpha q/\sqrt2,\,qt'_i}$; (3) run $\mathcal S$ on the first $m_0$ rows of $(A^{\top},b\bmod q)$, obtaining $\hat s$; (4) let $u:=q^{-1}(A^{\top}(\hat s-\tilde s)\bmod q)\in\R^m$ and output $u+\lfloor t-u\rceil\in\Lat^*$, rounding coordinatewise.

We bound the success probability $\bar p$ of $\mathcal D$ in the sense of \cref{lem:qstep-imperfect}. There $t=y-v$ with $v\in\Lat^*$ and $y$ a continuous Gaussian of parameter $d/\sqrt{2m}=\alpha/\sqrt2$ conditioned on $\|y\|<d$; the conditioning changes probabilities by at most $2^{-m}$ (\cref{lem:tails}(a)). Write $v=q^{-1}A^{\top}s_v+k_v$ with $s_v\in\Zq^n$ and $k_v\in\Z^m$. Then $qt'\equiv A^{\top}s'+qy\pmod q$ with $s':=\tilde s-s_v$, which is uniform and independent of $y$, and $qy$ has independent continuous Gaussian coordinates of parameter $\alpha q/\sqrt2\geq\sqrt{m/2}\geq\eta_{2^{-m}}(\Z)$ (\cref{lem:tails}(c)). By \cref{lem:round}, $b-A^{\top}s'$ is within statistical distance $2m\,2^{-m}$ of $D_{\Z^m,\alpha q}$. Thus the first $m_0$ rows of $(A^{\top},b)$ form, up to statistical distance $2^{-\Omega(m)}$, an LWE instance with a uniform matrix, uniform secret $s'$ and errors $D_{\Z^{m_0},\alpha q}$. If $\mathcal S$ returns $\hat s=s'$, then $u\equiv-q^{-1}A^{\top}s_v\pmod{\Z^m}$, hence $t-u\equiv y\pmod{\Z^m}$, and since $\|y\|_\infty<1/8$ the output is $t-y$, the closest point of $\Lat^*$. Hence $\bar p\geq p_{\mathcal S}(A)-2^{-\Omega(m)}$, where $p_{\mathcal S}(A)$ is the success probability of $\mathcal S$ on the first $m_0$ rows of the fixed matrix $A^{\top}$.

Apply \cref{lem:qstep-imperfect}(b) with $G:=\{z\in\Lat\setminus\{0\}:\|z\|\leq\beta'\}$. The width is $w=\sqrt m/(\sqrt2\,d)=1/(\sqrt2\,\alpha)=\beta'/\sqrt m$, so by \cref{lem:tails}(a) and the argument of \cref{prop:sis-reg} (here $w>4\sqrt2\sqrt m\geq4$), $1-D_{\Lat,w}(G)\leq2^{-m}+2^{-3m/2}$ and $\delta_G=2^{-\Omega(m)}$. The output lies in $G$ with probability at least $\max\{0,p_{\mathcal S}(A)-2^{-\Omega(m)}\}^2$. Averaging over $A$ and using the convexity of $x\mapsto\max\{0,x-\delta\}^2$, the success probability is at least $(p_{\mathcal S}-2^{-\Omega(m)})^2-q^{-n}\geq1/\poly(n)$. Solutions are checked ($Az\equiv0$, $z\neq0$, $\|z\|\leq\beta$) before being output.
\end{proof}

The idea of using an average-case LWE solver as the decoder in the quantum step is that of Stehl\'e, Steinfeld, Tanaka and Xagawa, who re-interpreted Regev's quantum reduction as a reduction from SIS to LWE \cite{SSTX09}; \cref{lem:qstep-imperfect} supplies the analysis for decoders that succeed only on average and with probability $1/\poly(n)$, and \cref{lem:round} reconciles the continuous offsets of the quantum step with the discrete errors of \cite{BKSW18}.

\paragraph{The thresholds in numbers.} The following model illustrates differences between the asymptotic sufficient bounds; it does not estimate finite-dimensional attack costs. \Cref{tab:thresholds} evaluates them, ignoring constants and logarithmic factors, for dimensions $512$ and $1024$ and three values of the failure parameter. Three observations are worth making. First, even in the normalisation $f\to0^+$, the $\SIVP$ threshold $\gs_{\mathrm{Reg}}\approx n^{3/2}$ is about $2^{13.5}$--$2^{15}$, larger than the signature proof factors in this normalisation (\cref{sec:sis}). Second, once the number of states and the radius condition are accounted for, the displayed sufficient bounds from the restricted ball composition are less favourable than the Regev-route bounds for LWE and SIS when $f<1/2$, and for $\SIVP$ when $f<3/4$: the restricted ball composition certifies noise scales at most about $n^{-2}$ and yields $\SIVP$ factors about $n^3$ and SIS bounds about $n^{5/2}$. Third, the normalised unique-SVP threshold $\sqrt n$ is only about $2^{4.5}$--$2^{5}$, comparable to the heuristic gaps of deployed key-recovery lattices (\cref{sec:parameters}); comparisons at this scale depend on constants that \cite{Reg04} does not make explicit, which is why we report them as heuristic comparisons, not as conclusions.

\begin{table}[htbp]
\centering\small
\renewcommand{\arraystretch}{1.2}
\begin{tabular}{@{}lcccccc@{}}
\toprule
& \multicolumn{3}{c}{$n=d=512$} & \multicolumn{3}{c}{$n=d=1024$}\\
\cmidrule(lr){2-4}\cmidrule(lr){5-7}
Threshold (exponent of 2) & $f\to0^+$ & $f=1/4$ & $f=1$ & $f\to0^+$ & $f=1/4$ & $f=1$\\
\midrule
\multicolumn{7}{@{}l}{\emph{Under $\HReg(f)$}}\\
$\gu\approx d^{1/2+2f}$ (unique-SVP) & 4.5 & 9 & 22.5 & 5 & 10 & 25\\
$\gamma_{\mathrm g}\approx d^{1+2f}$ ($\GapSVP$) & 9 & 13.5 & 27 & 10 & 15 & 30\\
$1/\alpha_{\mathrm{Reg}}\approx n^{1+2f}$ (LWE, many samples) & 9 & 13.5 & 27 & 10 & 15 & 30\\
$\gs_{\mathrm{Reg}}\approx n^{3/2+2f}$ ($\SIVP$) & 13.5 & 18 & 31.5 & 15 & 20 & 35\\
$\beta_{\mathrm{Reg}}\approx n^{3/2+2f}$ (SIS) & 13.5 & 18 & 31.5 & 15 & 20 & 35\\
\multicolumn{7}{@{}l}{\emph{Under $\HEDCP[\ell]$, with $\ell\approx n$}}\\
$1/\alpha^*_{\mathrm E}\approx\ell n\approx n^{2}$ (LWE, $m_{\max}$ samples) & \multicolumn{3}{c}{18} & \multicolumn{3}{c}{20}\\
$\gs_{\mathrm E}\approx\ell n^{2}\approx n^{3}$ ($\SIVP$) & \multicolumn{3}{c}{27} & \multicolumn{3}{c}{30}\\
$\beta_{\mathrm E}\approx\ell n^{3/2}\approx n^{5/2}$ (SIS) & \multicolumn{3}{c}{22.5} & \multicolumn{3}{c}{25}\\
\bottomrule
\end{tabular}
\caption{Illustrative values of the thresholds of \cref{thm:reach}, as exponents $e$ such that the threshold is about $2^e$, ignoring constants and logarithmic factors. The rows for $\HEDCP[\ell]$ take $\ell\approx n$, the information-theoretic minimum up to logarithmic factors (\cref{prop:ell-lower}), and do not depend on $f$. The column $f\to0^+$ is a normalisation (\cref{sec:technical}). These are normalisations of sufficient guarantees, not necessary attack thresholds. Applying them requires the finite-dimensional bound and every sample, distribution and instance promise.}
\label{tab:thresholds}
\end{table}

\begin{remark}\label{rem:sis-routes}
Items \ref{it:sis} and \ref{it:edcp} reach SIS through the quantum step applied to the $q$-ary lattice $\Lambda_q^\perp(A)$, not through $\SIVP$: under $\HReg(f)$ the decoder of the dual lattice is the worst-case $\BDD$ algorithm (\cref{prop:sis-reg}); under $\HEDCP[\ell]$ it is the average-case LWE solver, used as in the quantum reduction from SIS to LWE of \cite{SSTX09} (\cref{prop:sis-edcp}). Related quantum equivalences between LWE with quantum auxiliary states and (I)SIS appear in \cite{CLZ22,DFS24,CH26}. Under $\HEDCP[\ell]$, the threshold $\beta_{\mathrm E}=\Otil(\ell n^{3/2})$ coincides, up to polylogarithmic factors, with the norm $\gs_{\mathrm E}/\Otil(\sqrt n)=\Otil(\ell n^{3/2})$ at which the stated worst-case hardness premise becomes unavailable: this composition solves SIS at bounds comparable to those where its stated worst-case foundation becomes uninformative; below them it supplies no converse. Under $\HReg(f)$, the threshold $\beta_{\mathrm{Reg}}=O(m^{3/2+2f})$ lies above the worst-case-premise threshold $\Otil(n^{1+2f})$ by a factor of about $\sqrt n$. The claim that ``no reduction connects DCP to SIS'' should therefore be replaced by the sharper question of \cref{q:sis}: whether DCP connects to SIS \emph{at the small norm bounds used by signature security proofs}. Solving the exact SIS assumptions used at those bounds would remove the hardness premise from those security arguments; it would not by itself produce forgeries (\cref{sec:sis}).
\end{remark}

\begin{remark}\label{rem:robust}
The route of \cref{prop:edcp-lwe} loses twice. The factor $\ell$ in $\alpha_{\mathrm E}$ arises because the $\ell$ states produced by \cite[Thm. 3]{BKSW18} must be valid simultaneously; and the radius condition (B3) forces $m\leq m_{\max}=\Theta(n)$, so that $\sqrt m\,q^{n/m}$ is at least about $m_{\max}/(3\sqrt{2\pi e})$ instead of $O(\sqrt{n\log q})$. A solver whose robustness matched the preparation error model, including possible correlations, might remove the requirement that all states be valid simultaneously. If such a composition allowed constant per-state error growing linearly with width, the factor $\ell$ might be replaced by a constant. This is a conjectural modification of the reduction, not a theorem about arbitrary invalid states, and it leaves the second loss: even then the route would reach noise rates of about $1/(r_0\kappa\,m_{\max})=1/\Otil(n)$, not the standard rate $1/\Otil(\sqrt n)$. We do not build such robustness into $\HEDCP[\ell]$; whether a reduction from LWE to EDCP can avoid both losses is \cref{q:robust}.
\end{remark}

\begin{remark}\label{rem:simon-claim}
The draft \cite{Sim26} asserts, in its October revision as before, tolerance to faulty samples at rates up to $1/O(\log n)$, where its $n$ is $\log N$, and consequences for approximate SVP at factor $\sqrt n\,\mathrm{polylog}(n)$ in the lattice dimension (\cref{sec:simon-strategy}). Since $\log N=\Theta(n^2)$ in Regev's reduction, a failure rate $1/O(\log\log N)=1/O(\log n)$ corresponds to $(\log N)^f=O(\log n)$, that is to $f=\Theta(\log\log n/\log n)\to0$: a stronger tolerance than $\HReg(f)$ grants for any constant $f>0$. If Regev's theorem applied at such an $f$ with controlled constants, which is not established, since the theorem is stated for fixed $f$, it would yield the unique-SVP factor $\sqrt n\,\mathrm{polylog}(n)$ and, through \cref{lem:lm09}, $\GapSVP$ at factor $\Otil(n)$. Even the weaker reading with constant $f>0$ reaches every polynomial factor above $\gs_{\mathrm{Reg}}$.
\end{remark}

\begin{figure}[p]
\centering
\begin{minipage}{\textwidth}
\tikzset{
 reachhyp/.style={draw=wBlack,fill=wSky!12,rounded corners=2pt,align=center,
   text width=37mm,minimum height=10mm,inner sep=3pt},
 reachres/.style={draw=wVermillion,line width=1.1pt,align=center,
   text width=27mm,minimum height=18mm,inner sep=3pt},
 reachcond/.style={draw=wOrange,line width=1.1pt,align=center,inner sep=3pt},
 reachopen/.style={draw=wBlue,dash pattern=on 3pt off 2pt,align=center,
   text width=43mm,minimum height=12mm,inner sep=3pt},
 reachlab/.style={font=\scriptsize,fill=white,inner sep=1.5pt,align=center},
 reacharr/.style={->,>=Stealth,line width=.8pt}}
\small
\textbf{(a) Cyclic hypotheses and group actions}\par\smallskip
{\centering
\begin{tikzpicture}[font=\footnotesize]
 \node[reachhyp] (reg) at (0,0) {$\HReg(f)$\\fixed-$f$ hypotheses};
 \node[reachhyp] (hs) at (5,0) {$\HHS$\\noiseless cyclic states};
 \node[reachres,text width=40mm] (actions) at (10,0)
   {Free, transitive cyclic actions\\efficient inputs and evaluation\\Thm. \ref{thm:reach}\ref{it:hs}};
 \draw[reacharr] (hs.north west) to[out=140,in=40]
   node[reachlab,above]{some solver-dependent $f_{\mathcal A}$; L. \ref{lem:hs-rare}} (reg.north east);
 \draw[reacharr] (reg.south east) to[out=-40,in=-140]
   node[reachlab,below]{zero corruptions} (hs.south west);
 \draw[reacharr] (hs) -- (actions);
\end{tikzpicture}\par}
\smallskip
\textbf{(b) Proved lattice consequences with route-specific parameters}\par\smallskip
{\centering
\begin{tikzpicture}[font=\footnotesize]
 \node[reachhyp] (reg) at (3.1,0) {$\HReg(f)$};
 \node[reachhyp] (edcp) at (10.85,0) {$\HEDCP[\ell]$\\polynomial-modulus domain};
 \node[reachres,text width=40mm,minimum height=13mm] (bdd) at (3.1,-2.2)
   {$\uSVP_{\gu}$, $\BDD_{\delta_{\mathrm u}}$, $\GapSVP_{\gamma_{\mathrm g}}$\\worst-case oracles};
 \node[reachres,text width=40mm,minimum height=13mm] (lwee) at (10.85,-2.2)
   {Uniform LWE: $\alpha<\alpha_{\mathrm E}$\\$m_0\in\mathcal M_q(n)$};
 \draw[reacharr] (reg) -- node[reachlab,right]{\cite{Reg04,LM09}\\Thm. \ref{thm:reach}\ref{it:usvp}--\ref{it:bdd}} (bdd);
 \draw[reacharr] (edcp) -- node[reachlab,right]{P. \ref{prop:edcp-lwe}} (lwee);
 \node[reachres] (sivpr) at (0,-4.45)
   {$\SIVP_{\gs_{\mathrm{Reg}}}$\\$O(n^{3/2+2f})$ factor};
 \node[reachres] (sisr) at (3.1,-4.45)
   {$\SIS_{\beta}$, $\beta\geq\beta_{\mathrm{Reg}}$\\$O(m^{3/2+2f})$ bound};
 \node[reachres] (lwer) at (6.2,-4.45)
   {Uniform LWE\\$\alpha\leq\alpha_{\mathrm{Reg}}(m)$};
 \node[reachres] (wce) at (9.3,-4.45)
   {$\SIVP$, $\GapSVP$\\$\Otil(\ell n^2)$ factor\\$\BDD_{1/g_{\mathrm E}}$};
 \node[reachres] (sise) at (12.4,-4.45)
   {$\SIS_{\beta}$\\$\beta>\beta_{\mathrm E}$};
 \draw[reacharr] (bdd.south west) -- node[reachlab,above left,pos=.55]{P. \ref{prop:sivp-reg}} (sivpr.north);
 \draw[reacharr] (bdd.south) -- node[reachlab,right,pos=.55]{P. \ref{prop:sis-reg}} (sisr.north);
 \draw[reacharr] (bdd.south east) -- node[reachlab,above right,pos=.55]{Thm. \ref{thm:reach}\ref{it:lwe}} (lwer.north);
 \draw[reacharr] (lwee.south west) -- node[reachlab,left,pos=.55]{C. \ref{cor:edcp-sivp}} (wce.north);
 \draw[reacharr] (lwee.south east) -- node[reachlab,right,pos=.55]{P. \ref{prop:sis-edcp}} (sise.north);
\end{tikzpicture}\par}
\vspace{2pt}
{\footnotesize
Here $\delta_{\mathrm u}=1/(2\gu(d+1))$; all factors and bounds are
defined in \cref{sec:reach-thm}. Reg-LWE and both SIS routes require
$m\geq2n\lg q$ and the stated matrix promises; $q$ is prime.
EDCP-LWE uses (B1)--(B3), with extra samples for amplification;
the worst-case composition chooses $q,\alpha$ with $\alpha q>2\sqrt n$.
EDCP-SIS also requires $\beta_{\mathrm E}<q/\sqrt2$.
The SIS edges use dual decoding and Gaussian sampling, not a reversal
of a worst-case hardness reduction.\par}
\medskip
\textbf{(c) Conditional embedding applications}\par\smallskip
{\centering
\begin{tikzpicture}[font=\footnotesize]
 \node[reachres,text width=35mm,minimum height=12mm] (oracle) at (0,0)
   {$\uSVP_{\gu}$ or\\$\BDD_{1/g_{\mathrm E}}$};
 \node[reachcond,text width=65mm,minimum height=12mm] (mlwe) at (8.3,0)
   {Few-sample MLWE and other LWE-type instances\\Thm. \ref{thm:reach}\ref{it:lwe}; P. \ref{prop:samples}};
 \draw[reacharr,dashed] (oracle) -- node[reachlab,above,text width=28mm]{verified embedding\\or BDD promise} (mlwe);
\end{tikzpicture}\par}
\medskip
\textbf{(d) Structured route: state conversion and application checks}\par\smallskip
{\centering
\begin{tikzpicture}[font=\footnotesize]
 \node[reachhyp,text width=35mm] (large) at (0,0)
   {Ordinary EDCP solver at $(k,Q)$\\$Q=q^d+1$};
 \node[reachres,text width=35mm,minimum height=12mm] (ipm) at (5.15,0)
   {IP-M-EDCP solver\\same $Q$ and secret};
 \node[reachcond,text width=35mm,minimum height=12mm] (wz) at (10.3,0)
   {MLWE within the source\\preparation contract};
 \draw[reacharr] (large) -- node[reachlab,above]{L. \ref{lem:wz-state-map}} (ipm);
 \draw[reacharr,dashed] (ipm) -- node[reachlab,above,text width=12mm]{\cite{WZ26}\\Thm. 4} (wz);
\end{tikzpicture}\par}
\vspace{2pt}
{\footnotesize
The left solver is outside the stated $\HEDCP$ domain.
Filtering via \eqref{eq:wz-filter} supplies it under $\HHS^{\mathrm{ab}}$
($\HHS$ for $k=1$) at the stated widths and state budget.
The dashed application requires all source distributions, width,
sample and preparation-error conditions. The audited printed
certificate gives no ML-KEM application (\cref{sec:wz-audit}).\par}
\medskip
\textbf{(e) Research questions}\par\smallskip
{\centering
\begin{tikzpicture}[font=\footnotesize]
 \node[reachopen] at (0,0)
   {Cyclic-to-vector translation?\\$\HReg(f)\Rightarrow\HEDCP[\ell]$?\\Q. \ref{q:translate}};
 \node[reachopen] at (5,0)
   {NTRU key recovery; LIP\\Q. \ref{q:ntru}; Q. \ref{q:gapless}};
 \node[reachopen] at (10,0)
   {SIS at small norm bounds\\Q. \ref{q:sis}};
\end{tikzpicture}\par}
\end{minipage}
\caption{Detailed algorithmic implications. An arrow $A\to B$ means
that a solver for $A$ yields an algorithm for $B$, reversing the
problem-reduction convention of \cref{fig:dcp-map}. Solid arrows are
proved under all premises of the cited result; dashed arrows identify
instance or preparation conditions still to be checked for an
application. The separate research-question boxes assert no reduction.
P., C. and L. locate the explicit derivations; literature citations
identify their sources.
The $\HHS$ lattice factors depend on its witness through $f_{\mathcal A}$.
No general implication between $\HEDCP$ and the cyclic solver hypotheses
is established here. Solving SIS does not automatically yield a forgery.}
\label{fig:reach}
\end{figure}
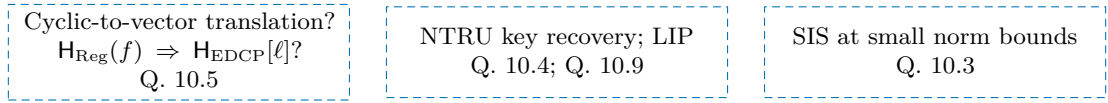

\subsection{The sample-count obstruction}\label{sec:samples}

Both routes to LWE in \cref{thm:reach}, \ref{it:lwe} and \ref{it:edcp}, need many more samples than the secret dimension: $m\geq2n\lg q$ for the first, and $m\in\mathcal M_q(n)$, which excludes $m$ close to $n$, for the second. That is the natural regime inside Regev's reduction, where the reduction itself chooses $m$, but not the regime of deployed schemes: a public key of ML-KEM contains $m=n$ (scalar) LWE samples, one of ML-DSA $m\approx n$ ($m/n\in\{1,6/5,8/7\}$), and the component $u$ of an ML-KEM ciphertext is again an LWE instance with $m=n$ samples, in a fresh secret. (The second component $v$ is not an LWE sample in that secret, because it contains the encoded message; but recovering the secret from $u$ decrypts $v$.) In this passive public-key setting, no additional supply of independent samples under the same secret is assumed. Interactive access or auxiliary keys require a separate analysis. The following proposition quantifies the loss.

\begin{proposition}[Sample-count obstruction]\label{prop:samples}
Consider normal-form LWE with $n$ secret coordinates, $m$ samples, modulus $q$, and secret and error coordinates of standard deviation $\sigma$.
\begin{enumerate}[label=(\alph*)]
  \item\label{it:gapm} \textbf{Heuristic}: For the embedding of dimension $n+m+1$ defined below, with scale $\tau\asymp\sigma$ and $\log q+|\log\sigma|+|\log(\tau/\sigma)|=o(n+m)$, the first-order heuristic gap estimate is
  \[
    g(m)\;\approx\;\frac{q^{m/(n+m)}}{\sigma\sqrt{2\pi e}},
  \]
  increasing from $g(n)\approx\sqrt q/(\sigma\sqrt{2\pi e})$ to $q/(\sigma\sqrt{2\pi e})$ as $m/n\to\infty$. Under the additional heuristic that the embedded
planted vector is shortest and its gap is represented
by the displayed estimate, comparison with
$\Gamma_{f,\mathcal A}(n+m+1)$ gives a heuristic
sufficient certificate. If the estimated gap is below
this bound, or the bound is unavailable, this
comparison is inconclusive.
  \item\label{it:bkswm} \textbf{Rigorous}: If $m=n$ and $\alpha q\geq1$, the restricted ball construction of \cref{prop:edcp-lwe} supplies no certificate at these parameters. The value $m=n$ is not in $\mathcal M_q(n)$: condition (B2) fails, and so does the lower bound on $\lambda_1(\Lambda_q(A))$ on which \cite[Thm. 3]{BKSW18} relies, since $\Lambda_q(A)=\Z^n$ for invertible $A$; and even taking the width bound of the theorem at face value, every state it would produce has width $r<1/(c\sqrt n\,\kappa\,\ell)\leq1/(c\sqrt n\,\kappa)$ and is within trace distance $2^{-\Omega(n)}$ of a state independent of $s$, so that polynomially many such states carry negligible information about $s$.
\end{enumerate}
\end{proposition}

\begin{proof}
\ref{it:gapm} Write $b=As+e\pmod q$, with integer representatives for $s$ and $e$, and put $D=n+m$. The homogeneous lattice
\[
 L_0=\{(x,y)\in\Z^n\times\Z^m:
                     Ax+y\equiv0\pmod q\}
\]
has determinant $q^m$, since the map
$(x,y)\mapsto Ax+y$ onto $\Z_q^m$ is surjective.
The vector $(s,e)$ belongs to the affine coset
$Ax+y\equiv b$, not to $L_0$ in general.

For a positive rational embedding scale $\tau$,
define the full-rank lattice
\[
 L_{\tau,b}
   =\{(x,y,\tau z):
       x\in\Z^n,\ y\in\Z^m,\ z\in\Z,\
       Ax+y\equiv zb\pmod q\}.
\]
Clearing a common denominator for rational $\tau$ does not change approximation ratios. The lattice has dimension $D+1$, determinant $\tau q^m$,
and contains the planted vector $(s,e,\tau)$.
The determinant follows from the same surjectivity
argument before scaling the last coordinate by
$\tau$.

Under the Gaussian heuristic for competing vectors,
and the model
$\|s\|^2+\|e\|^2\approx D\sigma^2$, the gap estimate is
\[
 g_\tau(m)\ \approx\
 \frac{\sqrt{(D+1)/(2\pi e)}\,
                 (\tau q^m)^{1/(D+1)}}
      {\sqrt{D\sigma^2+\tau^2}}.
\]
For $\tau$ comparable to $\sigma$, with
$\log q+|\log\sigma|+|\log(\tau/\sigma)|=o(D)$, this agrees
to first order with
$q^{m/D}/(\sigma\sqrt{2\pi e})$.
The heuristic also assumes that the planted
direction is shortest and that no atypically
short competing direction invalidates the gap.
Applying the uSVP theorem requires those promises;
the determinant computation alone does not
establish them.

\ref{it:bkswm}: for invertible $A\in\Zq^{n\times n}$, $A\Zq^n=\Zq^n$, so $\Lambda_q(A)=\Z^n$ and $\lambda_1=1$, below the bound $\sqrt n/(2\sqrt{2\pi e})$ that \cite[Lem. 6]{BKSW18} gives at $m=n$, for $n\geq69$. Taking the width bound at face value, with $m=n$, $q^{n/m}=q$ and $\alpha q^{n/m}=\alpha q\geq1$, so the admissible width is $r<1/(c\sqrt n\,\kappa\,\ell)$. A Gaussian amplitude of width $r$ places relative weight $\sum_{j\neq0}e^{-2\pi j^2/r^2}\leq 2e^{-2\pi c^2 n\kappa^2}(1+o(1))$ outside $j=0$; the state is therefore $2^{-\Omega(n)}$-close to $\ket{0}\ket{x_k}$, which is independent of $s$. By the triangle inequality, $\poly(n)$ such states are jointly $2^{-\Omega(n)}$-close to $s$-independent states.
\end{proof}

\begin{corollary}\label{cor:fewsample}
For a normal-form LWE instance with $m=n$, the
first-order gap estimate is
$\sqrt q/(\sigma\sqrt{2\pi e})$. Comparing it with
$\Gamma_{f,\mathcal A}(2n+1)$ gives a heuristic
sufficient certificate, subject to the planted-vector
and gap assumptions. No converse is claimed. For ML-KEM, the heuristic normalised gaps of keys and ciphertexts are about $0.3$--$0.4$ (\cref{tab:gaps}). These compare with the $f\to0^+$ model only; an explicit finite-dimensional bound and the actual embedding promise remain necessary for a rigorous certificate. At $m=n$, the restricted unstructured ball construction supplies no certificate (\cref{prop:samples}\ref{it:bkswm}). The printed structured forward bound supplies no certificate at the ML-KEM ring and modulus (\cref{prop:wz-certificate}); this is a limitation of that sufficient bound.
\end{corollary}

\paragraph{A worked example.} \Cref{fig:samples} plots the first-order normalised gap $g(m)/\sqrt{n+m}$ for ML-KEM-768-type parameters ($q=3329$, $\sigma=1$, $n=768$) and ML-DSA-44-type parameters ($q=8\,380\,417$, $\sigma=\sqrt2$, $n=1024$). The line at one is the comparison model $f\to0^+$, $c_0=1$, not a theorem threshold. The ML-KEM-type public-key value is about $0.4$ and its compressed ciphertext component about $0.3$; the model curve crosses one near $m=1.85n$. The ML-DSA-type value is about $11$ only in the separate scenario where full $t$ is supplied; the published-key estimate is about $0.3$. These comparisons illustrate sample sensitivity, without certifying actual gaps or attacks.

\begin{figure}[tbp]
\centering
\begin{tikzpicture}
\begin{axis}[
  width=0.86\textwidth, height=6.6cm,
  xmin=1, xmax=3, ymin=0.2, ymax=30000,
  ymode=log, log basis y=10,
  xtick={1,1.5,2,2.5,3},
  xlabel={samples relative to secret dimension, $m/n$},
  ylabel={normalised gap $g(m)/\sqrt{n+m}$},
  label style={font=\small}, tick label style={font=\footnotesize},
  grid=major, major grid style={line width=0.3pt, draw=black!12},
  axis line style={draw=black!50, line width=0.4pt},
  tick style={draw=black!50},
  legend style={at={(0.02,0.98)}, anchor=north west, draw=black!30, fill=white, font=\footnotesize, row sep=1pt},
  legend cell align=left,
]
% ML-DSA-44-type: q = 8380417, sigma = sqrt(2), n = 1024
\addplot[color=wVermillion, line width=1.4pt, dash pattern=on 6pt off 2pt on 1.5pt off 2pt, domain=1:3, samples=81, forget plot]
  {8380417^(x/(1+x))/(1.41421*4.13273*sqrt(1024*(1+x)))};
\addplot[color=wVermillion, only marks, mark=square*, mark size=2.4pt, mark options={draw=white, line width=0.6pt}]
  coordinates {(1,10.945) (1.5,48.205) (2,127.36) (2.5,251.91) (3,416.41)};
\addlegendentry{{ML-DSA-44-type: $q\approx2^{23}$, $\sigma=\sqrt2$, $n=1024$}}
% ML-KEM-768-type: q = 3329, sigma = 1, n = 768
\addplot[color=wBlue, line width=1.4pt, solid, domain=1:3, samples=81, forget plot]
  {3329^(x/(1+x))/(4.13273*sqrt(768*(1+x)))};
\addplot[color=wBlue, only marks, mark=*, mark size=2.4pt, mark options={draw=white, line width=0.6pt}]
  coordinates {(1,0.3562) (1.5,0.7170) (2,1.1239) (2.5,1.5310) (3,1.9133)};
\addlegendentry{{ML-KEM-768-type: $q=3329$, $\sigma=1$, $n=768$}}
% threshold
\addplot[color=black!75, line width=0.9pt, dashed, domain=1:3, samples=2] {1};
\addlegendentry{{threshold $\gu=\sqrt{n+m}$ ($f\to0^+$, $c_0=1$)}}
% annotations
\node[font=\scriptsize, anchor=west] at (axis cs:1.03,0.27) {public key, $m=n$: $0.36$};
\draw[black!60, line width=0.4pt] (axis cs:1.85,1.0) -- (axis cs:1.85,1.45);
\node[font=\scriptsize, anchor=south] at (axis cs:1.85,1.45) {crosses at $m\approx1.85n$};
\node[font=\scriptsize, anchor=west] at (axis cs:1.04,6.5) {full $t$ supplied: $10.9$};
\end{axis}
\end{tikzpicture}
\caption{The sample-count obstruction (\cref{prop:samples}). Normalised unique-SVP gap of the LWE embedding lattice as a function of the number of samples, for two parameter families, against the heuristic yardstick $g=\sqrt{n+m}$ ($f\to0^+$, $c_0=1$); points above it satisfy only this model comparison, subject to the actual gap and finite-dimensional guarantee. Series are distinguished by colour (Wong palette), line pattern, marker shape and legend; the plotted values for public keys also appear in \cref{tab:gaps}, together with those of ciphertexts.}
\label{fig:samples}
\end{figure}
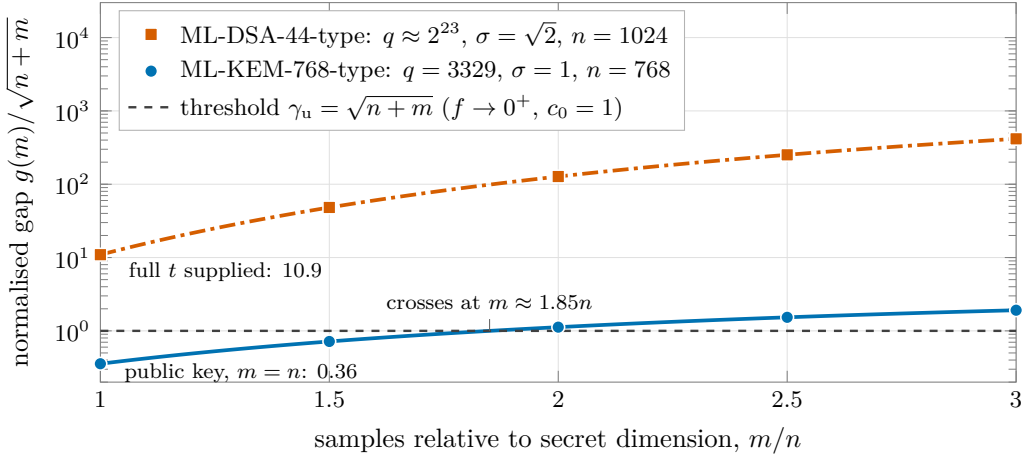

The two analysed unstructured direct constructions lose a factor $q^{n/m}$ in their sufficient noise guarantees. Large modulus-to-noise ratios can improve those comparisons, but a structured embedding still needs its actual gap and distribution promises. More samples do not suffice by themselves: at the standard noise rate $\alpha=1/\Otil(\sqrt n)$ of Regev's and dual-Regev encryption, the displayed direct guarantees supply no certificate even with $m=\Theta(n\log q)$ (\cref{sec:collapse-lwe}).

\subsection{The window outside the reach}\label{sec:window}

\Cref{thm:reach} gives sufficient guarantees for lattice instances with verified promises. A planted-shortest-vector instance with actual gap above $\gu$ is solved by the ordinary-uSVP route. An approximate lattice problem is solved when the available oracle's factor meets the target specification. Under $\HHS$, use the corresponding bounds for some solver-dependent $f_{\mathcal A}$. Under $\HEDCP[\ell]$, the restricted ball composition and the worst-case route have their own sample, distribution and gap conditions. Failure of one of these comparisons to certify an instance gives no converse about the solver's behaviour. We use ``window outside the reach'' only for the region without a certificate from these stated guarantees; its boundary is not a measured security margin. Solving a SIS proof assumption also remains distinct from forging a signature (\cref{sec:sis}). \Cref{fig:reach} retains the distinct quantitative routes. In particular, its two SIS branches use dual decoding and Gaussian sampling (\cref{rem:sis-routes}); they do not follow by reversing a worst-case-to-average-case hardness reduction. Its conditional application panels require the stated instance or preparation certificates, and its research-question boxes assert no additional reduction. The ordinary-EDCP solver in the structured panel has exponential modulus in the ring degree and is not supplied by the polynomial-modulus hypothesis.

% ======================================================================
\section{The Wen--Zheng route: recovered map and audited scope}\label{sec:wz-audit}\label{rem:wz-scope}\label{rem:wz-map}
% ======================================================================

We use the full version of \cite{WZ26} to recover the exact conversion of
structured states, transcribe the forward parameter conditions, and test
them at the public-key parameters. These are distinct results: the
state conversion is exact, the forward construction has restrictive
sufficient conditions, and its polynomial-modulus adapter needs a
complete resource ledger before it can invoke $\HEDCP[\ell]$ as defined here.
The theorem numbers and page references below refer to the
86-page full version of \cite{WZ26}.

\paragraph{Parameters and distributions.}
Write $k$ for the module rank (the source's $n$), $m$ for the number of
independent ring samples under one secret, $t$ for the amplified sample
count (the source's $m'$), and $L$ for the number of structured states
(the source's $l$). Thus the scalar secret dimension is $dk$, not $k$.
Set
\[
 F(X)=X^d+1,\qquad Q=F(q)=q^d+1,\qquad
 S_q=\frac{q^d-1}{q-1},\qquad z=(1,q,\ldots,q^{d-1}).
\]
Theorem 4 of \cite{WZ26} assumes power-of-two $d=\poly(\kappa)$,
constant $k$, $m=O(\log\kappa)$, and a polynomially bounded unramified
prime $q$. It requires the prime ideals above $q$ to satisfy
$\min_i\operatorname{N}(\mathfrak p_i)\geq
(kd)^{\omega(1/(m-k+1))}$. Its polynomial-LWE input has independent
coefficient discrete Gaussians of widths $\sigma_s,\sigma_e\geq\sqrt d$
for the secret and errors, respectively; these are Gaussian parameters,
not coefficient standard deviations. Bounded, centred-binomial,
compressed or rounded inputs need a separate conversion.
For power-of-two cyclotomics the unnormalised canonical embedding
scales coefficient Euclidean norms by $\sqrt d$; Gaussian covariances
must be converted accordingly. This does not identify different
error distributions. In particular, $O(\log\kappa)$ is an upper bound
on $m$, not a promise that every smaller sample count satisfies the
regularity argument: the invoked Lemma 49 and Corollary 6 require $m>k$.

\paragraph{The printed forward certificate.}
Lemma 28 and Theorem 4, pp. 39 and 44--45 of \cite{WZ26}, use
\begin{align}
 \rho_2&\geq 2d\,q^{k/m+2/(dm)},\nonumber\\
 B&=\frac{15}{2}+3\sigma_s\sqrt\kappa\,k^2+18k
       +2m\sigma_e\rho_2d^2\sqrt\kappa,\label{eq:wz-B}\\
 \left(3+r\sqrt\kappa\,Bd(\sqrt2+2/q)\right)S_q
   &\leq \frac{Q^{1-k/t}}{4tLC},\qquad C\geq8.
 \label{eq:wz-printed}
\end{align}
Here $B$ is the source's $c'_{\max}$ and $r$ is its structured-state
width. The amplified matrix is statistically close to uniform over
$\Z_Q^{t\times k}$; its error may depend on that matrix, which is
explicitly allowed in Lemma 30. The evaluation encoding of the secret
is injective because $F(q)\geq q^d$, so recovering the encoded secret
recovers the polynomial secret by the source's centred digit map.
The construction is claimed to produce $L$ approximate states with
negligible individual trace error and constant joint good-branch
probability, using polynomial time in the full input length and
$k,d,t,L,\log q,\kappa$.
We use \eqref{eq:wz-printed} exactly as a \emph{printed sufficient
certificate}, not as a necessary condition for any algorithm.
Lemma 30 instead retains the expansion factor $EF(F)$, and the proof
uses $EF(F)\leq2$; using $2$ in place of $\sqrt2$ is a conservative
strengthening. All negative certificate checks below already use the
more permissive printed constant.

The source's ideal all-good product is
$(1-1/(tL))^{tL}$, which is below $1/e$, not above it.
For $tL\geq2$ it is at least $1/4$; this repairs the stated constant
without changing polynomial-time claims, after accounting separately
for the preprocessing and matrix events. If the usable good-branch
probability is $c$, the target solver succeeds with probability $p$,
and the full preparation instrument has trace error at most $\varepsilon$,
the composed recovery probability is at least $cp-\varepsilon$.
Thus individual errors must be budgeted against $L$ and $p$, not merely
called small after conditioning on a measurement outcome.

\begin{lemma}\label{lem:wz-state-map}
An ideal product-amplitude IP-M-EDCP state
\[
 \sum_{c\in\Z^d}\left(\prod_{i=0}^{d-1}a(c_i)\right)
   \ket{c}\ket{x+\langle c,z\rangle\widetilde s\bmod Q}
\]
reduces by measurement to one ordinary EDCP state of dimension $k$,
modulus $Q$ and one-dimensional amplitude $a$, with the same secret.
The conversion preserves the number of states and cannot increase
trace distance. It holds for every secret and arbitrary offsets.
\end{lemma}
\begin{proof}
Measure $c_1,\ldots,c_{d-1}$, obtaining $u_1,\ldots,u_{d-1}$.
With $x'=x+(\sum_{i=1}^{d-1}q^iu_i)\widetilde s$, the remaining state is
\[
 \sum_{j\in\Z}a(j)\ket{j}\ket{x'+j\widetilde s\bmod Q}.
\]
This is the construction of \cite[Thm. 5, pp. 45--46]{WZ26}.
Its action is independent of the secret distribution; in particular it
also applies to the encoded Gaussian secret produced by Theorem 4,
although the formal statement of Theorem 5 names a uniform secret.
Trace-distance contraction follows because measurement is a quantum
channel when its classical outcome is retained.
\end{proof}

\paragraph{What solver this actually reaches.}
The output has $\log Q=\Theta(d\log q)$, but $Q$ is exponential in $d$.
Polynomial bit length alone does not place this state in
$\HEDCP[\ell]$, which assumes a polynomial modulus in its EDCP dimension.
The source's Definition 22 uses amplitudes $\sqrt{D_{\Z,r}}$, whereas
its construction writes $\rho_r$. Under the former convention the
amplitude width in \cref{def:edcp} is $\sqrt2r$; under the latter it is
$r$. One must fix this convention in a quantitative adapter. The
constant factor does not change the exclusions below.

There is nevertheless a precise positive consequence under the already
stated abelian hypothesis. For amplitude width $R$ in this survey's
convention, project an EDCP register onto $j\in\{0,1\}$ and attenuate
the $j=0$ amplitude to equal the $j=1$ amplitude. This produces a
noiseless coset state over $(\Z_Q)^k$ with heralded probability
\begin{equation}\label{eq:wz-filter}
 p_R=\frac{2e^{-2\pi/R^2}}{\sum_{j\in\Z}e^{-2\pi j^2/R^2}}.
\end{equation}
For $R$ bounded below by a positive constant and bounded above by a
polynomial in the input length, $p_R\geq1/\poly(kd\log q)$.
If the witnessing $\HHS^{\mathrm{ab}}$ solver needs $J$ registers,
$O((J+\log(1/\delta))/p_R)$ independently prepared ideal EDCP states
suffice to obtain $J$ coset states except with probability $\delta$.
Consequently, an efficient forward preparation meeting all its
premises, this larger state budget, and the $cp-\varepsilon$ error
budget gives a search algorithm under $\HHS^{\mathrm{ab}}$.
For $k=1$ the group is cyclic, so $\HHS$ suffices.
This consequence does not assume that cyclic and general abelian
state-input problems are interchangeable, and does not insert an
extra clause into $\HEDCP$.

For the original $\HEDCP[\ell]$ route, Appendix E, Lemma 58 of \cite{WZ26}
proposes a modulus-to-dimension adapter. Its intermediate dimension
satisfies $n'\geq(k+1)d\lceil\log q\rceil+2\kappa$; it reserves
$16k+4\lceil\ln\ln Q\rceil$ input states for normal form, requires the
remaining sample count to be at least $n'\log q$, and incurs a further
width loss and output-state budget. These are additional requirements,
not properties of the exact map above. In particular the invoked
EDCP-to-LWE lemma requires $r=\Omega(\sqrt\kappa)$, and the unknown-noise
routines invoked twice have their own polynomial sample overheads.
The displayed fixed sample ledger in Lemma 58 does not explicitly
propagate those overheads. We therefore retain this adapter as a
source-stated route requiring a completed search, width, success and
sample audit before it supports a new $\HEDCP[\ell]$ consequence.

\begin{proposition}\label{prop:wz-certificate}
For the printed certificate \eqref{eq:wz-printed}, define
\[
 A_{\max}=\frac{Q}{4eCkL S_q\ln Q}.
\]
For every real $t>0$, hence also every admissible integer amplified
sample count, the certificate requires
\begin{equation}\label{eq:wz-optimized}
 3+r\sqrt\kappa\,Bd(\sqrt2+2/q)\leq A_{\max}.
\end{equation}
At $d=256$, $q=3329$, $k\in\{2,3,4\}$, $C\geq8$, and $L\geq1$,
$A_{\max}<0.01$. Thus this certificate has no feasible positive width
at the ML-KEM ring and modulus, regardless of $t$, even before
checking input distributions and amplification premises.
\end{proposition}
\begin{proof}
Divide \eqref{eq:wz-printed} by $S_q$. The only dependence of its
right-hand side on $t$ is $e^{-k\ln Q/t}/t$, whose maximum for $t>0$
is $1/(e k\ln Q)$, attained at $t=k\ln Q$.
For $q\geq3$, $d\geq2$, we have
\[
 \frac{Q}{S_q}=(q-1)\frac{q^d+1}{q^d-1}<q,
 \qquad \ln Q>d\ln q.
\]
Hence $A_{\max}<q/(4eCkLd\ln q)$.
Substituting $q=3329$, $d=256$, and even the most permissive
$C=8$, $k=2$, $L=1$ gives less than $0.009217$.
The left-hand side of \eqref{eq:wz-optimized} exceeds $3$.
\end{proof}

\begin{table}[htbp]
\centering\small
\renewcommand{\arraystretch}{1.15}
\begin{tabular}{@{}lrrll@{}}
\toprule
\textbf{Input parameters} & $k$ & $m$ & $A_{\max}$, approximately & \textbf{Width required by certificate}\\
\midrule
ML-KEM-512 & 2 & 2 & 0.009214 & no feasible $r$\\
ML-KEM-768 & 3 & 3 & 0.006142 & no feasible $r$\\
ML-KEM-1024 & 4 & 4 & 0.004607 & no feasible $r$\\
ML-DSA-44 & 4 & 4 & 5.901936 & $r<2.17\cdot10^{-19}$\\
ML-DSA-65 & 5 & 6 & 4.721549 & $r<1.24\cdot10^{-18}$\\
ML-DSA-87 & 7 & 8 & 3.372535 & $r<1.04\cdot10^{-19}$\\
\bottomrule
\end{tabular}
\caption{Optimistic tests of the printed Wen--Zheng certificate, not
attack costs or security bounds. We set $C=8$, $L=1$, relax $\kappa$
to $1$, take $\sigma_s=\sigma_e=\sqrt d=16$, minimise $\rho_2$, and
maximise over all real $t>0$. The ML-DSA bounds use
$(A_{\max}-3)/(Bd(\sqrt2+2/q))$ and are rounded upwards.
These Gaussian surrogates are not the deployed distributions.
Real state budgets and statistical parameters only tighten the test.}
\label{tab:wz-audit}
\end{table}

The positive widths in the last three rows are far below the
$\Omega(\sqrt\kappa)$ input requirement of the source's
polynomial-modulus adapter; no efficient two-point-filter application is certified here by \eqref{eq:wz-filter}. None of these rows is a recovered
attack. For ML-KEM and ML-DSA-44, $m=k$ also lies outside the strict
$m>k$ premise of the regularity lemma used in amplification.
Ciphertext compression and ML-DSA's omitted low bits introduce further
conversion obligations; they cannot repair the displayed sufficient
tests by themselves.

\paragraph{The route with more input samples.}
Appendix D.5, Lemma 52 of \cite{WZ26}, avoids amplification when
sufficiently many independent samples are supplied, with the source
suggesting order $k\log q$. It allows nonconstant rank and a broader
class of polynomials. In the cyclotomic case above, one may take the carry bound
\[
 B_{\mathrm{many}}=\frac{15}{2}+3\sigma_s\sqrt\kappa\,k^2+18k
                         +\sigma_e\sqrt\kappa.
\]
Lemma 30 then uses $t=m$, $B=B_{\mathrm{many}}$, and $EF(F)$ in place
of $\sqrt2$ in \eqref{eq:wz-printed}. This removes the amplification
error but not the modulus, state count, width or distribution
conditions. In particular, extra states generated from the same
input are not extra independent ring samples, and the additive-$3$
obstruction of \cref{prop:wz-certificate} still applies to this printed
bound at the ML-KEM ring and modulus. For any parameter family satisfying the forward inequalities and
the full preparation contract, the positive solver composition above
applies. We do not infer such a family merely from the failure of the
displayed deployed-parameter tests.

% ======================================================================
\section{Consequences for the LWE family}\label{sec:collapse}
% ======================================================================

This section records sufficient consequences for LWE-type assumptions under each hypothesis and the checks still needed for particular instances. Most entries follow from \cref{thm:reach}. Their applicability depends on the number of samples under one secret, the modulus-to-noise ratio, the algebraic structure of the public matrix, and the promises of the chosen embedding; the name ``Module-LWE'' alone fixes none of these.

An algorithm for LWE with uniform $A$ does not automatically solve a structured instance. The worst-case oracles of \cref{thm:reach} can apply to structured lattices, but their promises must be established for the particular embedding. Under $\HHS$, \cref{cor:hs-lattice} supplies these oracles at some solver-dependent polynomial factors. Under $\HEDCP[\ell]$, \cref{cor:edcp-sivp} supplies a separate worst-case route. The exact structured state conversion of Wen and Zheng \cite{WZ26} is restated in \cref{lem:wz-state-map}; its first output has exponential modulus, and the numerical certificate audit gives no deployed application (\cref{sec:wz-audit}).

\paragraph{How conclusions are assigned.} For $\HReg(f)$, check a sufficient finite-dimensional promise $\Gamma_{f,\mathcal A}$ and the actual planted vector, or the sample and distribution premises of the direct LWE result. For $\HHS$, make the same checks with the solver-dependent exponent from \cref{lem:hs-rare}. For $\HEDCP[\ell]$, select an admissible number $m_0$ of uniform samples with $\alpha<\alpha_{\mathrm E}(n,m_0,\ell)$; reserve $m'$ additional samples if using the amplified guarantee. Alternatively verify the worst-case embedding promise at $g_{\mathrm E}$. For the structured route, check the source-specific inequalities and solver domain in \cref{sec:wz-audit}. The remaining polynomial-modulus adapter audit is recorded in \cref{q:wz}; its status is \textsuperscript{U}. Heuristic gap and norm calculations are marked \textsuperscript{H} and remain conditional. The conjunction $\Hstar$ makes each certified route available but does not remove its premises.

\begin{table}[htbp]
\centering\footnotesize
\renewcommand{\arraystretch}{1.08}
\setlength{\tabcolsep}{4pt}
\begin{tabularx}{\textwidth}{@{}L{3.3cm}cccY@{}}
\toprule
\textbf{Assumption} & $\HHS$ & $\HReg(f)$ & $\HEDCP$ & \textbf{Sufficient guarantee or remaining check}\\
\midrule
Worst-case lattice problems & \vC\bR & \vC\bR & \vC\bR & Some solver-dependent polynomial factors under $\HHS$; the explicit factors of \cref{thm:reach,cor:hs-lattice,cor:edcp-sivp} otherwise.\\
LWE, many samples, small noise & \vK\bR & \vC\bR & \vC\bR & Uniform-matrix and distribution conditions; $\alpha\leq\alpha_{\mathrm{Reg}}$ or $\alpha<\alpha_{\mathrm E}$ with admissible subset and holdout samples. Under $\HHS$, substitute $f_{\mathcal A}$.\\
LWE, many samples, standard noise & \vS & \vS & \vS & No certificate from the stated sufficient bounds for $\alpha=1/\Otil(\sqrt n)$ (\cref{sec:collapse-lwe}).\\
LWE, few samples, large $q/\sigma$ & \vK\bH & \vK\bH & \vK\bH & Verify the planted-shortest-vector promise and actual gap against $\Gamma_{f,\mathcal A}$ or $g_{\mathrm E}$. The FHE row is one illustrative comparison.\\
Few-sample unstructured LWE (FrodoKEM) & \vK\bH & \vK\bH & \vS & Finite-dimensional embedding comparison unresolved; restricted ball construction supplies no certificate at $m=n$.\\
Few-sample MLWE (ML-KEM, ML-DSA keys) & \vK\bH & \vK\bH & \vK\bH & Embedding promises remain conditional; the audited structured bounds give no deployed certificate (\cref{tab:wz-audit}).\\
Gaussian EDCP & \vK\bR & \vK\bR & \vC\bR & By hypothesis in its domain. Cyclic hypotheses cover parameters whose EDCP-to-LWE image \cite[Thm. 4]{BKSW18} meets the inherited LWE bounds; no general cyclic-to-vector implication asserted.\\
Ring-LWE, Module-LWE with a large gap & \vK\bH & \vK\bH & \vK\bH & Worst-case oracles apply when the actual structured embedding promise is verified; uniform-matrix estimates do not establish it. The structured route has separate width, sample and modulus conditions (\cref{sec:wz-audit}).\\
MLWE, short secrets & \multicolumn{3}{c}{as MLWE} & Restricted distributions and their conversions must be checked \cite{BJRW23}.\\
MP-LWE, Order-LWE & \vK\bH & \vK\bH & \vK\bH & Conditional worst-case embedding route; hardness reductions from Ring-LWE give no solver for the target \cite{RSSS17,NP26,PX26,PP24,BBPS19}.\\
LWR & \vK\bR & \vK\bR & \vK\bH & Many-sample BDD certificate for $p\geq4\sqrt m\,\gu(m+1)$; otherwise verify the embedding promise (\cref{sec:collapse-variants}).\\
Succinct / decomposed LWE & \multicolumn{3}{c}{as LWE} & Ignore the hint for succinct LWE; use the parameter-matched equivalence for decomposed LWE \cite{AACLMR26}.\\
CLWE & \multicolumn{3}{c}{as LWE / \vS} & The discrete-secret reduction to LWE \cite{GVV22}; general secrets lack a certificate here.\\
$S|\LWE\rangle$ & \multicolumn{3}{c}{as LWE} & Measurement gives corresponding classical samples \cite{CLZ22,CH26}.\\
Stehl\'e--Steinfeld NTRU variant & \vK\bH & \vK\bH & \vK\bH & Conditional Ring-LWE ciphertext embedding route \cite{SS11}; ordinary deployed NTRU is treated separately in \cref{sec:ntru}.\\
\bottomrule
\end{tabularx}
\caption{Consequences under each hypothesis, with all stated parameter and promise conditions. \vC: proved consequence; \vK: conditional; \vO: unresolved connection; \vS: no certificate from the analysed reductions. Basis marks: R rigorous, H heuristic, U unverified parameter application. The $\HHS$ lattice bounds depend on its witness through $f_{\mathcal A}$; they are not numerical certificates for a deployed scheme. No absence of a certificate asserts security or failure of all attacks.}
\label{tab:collapse}
\end{table}

% ==================================================================
\subsection{Unstructured LWE}\label{sec:collapse-lwe}

\paragraph{Four regimes.} With sufficiently many samples and sufficiently small noise, \cref{thm:reach} certifies LWE under $\HReg(f)$; \cref{prop:edcp-lwe} gives a second sufficient range under $\HEDCP[\ell]$, using $m_0\in\mathcal M_q(n)$ and, for amplification, $m'$ holdout samples. The displayed bounds do not certify many-sample LWE at the standard noise rate. For few samples, the ordinary-uSVP comparison depends on the embedded planted vector and its actual gap: a large modulus-to-noise ratio improves the heuristic estimate, while the estimates for the deployed small-ratio instances remain inconclusive. At $m=n$ the restricted ball construction provides no certificate. These are regimes of the stated guarantees, not necessary conditions on all possible attacks.

\paragraph{Regev's cryptosystem and its dual.} Regev's cryptosystem uses $q\asymp n^2$, $m=\Theta(n\log q)$ and $\alpha=1/(\sqrt n\log^2n)$ \cite{Reg09}. The $q/4$ lower bound used in the displayed BDD certificate gives a ratio of only $\Theta(\log^{3/2}n)$ against the error norm, which does not verify the required polynomial promise. Likewise $\alpha$ exceeds the sufficient ball-reduction noise bound. These failed comparisons do not prove that an actual lattice gap or another attack cannot be better. The dual cryptosystem of \cite{GPV08} supplies many samples in a fresh secret per ciphertext; its standard-noise regime is also uncertified by the stated bounds. We therefore record no attack certificate here under any of the three hypotheses, without asserting an impossibility result.

\paragraph{Lindner--Peikert-type designs and FrodoKEM.} Lindner and Peikert's compact design publishes $b=As+e$ with square $A$; a ciphertext supplies an instance in a fresh secret and a message-bearing component \cite{LP11}. FrodoKEM uses unstructured $A$ and parameters $(n,q,\sigma)=(640,2^{15},2.8)$, $(976,2^{16},2.3)$ and $(1344,2^{16},1.4)$ \cite{Frodo16,FrodoKEM21}. Its eight secret columns share $A$ but each has only $m=n$ samples. The heuristic gap estimates are about $16$, $27$ and $44$, or $0.4$--$0.9$ times $\sqrt d$ (\cref{tab:gaps}). These numbers specify comparisons that would suffice under the model $c_0d^{1/2+2f}$; they do not prove failure below a bound. The restricted ball composition does not certify these samples, and the structured connection does not apply to unstructured $A$. This leaves quantitative applicability unresolved, rather than establishing a comparative security margin.

\paragraph{Normal form, small errors and the role of samples.} The normal-form transformation of Applebaum, Cash, Peikert and Sahai \cite{ACPS09} converts LWE with a uniform secret into LWE with a secret drawn from the error distribution, at the cost of $n$ samples; it is why short secrets are standard and why the gap of \cref{prop:samples} is computed with secret and error of comparable size. Two results discussed in \cref{sec:prelim-lwe} show that the number of samples is a security parameter in its own right: the Arora--Ge linearisation \cite{AG11} breaks LWE with very small errors given many samples, while Micciancio and Peikert prove LWE with binary errors hard when the number of samples is limited \cite{MP13}. The sample-count obstruction is the counterfactual analogue of the same phenomenon.

% ==================================================================
\subsection{Ring- and Module-LWE}\label{sec:collapse-structured}

\paragraph{What the worst-case reductions say.} Lyubashevsky, Peikert and Regev defined Ring-LWE over the ring of integers $R$ of an arbitrary number field $K$, with the secret in the dual ideal $R^\vee$ and errors that are Gaussian in the \emph{canonical embedding} of $K$ into $\R^n$ \cite{LPR13}. Their first theorem is a quantum reduction, modelled on Regev's, from worst-case approximate $\SVP$ on \emph{ideal} lattices to search Ring-LWE. Their second, for cyclotomic fields and primes $q\equiv1$ modulo the conductor, is a search-to-decision reduction: such $q$ split completely, the Chinese remainder theorem decomposes $R_q$ into $n$ copies of $\Zq$, and the Galois automorphisms of the field permute the copies transitively, which lets an algorithm that solves decision on one copy recover the whole secret. Peikert, Regev and Stephens-Davidowitz later reduced worst-case ideal-lattice problems directly to decision Ring-LWE, for any number field and any modulus, removing the algebraic restrictions of the search-to-decision step \cite{PRS17}. Langlois and Stehlé proved the analogous worst-case to average-case reductions for Module-LWE and Module-SIS from problems on module lattices \cite{LS15}, and Boudgoust, Jeudy, Roux-Langlois and Wen gave hardness reductions for bounded secret and error distributions under the conditions of their theorems \cite{BJRW23}. These reductions establish conditional hardness of the structured problems from their source assumptions. They also show that the worst-case problems behind Ring- and Module-LWE are problems on structured lattices, which is why the quantum algorithms for ideal lattices of \cref{sec:intro-alarms} \cite{CDPR16,CDW17,PHS19} matter for their foundations, though not, at subexponential approximation factors, for their concrete security.

\paragraph{Conditional structured consequences.} The worst-case ordinary-uSVP and BDD oracles apply to Ring- or Module-LWE embeddings satisfying their actual planted-vector and gap promises. The average-case uniform-matrix theorem cannot supply those promises for structured $A$. Under $\HHS$ the same statement holds with a solver-dependent polynomial factor; under $\HEDCP[\ell]$ the sufficient worst-case bound is $g_{\mathrm E}$. Wen and Zheng's structured route \cite{WZ26} is analysed with its explicit forward certificate and exact state map in \cref{sec:wz-audit}. The numerical audit supplies no deployed attack: ML-KEM fails the printed width inequality even after optimising the amplified sample count, and the ML-DSA Gaussian surrogate bounds are too small for the source's polynomial-modulus adapter. This is a limitation of those sufficient certificates, not an impossibility result.

\paragraph{Rank changes in the heuristic gap comparison.} Albrecht and Deo reduced Module-LWE of rank $k$ and modulus $q$ to Ring-LWE with modulus $q^k$, sample by sample, at the cost of increasing the noise rate by a factor $n^{c+1/2}\sqrt k$ for ring degree $n$, rank $k$ and any constant $c>0$, in power-of-two cyclotomics \cite{AD17}; together with the trivial converse this makes the two problems polynomially equivalent. In the normal-form heuristic used here, this map does not increase the estimated gap. For an instance with secret dimension $n$, $m$ scalar samples and noise rate $\alpha$, \cref{prop:samples}\ref{it:gapm} gives $g\approx q^{-n/(n+m)}/(\alpha\sqrt e)$, so that at fixed noise rate a larger modulus \emph{decreases} this heuristic estimate for a few-sample instance. For an ML-KEM public key ($n=m=256k$) the gap is $q^{-1/2}/(\alpha\sqrt e)$; after the Albrecht--Deo map the secret dimension is $256$, the number of samples $256k$, the modulus $q^k$ and the noise rate $\alpha'\geq\alpha$, so the gap becomes $q^{-k/(k+1)}/(\alpha'\sqrt e)$, smaller by at least a factor $q^{k/(k+1)-1/2}$. Thus this map does not improve the displayed heuristic comparison. This is not a theorem about actual structured gaps or every possible rank-changing reduction.

% ==================================================================
\subsection{Middle-product, order-based and unified variants}\label{sec:collapse-mplwe}

\paragraph{The constructions.} Middle-Product LWE, introduced by Roşca, Sakzad, Stehlé and Steinfeld, takes $a\in\Zq^{<n}[x]$ uniform, a secret $s\in\Zq^{<2n-1}[x]$, and returns $b=a\odot_n s+e$, where $a\odot_n s$ consists of the $n$ middle coefficients of the product $a\cdot s$, which has degree less than $3n-2$ \cite{RSSS17}. The point of the definition is that it refers to no particular polynomial modulus: the authors proved that MP-LWE is at least as hard as Polynomial-LWE modulo $f$ for \emph{every} $f$ in an exponentially large family of monic polynomials with bounded expansion factor, so a suitable MP-LWE solver yields solvers for each source family at the parameters of its reduction, and they built a public-key encryption scheme from it. Njah Nchiwo and Pellet-Mary improved the parameters of the reduction from Ring-LWE \cite{NP26}, and Pellet-Mary and Xia gave a single reduction, parameterised by field elements, that recovers the earlier ones and applies to number fields satisfying their polynomial-boundedness conditions \cite{PX26}. Order-LWE, introduced by Bolboceanu, Brakerski, Perlman and Sharma, replaces the ring of integers by an arbitrary order of the field; they reduced worst-case problems on lattices associated with ideals of the order to it and used it to study Ring-LWE with secrets from non-uniform, entropic distributions \cite{BBPS19}. Peikert and Pepin placed Ring-, Polynomial-, Module-, Order- and Middle-Product LWE in one framework and gave simpler and tighter reductions among them \cite{PP24}.

\paragraph{Conclusion.} Hardness reductions from Ring-LWE or worst-case problems to these variants do not supply algorithms in the reverse direction. The worst-case ordinary-uSVP route applies if the chosen embedding has the planted vector as a shortest vector and its actual gap meets the oracle promise. Uniform-matrix estimates do not establish that promise for middle-product or order-based instances; we have not verified it here. Under $\HEDCP[\ell]$, the worst-case route of \cref{cor:edcp-sivp} similarly requires the appropriate embedding promise at $g_{\mathrm E}$. We establish no additional EDCP-based certificate for these variants. A future structured reduction would require its own parameter and state-preparation audit.

% ==================================================================
\subsection{Rounding, short secrets and strengthened assumptions}\label{sec:collapse-variants}

\paragraph{Learning with Rounding.} An LWR sample $(a,\lfloor(p/q)\langle a,s\rangle\rceil\bmod p)$, multiplied by $q/p$, is an LWE-type sample whose error is deterministic and bounded by $q/(2p)$ in absolute value \cite{BPR12}. Banerjee, Peikert and Rosen proved LWR at least as hard as LWE when $q/p$ exceeds the LWE error by a sufficiently large (in their reduction, super-polynomial) factor, and built pseudorandom functions from it. For the algorithmic direction, take prime $q$, uniform independent rows and $m\geq2n\lg q$. Outside the exceptional matrices in \cref{lem:qary}, the rescaled instance meets a $\BDD$ promise of gap at least $(q/4)/((q/2p)\sqrt m)=p/(2\sqrt m)$, which \cref{thm:reach}\ref{it:bdd} solves when $p\geq4\sqrt m\,\gu(m+1)=\Otil(n^{1+2f})$. Few-sample instances require a separate embedding analysis. For structured Module-LWR matrices, the uniform-matrix minimum bound cannot be assumed; the actual lattice minimum and planted-vector promise must be checked.

\paragraph{Short secrets.} Shorter secrets and errors increase the displayed heuristic gap estimate. With the usual rescaling of the secret coordinates, an instance with secret standard deviation $\sigma_s$ and error standard deviation $\sigma_e$ has heuristic gap $q^{m/(n+m)}/(\sigma_s^{n/(n+m)}\sigma_e^{m/(n+m)}\sqrt{2\pi e})$, which reduces to \cref{prop:samples}\ref{it:gapm} when $\sigma_s=\sigma_e$; smaller widths improve that model comparison. Applying an oracle still requires the actual distribution and embedding promises. The hardness reductions of \cite{BJRW23} alone do not establish those algorithmic premises.

\paragraph{Succinct and decomposed LWE.} Wee's $\ell$-succinct LWE \cite{Wee24} and the decomposed LWE assumption of Abram, Malavolta and Roy \cite{AMR25}, proved equivalent by Abram, Arnon, Cini, Lou, Malavolta and Roy \cite{AACLMR26}, assert that LWE remains hard when the adversary also receives structured auxiliary information, in the case of decomposed LWE, information that can be read as circular-security leakage of a natural encryption scheme. They are \emph{strengthenings} of LWE: their hardness implies that of LWE. The converse direction, which is the one an algorithm needs, requires an argument. For succinct LWE it is immediate: an instance consists of a standard LWE instance with a uniform matrix together with a hint, and a solver for that LWE instance, run while ignoring the hint, succeeds with the same probability, because the marginal distribution of the instance is unchanged. For decomposed LWE, whose samples may encrypt functions of the secret, a standard LWE solver does not apply directly; the transfer goes through the reduction of \cite{AACLMR26}, which turns an adversary against succinct LWE into one against decomposed LWE, and holds only at the parameters for which their equivalence is proved. These arguments transfer an algorithm to the corresponding assumption only with matched parameters and access models. An attack on an attribute-based encryption, laconic function evaluation or succinct randomised encoding construction requires a further construction-specific argument; it does not follow merely because its security proof uses that assumption.

\paragraph{Continuous LWE.} Bruna, Regev, Song and Tang reduced worst-case lattice problems to CLWE by a quantum reduction \cite{BRST21}, and Gupte, Vafa and Vaikuntanathan reduced LWE to CLWE classically \cite{GVV22}. These reductions show that CLWE is at least as hard as the source problems; they give no algorithm for CLWE. For an algorithm the relevant result is the reduction in the opposite direction, also given by Gupte, Vafa and Vaikuntanathan, from CLWE with a \emph{discrete} secret to LWE, which, as they note, nearly completes the picture modulo the discrete-secret requirement \cite[App. C]{GVV22}. Under $\Hstar$, discrete-secret CLWE therefore falls whenever the LWE instance produced by that reduction is in the reach, while for general-secret CLWE the reductions reviewed here provide no algorithmic certificate, recorded as \Survive.

\paragraph{Quantum-state LWE.} In $S|\LWE\rangle$ the adversary receives, instead of a classical sample, a quantum state that superposes the possible errors. Measuring the state yields a classical LWE sample with error distributed according to the squared amplitudes, so a classical solver applies when its required distribution, parameters and sample budget match that measured marginal. Chailloux and Hermouet proved a quantum equivalence between $S|\LWE\rangle$ and ISIS \cite{CH26}, building on the filtering technique of Chen, Liu and Zhandry \cite{CLZ22} and on the quantum oblivious LWE sampling of Debris-Alazard, Fallahpour and Stehlé \cite{DFS24}; transferring an algorithm to ISIS requires the parameters and resources of their equivalence. This does not itself give a forgery for the SIS-based signatures discussed in \cref{sec:sis}.

% ==================================================================
\subsection{Search, decision and quantum auxiliary states}\label{sec:collapse-sd}

Search and decision LWE are related by reductions with explicit modulus, distribution and resource conditions \cite{Reg09,MM11}. The quantum reduction of Sudo, Tezuka, Hara and Yoshida, built on a quantum Goldreich--Levin theorem over $\Zq$ for prime $q$, recovers the secret with only two queries to a decision oracle \cite{STHY23}, but the query count alone does not establish preservation of a fixed classical sample budget. A use under $\Hstar$ must check the decision-oracle interface, coherent access and the reduction parameters; sample-preserving results such as \cite{MM11} apply under their own conditions. In the other direction the relevant works are the quantum bridges between LWE and SIS. Chen, Liu and Zhandry's \emph{filtering} technique converts states with a known amplitude profile into states with a profile suited to the problem at hand; it yields polynomial-time quantum algorithms for $\SIS^\infty$ with very wide matrices, polynomial modulus and bound $q/2-O(1)$, and for EDCP with suitable amplitude functions \cite{CLZ22}. Kothari, O'Donnell and Wu later gave classical deterministic algorithms for all the $\SIS^\infty$ variants treated there \cite{KOW26}. Debris-Alazard, Fallahpour and Stehlé showed how to sample LWE instances obliviously, that is, without knowing a secret for them, with a quantum algorithm, which contradicts the lattice knowledge assumptions used by some standard-model SNARKs \cite{DFS24}. These results matter here because they are explicit quantum connections between the LWE and SIS worlds. None of them gives a DCP-based algorithm for SIS at the small norm bounds of signatures (\cref{q:sis}).

% ==================================================================
\subsection{Advanced primitives}\label{sec:collapse-advanced}

Fully homomorphic encryption was first realised over ideal lattices \cite{Gen09}; later leveled and approximate-eigenvector constructions use LWE or its structured variants \cite{BGV12,GSW13}. Evaluation and bootstrapping keys encrypt functions of a secret and cannot automatically be counted as independent LWE samples. For the illustrative FHE parameters in \cref{tab:gaps},
the heuristic public-key gap is much larger than the
model $d^{1/2+2f}$ over the displayed range of $f$.
This is a conditional comparison for that parameter
set. A rigorous attack certificate additionally
requires a finite-dimensional reduction bound and
the appropriate embedding promise; we do not infer
a uniform classification of all FHE parameters from
this example. The table's illustrative public key uses ring degree $2^{15}$, modulus $2^{880}$ and error standard deviation $3.2$. Attribute-based encryption, pseudorandom functions and identity-based encryption likewise require their own sample and noise analysis \cite{Wee24,AMR25,BPR12,GPV08,Pei16}; no family-wide conclusion follows from this example.

% ==================================================================
\subsection{Consequences for the lattice standards}\label{sec:collapse-standards}

ML-KEM \cite{FIPS203,Kyber18} and ML-DSA \cite{FIPS204,ducas2017} are built on Module-LWE over $R_q=\Zq[x]/(x^{256}+1)$, the ring setting of \cite{WZ26}, and ML-DSA's unforgeability rests additionally on Module-SIS. \Cref{tab:standards} lists the parameters needed for their comparisons.

\begin{table}[htbp]
\centering\small
\renewcommand{\arraystretch}{1.2}
\setlength{\tabcolsep}{4pt}
\begin{tabular}{@{}llrlll@{}}
\toprule
\textbf{Parameter set} & \textbf{rank} & $q$ & \textbf{distributions} & \textbf{$m/n$: key; ctxt} & \textbf{other}\\
\midrule
ML-KEM-512 & $k=2$ & 3329 & CBD, $\eta_1=3$ / $\eta_2=2$ & $1$; $1$ ($u$) & $(d_u,d_v)=(10,4)$\\
ML-KEM-768 & $k=3$ & 3329 & CBD, $\eta_1=2$ / $\eta_2=2$ & $1$; $1$ ($u$) & $(d_u,d_v)=(10,4)$\\
ML-KEM-1024 & $k=4$ & 3329 & CBD, $\eta_1=2$ / $\eta_2=2$ & $1$; $1$ ($u$) & $(d_u,d_v)=(11,5)$\\
ML-DSA-44 & $(4,4)$ & $\approx2^{23}$ & uniform, $\eta=2$ & $1$ & $d=13$, $\tau=39$\\
ML-DSA-65 & $(6,5)$ & $\approx2^{23}$ & uniform, $\eta=4$ & $6/5$ & $d=13$, $\tau=49$\\
ML-DSA-87 & $(8,7)$ & $\approx2^{23}$ & uniform, $\eta=2$ & $8/7$ & $d=13$, $\tau=60$\\
FrodoKEM-640 & none & $2^{15}$ & $\sigma=2.8$ & $1$; $1$ & $\bar n=8$\\
FrodoKEM-976 & none & $2^{16}$ & $\sigma=2.3$ & $1$; $1$ & $\bar n=8$\\
FrodoKEM-1344 & none & $2^{16}$ & $\sigma=1.4$ & $1$; $1$ & $\bar n=8$\\
\bottomrule
\end{tabular}
\caption{Parameters of the LWE-based standards \cite{FIPS203,FIPS204} and of FrodoKEM \cite{FrodoKEM21} that determine the conclusions. CBD is the centred binomial distribution; for ML-KEM, $\eta_1$ governs the key-generation secret and error and the encryption secret $r$, and $\eta_2$ the encryption errors $e_1,e_2$; ML-DSA uses $q=8\,380\,417$ and module rank $(k,\ell')$, written $(k,\ell)$ in \cite{FIPS204}; FrodoKEM's matrix is unstructured and its error distribution approximates a rounded Gaussian of standard deviation $\sigma$. The column $m/n$ gives the number of scalar samples per secret coordinate in a public key and, for KEMs, in the ciphertext component that is an LWE instance in the fresh secret ($u$ for ML-KEM, $C_1$ for FrodoKEM); the other component carries the encoded message and is not an LWE sample in that secret. ML-KEM ciphertexts are compressed to $d_u$ and $d_v$ bits per coefficient, and ML-DSA public keys omit the $d$ low-order bits of $t$; both add deterministic noise and lower the gap estimates in the heuristic model of \cref{tab:gaps}.}
\label{tab:standards}
\end{table}

\begin{itemize}
  \item \textbf{ML-KEM} (\vK\bH{}\,/\,\vK\bH{}\,/\,\vK\bH{} under $\HHS$\,/$\HReg(f)$\,/$\HEDCP[\ell]$). A key publishes $t=As+e$ and a ciphertext includes compressed $u=A^\top r+e_1$ and a message-bearing $v$ \cite{FIPS203}. The key and $u$ each have $m/n=1$; recovering $r$ from $u$ would decrypt the message. Their normalised heuristic gaps are about $0.3$--$0.4$. An explicit finite-dimensional guarantee and the embedding promises are not established here. The printed Wen--Zheng certificate has no feasible width here (\cref{prop:wz-certificate}).
  \item \textbf{ML-DSA} (\vK\bH{}\,/\,\vK\bH{}\,/\,\vK\bH). Its public key supplies the high bits $t_1$ of $t=As_1+s_2$ and has $m/n\in\{1,6/5,8/7\}$ \cite{FIPS204}. Dropping thirteen bits lowers the published-key heuristic gap to about $0.3\sqrt d$. The full-$t$ rows describe a separate scenario in
which $t_0$ is supplied to the adversary. We do not
derive this additional input from ordinary signatures.
The dependence of hints on $ct_0$ does not by itself
give an efficient reconstruction algorithm for
$t_0$. Conclusions for the published key therefore
use the published-key rows; the full-$t$ rows are
reported only as conditional sensitivity estimates. The Module-SIS norm comparisons in \cref{tab:sisreach} concern a proof assumption, and do not themselves yield a forgery.
  \item \textbf{FN-DSA} (Falcon, \vO{}\,/\,\vO{}\,/\,\vO). The power-of-two cyclotomic key lattice has $\lambda_2=\lambda_1$, so it does not meet the ordinary-uSVP promise (\cref{lem:ntru-sym}). The lossiness connection does not supply a verified deployed key attack, and the displayed sampler norm calculation does not certify the proof's SIS bound (\cref{sec:ntru,sec:sis}).
  \item \textbf{FrodoKEM} (\vK\bH{}\,/\,\vK\bH{}\,/\,\vS). The sufficient embedding comparison is unresolved; the restricted unstructured ball construction gives no certificate for its square-matrix instances (\cref{sec:collapse-lwe}).
  \item \textbf{SLH-DSA} (\vS{}\,/\,\vS{}\,/\,\vS) has no lattice dependency \cite{FIPS205}. XMSS and LMS share that property \cite{SP800208}; HQC is a code-based KEM selected by NIST \cite{HQC,NISTIR8545}.
\end{itemize}

For the few-sample Module-LWE instances considered here,
the unstructured compositions do not furnish a
deployed-parameter attack certificate. The printed structured forward certificate is audited in
\cref{sec:wz-audit} and gives no certificate at the ML-KEM
ring and modulus; the polynomial-modulus adapter and
conversions for deployed distributions remain unverified here. The embedding estimates
also require a finite-dimensional reduction guarantee
and a justified gap estimate. These are limits of the
present classification, not a proof of security or a
measured security margin.

% ======================================================================
\section{Other lattice assumptions: reach and limitations}\label{sec:lattice-survivors}
% ======================================================================

This section examines SIS at small norm bounds, CVP-based encryption, NTRU and lattice isomorphism. Each presents additional requirements or limitations for the compositions of \cref{thm:reach}. For NTRU, the obstruction to an ordinary-uSVP promise depends on the key ring and its symmetries. Few-sample LWE and Module-LWE are treated in \cref{sec:collapse,sec:parameters}. The absence of a certificate from these routes supplies neither a separation nor positive evidence of hardness. \Cref{tab:survivors} summarises the conclusions of this section.

\begin{table}[htbp]
\centering\footnotesize
\renewcommand{\arraystretch}{1.08}
\setlength{\tabcolsep}{4pt}
\begin{tabularx}{\textwidth}{@{}L{3.3cm}cccY@{}}
\toprule
\textbf{Assumption or scheme} & $\HHS$ & $\HReg(f)$ & $\HEDCP$ & \textbf{Reason or remaining check}\\
\midrule
SIS, $\beta\geq\beta_{\mathrm{col}}$ & \vK\bR & \vC\bR & \vC\bR & Solved under the modulus, matrix, sample and norm conditions of \cref{prop:sis-reg,prop:sis-edcp}: $\beta_{\mathrm{Reg}}(m)=O(m^{3/2+2f})$; $\beta_{\mathrm E}=\Otil(\ell n^{3/2})$ for $m=\Theta(n\log q)$. Under $\HHS$, substitute $f_{\mathcal A}$. For a scheme this is kind (B), which does not by itself give a forgery.\\
SIS, $\beta_{\mathrm{void}}\leq\beta<\beta_{\mathrm{col}}$ & \vK\bR & \vK\bR & \vK\bR & Kind (C): $\SIVP_{\beta\cdot\Otil(\sqrt n)}$, the worst-case premise of \cite{MR07,MP13}, is solved, but no SIS algorithm is certified; $\beta_{\mathrm{void}}=\Otil(n^{1+2f})$ under $\HReg(f)$. Under $\HEDCP[\ell]$ the two thresholds coincide up to polylogarithmic factors.\\
SIS, $\beta<\beta_{\mathrm{void}}$ & \vO & \vO & \vO & Neither guarantee is reached; the DCP--SIS question at small norm bounds is open in both directions (\cref{q:sis}).\\
SIS-only primitives (Ajtai's function, GPV signatures) & \vO & \vO & \vO & Third regime: $\beta=\sqrt m$ for Ajtai's function and $\Otil(n)$ for GPV signatures with the trapdoors of \cite{MP12}, a margin of about $n^{2f}$ under $\HReg(f)$. SWIFFT's collisions also have coefficients in $\{-1,0,1\}$, but its matrix is ring-structured, outside the uniform-matrix thresholds.\\
ML-DSA (Module-LWE key; Module-SIS proof assumption) & \vK\bH & \vK\bH & \vK\bH & Published-key gap about $0.3\sqrt D$ (\cref{tab:gaps}). Module-SIS (SUF): heuristic comparison $c_0\lesssim0.6$--$0.7$ at $f\to0^+$ (\cref{tab:sisreach}), kind (B). Under $\HEDCP$ the glyph rests on the key; for Module-SIS neither \cref{prop:sis-edcp} nor the audited Wen--Zheng route supplies a certificate (\cref{sec:wz-audit}).\\
FN-DSA (Falcon) & \vO & \vO & \vO & $\lambda_2=\lambda_1$ excludes the ordinary-uSVP route (\cref{lem:ntru-sym}); the estimated typical norm of the quantum-step samples exceeds the proof's SIS bound even at $\gu=1$ (our norm-model estimate, \cref{tab:sisreach}), so it certifies no solution at that bound; the lossiness connection gives no verified key attack.\\
GGH-type encryption & \vO & \vO & \vO & The $\BDD$ gap is a small constant rather than $\gu$; no worst-case foundation, and classical breaks of GGH encryption \cite{Ngu99} and of GGH signatures \cite{NR06}.\\
NTRU-HPS, NTRU-HRSS, Streamlined NTRU Prime & \vO & \vO & \vO & Glyphs: key recovery and decryption, not established by these reductions. For keys, rotations exclude an ordinary-uSVP promise above $1$ for $x^n-1$ (keys not fixed by $x$) and above $\sqrt3$ for NTRU Prime (\cref{lem:ntru-sym}); density ratios are not uSVP gaps (\cref{tab:gaps}). Decision NTRU: conditional gap-based and lossiness routes (\cref{lem:ntru-lossiness}), neither certified here for deployed parameters. Overstretched NTRU: independent classical attacks \cite{ABD16,KF17}.\\
LIP-based schemes (e.g.\ Hawk) & \vO & \vO & \vO & For the class of $\Z^n$, no ordinary-uSVP promise above one, and the guarantee of a polynomial-factor $\SIVP$ oracle alone does not recover an exact isometry; these observations concern that class, and the compositions give no general attack certificate. Hawk has been withdrawn \cite{NIST26r3}; the structural attack of \cite{SW26} reduces its key recovery to exact $\SVP$ in dimension $n/2+1$, which the polynomial-factor guarantees do not solve automatically.\\
ML-KEM, FrodoKEM (keys, ciphertexts) & \multicolumn{3}{c}{as in \cref{sec:collapse-standards}} & Normalised heuristic gaps $0.28$--$0.85$ at $f\to0^+$, $c_0=1$ for ML-KEM keys and $u$ and FrodoKEM keys; the restricted ball construction gives no certificate (\cref{tab:gaps}).\\
\bottomrule
\end{tabularx}
\caption{Summary of the conclusions of this section under each hypothesis, with all stated parameter and promise conditions; glyphs and basis marks as in \cref{tab:collapse}. SIS rows classify instances by their norm bound relative to the thresholds of \cref{sec:sis}. For SIS-based primitives a conclusion concerns the assumption of the security proof, kind (B), unless stated, and does not by itself give a forgery. No absence of a certificate asserts security or failure of all attacks.}
\label{tab:survivors}
\end{table}

% ==================================================================
\subsection{SIS and Module-SIS}\label{sec:sis}

\paragraph{Three regimes.} Recall from \cref{sec:prelim-sis} that $\SIS_{n,m,q,\beta}$ asks for a nonzero $z\in\Lambda_q^\perp(A)$ with $\|z\|\leq\beta$, an \emph{absolute} length, and that for uniform $A$ with $m\geq2n\lg q$ the lattice $\Lambda_q^\perp(A)$ has $\lambda_m\leq\sqrt{2m}\,\omega(\sqrt{\log m})$ (\cref{lem:qary,lem:tails}). Under the modulus, matrix, sample and norm conditions of the cited propositions, we compare two sufficient bounds. Their asymptotic forms suppress constants and logarithms; they are not sharp finite-dimensional thresholds. The first is the norm $\beta_{\mathrm{col}}$ above which SIS instances are solved: $\beta_{\mathrm{Reg}}(m)=O(m^{3/2+2f})$ under $\HReg(f)$ (\cref{prop:sis-reg}), $\beta_{\mathrm E}(n,m,\ell)=\Otil(\ell n^{3/2})$ under $\HEDCP[\ell]$ for $m=\Theta(n\log q)$ (\cref{prop:sis-edcp}), and the smaller of the two under $\Hstar$. The second is the norm $\beta_{\mathrm{void}}$ above which the stated worst-case hardness premise becomes unavailable, because \cite{MR07,MP13} base $\SIS_\beta$ on $\SIVP_{\beta\cdot\Otil(\sqrt n)}$ and that problem has become easy: $\beta_{\mathrm{void}}:=\gs/\Otil(\sqrt n)$, which is $\Otil(n^{1+2f})$ under $\HReg(f)$, $\Otil(\ell n^{3/2})$ under $\HEDCP[\ell]$, and the smaller of the two under $\Hstar$. With those qualifications, the following three cases organise the sufficient guarantees.
\begin{enumerate}
  \item $\beta\geq\beta_{\mathrm{col}}$: \Collapse. The SIS instance is solved.
  \item $\beta_{\mathrm{void}}\leq\beta<\beta_{\mathrm{col}}$: \Cond. The stated direct SIS algorithm supplies no certificate, while the cited worst-case hardness premise becomes unavailable. The reduction remains mathematically valid; other hardness evidence is not excluded.
  \item $\beta<\beta_{\mathrm{void}}$: \Open. Neither guarantee is reached by these comparisons; this does not prove the instance or its foundation hard.
\end{enumerate}
Under $\HReg(f)$, the displayed asymptotic bounds differ by about $\sqrt n$, before accounting for suppressed factors and side conditions. Under $\HEDCP[\ell]$ alone, for uniform $A$ with $\Theta(n\log q)$ columns (a wider matrix can be restricted to that many columns), both thresholds are $\Otil(\ell n^{3/2})$, at least about $n^{5/2}$ by \cref{cor:edcp-ceiling}, so their displayed orders coincide up to polylogarithmic factors. This comparison neither fixes a finite-dimensional boundary nor proves hardness below it. SIS with norm bound $\beta$ and LWE with noise rate $\alpha\approx\sqrt m/\beta$ rest on the same worst-case problem, $\SIVP$ within $\Otil(n/\alpha)=\Otil(\sqrt n\,\beta)$ \cite{Reg09,MR07,MP13}, and \cref{prop:sis-edcp}, following \cite{SSTX09}, turns an LWE solver at noise rate $\alpha$ into a SIS solver at norm $\sqrt m/(\sqrt2\,\alpha)$; the proposition certifies this connection when its full modulus, distribution, sample and norm conditions hold. Under $\HReg(f)$ the band is wider because the direct $\SIVP$ algorithm of \cref{prop:sivp-reg} is stronger, by a factor of about $\sqrt m$, than the one obtained through LWE. Under $\Hstar$ with $\ell\approx n$, the solution threshold is about $n^{\min\{3/2+2f,\,5/2\}}$ and the worst-case-premise threshold about $n^{\min\{1+2f,\,5/2\}}$, up to logarithmic factors.

\paragraph{Solving SIS is not forging.} The regimes classify SIS instances, not schemes. For a primitive whose security proof reduces to SIS we distinguish three kinds of consequence, and conclusions on SIS-based primitives in this survey refer to the second unless we say otherwise.
\begin{enumerate}[label=(\Alph*)]
  \item \emph{Explicit attack}: an algorithm that recovers keys, decrypts, forges signatures or finds collisions for the scheme as deployed.
  \item \emph{Proof assumption solvable}: the problem to which the security proof reduces, at the parameters of the proof, is solved (regime 1). The hardness premise is then unavailable, but an attack follows only if a solution of that norm and form is itself a break, which needs a reduction in the converse direction.
  \item \emph{Guarantee no longer informative}: the worst-case problem on which the assumption is based is solved (regime 2), while the assumption itself is not known to be.
\end{enumerate}
The gap between (A) and (B) is concrete. A SIS solution of norm $\sqrt m$ is a collision of Ajtai's function only if it lies in $\{-1,0,1\}^m$, and a forgery of a hash-and-sign signature needs a short preimage of a \emph{hashed} target, that is, a solution of an inhomogeneous problem, which a solver for homogeneous SIS does not provide. A forger yields a SIS solution; a SIS solver does not, by itself, yield a forger.

\paragraph{SIS-only primitives.} The primitives whose security rests on SIS alone have norm bounds of very different sizes. Ajtai's function $f_A(x)=Ax\bmod q$ on binary inputs $x\in\{0,1\}^m$ has a collision-resistance reduction to the stated SIS assumption because a collision $x\neq x'$ yields $z=x-x'\in\{-1,0,1\}^m$ with $Az\equiv0$ and $\|z\|\leq\sqrt m$ \cite{ajtai1996}; with $\beta=\sqrt m=\Otil(\sqrt n)$ it is deep in the third regime. SWIFFT is a fast instantiation of the same idea over the ring $\Z_{257}[x]/(x^{64}+1)$, compressing $1024$ input bits through sixteen ring elements, with collisions again giving solutions with coefficients in $\{-1,0,1\}$ \cite{LMPR08}. Hash-and-sign signatures in the style of Gentry, Peikert and Vaikuntanathan sign a message by sampling, with a trapdoor, a short preimage $x$ of the hashed message from a discrete Gaussian; a forgery, combined with the signer's own preimage of the same hash, gives a SIS solution of norm about twice the signature bound \cite{GPV08}. With the trapdoors of Micciancio and Peikert \cite{MP12} the Gaussian width is $\Otil(\sqrt n)$ and the signature bound $\Otil(n)$, so these signatures lie in the third regime: clearly under $\HEDCP[\ell]$, whose thresholds are at least about $n^{5/2}$, and under $\HReg(f)$ by a margin of $n^{2f}$, up to logarithmic factors and constants. Lyubashevsky's Fiat--Shamir with aborts reduces to SIS alone only when the public key is statistically close to uniform \cite{Lyu12}; in practical instantiations, ML-DSA included, the key is only computationally hidden and its secrecy rests on Module-LWE.

\paragraph{The deployed signatures.} \Cref{tab:sisreach} compares estimated typical sampler norms with the bounds in the security proofs of FN-DSA and ML-DSA. The model replaces dual minima by Gaussian-heuristic values unless a known short vector invalidates that substitution. For Falcon, the dual contains the conjugate secret scaled by $q^{-1}$, of norm at most $1.17/\sqrt q\approx0.011$ \cite{Falcon20}. With $\gu=1$, the resulting typical $\ell_2$ norms are about $5.5\cdot10^4$ and $1.1\cdot10^5$, versus proof bounds about $1.2\cdot10^4$ and $1.7\cdot10^4$.

The short vector in the dual invalidates the use of
the larger Gaussian-heuristic dual minimum in this
calculation. With the resulting dual bound, the
displayed sampler's estimated typical output norm
exceeds the target bound even in the optimistic
comparison $\gu=1$. Thus this norm calculation
does not certify a solution at the required bound.
We do not claim an impossibility result for this
sampler or for other attacks; that would require
an appropriate bound on the probability of a useful
short output.

For ML-DSA the strong-unforgeability proof uses Module-SIS with the displayed $\ell_\infty$ bounds \cite{DilSpec21,FIPS204}. The model comparisons at $c_0\lesssim0.6$--$0.7$ are heuristic sufficient comparisons, not attack thresholds. Ordinary unforgeability also involves SelfTargetMSIS and the hash function. Under $\HEDCP[\ell]$, \cref{prop:sis-edcp}'s uniform-matrix and sample premises do not certify these structured instances; the audited Wen--Zheng search-MLWE route does not supply a small-norm Module-SIS solver (\cref{sec:wz-audit}). A homogeneous SIS solution still does not supply an inhomogeneous hashed-target preimage or an actual forgery.

\begin{table}[htbp]
\centering\footnotesize
\renewcommand{\arraystretch}{1.2}
\setlength{\tabcolsep}{3.5pt}
\begin{tabular}{@{}lrrrlll@{}}
\toprule
\textbf{Instance} & $m$ & $n$ & $\lambda_1(\Lat^*)$ & \textbf{Required} & \shortstack[l]{\textbf{Heuristic}\\\textbf{comparison}} & \textbf{$\HReg$, $f\to0^+$}\\
\midrule
Falcon-512 (as SIS) & 1024 & 512 & $\leq0.011$ & $\ell_2\leq11\,668$ & not certified & \Open\\
Falcon-1024 (as SIS) & 2048 & 1024 & $\leq0.011$ & $\ell_2\leq16\,765$ & not certified & \Open\\
ML-DSA-44 MSIS (SUF) & 2304 & 1024 & 0.0097 & $\ell_\infty\leq2^{18.5}$ & $c_0\lesssim0.7$ & \Cond\textsuperscript{H}\\
ML-DSA-65 MSIS (SUF) & 3072 & 1536 & 0.0046 & $\ell_\infty\leq2^{20}$ & $c_0\lesssim0.7$ & \Cond\textsuperscript{H}\\
ML-DSA-87 MSIS (SUF) & 4096 & 2048 & 0.0054 & $\ell_\infty\leq2^{20}$ & $c_0\lesssim0.6$ & \Cond\textsuperscript{H}\\
\bottomrule
\end{tabular}
\caption{Heuristic reach of \cref{prop:sis-reg} on the lattice problems behind the security proofs of deployed signatures, for $f\to0^+$ (our norm-model estimates; Falcon's displayed dual-minimum upper bound comes from a known short vector). $\Lat^*$ is the dual of the relevant $q$-ary lattice of dimension $m$ with $n$ constraints, normalised as in the proof of \cref{prop:sis-reg}; with the decoding distance $d=\lambda_1(\Lat^*)/(2\gu(m+1))$, the quantum step produces samples of width $w=\sqrt{2m}\,\gu(m+1)/\lambda_1(\Lat^*)$, and the last but one column gives the largest $c_0$ for which their typical norm meets the required bound. ``SUF'' is strong unforgeability. A conclusion here is of kind (B): the assumption of the proof would be solvable, which does not by itself give a forgery. Falcon's problem is inhomogeneous and structured and is treated, optimistically for the attacker, as homogeneous SIS with the bound $2\beta$ of the proof; for it, $\lambda_1(\Lat^*)$ is bounded by the conjugate of the secret key, not taken from the Gaussian heuristic ($0.070$ and $0.099$), and ``not certified'' means that the estimated typical norm exceeds the bound even at $\gu(m+1)=1$. This does not bound the probability of a useful short output or rule out other attacks.}
\label{tab:sisreach}
\end{table}

\paragraph{Two caveats.} First, many practical ``SIS-based'' schemes are SIS \emph{and} LWE schemes: ML-DSA's public key is a Module-LWE sample, so key recovery is reached whenever that instance is (\cref{sec:parameters}), whatever the status of Module-SIS. Second, the SIS-only constructions reviewed here supply hash functions, commitments and signatures; they do not establish a SIS-only KEM. This is a limit of the reviewed constructions, not an impossibility claim for encryption. For Module-SIS \cite{LS15}, the worst-case decoder can still be used, but the direct SIS propositions also require distribution and dual-minimum conditions. Uniform-matrix bounds do not transfer automatically to a structured matrix; the signature calculations above remain heuristic.

\paragraph{The $\SIS^\infty$ results.} The algorithms of Chen, Liu and Zhandry for $\SIS^\infty$ with very wide matrices, polynomial modulus and bound $q/2-O(1)$ \cite{CLZ22} are sometimes cited as evidence that SIS is quantumly fragile. They concern a regime far from cryptographic SIS, pass through LWE with quantum auxiliary states rather than through DCP, and have been fully dequantised: Kothari, O'Donnell and Wu solve all the $\SIS^\infty$ variants treated there with classical deterministic algorithms \cite{KOW26}. They are relevant here only as a reminder that the DCP--SIS question at small norm bounds is open in both directions (\cref{q:sis}).

% ==================================================================
\subsection{Unstructured CVP-based encryption}\label{sec:ggh}

An early lattice-based public-key scheme designed for efficiency, by Goldreich, Goldwasser and Halevi (GGH), uses a ``good'' private basis, nearly orthogonal, and publishes a ``bad'' basis of the same lattice. A message is encoded as a lattice point $v$, and the ciphertext is $c=v+e$ for a small error $e$ with entries $\pm\sigma$; the owner of the good basis decodes by Babai's rounding, which succeeds because the error is small compared with the good basis \cite{GGH97}. A signature scheme follows the same pattern, signing a hashed message by finding a nearby lattice point. Nguyen showed that the encryption scheme leaks information: since every error entry is $\pm\sigma$, the ciphertext determines the message modulo $2\sigma$, and the remaining problem is a closest-vector instance with a much smaller error; with this observation he decrypted the authors' challenge ciphertexts in dimensions up to $350$ \cite{Ngu99}. Nguyen and Regev later showed that GGH and NTRUSign signatures leak the shape of the secret basis, each signature is a point in a hidden parallelepiped, which can be learned from enough of them, and recovered secret keys experimentally \cite{NR06}. This is the attack that hash-and-sign schemes with Gaussian sampling \cite{GPV08} were designed to avoid.

Under $\Hstar$, GGH decryption is a $\BDD$ instance whose gap is fixed by the quality of the private basis. For the nearly orthogonal private bases recommended by its authors, a scaled identity plus a small perturbation, and an error spread over all coordinates, the Euclidean gap $\lambda_1/\|e\|$ is a small constant rather than the polynomial factor $\gu$; decryption works because Babai's rounding tolerates errors that are small in the coordinates of the private basis, not because the instance is a unique-SVP instance with a large gap. GGH-type encryption is therefore \Open{} under all three hypotheses: outside the reach of the reductions, but with no worst-case foundation and with a record of classical breaks, which makes it an unattractive candidate.

% ==================================================================
\subsection{NTRU and NTRU Prime}\label{sec:ntru}

\paragraph{The schemes.} NTRU (\cref{sec:prelim-ntru}) is used in three forms. NTRU-HPS and NTRU-HRSS, merged as the NTRU submission to the NIST process, use $\Z_q[x]/(x^n-1)$ with $q$ a power of two and ternary keys; Streamlined NTRU Prime uses the field $\Q[x]/(x^p-x-1)$ of prime degree with a large Galois group, to avoid algebraic structure that attacks might exploit, and NTRU LPRime is the corresponding Ring-LWE-style variant \cite{BCLV17}; and Falcon, selected by NIST for standardisation as FN-DSA (FIPS 206, in preparation), uses NTRU lattices over $\Z[x]/(x^n+1)$ with $q=12289$ as the basis of a hash-and-sign signature with fast Fourier sampling \cite{Falcon20}.

\paragraph{No unique shortest vector.} NTRU key recovery asks for a short vector of the lattice $\Lat_h:=\{(u,v)\in R^2:uh\equiv v\pmod q\}$, which contains the secret key $(f,g)$. It is tempting to read it as a unique-SVP instance with gap $\mathrm{gh}(\Lat_h)/\|(f,g)\|=\Theta(\sqrt q/\sigma)$, but $\Lat_h$ is an $R$-module, and multiplication by $x$ maps the key to other lattice vectors that are as short, or almost as short.

\begin{lemma}[Rotations of NTRU keys]\label{lem:ntru-sym}
Let $R=\Z[x]/(\Phi)$ with $\Phi$ monic and irreducible over $\Q$ of degree $n\geq2$, and let $\Lat\subseteq R^k$ be a nonzero $R$-submodule, viewed as a lattice through the coefficient embedding. For every nonzero $v\in\Lat$, the vectors $v,xv,\dots,x^{n-1}v$ lie in $\Lat$ and are linearly independent.
\begin{enumerate}[label=(\alph*)]
  \item If $\Phi=x^n+1$, they all have norm $\|v\|$; hence $\lambda_1(\Lat)=\lambda_2(\Lat)=\dots=\lambda_n(\Lat)$.
  \item If $\Phi=x^p-x-1$, then $\|xv\|\leq\sqrt3\,\|v\|$; hence $\lambda_2(\Lat)\leq\sqrt3\,\lambda_1(\Lat)$.
\end{enumerate}
In particular, the NTRU lattice $\Lat_h$ of Falcon's ring or of NTRU Prime's field is never a $\gamma$-unique-SVP instance for $\gamma>1$, respectively $\gamma>\sqrt3$, whatever $h$ and $q$.
\end{lemma}

\begin{proof}
$\Lat$ is closed under multiplication by $x$. If $\sum_{i<n}c_ix^iv=0$ with $c\in\Q^n$ not all zero, then $c(x)\,v=0$ in $(\Q[x]/(\Phi))^k$, a vector space over the field $\Q[x]/(\Phi)$ in which $c(x)\neq0$ is invertible, so $v=0$; this proves linear independence. (a) Multiplication by $x$ sends the coefficient vector $(a_0,\dots,a_{n-1})$ of each component to $(-a_{n-1},a_0,\dots,a_{n-2})$, an isometry; applying it to a shortest nonzero vector gives $n$ linearly independent vectors of norm $\lambda_1$. (b) Since $x^p=x+1$, multiplication by $x$ sends $(a_0,\dots,a_{p-1})$ to $(a_{p-1},a_0+a_{p-1},a_1,\dots,a_{p-2})$, of squared norm at most $\|a\|^2+a_0^2+2a_{p-1}^2\leq3\|a\|^2$, using $(a_0+a_{p-1})^2\leq2a_0^2+2a_{p-1}^2$; applying this to a shortest nonzero vector $v$ gives a second, linearly independent vector of norm at most $\sqrt3\,\lambda_1$.
\end{proof}

The observation is elementary, and folklore among cryptanalysts; we state it to identify an obstruction to the direct ordinary-uSVP route. For $\Z[x]/(x^n-1)$ with $n$ an odd prime, used by NTRU-HPS and NTRU-HRSS, the ring is not a field, but multiplication by $x$ is still a permutation of coefficients: if $xv=\lambda v$ for a nonzero $v$, then $|\lambda|=1$; $\lambda=-1$ is impossible, since it would give $v=x^nv=-v$; so every nonzero lattice vector not fixed by multiplication by $x$ has a linearly independent rotation of the same norm. The Kannan embedding of an LWE, Ring-LWE or Module-LWE instance is not an $R$-module, its embedding coordinate is a single integer, so the rotation argument does not automatically obstruct those affine embeddings. For NTRU keys, the exact obstruction and excluded approximation factors depend on the ring, as specified above.

\paragraph{Ordinary and module uniqueness.}
In a rank-two module, many short vectors may belong
to one line over the number field while being linearly
independent over $\mathbb Q$. Ordinary Euclidean
uniqueness and module uniqueness are therefore different
promises. Felderhoff, Pellet-Mary and Stehl\'e study
module-uSVP and its relation to NTRU, including reductions
with ideal-SVP oracle requirements \cite{FPS22}.
These results do not directly instantiate Regev's
ordinary-uSVP promise, but they are relevant to any
proposed module-aware replacement for that route.

For the power-of-two cyclotomic NTRU lattice,
\cref{lem:ntru-sym} gives $\lambda_2=\lambda_1$.
For the NTRU Prime lattice it instead gives
$\lambda_2\leq\sqrt3\,\lambda_1$, excluding the ordinary
$\gamma$-uSVP promise for $\gamma>\sqrt3$.
The latter statement does not assert equality of the
two minima or exclude every $\gamma>1$.

\paragraph{What the subfield attack does not supply.}
The attack of Albrecht, Bai and Ducas \cite{ABD16} concerns
$h=g/f\bmod q$ and uses the multiplicative norm
$N_{K/L}(h)=N_{K/L}(g)/N_{K/L}(f)\bmod q$ in a proper subfield.
Their Theorem 2 requires coprimality of the normed ideals and a
sufficient length separation; it recovers a multiple of the normed
key, which must then be lifted. The asymptotic attack analysis uses
Heuristic 1 to control the sizes of norms and inverse norms.
As the authors emphasise in their introduction, the relative norm
is not additive. It therefore does not turn $As+e$ into a fresh
lower-dimensional Module-LWE sample with a controlled additive error,
and supplies neither the sample amplifier nor the modulus adapter
missing from a deployed Wen--Zheng application. Its overstretched
NTRU attacks are independent of the present DCP hypotheses and do
not assert ordinary-uSVP uniqueness of a ring-generated short submodule.

\paragraph{Reach.} The direct ordinary-uSVP route is excluded for Falcon's cyclotomic keys at factors above one and for NTRU Prime keys at factors above $\sqrt3$. The ratio $\mathrm{gh}(\Lat_h)/\|(f,g)\|$ in \cref{tab:gaps} is a density comparison, not that promise. A decision-NTRU distinguisher based on $\GapSVP$ could apply if the structured and uniform branches satisfy a proved length separation exceeding the available approximation bound; the displayed density estimates alone do not prove it. The Stehl\'e--Steinfeld variant \cite{SS11} has a separate Ring-LWE ciphertext connection, again requiring the embedding or oracle promises. Overstretched NTRU has independent classical subfield and dense-sublattice attacks \cite{ABD16,KF17}. None of these observations establishes deployed NTRU key recovery or decryption from the present hypotheses.

\paragraph{The indirect Ring-LWE route.}
The lossiness reduction gives a possible route to decision
NTRU through a search Ring-LWE solver, subject to its
sample and width requirements. It does not interpret
$(h,0)$ as an ordinary Ring-LWE sample and does not
recover an NTRU key. We have not established that the
dihedral algorithms considered here meet those oracle
requirements for a deployed NTRU distribution.
Consequently this connection gives no deployed-parameter
attack in the present analysis. The rotation obstruction
to the direct ordinary-uSVP argument remains separate. The precise sufficient variant below separates the sample and width requirements from the rotation obstruction.

\begin{lemma}
\label{lem:ntru-lossiness}
Let $R$ be a ring whose additive group has an integral basis
of rank $d$, let $q\geq2$, and use the Euclidean coefficient
norm in that basis. Write $R_q=R/qR$ and
\[
 D_{R,w}(e)\ \propto\ \exp(-\pi\|e\|^2/w^2).
\]
Suppose a polynomial-time quantum algorithm $\mathcal O$
recovers a uniform $s\in R_q$ with probability at least
$1-\varepsilon$ from $t$ Ring-LWE samples
$(a_i,a_is+e_i)$, where the $a_i$ are independent uniform
elements of $R_q$ and the $e_i$ are independent
$D_{R,w}$ errors.

Consider a distribution on tuples $(a_1,\ldots,a_t)$
such that every tuple admits $f,g_1,\ldots,g_t\in R$ with
$f\not\equiv0\pmod{qR}$,
\[
 a_if=g_i\pmod{qR}\quad(1\leq i\leq t),\qquad
 \sum_{i=1}^t\|g_i\|^2\leq B^2.
\]
The lossiness distinguisher below has, under ideal
Gaussian sampling, a difference in acceptance
probabilities between independent uniform $a_i$
and this distribution of at least
\[
 \frac12-\varepsilon-\sqrt{\frac{\pi}{8}}\frac{B}{w}.
\]
In particular, if $w\geq\sqrt{2\pi}\,B$ and
$\varepsilon\leq1/8$, the advantage is at least $1/8$.
Efficient approximate Gaussian sampling preserves
this conclusion up to its statistical error: joint
sampling error at most $\eta_{\rm samp}$ costs at
most $2\eta_{\rm samp}$ in the advantage.
\end{lemma}

\begin{proof}
Given $(a_i)$, sample $s$ uniformly and the errors
independently as above, call $\mathcal O$ on
$(a_i,a_is+e_i)_{i=1}^t$, and accept exactly when its
answer is $s$. Uniform inputs are accepted with
probability at least $1-\varepsilon$.

Fix a tuple satisfying the short-relation promise, and
let $P$ be the distribution of $(s,e_1,\ldots,e_t)$.
Define the translation
\[
 T(s,e_1,\ldots,e_t)
  =(s+f\bmod qR,e_1-g_1,\ldots,e_t-g_t).
\]
It preserves all pairs given to $\mathcal O$, while
changing the correct secret. Let $P'=T_*P$. The uniform
secret is invariant under translation. For the errors,
symmetry gives zero mean and hence, using natural
logarithms,
\[
 D_{\rm KL}(P\|P')
   =\frac{\pi}{w^2}\sum_i\|g_i\|^2.
\]
Pinsker's inequality therefore gives
\[
 \Delta(P,P')\leq\sqrt{\frac{\pi}{2}}\frac{B}{w}.
\]
Let $u$ be the recovery probability under $P$. Under
$P'$, the same observed pairs have the alternative
target $s+f$, distinct from $s$. Thus the two recovery
probabilities sum to at most one. They differ by at
most $\Delta(P,P')$, implying
$u\leq(1+\Delta(P,P'))/2$. Subtract this upper bound
from the uniform-input acceptance probability.
\end{proof}

For a single public key the lemma supplies only $t=1$: reusing $h$ does not give independent uniform $a_i$ in the uniform branch. If each $\|g_i\|\leq B_0$, the sufficient width is $w\geq\sqrt{2\pi t}\,B_0$. A short-relation promise holding with probability $1-\eta$ costs at most a further $\eta$ in advantage. The conclusion is a decision distinguisher, not recovery of a denominator.

\paragraph{Conclusion.} Deployed NTRU key recovery and decryption are not established by these reductions under any of the three hypotheses; they remain \Open{}. The direct ordinary-uSVP obstruction is exact for power-of-two cyclotomics and excludes factors above $\sqrt3$ for NTRU Prime. Decision NTRU has conditional gap-based and lossiness routes, with their respective promises and oracle requirements; neither is certified here for deployed parameters. Overstretched NTRU has independent classical attacks \cite{ABD16,KF17}. Module-uSVP and ideal-lattice algorithms remain relevant background \cite{FPS22,EHKS14,BS16,CDPR16,CDW17,PHS19}, without yielding the missing deployed key-recovery certificate.

% ==================================================================
\subsection{The Lattice Isomorphism Problem}\label{sec:lip}

\paragraph{The problem.} Two lattices are isomorphic if an orthogonal transformation maps one onto the other. Given bases $B$ and $B'=OBU$ with $O$ orthogonal and $U$ unimodular, the search Lattice Isomorphism Problem (LIP) asks for $O$, or equivalently for $U$. In the language of quadratic forms, which Ducas and van Woerden adopted because it removes the real-valued $O$ from the picture, the Gram matrices $Q=B^{\top}B$ and $Q'=B'^{\top}B'$ satisfy $Q'=U^{\top}QU$, and LIP asks for $U\in\mathrm{GL}_n(\Z)$; the decision version asks whether two forms are equivalent \cite{DvW22}.

\paragraph{Cryptography from LIP.} Ducas and van Woerden defined an average-case distribution on an isomorphism class by sampling a basis through discrete Gaussian sampling, proved a worst-case to average-case reduction \emph{within} a class, and built from LIP an identification scheme in the style of the graph-isomorphism protocol, a key encapsulation mechanism and a hash-and-sign signature scheme \cite{DvW22}. The idea of the KEM and of the signature is to hide a lattice with good decoding or sampling properties, for instance $\Z^n$ or a remarkable lattice, behind a random isomorphism, so that the secret isomorphism gives the decoder and the public form does not. Bennett, Ganju, Peetathawatchai and Stephens-Davidowitz studied rotations of $\Z^n$ specifically, the simplest case, with both algorithms and cryptographic constructions \cite{BGPS23}. Hawk, by Ducas, Postlethwaite, Pulles and van Woerden, is a signature based on module LIP over power-of-two cyclotomics: the secret is a basis of a rank-two module isomorphic to $\Z^{2n}$, and signing requires only discrete Gaussian sampling over $\Z^{2n}$, which is simple and fast \cite{DPPvW22}.

\paragraph{Reach.} For the isomorphism class of $\Z^n$, every rotation has all successive minima equal. Such an instance supplies no ordinary-uSVP promise above one. The guarantee of a polynomial-factor $\SIVP$ oracle alone also does not recover an exact isometry. These observations concern this class, not all LIP instances, and do not exclude other reductions. We record LIP-based schemes as \Open{} because the compositions established here provide no general attack certificate.

\paragraph{The 2026 attack on Hawk.} The concrete history is cautionary. Straznickas and Weis showed in July 2026 that key recovery for Hawk-$n$ reduces in polynomial time to $\SVP$ in dimension $n/2+1$, by exploiting the structure of the power-of-two cyclotomic field; this lowers the estimated key-recovery cost from $2^{150}$ to $2^{108}$ for Hawk-512 and from $2^{288}$ to $2^{182}$ for Hawk-1024, and a Hawk-256 key was recovered in hours on a single server \cite{SW26}. The attack does not apply to Falcon. NIST's page for the third round of the additional-signature process, which had admitted Hawk in May \cite{NIST26sig}, records that the submission team has withdrawn it \cite{NIST26r3}; the date of the team's announcement, July 29, 2026, is taken from secondary reporting \cite{HawkWd26}. Two observations matter here. First, the attack targets the algebraic structure of the module instance, not LIP as such. Second, it reduces key recovery to \emph{exact} $\SVP$ in half the dimension, the polynomial-factor guarantees of \cref{thm:reach} do not automatically solve that exact problem. Thus this published attack does not by itself supply an additional efficient attack under $\Hstar$. The episode shows that structured instances of young assumptions need independent scrutiny before they can carry the weight that ML-KEM and ML-DSA would leave behind.

% ==================================================================
\subsection{Heuristic gaps of deployed parameter sets}\label{sec:parameters}

\paragraph{Sufficient certificates and heuristic comparisons.}
Fix a failure exponent $f$ and a witnessing algorithm
$\mathcal A$. Let $\Gamma_{f,\mathcal A}(d)$ be an upper
bound on a sufficient ordinary-uSVP promise for the
resulting reduction. A key-recovery application requires
that its embedded planted vector be shortest and that
the actual gap exceed this bound. Failure to establish
that inequality gives no certificate from this theorem;
it does not prove failure of the algorithm on the
instance.

The model $c_0d^{1/2+2f}$ and the limit $f\to0^+$ are
used only for comparisons. The asymptotic constants
and the starting dimension may depend on $f$ and
$\mathcal A$. Gaussian-heuristic gaps are additional
estimates. None of these comparisons is a necessary
and sufficient cryptanalytic threshold.

\paragraph{Method.} \Cref{tab:gaps} applies \cref{prop:samples}\ref{it:gapm} to deployed parameter sets. For an instance with $n$ secret coordinates of standard deviation $\sigma_s$ and $m$ samples with error standard deviation $\sigma_e$, the embedding lattice has exact dimension $n+m+1$. We report $D=n+m$ as the first-order dimension in the numerical table. After rescaling the secret coordinates, its first-order heuristic gap is
\[
  g\;\approx\;\frac{q^{m/(n+m)}}{\sigma_s^{\,n/(n+m)}\,\sigma_e^{\,m/(n+m)}\,\sqrt{2\pi e}},
\]
which for keys with $m=n$ and $\sigma_s=\sigma_e=\sigma$ is $\sqrt{q/(2\pi e)}/\sigma$. For ML-KEM ciphertexts we take the component $u=A^{\top}r+e_1$ alone: an instance with $m=n$ samples in the fresh secret $r$, whose error is $e_1$ plus the compression error, of variance about $(q/2^{d_u})^2/12$; with the rounding rules of \cite{FIPS203} its exact variance is $0.92$ for $d_u=10$ and $0.38$ for $d_u=11$, the values we use. The component $v$ is not an LWE sample in $r$, since it contains the encoded message, but recovering $r$ from $u$ decrypts it; for ML-DSA we give the gap both for the published key, where the dropped low bits $t_0$ contribute an error of standard deviation $2^{13}/\sqrt{12}$, and for the full $t$. For NTRU keys the same formula, with secret $f$ and error $g$, gives the ratio $\mathrm{gh}(\Lat_h)/\|(f,g)\|$, which is not a unique-SVP gap (\cref{lem:ntru-sym}); we list it to show how dense the planted sublattice is. The column $g/\sqrt D$ compares the gap with the model $\sqrt D$ at $f\to0^+$ and $c_0=1$; it records the $c_0$ comparison that would suffice within this heuristic model at $f\to0^+$, without certifying a finite-dimensional bound (the NTRU density ratios are not ordinary-uSVP gaps). These are first-order estimates: they ignore sublattice selection and other refinements that a real attacker would optimise, treat compression errors as independent noise, and are given to two significant digits because nothing in them is more precise.

\begin{table}[htbp]
\centering\small
\renewcommand{\arraystretch}{1.15}
\setlength{\tabcolsep}{4pt}
\begin{tabular}{@{}lrrrrll@{}}
\toprule
\textbf{Instance} & $q$ & $D$ & $g$ & $g/\sqrt D$ & \textbf{$\HReg(f)$} & \textbf{Ball}\\
\midrule
ML-KEM-512 key & 3329 & 1024 & 11 & 0.36 & \Cond & no cert.\\
ML-KEM-768 key & 3329 & 1536 & 14 & 0.36 & \Cond & no cert.\\
ML-KEM-1024 key & 3329 & 2048 & 14 & 0.31 & \Cond & no cert.\\
ML-KEM-512 ctxt ($u$) & 3329 & 1024 & 11 & 0.33 & \Cond & no cert.\\
ML-KEM-768 ctxt ($u$) & 3329 & 1536 & 12 & 0.30 & \Cond & no cert.\\
ML-KEM-1024 ctxt ($u$) & 3329 & 2048 & 13 & 0.28 & \Cond & no cert.\\
FrodoKEM-640 key & $2^{15}$ & 1280 & 16 & 0.44 & \Cond & no cert.\\
FrodoKEM-976 key & $2^{16}$ & 1952 & 27 & 0.61 & \Cond & no cert.\\
FrodoKEM-1344 key & $2^{16}$ & 2688 & 44 & 0.85 & \Cond & no cert.\\
ML-DSA-44 key ($t_1$ only) & $\approx2^{23}$ & 2048 & 12 & 0.27 & \Cond & no cert.\\
ML-DSA-65 key ($t_1$ only) & $\approx2^{23}$ & 2816 & 14 & 0.26 & \Cond & no cert.\\
ML-DSA-87 key ($t_1$ only) & $\approx2^{23}$ & 3840 & 16 & 0.26 & \Cond & no cert.\\
ML-DSA-44 full $t$ supplied & $\approx2^{23}$ & 2048 & 500 & 11 & \Cond{} ($f\lesssim0.15$) & no cert.\\
ML-DSA-65 full $t$ supplied & $\approx2^{23}$ & 2816 & 560 & 11 & \Cond{} ($f\lesssim0.15$) & no cert.\\
ML-DSA-87 full $t$ supplied & $\approx2^{23}$ & 3840 & 840 & 14 & \Cond{} ($f\lesssim0.15$) & no cert.\\
Falcon-512 key (NTRU)\,$\ddag$ & 12289 & 1024 & 6.6 & 0.21 & \Open & n/a\\
Falcon-1024 key (NTRU)\,$\ddag$ & 12289 & 2048 & 9.4 & 0.21 & \Open & n/a\\
sntrup761 key (NTRU Prime)\,$\ddag$ & 4591 & 1522 & 23 & 0.59 & \Open & n/a\\
FHE example RLWE key & $2^{880}$ & $2^{16}$ & $\approx2^{436}$ & $\gg1$ & \Cond & no cert.\\
Many-sample, $q{=}3329$, $\sigma{=}1$ & 3329 & $\gg n$ & $\approx800$ & -- & $n$, $f$ & no cert.\\
\bottomrule
\end{tabular}
\caption{First-order heuristic gap estimates and NTRU density ratios, with parameters from \cite{FIPS203,FIPS204,FrodoKEM21,Falcon20,BCLV17}. The dimension column is $D=n+m$; the exact affine embedding has dimension $D+1$ and determinant $\tau q^m$ (\cref{prop:samples}). Its exact $g_\tau$ expression applies when the first-order approximation is unjustified. All numerical comparisons are conditional; $g/\sqrt D$ is a model normalisation, not an attack threshold. Full-$t$ rows assume that additional input is supplied; the displayed $f$ range also assumes $c_0=1$. ``No cert.'' means the restricted ball construction or its displayed width and noise bounds supply no certificate; it is not a claim about every construction in \cite{BKSW18} or all EDCP algorithms; ``n/a'' means no direct application is made here. $\ddag$: NTRU values are density ratios, not ordinary-uSVP gaps; Falcon has $\lambda_2=\lambda_1$, and NTRU Prime has $\lambda_2\leq\sqrt3\lambda_1$. The FHE row is an illustrative parameter set with error standard deviation $3.2$.}
\label{tab:gaps}
\end{table}

\paragraph{Interpretation.} The tables give sufficient comparisons under explicit models, not necessary attack thresholds. The cyclic hypothesis $\HHS$ supplies polynomial-factor lattice algorithms, but the unknown witnessing solver leaves its quantitative factor unspecified. The structured certificate audit in \cref{sec:wz-audit} supplies no deployed attack; the remaining polynomial-modulus adapter questions are recorded in \cref{q:wz}. An explicit $\Gamma_{f,\mathcal A}$ would resolve one input to the embedding comparison; actual gaps, distributions and sample availability remain to be checked. The FHE row is an illustrative parameter set, and the full-$t$ ML-DSA rows assume extra input. NTRU's density ratios do not remove its ring-dependent uniqueness obstruction. Absence of a certificate records the limits of this analysis.

% ==================================================================
\section{Alternatives outside lattices}\label{sec:non-lattice}
% ==================================================================

\Cref{tab:nonlattice} summarises this section. The common thread is one question that each family must answer: does a complete efficient reduction to the assumed coset-state problem exist, rather than just a hidden-shift subroutine? For each family we describe the underlying problem and the main constructions, representative classical and quantum attacks, and the reason for its conclusion.

\begin{table}[htbp]
\centering\footnotesize
\renewcommand{\arraystretch}{1.02}
\setlength{\tabcolsep}{3.5pt}
\begin{tabularx}{\textwidth}{@{}L{2.9cm}cccL{1.7cm}Y@{}}
\toprule
\textbf{Family (examples)} & $\HHS$ & $\HReg$ & $\HEDCP$ & \textbf{Evidence} & \textbf{Problem and cited attacks}\\
\midrule
Goppa-code McEliece (Classic McEliece) & \vS & \vS & \vS & under revision\,\dag & Decoding random linear codes (NP-hard in the worst case \cite{BMvT78}); Grover-accelerated information-set decoding \cite{Ber10,KT17}; the hiding resists strong Fourier sampling on single coset states \cite{DMR11}. Classical structural recovery estimates (per run $2^{89}$--$2^{128}$ across cost models \cite{Wei26,FV26mceliece}) depend on unverified structural premises, repetitions and the memory model (\cref{sec:mceliece-status}); small challenge recoveries are distinct evidence.\\
Quasi-cyclic codes (HQC, BIKE) & \vS & \vS & \vS & moderate & Quasi-cyclic syndrome decoding; ISD with Grover. The algebraic structure is the analogue of ideal lattices and deserves monitoring.\\
LPN-based PKE (Alekhnovich) & \vS & \vS & \vS & moderate & LPN at noise rate about $1/\sqrt n$ \cite{Ale03}; noisy binary dihedral cosets are LPN (\cref{lem:binary}), an equivalence rather than a barrier.\\
Rank-metric codes (RYDE, Mirath) & \vS & \vS & \vS & moderate & Rank decoding, MinRank; the $\BDD$-to-EDCP reduction of \cite{CMY26} is effective only where decoding is already easy.\\
Lee-metric codes & \vK\bR & \vK\bR & \vK\bR & limited & The Lee-metric connections to lattices and EDCP \cite{CMY26} require the stated metric, distribution, parameter and solver-domain promises; no family-wide conclusion follows.\\
Hash-based (SLH-DSA, XMSS, LMS) & \vS & \vS & \vS & high & Second preimage and related properties; generic Grover and BHT benchmarks \cite{Gro96,BHT98,Ber09}.\\
Symmetric ZK (FAEST, Picnic) & \vS & \vS & \vS & high (AES) & AES or LowMC relation plus hash, commitment and proof-transform assumptions; MPC- and VOLE-in-the-head \cite{IKOS07,Picnic17,FAEST23}.\\
Multivariate (UOV, MAYO, QR-UOV, SNOVA; MQOM) & \vS & \vS & \vS & moderate & MQ (NP-hard \cite{GJ79}) with trapdoor structure; Grover-enhanced algebraic attacks \cite{BY18}. Structural breaks have been classical \cite{Beu22,TPD21}.\\
Code-in-the-head (SDitH) & \vS & \vS & \vS & high & Syndrome decoding of random codes \cite{FJR22}.\\
Cyclic isogeny actions (CSIDH-type, CSI-FiSh) & \vC\bR & \vC\bR & \vS & limited & Cyclic DCP via noiseless cosets, conditional on the complete action, public-input and coherent-evaluation requirements (\cref{sec:isogenies}) \cite{CJS14,BKV19}. Non-cyclic actions require $\HHS^{\mathrm{ab}}$. Known subexponential attacks \cite{Kup05,Pei20}. The polynomial-modulus $\HEDCP$ does not supply this route.\\
Supersingular endomorphism ring (SQIsign) & \vS & \vS & \vS & moderate & Endomorphism ring and isogeny path problems \cite{EHLMP18,Wes21}; generic quantum attacks in time $\Otil(p^{1/4})$ \cite{BJS14}, claw finding \cite{JS19}.\\
Higher-dimensional isogeny PKE (FESTA, POK\'E) & \vS & \vS & \vS & limited & Non-commutative isogeny problems with torsion information \cite{BMP23,BM25}; young assumptions.\\
Equivalence group actions (LESS, MEDS, ALTEQ) & \vS & \vS & \vS & moderate to limited & Code, matrix-code and trilinear-form equivalence. Measurement limitations are proved for specific symmetric-group and code-equivalence settings \cite{MRS08,HMRRS10,DMR11}; no corresponding general bound is established here for MEDS or ALTEQ.\\
\bottomrule
\end{tabularx}
\caption{Non-lattice families under the three hypotheses; glyphs and basis marks as in \cref{tab:collapse}. The glyph \vS{} records no certificate from the analysed reductions, not security. The fifth column is a qualitative assessment of the cited evidence, not a security level or ranking: \emph{high}, an extensively studied underlying problem with generic quantum benchmarks in the cited analyses; \emph{moderate}, a studied problem with structured variants, algorithmic progress or breaks of related schemes; \emph{limited}, a young assumption or known subexponential attack; \emph{under revision} \dag, active structural attacks (\cref{sec:mceliece-status}). Primitive-specific assumptions and measurement-model qualifications remain essential.}
\label{tab:nonlattice}
\end{table}

% ==================================================================
\subsection{Code-based cryptography}\label{sec:codes}

\paragraph{Decoding problems.} A linear $[n,k]$ code over $\F_q$ is a $k$-dimensional subspace $C\subseteq\F_q^n$, given by a generator matrix $G\in\F_q^{k\times n}$ or a parity-check matrix $H\in\F_q^{(n-k)\times n}$ with $C=\{c:Hc^{\top}=0\}$. The \emph{syndrome decoding} problem asks, given $H$, a syndrome $s$ and a weight $w$, for an error $e$ of Hamming weight $w$ with $He^{\top}=s$. Berlekamp, McEliece and van Tilborg proved its decision version NP-complete \cite{BMvT78}, whereas cryptographic security requires hardness of the particular average-case distribution. Worst-case NP-completeness alone does not establish average-case hardness for random codes, Goppa codes or deployed parameters.

\paragraph{McEliece and Niederreiter.} In McEliece's 1978 cryptosystem the secret is a binary Goppa code, defined by a support and a Goppa polynomial of degree $t$ over $\F_{2^m}$, for which an efficient decoder corrects $t$ errors; the public key is a disguised generator matrix $G_{\mathrm{pub}}=SGP$ with $S$ invertible and $P$ a permutation; a message $\mu$ is encrypted as $c=\mu G_{\mathrm{pub}}+e$ with $e$ of weight $t$, and decrypted by undoing $P$, decoding, and undoing $S$ \cite{McE78}. Niederreiter's dual version publishes a disguised parity-check matrix and encrypts by encoding the message into the error vector, the ciphertext being its syndrome \cite{Nie86}. Classic McEliece, a modern instantiation, is a Niederreiter-style KEM with binary Goppa codes and a public key in systematic form; its smallest parameter set, \texttt{mceliece348864}, has $n=3488$, $m=12$ and $t=64$, a public key of $261\,120$ bytes and a ciphertext of $96$ bytes \cite{ClassicMcEliece}. Classic McEliece was a fourth-round NIST candidate that NIST did not select \cite{NISTIR8545}; it was standardised by ISO/IEC in June 2026 and has been recommended by the BSI since 2020, a recommendation now restricted (\cref{sec:mceliece-status}) \cite{ISO18033Amd2,BSI26}.

\paragraph{Information-set decoding.} The generic attack on syndrome decoding goes back to Prange \cite{Prange62}: choose $k$ coordinates at random, hope that the error vanishes on them, and solve the resulting linear system; each attempt succeeds with probability $\binom{n-k}{t}/\binom{n}{t}$, so the expected cost is exponential in $t$. Its refinements lower the exponent but keep the structure of a search over a large space of guesses. That structure is exactly what Grover's algorithm accelerates: Bernstein showed that a Grover search over Prange's guesses halves the exponent \cite{Ber10}, and Kachigar and Tillich obtained small further gains from quantum walks applied to the advanced variants \cite{KT17}. Until recently no structural attack came close: the best key recovery for Classic McEliece parameters, a brute-force search combined with Sendrier's support splitting algorithm, cost more than $2^{3000}$ even for the smallest parameter set \cite{Sen00,BSI26}, and the security levels of the scheme were set by generic decoding \cite{ClassicMcElieceSec}. \Cref{sec:mceliece-status} summarises the subsequent structural attacks and their evidential limits, which a companion paper analyses in detail \cite{FV26mceliece}.

\paragraph{Quantum Fourier sampling.} In a code-equivalence formulation where both $G$ and $G_{\mathrm{pub}}$ are supplied, recovering a hiding pair $(S,P)$ can be expressed as a hidden-subgroup problem over a group built from $\mathrm{GL}_k(\F_2)$ and the symmetric group $S_n$. Dinh, Moore and Russell showed that, for well-permuted and well-scrambled codes, this hidden-subgroup problem cannot be solved by strong Fourier sampling, the measurement strategy behind the efficient abelian algorithms, or by any measurement of a single coset state, extending the negative results for graph isomorphism to subgroups of arbitrary structure, including automorphism groups of codes \cite{DMR11}, in line with the general limitations of coset states for $S_n$ \cite{MRS08,HMRRS10} discussed in \cref{sec:frontier}. That formulation supplies the hidden generator $G$ as an additional input; it is not itself a key-recovery reduction from the McEliece public key alone. These results limit a class of measurements on coset states; they are not lower bounds against every quantum algorithm, and they do not by themselves establish quantum security.

\paragraph{Quasi-cyclic codes: HQC and BIKE.} HQC works in $\F_2[x]/(x^n-1)$ with $n$ prime. The secret key is a pair of sparse polynomials $(x,y)$ and the public key is $(h,s=x+hy)$ for uniform $h$; to encrypt $\mu$, one samples sparse $r_1,r_2,e$ and sends $u=r_1+hr_2$ and $v=\mu G+sr_2+e$, where $G$ generates a \emph{public} decodable code; the receiver computes $v-uy=\mu G+(xr_2-r_1y+e)$, whose error term has low weight, and decodes \cite{HQC}. Since the decoding code is public, security rests on the hardness of decoding random quasi-cyclic codes rather than on hiding a structured code. NIST selected HQC for standardisation in March 2025 \cite{NISTIR8545}, with a draft standard (FIPS 207) in preparation. BIKE uses quasi-cyclic moderate-density parity-check codes: the secret is a pair of sparse polynomials $(h_0,h_1)$, the public key is $h=h_1h_0^{-1}$, and decryption uses bit-flipping decoding, whose failure rate must be negligible for CCA security \cite{BIKE}; it was not selected \cite{NISTIR8545}. The BSI notes that HQC is not affected by the 2024--2026 structural attacks on McEliece (\cref{sec:mceliece-status}) \cite{BSI26}. Quasi-cyclic codes are the code-based analogue of ideal lattices; no attack certificate is established by the DCP routes reviewed here, and structured decoding requires its own analysis (\cref{sec:open}).

\paragraph{LPN.} Learning Parity with Noise, $\mathrm{LPN}_{n,\tau}$, asks to recover $s\in\F_2^n$ from samples $(a,\langle a,s\rangle+\epsilon)$ with uniform $a$ and Bernoulli noise $\epsilon$ of rate $\tau$; it is the binary analogue of LWE and the average-case form of decoding a random linear code. Alekhnovich built public-key encryption from LPN with noise rate about $1/\sqrt n$ \cite{Ale03}. Brakerski, Lyubashevsky, Vaikuntanathan and Wichs gave a worst-case hardness result for LPN, through a smoothing lemma for codes: a worst-case nearest-codeword problem for balanced codes, with relative error about $\log^2n/n$, can be solved with an LPN distinguisher at the very high noise rate $1/2-1/\poly(n)$; the worst-case problem is only mildly hard, it has a quasi-polynomial algorithm and lies in $\mathrm{Search}\text{-}\mathsf{BPP}^{\mathsf{SZK}}$, and LPN at noise rate $1/2-1/\poly(n)$ already implies symmetric cryptography. They also showed that LPN with the very low noise rate $\log^2(n)/n$ implies collision-resistant hash functions \cite{BLVW19}. These results give the code-based family a foundation independent of lattices at the extremes of the noise range, but not a worst-case foundation for the noise rates used by public-key encryption.

\paragraph{The dihedral route beyond lattices.} The $q$-ary lattice of an LPN instance contains $2\Z^m$, so its minimum is at most $2$ while the error has norm $\Theta(\sqrt{\tau m})$; when $\tau m$ grows, this typically violates the small-radius Euclidean $\BDD$ promise used here. This embedding therefore supplies no certificate from \cref{thm:reach}. The dihedral route is nevertheless not intrinsically lattice-specific. Debris-Alazard, Remaud and Tillich gave a quantum reduction from finding short codewords in a random linear code to decoding, by adapting to the Hamming metric the reading of Regev's quantum step as a reduction from short vectors to closest vectors \cite{DRT21}; the coarser metric required decoding beyond the unique-decoding radius, errors uniform of fixed weight instead of Bernoulli, and an analysis with orthogonal polynomials. In September 2026 Cartor, Manganiello and Youmans gave a general quantum reduction from $\BDD$ over finite abelian groups with translation-invariant metrics to EDCP, using locality-sensitive hashing to play the role of Regev's partition into cubes; they instantiated it for the Hamming, rank and $p$-Lee metrics and recovered the LWE-to-EDCP reduction through the equivalence between lattice problems and Lee-metric decoding \cite{CMY26}. The authors note that for the Hamming and rank metrics their reductions yield a non-negligible number of EDCP states only in parameter regimes where known decoding algorithms are already efficient. These instantiations supply no attack certificate for cryptographic Hamming- or rank-metric parameters here. The Lee-metric connection must be matched to the exact metric, distributions, parameters and EDCP solver domain before transferring a lattice consequence. The following lemma describes what a dihedral problem over a binary group amounts to, in one natural noise model.

\begin{lemma}[Binary dihedral cosets are LPN]\label{lem:binary}
Over the group $\F_2^n$, consider the dihedral coset problem in which each register is, independently, with probability $1-\tau$ a good register $\tfrac{1}{\sqrt2}(\ket{0,y}+\ket{1,y+s})$ with uniformly random $y$, and with probability $\tau$ a uniformly random computational-basis state $\ket{b,x}$. This problem is equivalent, under quantum polynomial-time reductions in both directions, to search-$\mathrm{LPN}_{n,\tau/2}$. In particular the noiseless problem ($\tau=0$) is solved by Simon's algorithm \cite{Sim97}.
\end{lemma}

\begin{proof}
Write $\ket{\hat c,\hat z}:=H^{\otimes(n+1)}\ket{c,z}$. Expanding a good register in this basis, the coefficient of $\ket{\hat c,\hat z}$ is proportional to $(-1)^{\langle z,y\rangle}\big(1+(-1)^{c+\langle z,s\rangle}\big)$. Averaging over uniform $y$ kills every off-diagonal term between distinct $z,z'$, since it carries the phase $(-1)^{\langle z\oplus z',y\rangle}$; hence the density matrix of a good register is the classical mixture of $\ket{\hat c,\hat z}$ with $z$ uniform and $c=\langle z,s\rangle$. A bad register is maximally mixed, i.e.\ the uniform mixture over all $\ket{\hat c,\hat z}$, for which $c=\langle z,s\rangle$ with probability $1/2$. Each register is therefore \emph{exactly} the mixture $\sum_{z,c}\Pr[z,c]\,\ket{\hat c,\hat z}\!\bra{\hat c,\hat z}$ with $(z,c)$ distributed as an LPN sample of noise rate $\tau/2$. Measuring in the Hadamard basis turns registers into LPN samples, and preparing $\ket{\hat c,\hat z}$ from an LPN sample $(z,c)$ turns samples into registers with the correct density matrix; both maps are efficient. For $\tau=0$, $n+O(1)$ samples determine $s$ by linear algebra.
\end{proof}

The lemma is an equivalence, not a barrier. In its noise model, uniform independent offsets and maximally mixed bad registers; adversarial bad registers are not covered, a noise-tolerant dihedral algorithm over $\F_2^n$ is an LPN algorithm and conversely, and the noiseless problem is easy. It says nothing about whether LPN or decoding is hard, and it does not exclude attacks on binary codes through reductions of another kind. What it shows is narrower: this noisy binary-coset formulation supplies no attack without an LPN solver, so a reduction from binary decoding to a dihedral problem over $\F_2^n$ would lead back to a decoding problem. The hypotheses of \cref{sec:hyp-def} specify cyclic coset-state or Gaussian extrapolated-state interfaces; they do not automatically provide this noisy binary solver. The surveyed reductions supply no attack certificate for the cryptographic Hamming-decoding regimes considered \cite{CMY26}. We therefore record no DCP-based certificate for Classic McEliece, HQC or BIKE here. Their independent attack analyses remain necessary (\cref{tab:nonlattice}).

\paragraph{Rank and Lee metrics.} Rank-metric codes are subspaces of $\F_{q^m}^n$ whose vectors are viewed as $m\times n$ matrices over $\F_q$, with the rank as weight; their decoding problem and the related MinRank problem underlie the signature schemes RYDE and Mirath, second-round candidates of the additional-signature process that did not advance to the third round \cite{NIST26sig}. As noted above, the reduction of \cite{CMY26} reaches rank-metric decoding only in easy regimes, so these assumptions remain outside the reach. The Lee-metric connections are conditional on the metric, distributions and parameter promises of \cite{CMY26}; their applicability must be checked for each proposed cryptographic instance (\Cond).

\paragraph{Code-based signatures.} Syndrome Decoding in the Head, by Feneuil, Joux and Rivain, is a Fiat--Shamir signature from a zero-knowledge proof of knowledge of a syndrome-decoding solution, built with the MPC-in-the-head technique described in \cref{sec:hash} \cite{FJR22}; it is a third-round candidate \cite{NIST26sig}.

% ==================================================================
\subsection{Classic McEliece: structural cryptanalysis and current status}\label{sec:mceliece-status}

Between 2024 and October 2026 the structural cryptanalysis of McEliece with binary Goppa codes moved further, at the parameters of Classic McEliece, than in the previous four decades. A companion paper \cite{FV26mceliece} analyses these works in detail, version by version; here we summarise what matters for the classification. The attacks are classical attacks on the hidden Goppa description, not applications of a DCP solver; none uses a hidden-shift structure, so none changes a classification of this survey. They bear, however, on the independent evidence recorded for Classic McEliece in \cref{tab:nonlattice}.

\paragraph{What has been shown.} Randriambololona's syzygy distinguisher is asymptotically subexponential in the error-correcting capacity $t$ \cite{Ran25,BSI26}. The higher-order-vanishing distinguisher of Hemmert and Wiemers led to key recoveries by Hemmert and by Saarinen's HOVER algorithm, which solved public TII challenges up to TII-252 without threatening Classic McEliece parameters \cite{HW26,Hem26,Saa26a}. A Pfaffian attack breaks very high-rate Goppa codes, including the CFS signature scheme, but does not apply at the rates of Classic McEliece \cite{BLT26,BSI26}, and a heuristic key recovery conjectured to be subexponential remains far above the targets, its distinguishing step alone costing about $2^{454}$ at level 1 \cite{BLRT26,Wei26}. The Holdout/Hermite distinguisher of Ghoshal, Ishai, Jain and Sun is quasi-polynomial in the asymptotic Classic McEliece regime, at an estimated $2^{114}$--$2^{124}$ bit operations for the five parameter sets \cite{GIJS26}, and the lowest estimates for key recovery rest on it. Saarinen's itemised conditional costing gives $2^{126.8}$--$2^{145.2}$ bit operations in a model that excludes memory traffic \cite{Saa26cost}. Weis estimates one run of the underlying sparse linear algebra at $2^{89}$--$2^{98}$ bit operations, at $2^{107}$--$2^{117}$ when each operation is charged the square root of the memory it addresses, and at $2^{117}$--$2^{128}$ under the rule of the Classic McEliece security guide, with one run or $103$ to $1443$ runs per key recovery \cite{Wei26}; Vedenev's revision of October 7 reports conditional estimates of $2^{89.9}$--$2^{98.6}$ \cite{Ved26}. Public challenges with brute-force labels up to $254$ have been solved along the way \cite{BLT26,Hem26,Saa26a,Wei26,Ved26,Saa26b}; these labels are brute-force work factors, not security levels.

\paragraph{What has not been shown.} No key recovery has been run at a Classic McEliece size. The lowest estimates rest on heuristic assumptions, on kernel dimensions, the span of gradients, the structure of a column kernel and the behaviour of the iterative solver, tested on instances with at most about $10^7$ unknowns, whereas a full-size run has $2^{38}$ to $2^{42}$. For one and the same run, the published figures differ by about thirty bits between cost models, and a run would need $100$ TiB to $1.5$ PiB of memory and, on a mesh at $10^9$ steps per second, between $7$ and $420$ years \cite{Wei26}. Against Grover-accelerated information-set decoding, about $2^{101}$ to $2^{183}$ logical gates for the five sets \cite{Ber10,Wei26}, the conditional structural estimates are lower for the four larger sets in essentially every cost model and comparable for \texttt{mceliece348864}; Weis knows of no quantum speed-up for the structural computation \cite{Wei26}.

\paragraph{Standardisation and guidance.} NIST did not select Classic McEliece \cite{NISTIR8545}; ISO/IEC standardised selected parameter sets in June 2026, excluding \texttt{mceliece348864} \cite{ISO18033Amd2,ClassicMcElieceISO}. In its note of 29 September 2026 the BSI records that the published results do not yet enable a practical attack on its recommended parameter sets. It nevertheless advises against Classic McEliece in new developments, recommends FrodoKEM, ML-KEM or, once its standard is published, HQC for key agreement, limits the current recommendation for Classic McEliece to hybrid use with a classical scheme, and expects to revise it in early 2027; it states that FrodoKEM, ML-KEM and HQC are not affected by these results \cite{BSI26}. A classical component can preserve a classical security guarantee if the combiner has the required property; it alone cannot provide a post-quantum guarantee against Shor's algorithm.

\paragraph{How the evidence fits together.} \Cref{tab:mceliece-evidence} separates the questions that are easily conflated.

\begin{table}[htbp]
\centering\small
\renewcommand{\arraystretch}{1.15}
\begin{tabularx}{\textwidth}{@{}L{3.15cm}Y Y@{}}
\toprule
\textbf{Evidence} & \textbf{What it supports} & \textbf{What remains separate}\\
\midrule
GIJS distinguisher \cite{GIJS26} & Recognition of Goppa public keys in the stated asymptotic regime. & Efficient decryption and equivalent-key recovery need additional arguments.\\
HOVER and TII recoveries \cite{Saa26a,Saa26b,Wei26,Ved26} & Demonstrated recovery of particular small public challenge keys. & Success and resources for a Classic McEliece parameter distribution.\\
Full-size cost analyses \cite{Saa26cost,Wei26,Ved26} & Conditional arithmetic estimates for specified algorithms. & Structural success, repetitions, memory costs and implementation validation.\\
Fourier-sampling limitations \cite{DMR11} & Limits on specified quantum measurements in a code-equivalence formulation. & Security against classical structure recovery or arbitrary quantum algorithms.\\
BSI guidance \cite{BSI26} & A dated deployment recommendation responding to the developing evidence. & A practical key-recovery demonstration or a general security theorem.\\
\bottomrule
\end{tabularx}
\caption{Evidence categories for Classic McEliece, as checked on October 11, 2026. No row by itself settles the other questions.}
\label{tab:mceliece-evidence}
\end{table}

\paragraph{Conclusion.} We retain \SurviveDag. The surveyed DCP reductions supply no attack certificate for Classic McEliece (\cref{sec:codes}, with \cref{lem:binary} and the discussion following it). What changes is the independent evidence recorded in the fifth column of \cref{tab:nonlattice}: the lowest published estimates are no longer those of generic decoding with a Grover speed-up but those of classical structural attacks whose concrete cost is under revision, and the dagger in \cref{tab:nonlattice,tab:sizes} records exactly this. This is neither an endorsement for new deployment nor a claim of a practical full-size break. Two consequences follow for the rest of the survey. First, the KEM layer of a post-DCP world becomes thinner: among standardised and standard-track KEMs, HQC, whose decoding code is public and whose security rests on decoding random quasi-cyclic codes, is the code-based KEM outside the reach of the counterfactual that, according to the BSI, is not affected by these attacks \cite{BSI26}. This raises the stakes of \cref{q:codes}; HQC still needs its own security analysis, since its lack of exposure to these particular Goppa attacks is not a blanket security guarantee. FrodoKEM and ML-KEM keep the status given in \cref{sec:collapse-standards}. Second, the hybrid designs of \cref{sec:hybrids} would take HQC as the code-based component to analyse first, with Classic McEliece used, if at all, as an additional component pending re-evaluation, in line with the BSI's advice. The episode also illustrates again a lesson of \cref{sec:intro-alarms}: weaknesses have appeared where algebraic structure is hidden, in ideal lattices, in overstretched NTRU (\cref{sec:ntru}), and now in the Goppa structure of McEliece, and step 9 of the scenario-planning checklist (\cref{tab:checklist}) is not a formality.

% ==================================================================
\subsection{Hash-based and symmetric-key signatures}\label{sec:hash}

\paragraph{One-time and few-time signatures.} Lamport's one-time signature takes as secret key $2\ell$ random strings $x_{i,b}$ and as public key their images $y_{i,b}=H(x_{i,b})$; a message digest $d\in\{0,1\}^\ell$ is signed by revealing $x_{i,d_i}$ for each $i$ \cite{lamport1979}. A key can be used once. Winternitz-type schemes replace the $2\ell$ strings by hash chains, trading signature size for computation; WOTS$^+$ adds randomising bitmasks so that collision resistance of the hash function is not needed \cite{hulsing2013w}.

\paragraph{Merkle trees.} Merkle authenticated $2^h$ one-time public keys with a binary hash tree whose root is the long-term public key; a signature consists of a one-time signature and the authentication path of its leaf \cite{Mer89}. The resulting schemes are \emph{stateful}: the signer must never reuse a leaf. XMSS \cite{RFC8391} and LMS \cite{RFC8554} are standardised stateful schemes, approved by NIST together with requirements on state management \cite{SP800208}.

\paragraph{Stateless signatures: SPHINCS$^+$ and SLH-DSA.} SPHINCS removed the state by using a hypertree of Merkle trees, a pseudorandomly chosen leaf for each message, and a few-time signature at the bottom \cite{BHHLNPSW15}; SPHINCS$^+$ improved it with the FORS few-time scheme and tweakable hash functions \cite{bernstein2019}. FIPS 205 standardises it as SLH-DSA, with SHA-2 and SHAKE instantiations at three security levels, each in a ``small'' and a ``fast'' variant; at level 1 the public key has $32$ bytes and signatures $7\,856$ bytes (small) or $17\,088$ bytes (fast) \cite{FIPS205}. Security reduces to properties of the hash function, pseudorandomness, second-preimage resistance and related properties of tweakable hashes, without a lattice or class-group hardness premise.

\paragraph{Quantum attacks.} Generic quantum attack benchmarks for hash functions include the following: Grover search finds preimages in about $2^{n/2}$ evaluations \cite{Gro96}, and the Brassard--H{\o}yer--Tapp algorithm finds collisions in about $2^{n/3}$ evaluations using comparable quantum-accessible memory \cite{BHT98}. Bernstein argued that, once hardware costs are accounted for, the latter is no better than parallel classical collision search \cite{Ber09}. These benchmarks do not use a DCP oracle. The reductions analysed here establish no additional hash attack; this does not rule out attacks exploiting a particular hash function or future algorithms.

\paragraph{Signatures from symmetric primitives.} The MPC-in-the-head paradigm of Ishai, Kushilevitz, Ostrovsky and Sahai turns a secure multiparty computation of a relation into a zero-knowledge proof: the prover simulates the parties ``in its head'', commits to their views, and opens a subset chosen by the verifier \cite{IKOS07}. Applied to the relation $y=E_k(x)$ for a block cipher $E$, with Fiat--Shamir, it yields signatures based on the one-wayness of $k\mapsto E_k(x)$ together with the commitment, hash and proof-transform assumptions: Picnic used LowMC \cite{Picnic17}, and FAEST uses AES with the more efficient VOLE-in-the-head technique \cite{FAEST23}. FAEST, MQOM (the same paradigm for multivariate quadratic systems) and SDitH are third-round candidates \cite{NIST26sig}.

\paragraph{The encryption question.} The signature constructions discussed here do not supply a KEM. Merkle-puzzle key agreement illustrates a restricted query-complexity gap, whose quantum treatment depends on the access model \cite{BHKKLS11}; it is not an impossibility theorem for every construction using symmetric primitives. This survey establishes no KEM based solely on the one-wayness premise used by these signatures.

% ==================================================================
\subsection{Multivariate cryptography}\label{sec:mq}

\paragraph{The MQ problem and its trapdoors.} Multivariate schemes rely on the problem of solving a system of $m$ quadratic equations in $n$ variables over a finite field (MQ), which is NP-hard in the worst case \cite{GJ79}. A trapdoor is introduced by composing an easily invertible quadratic map $\mathcal F$ with secret affine maps, $\mathcal P=\mathcal T\circ\mathcal F\circ\mathcal S$. In Patarin's oil-and-vinegar construction the variables are split into ``oil'' and ``vinegar'' variables and the central polynomials contain no oil$\times$oil terms, so that fixing the vinegar variables leaves a linear system in the oil variables \cite{Pat97}. Kipnis and Shamir broke the balanced version, with as many oil as vinegar variables, by recovering the hidden oil subspace \cite{KS98}; Kipnis, Patarin and Goubin proposed the unbalanced version, UOV, with more vinegar than oil variables, which has resisted for over twenty-five years \cite{KPG99}.

\paragraph{The record of the family.} The failures of the family are instructive. Rainbow, a multi-layer version of UOV \cite{DS05} and a NIST finalist, was broken by Beullens with a key-recovery attack that takes about a weekend on a laptop at security level 1 \cite{Beu22}. HFE uses a low-degree univariate central map over an extension field \cite{Pat96}. Tao, Petzoldt and Ding give key-recovery attacks on the HFE variants analysed in \cite{TPD21}; this does not justify a blanket claim about every HFE-based signature or parameter set. These attacks and the Rainbow break exploit structure classically. The candidates considered here include UOV and compact variants: MAYO, which uses a small oil space and recovers invertibility by ``whipping'' the public map into several copies \cite{Beu21}, QR-UOV and SNOVA, all third-round candidates together with MQOM \cite{NIST26sig}.

\paragraph{Quantum attacks.} The cited quantum improvements include Grover-enhanced algebraic methods, which accelerate the exhaustive-search component of hybrid Gröbner-basis and linearisation attacks \cite{BY18}. The reviewed routes give no DCP-based attack certificate. They also do not establish a multivariate KEM as a replacement for the schemes compared here.

% ==================================================================
\subsection{Isogenies: commutative versus non-commutative}\label{sec:isogenies}

The relevant distinction is between complete action-based reductions and attacks that contain only a hidden-shift subroutine. Calling a construction commutative or non-commutative does not alone verify either its applicability or its resistance to a DCP solver.

\paragraph{Background.} An isogeny is a non-constant morphism between elliptic curves that preserves the point at infinity; a separable isogeny is determined, up to isomorphism of the target, by its finite kernel, and Vélu's formulas compute it in time polynomial in the kernel size. The endomorphism ring of an elliptic curve over a finite field is either an order in an imaginary quadratic field (ordinary curves, and the $\F_p$-rational endomorphisms of supersingular curves over $\F_p$) or a maximal order in the quaternion algebra ramified at $p$ and $\infty$ (supersingular curves over $\F_{p^2}$). The first case carries a commutative group action; the second does not.

\paragraph{Commutative group actions.} When the relevant endomorphism ring is an order $\mathcal O$ in an imaginary quadratic field, the class group $\mathrm{cl}(\mathcal O)$ acts freely and transitively on the curves with that ring: an ideal class $[\mathfrak a]$ maps $E$ to the target of the isogeny with kernel $E[\mathfrak a]$. Couveignes \cite{Cou06} and Rostovtsev and Stolbunov \cite{RS06} proposed Diffie--Hellman key exchange from this ``hard homogeneous space'' on ordinary curves. CSIDH made it efficient with supersingular curves over $\F_p$: for $p=4\ell_1\cdots\ell_n-1$ with small odd primes $\ell_i$, the ideals above $\ell_i$ act through $\ell_i$-isogenies computable over $\F_p$, a secret key is an exponent vector $(e_1,\dots,e_n)$ in a small box, and the public key is the curve $[\mathfrak l_1^{e_1}\cdots\mathfrak l_n^{e_n}]\star E_0$, represented by one coefficient of a Montgomery equation \cite{CLMPR18}. CSIDH-512 uses $n=74$ primes and exponents in $[-5,5]$, with $64$-byte public keys. Alamati, De Feo, Montgomery and Patranabis abstracted the properties such actions provide into a framework of cryptographic group actions and built further primitives on it \cite{ADMP20}. Beullens, Kleinjung and Vercauteren computed the class group of the CSIDH-512 order, it is cyclic, of order roughly $\sqrt p$, together with a reduced basis of its relation lattice, which gives every class a canonical representation and makes Fiat--Shamir signatures with uniformly random challenges possible (CSI-FiSh) \cite{BKV19}.

\paragraph{Quantum attacks: a hidden shift.} Childs, Jao and Soukharev observed that recovering the class $[\mathfrak a]$ with $E'=[\mathfrak a]\star E$ is an abelian hidden shift problem: the functions $x\mapsto x\star E$ and $x\mapsto x\star E'$ on the class group differ by a shift by $[\mathfrak a]$. With Kuperberg's or Regev's sieve this gives a subexponential quantum algorithm, under the generalised Riemann hypothesis for the evaluation of the action \cite{CJS14}. This attack supplies a concrete quantum benchmark for CSIDH, analysed by Bonnetain and Schrottenloher \cite{BS20} and, with Kuperberg's collimation sieve, by Peikert, who estimated that CSIDH-512 key recovery needs only about $2^{16}$ quantum evaluations of the group action using $2^{40}$ bits of quantumly accessible classical memory, which challenges the level-1 estimate under that evaluation and memory model; an evaluation count is not an end-to-end gate or hardware cost \cite{Pei20}; the cost of a single quantum evaluation of the action was analysed by Bernstein, Lange, Martindale and Panny \cite{BLMP19}. Since the CSIDH-512 class group is cyclic \cite{BKV19}, coset sampling turns its key recovery into a cyclic DCP instance with perfect coset states (\cref{lem:hs-dcp}), so it is covered by $\HHS$ when the complete evaluation and public-input conditions below are met.

\paragraph{Actions satisfying the evaluation hypotheses.}
The hidden-shift argument applies to free and
transitive actions with an efficiently available
group decomposition and efficient coherent evaluation
on arbitrary group elements, using a consistent
encoding of the acted-on objects. The inputs required
by the evaluation algorithm must be public or
efficiently recoverable from the public instance.
For example, the cited Clapoti result assumes an
effective primitive orientation \cite{PR23}.
The application must verify this requirement, as well
as the relevant cyclic or general abelian group
hypothesis. We therefore state the consequence for
actions satisfying these conditions, rather than
for every construction described as commutative.

\paragraph{SIDH and its break.} The Supersingular Isogeny Diffie--Hellman protocol of Jao and De Feo avoided commutativity: each party computes a secret isogeny of smooth degree from a fixed supersingular curve and publishes its target together with the images of a basis of a torsion subgroup, which the other party needs to complete the exchange \cite{JDF11}. Its KEM, SIKE, reached the fourth round of the NIST process. In 2022 it was broken classically: Castryck and Decru recovered keys in polynomial time when the starting curve has a known endomorphism ring, Maino, Martindale, Panny, Pope and Wesolowski gave a direct attack, and Robert extended the attack to polynomial time for every starting curve \cite{CD23,MMPPW23,Rob23}. All three exploit the torsion-point images through higher-dimensional isogenies, without using a DCP oracle.

\paragraph{Non-commutative problems.} The supersingular $\ell$-isogeny graph over $\F_{p^2}$ is a Ramanujan graph, and walks in it define the hash function of Charles, Goren and Lauter, whose collision resistance rests on the difficulty of finding paths \cite{CGL09}. Eisentr\"ager, Hallgren, Lauter, Morrison and Petit showed, heuristically, that finding isogeny paths and computing endomorphism rings reduce to each other \cite{EHLMP18}, and Wesolowski proved the equivalence under the generalised Riemann hypothesis \cite{Wes21}. The quaternion side is easy: Kohel, Lauter, Petit and Tignol gave a heuristic polynomial-time algorithm that connects maximal orders by ideals of $\ell$-power norm \cite{KLPT14}; the hardness lies in passing from a curve to its endomorphism ring. Generic attacks find a path between two supersingular curves in time $\Otil(p^{1/2})$ classically; Delfs and Galbraith gave an $\Otil(p^{1/4})$ classical algorithm for curves defined over $\F_p$ \cite{DG16}, and Biasse, Jao and Sankar a quantum algorithm in time $\Otil(p^{1/4})$ for arbitrary supersingular curves, which walks from each curve to one defined over $\F_p$ with a Grover-accelerated search and then solves the commutative problem between the two $\F_p$-curves with a subexponential hidden-shift algorithm \cite{BJS14}. For the specific claw-finding instances of SIKE, Jaques and Schanck showed that quantum claw-finding is no better than the classical parallel collision search once memory costs are accounted for \cite{JS19}.

\paragraph{What the counterfactual changes.} In the algorithm of \cite{BJS14}, a hidden-shift subroutine handles the commutative problem between $\F_p$-rational curves. Replacing that subroutine under $\HHS$, or $\HHS^{\mathrm{ab}}$ for a non-cyclic class group, requires the complete public-input and evaluation conditions above. If those conditions hold, the unchanged outer Grover search still has cost $\Otil(p^{1/4})$. This describes that particular composition; it is neither an optimality statement for supersingular path finding nor a bound on every attack a hypothetical solver might enable.

\paragraph{SQIsign.} SQIsign, by De Feo, Kohel, Leroux, Petit and Wesolowski, is a Fiat--Shamir signature from an identification protocol: the prover knows a secret isogeny from a curve $E_0$ with known endomorphism ring to its public curve $E_A$, hence the endomorphism ring of $E_A$; it commits to a random isogeny $E_0\to E_1$, receives a challenge isogeny $E_1\to E_2$, and responds with an isogeny $E_A\to E_2$ of smooth degree computed on the quaternion side with a variant of KLPT and translated back to curves \cite{DKLPW20}. Security rests on the endomorphism ring problem and on the response distribution revealing nothing useful. Later versions use higher-dimensional isogeny representations to make signing and verification efficient \cite{SQIsignSpec}; SQIsign is a third-round candidate, with the smallest combined size of public key and signature among the schemes compared in \cref{tab:sizes} \cite{NIST26sig}.

\paragraph{Encryption after SIDH.} Post-SIDH public-key encryption uses higher-dimensional isogenies constructively: FESTA turns the machinery of the torsion attacks into a trapdoor \cite{BMP23}, and POK\'E is a compact PKE built on similar techniques \cite{BM25}. Both rely on non-commutative problems with controlled torsion information; they are \Survive{} in our classification, but their assumptions are young.

% ==================================================================
\subsection{Equivalence problems and non-abelian group actions}\label{sec:equiv}

\paragraph{The problems.} Three families of signature schemes use group actions of non-abelian groups. In linear code equivalence, given generator matrices $G,G'$ of two $[n,k]$ codes, one seeks $S\in\mathrm{GL}_k(\F_q)$ and a monomial matrix $Q$ (a permutation combined with nonzero scalings) such that $G'=SGQ$; LESS builds a Fiat--Shamir signature from the corresponding identification protocol, in which the prover commits to a random code equivalent to $G$ and reveals, according to the challenge, an equivalence with $G$ or with $G'$ \cite{BMPS20}. In matrix code equivalence the codes are subspaces of $\F_q^{m\times n}$ and the acting group is $\mathrm{GL}_m(\F_q)\times\mathrm{GL}_n(\F_q)$, acting by $C\mapsto ACB$; this is the basis of MEDS \cite{MEDS23}. In alternating trilinear form equivalence, $\mathrm{GL}_n(\F_q)$ acts on alternating trilinear forms by change of variables; this is the basis of ALTEQ \cite{ALTEQ22}. LESS was a second-round candidate and MEDS and ALTEQ first-round candidates of NIST's additional-signature process; none is among the third-round candidates \cite{NIST26sig}; that selection outcome alone establishes neither hardness nor a cryptanalytic break.

\paragraph{Conclusion.} The natural quantum approach to these problems is a hidden-subgroup or hidden-shift problem in non-abelian groups such as $\mathrm{GL}_k\times S_n$, for which the negative results proved for the symmetric group and for code equivalence, strong Fourier sampling on single coset states, and entangled measurements on few of them, fail \cite{MRS08,HMRRS10,DMR11}, suggest similar limits; for matrix-code and trilinear-form equivalence we know of no corresponding proof. The reductions analysed here give no attack certificate (\Survive); the Fourier-sampling results are independent evidence that limits a class of measurements, not a consequence of the counterfactual and not a proof of quantum security. They are the natural candidates when one wants a group-action interface, for ring signatures or threshold constructions, say, without the commutativity that makes CSIDH fragile under $\HHS$.

% ==================================================================
\subsection{Efficiency and concrete security of the alternatives}\label{sec:efficiency}

\Cref{tab:sizes} reports public-key and ciphertext or signature sizes for selected parameter sets; ML-KEM-768 is included as a level-3 reference. Classic McEliece has a large public key and short ciphertext, HQC a smaller key and larger ciphertext, and FrodoKEM keys and ciphertexts near ten kilobytes. These sizes affect protocol design, but do not establish deployability or a security ranking. In particular, FrodoKEM's lack of a certificate from the restricted ball route does not prove a stronger position than ML-KEM under the counterfactual. Their unresolved checks differ (\cref{sec:collapse-lwe,sec:wz-audit}). Classic McEliece also requires the independent structural-attack assessment summarised in \cref{sec:mceliece-status}; the companion paper \cite{FV26mceliece} compares the structural estimates with classical and Grover-accelerated generic decoding for each parameter set.

For signatures, SLH-DSA uses small public keys and kilobyte-scale signatures, UOV trades larger keys for short signatures, and MAYO and SQIsign have the displayed size profiles. FAEST additionally illustrates signatures using a symmetric relation, with the hash, commitment and proof-transform assumptions described above. Security comparisons require family-specific attacks and cost models: information-set decoding \cite{Ber10,KT17}, generic preimage and collision search \cite{Gro96,BHT98,Ber09}, algebraic attacks on MQ \cite{BY18}, and isogeny-path algorithms \cite{BJS14,JS19}. Their resource measures and hypotheses differ; no uniform ``at most Grover'' bound or common security ordering follows from this table.

\begin{table}[htbp]
\centering\small
\renewcommand{\arraystretch}{1.2}
\setlength{\tabcolsep}{4pt}
\begin{tabular}{@{}llrrccc@{}}
\toprule
\textbf{Scheme (level 1)} & \textbf{Family} & \textbf{pk (B)} & \textbf{ct / sig (B)} & $\HHS$ & $\HReg$ & $\HEDCP$\\
\midrule
\multicolumn{7}{@{}l}{\emph{Key encapsulation}}\\
ML-KEM-768 (ref., level 3) & module lattice & 1\,184 & 1\,088 & \vK\bH & \vK\bH & \vK\bH\\
FrodoKEM-640 & unstructured lattice & 9\,616 & 9\,720 & \vK\bH & \vK\bH & \vS\\
Classic McEliece 348864 & binary Goppa codes\,\dag & 261\,120 & 96 & \vS & \vS & \vS\\
HQC-1 & quasi-cyclic codes & 2\,241 & 4\,433 & \vS & \vS & \vS\\
\midrule
\multicolumn{7}{@{}l}{\emph{Signatures}}\\
ML-DSA-44 (reference) & module lattice & 1\,312 & 2\,420 & \vK\bH & \vK\bH & \vK\bH\\
FN-DSA-512 (reference) & NTRU lattice & 897 & $\approx$666 & \vO & \vO & \vO\\
SLH-DSA-128s / 128f & hash-based & 32 & 7\,856 / 17\,088 & \vS & \vS & \vS\\
UOV (OV-Ip-pkc) & multivariate & 43\,576 & 128 & \vS & \vS & \vS\\
MAYO-1 & multivariate & 1\,420 & 454 & \vS & \vS & \vS\\
SQIsign (NIST-I, v3.0) & isogenies & 83 & 200 & \vS & \vS & \vS\\
\bottomrule
\end{tabular}
\caption{Sizes of level-1 parameter sets (ML-KEM-768 shown as the commonly deployed reference). Sources: the normative standards FIPS 203--205 \cite{FIPS203,FIPS204,FIPS205}; for FN-DSA, the Falcon specification \cite{Falcon20}, the predecessor of the forthcoming FIPS 206; the FrodoKEM specification \cite{FrodoKEM21}; the Open Quantum Safe parameter tables for Classic McEliece, HQC, UOV and MAYO \cite{OQS}; and the SQIsign specification v3.0 \cite{SQIsignSpec}. Sizes of round-3 candidates are those of their specifications as of October 2026 and may change during the NIST process. The last three columns give the conclusion under each hypothesis, with the glyphs and marks of \cref{tab:collapse}; the dagger is explained in \cref{tab:nonlattice}.}
\label{tab:sizes}
\end{table}

% ==================================================================
\subsection{Subset sum}\label{sec:subsetsum}

In the subset-sum problem modulo $N$ one is given $a_1,\dots,a_n\in\Z_N$ and a target $t$ and asked for a subset $S$ with $\sum_{i\in S}a_i\equiv t\pmod N$; the \emph{density} $n/\log N$ governs its difficulty. Regev showed that an average-case subset-sum algorithm yields a DCP algorithm with failure parameter $f=1$ \cite[Thm. 1.3]{Reg04}; the implication runs from subset sum to DCP, so this direction alone does not turn a DCP algorithm into a subset-sum algorithm, and the 2026 claim tried precisely to replace this step (\cref{sec:simon}). Lyubashevsky, Palacio and Segev built public-key encryption and other primitives provably as secure as subset sum \cite{LPS10}. In the other direction, Lagarias and Odlyzko showed that subset sum of sufficiently low density reduces to a short-vector problem in a lattice \cite{LO85}, and the densities that make the cryptosystems of \cite{LPS10} secure are close to the lattice regime. We classify subset-sum public-key encryption as \Open.

% ======================================================================
\section{Notes on non-commutative hidden structures}\label{sec:frontier}
% ======================================================================

Non-abelian hidden-subgroup and hidden-shift problems are one research direction for assumptions outside the dihedral routes analysed here. The evidence includes limitations on particular quantum measurements, alongside algorithms that solve some non-abelian cases. Neither noncommutativity nor a measurement lower bound establishes resistance to every DCP-enabled attack. This section reviews those results, the record of group-based constructions and relaxed dihedral algorithms. It does not derive a new deployable KEM or a formal separation from DCP.

\subsection{Where the hidden subgroup problem is easy}\label{sec:frontier-easy}

\paragraph{Query complexity is not the obstacle.} Ettinger, H{\o}yer and Knill showed that for every finite group polynomially many coset states determine the hidden subgroup information-theoretically \cite{EHK04}. The difficulty of the HSP is therefore computational: the question is always whether an efficient measurement extracts the information that the states contain.

\paragraph{Normal subgroups and weak Fourier sampling.} Hallgren, Russell and Ta-Shma showed that measuring only the name of the irreducible representation in the Fourier transform of a coset state, \emph{weak} Fourier sampling, determines the normal core of the hidden subgroup, and hence solves the HSP for normal subgroups whenever the Fourier transform over the group can be implemented efficiently \cite{HRT03}. Dihedral reflections are not normal, which is why the method fails for DCP.

\paragraph{Strong Fourier sampling and optimal measurements.} Moore, Rockmore, Russell and Schulman showed that \emph{strong} Fourier sampling, which also measures inside the representation, in a suitably chosen basis, solves the HSP in affine groups and in $q$-hedral groups $\Z_p\rtimes\Z_q$ with $q$ large, and gives algorithms for related hidden-shift problems \cite{MRRS07}. Bacon, Childs and van Dam took the opposite route: they computed the optimal (pretty good) measurement on several coset states at once for semidirect products $\Z_p^r\rtimes\Z_p$, including the Heisenberg group, and reduced its implementation to solving systems of polynomial equations, which is efficient for constant $r$ \cite{BCvD05}. Ivanyos, Sanselme and Santha solved the HSP in extraspecial groups \cite{ISS07}, and Friedl, Ivanyos, Magniez, Santha and Sen solved the hidden translation problem, the non-abelian analogue of the hidden shift, in solvable groups of bounded exponent and derived length \cite{FIMSS03}. Most recently, Morales gave polynomial-time algorithms for semidirect products $A\rtimes\Z_{p^k}$ with scalar action and for quasi-Hamiltonian groups given in a structured presentation \cite{Mor26}.

\paragraph{Scope of these algorithms.} A cryptographic problem with a complete efficient reduction to one of these efficiently solvable HSP families is exposed when the reduction meets the algorithm's subgroup, presentation and access promises. A group name alone is insufficient: these results neither solve every problem over that group nor exclude all cryptographic assumptions defined using it.

\subsection{Limitations of known measurements}\label{sec:frontier-hard}

\paragraph{The symmetric group.} For the order-two hidden subgroups used to express graph isomorphism, Moore, Russell and Schulman show that measuring a single coset state in any basis yields only exponentially small information \cite{MRS08}. Hallgren, Moore, R\"otteler, Russell and Sen extended this to entangled measurements: in their coset-state measurement model, polynomial-resource procedures require measurements entangled across $\Omega(n\log n)$ registers for the specified hidden-subgroup family \cite{HMRRS10}. Graph isomorphism itself is classically solvable in quasi-polynomial time by Babai's algorithm \cite{Bab16}, which shows that hardness of the quantum HSP route does not imply hardness of the underlying problem; code equivalence has its own negative results with specific code and measurement premises \cite{DMR11}. We establish no corresponding general lower bound for matrix-code or trilinear-form equivalence in \cref{sec:equiv}.

\paragraph{Sample complexity of hidden shifts.} Zhou studied the generalised hidden shift problem over arbitrary finite groups and its sample complexity, obtaining for the symmetric group $S_n$ a lower bound of $\Omega(\sqrt{n\log n})$ in the saturated case \cite{Zho26}. Such results bound sample requirements in their specified models. They should be distinguished from computational lower bounds and from the state budgets of cryptographic reductions.

\paragraph{Infinite groups.} Kuperberg extended the hidden subgroup problem to infinite groups, including Lie groups and arithmetic groups \cite{Kup25}. This suggests non-abelian analogues of lattice problems, defined by discrete subgroups of non-abelian Lie groups, to which the Euclidean-lattice reduction of \cref{sec:regev02} does not directly apply; a separation from other dihedral routes remains unproved. Dutto, Mercuri, Murru and Romano surveyed the hidden subgroup problem from a mathematical and cryptographic perspective \cite{DMM25}.

\subsection{Group-based cryptography}\label{sec:frontier-groups}

Braid-group key exchange is the classical example of a non-commutative assumption. Anshel, Anshel and Goldfeld proposed a key exchange based on commutators in a non-commutative group \cite{AAG99}, and Ko, Lee, Cheon, Han, Kang and Park a Diffie--Hellman analogue based on the conjugacy search problem in braid groups \cite{KLCHKP00}. Most such proposals were broken by length-based attacks and by attacks through linear representations of braid groups, surveyed by Myasnikov, Shpilrain and Ushakov \cite{MSU08}. The lesson is that non-commutativity is not sufficient: an assumption meant to survive a DCP break also needs average-case hardness under classical attack, which must be assessed independently for each equivalence problem in \cref{sec:equiv}. Nor is it necessary: hash-based and code-based assumptions lie outside the reach of the hypotheses without any non-abelian structure.

\subsection{Relaxed dihedral problems}\label{sec:frontier-relaxed}

The dihedral family itself is not monolithic. Childs and van Dam solved the generalised hidden shift problem in polynomial time once the extrapolation range is a polynomial power of the modulus \cite{CvD07}; Chen, Liu and Zhandry solved EDCP for suitable amplitude functions by filtering \cite{CLZ22}; Bai, Jangir, Kirshanova, Ngo and Youmans solved EDCP over power-of-two moduli in quasi-polynomial time with quasi-polynomially many samples \cite{BJKNY25}; and Boucher, Fouque and Shen solved a cyclotomic generalisation of DCP with a quantum sieve in time $2^{O_p(\log n\log q)}$ \cite{BFS26}. The direct compositions examined in \cref{rem:kuperberg} give no attack on the standard LWE regimes considered here: their sample budgets, amplitudes, widths, moduli and running times do not all match. This does not exclude every LWE parameter family or alternative reduction. Characterising compatible amplitude profiles, moduli and sample counts could identify further positive compositions.

\subsection{Criteria for evaluating alternative assumptions}\label{sec:frontier-needs}

\paragraph{Research considerations.} Three checks can guide the assessment of alternatives, without providing necessary or sufficient conditions for DCP resistance. First, audit all known complete reductions and their access requirements; non-abelian structure and limits on Fourier sampling do not establish a separation. Second, analyse average-case classical and quantum attacks at the intended distribution and parameters. Third, identify the primitive required: an assumption supporting signatures need not supply a KEM. Equivalence-based signatures and higher-dimensional isogeny encryption illustrate different goals and different unresolved security questions. Formal independence from DCP would require a precisely stated model and a separate separation argument.

% ======================================================================
\section{Future research and conclusions}\label{sec:open}
% ======================================================================

This section collects questions that could clarify or extend the classification, draws scenario-planning consequences for hybrid designs, turns the classification into a scenario-planning checklist with worked examples, and concludes.

\subsection{Open questions}\label{sec:open-questions}

The questions below are stated with the background needed to attack them. They are ordered by their consequences for the classification.

\begin{question}[Few-sample reductions]\label{q:fewsample}
Is there a quantum reduction from LWE with $m=n+O(n/\log n)$ samples and $\alpha q=O(1)$ to an extrapolated dihedral problem with non-negligible extrapolation?
\end{question}

\emph{Background.} Both unstructured routes lose a factor $q^{n/m}$ (\cref{sec:samples}), and this loss contributes to the failure of the displayed sufficient bounds for FrodoKEM and the published keys of ML-DSA. Finite-dimensional factors, structure and rounded distributions require additional checks. The loss reflects the fact that the error must be small relative to the minimum of $\Lambda_q(A)$, whose Gaussian-heuristic value scales as $q^{1-n/m}$. A reduction exploiting the short secret, as the normal-form transformation \cite{ACPS09} and affine embedding do, might avoid this particular loss. We establish no such few-sample EDCP construction here.

\begin{question}[The finite-dimensional factor of Regev's theorem]\label{q:constant}
What finite-dimensional factor $\Gamma_{f,\mathcal A}(d)$ does Regev's reduction \cite[Thm. 1.1]{Reg04} guarantee, as a function of $f$, of the dimension and of the parameters of the postulated solver, and how far can it be reduced by optimising the partition, the embedding into $\Z_N$ and the analysis?
\end{question}

\emph{Background.} \Cref{tab:gaps,tab:sisreach} compare the estimates with $c_0d^{1/2+2f}$. In the formal normalisation $f\to0^+$, the sufficient model comparisons correspond to $c_0$ below about $0.3$--$0.9$ for the LWE examples and $0.6$--$0.7$ for ML-DSA's Module-SIS norm model. The constants and starting dimension need not be uniform in $f$. An explicit finite-dimensional bound would resolve one input to the comparison; the actual planted vector, gap, distribution and sample conditions would still be required. It would neither remove NTRU's ordinary-uniqueness obstruction nor implement the hypothetical solver.

\begin{question}[DCP and SIS at small norm bounds]\label{q:sis}
Does a DCP algorithm, under $\HReg(f)$ or $\HEDCP[\ell]$, yield an algorithm for $\SIS_{n,m,q,\beta}$ with $\beta$ below $\beta_{\mathrm{void}}$ ($\Otil(n^{1+2f})$ under $\HReg(f)$, $\Otil(\ell n^{3/2})$ under $\HEDCP[\ell]$), in particular at the norm bounds of the signatures of \cref{tab:sisreach}? Does it yield forgeries, which need inhomogeneous solutions for hashed targets and, for ML-DSA, solutions compatible with its rounding and hints? Conversely, can one establish a separation from the specified DCP capabilities in a precisely defined computational or oracle model?
\end{question}

\emph{Background.} Under $\HEDCP[\ell]$, SIS is solved for $\beta>\beta_{\mathrm E}=\Otil(\ell n^{3/2})$, essentially where its worst-case foundation fails under that hypothesis; under $\HReg(f)$, for $\beta\geq\beta_{\mathrm{Reg}}=O(m^{3/2+2f})$, a factor of about $\sqrt n$ above it (\cref{prop:sis-reg,prop:sis-edcp}, \cref{sec:sis}). Solving the exact SIS assumption of a security proof would remove its hardness premise; a forgery needs more (\cref{sec:sis}). The compositions reviewed here do not certify the signature parameters below those bounds \cite{SSTX09,CLZ22,DFS24,CH26}. The dequantisation in \cite{KOW26} concerns particular wide-matrix $\SIS^\infty$ algorithms; it does not preclude every quantum advantage in such regimes.

\begin{question}[An EDCP characterisation of NTRU]\label{q:ntru}
Is small-modulus NTRU equivalent, or reducible, to an extrapolated dihedral problem, in the way that LWE and Module-LWE are \cite{BKSW18,WZ26}?
\end{question}

\emph{Background.} The exact rotation bounds in \cref{lem:ntru-sym} obstruct a direct ordinary-uSVP application at the stated factors. The homogeneous relation $hf-g=0$ is not a standard Ring-LWE sample. The lossiness route of \cref{lem:ntru-lossiness} requires a search Ring-LWE solver at its particular sample count and width; a single key gives only one candidate ring element. Module-uSVP is a different promise \cite{FPS22}. A new structured route or a verified composition meeting these oracle requirements could change the decision classification; key recovery requires its own reduction.

\begin{question}[From cyclic DCP to vector EDCP]\label{q:translate}
Does an algorithm for $\DCP_N^f$ on cyclic groups imply one for $\mathrm{G}\text{-}\EDCP_{n,N,r}$ over $\Z_N^n$ on the parameter families of \cite{BKSW18}?
\end{question}

\emph{Background.} \Cref{lem:dcp-edcp} reduces EDCP with smooth amplitudes to robust two-point solvers over $\Z_N^n$, so the question is whether such solvers follow from cyclic ones. Regev's embedding of $\{0,\dots,M-1\}^n$ into $\Z_N$ (\cref{sec:regev02}) handles differences with small coordinates, which is the situation of the two-point problem but not of an arbitrary secret in $\Z_N^n$. An affirmative answer on selected parameter families would establish only a restricted implication. To derive $\HReg(f)\Rightarrow\HEDCP[\ell]$, it would need to cover the full latter domain with its secret, offset, width, state-budget and success guarantees; the compositions of \cref{cor:edcp-sivp,prop:sis-edcp} require their respective subdomains.

\begin{question}[Robust extrapolated solvers]\label{q:robust}
Can the reduction from LWE to Gaussian EDCP be arranged so that a solver tolerating a constant fraction of invalid states suffices, removing the factor $\ell$ from the admissible width, and so that the radius condition (B3) no longer forces $m\leq m_{\max}$ (\cref{rem:robust})? Conversely, is robustness to invalid states harder to achieve for extrapolated states than for the two-point problem behind $\HReg(f)$?
\end{question}

\emph{Background.} At fixed width and inverse-polynomial success, \cref{prop:ell-lower} shows that $\ell=\Omega(n\log q)$ states are necessary, so requiring all of them to be valid costs that factor in the noise rate reached; the radius condition costs a further factor of about $\sqrt n$ up to logarithms (\cref{cor:edcp-ceiling,rem:robust}). Removing the first loss alone would bring the route to noise rates of about $1/\Otil(n)$. Within this width calculation, removing both losses would be needed to certify the standard noise rate $\alpha=1/\Otil(\sqrt n)$: the regime of Regev's and dual-Regev encryption and of identity-based encryption in the style of \cite{GPV08}, regimes without a certificate from the stated compositions (\cref{sec:collapse-lwe}).

\begin{question}[Structured codes]\label{q:codes}
Is there a hidden-structure attack on quasi-cyclic syndrome decoding analogous to the quantum attacks on principal ideals \cite{CDPR16,BS16}? Can the $\BDD$-to-EDCP reductions of \cite{CMY26} be made effective for Hamming-metric decoding in cryptographic regimes over large alphabets, where \cref{lem:binary} does not apply?
\end{question}

\emph{Background.} Code-based KEMs are candidates for independent assessment in this scenario (\cref{sec:non-lattice}), so their structured variants deserve the scrutiny that ideal lattices received after 2014. Over $\F_2^n$ the noisy dihedral problem of \cref{lem:binary} is LPN, so that formulation gives no new route; this is an equivalence, not a protection, and it does not exclude reductions of another kind. Over larger alphabets the EDCP route of \cite{CMY26} is a genuinely different problem.

\begin{question}[Structural attacks on Classic McEliece]\label{q:goppa}
Do the Holdout/Hermite key-recovery methods \cite{GIJS26,Saa26cost,Saa26b,Ved26,Wei26} succeed on the Classic McEliece parameter distributions, with what structural success probabilities and end-to-end costs including repetitions and realistic memory access? In particular, do the heuristic assumptions behind the lowest estimates, on the dimension of the row kernel, the span of the gradients, the structure of the column kernel and the behaviour of block Lanczos over $\F_2$ \cite{Wei26}, hold at $2^{38}$ to $2^{42}$ unknowns? Can the reported small-instance recoveries be extended without assuming the full-size rank and locator-uniqueness properties? Does a quantum computer accelerate the underlying sparse linear algebra over $\F_2$? Is there a regime of binary Goppa codes, for example with $n<2^m$ and larger $m$, or another family of codes with an efficient decoder, that resists all known structural attacks with public keys of comparable size?
\end{question}

\emph{Background.} These attacks are classical and Goppa-specific, so they do not rely on the counterfactual, but they remain relevant to assessing the alternative KEMs (\cref{sec:mceliece-status}); the companion paper \cite{FV26mceliece} states this question in more detail. Their asymptotic claims are strong, subexponential and quasi-polynomial distinguishers \cite{Ran25,GIJS26}, while their concrete estimates rest on unreviewed analyses and on heuristic assumptions tested only on small instances, of the order of $10^7$ unknowns at most \cite{Saa26a,Saa26b,Wei26}. For one and the same run, the published figures differ by about thirty bits between cost models \cite{Wei26,FV26mceliece}. The scaling of the run with $m$ suggests that larger extension degrees with $n<2^m$ could restore the margins at public keys several times larger, but Weis presents such sets as illustrations, not proposals \cite{Wei26}. Weis knows of no quantum speed-up for a run; the question matters here because the counterfactual presupposes large quantum computers, and for \texttt{mceliece348864} the structural attack and Grover-accelerated generic decoding are already of comparable cost under the stated assumptions (\cref{sec:mceliece-status}). The BSI expects to revise its recommendation in early 2027 \cite{BSI26}.

\begin{question}[Gapless lattice problems]\label{q:gapless}
Can LIP for $\Z^n$ or for remarkable lattices be related to $\uSVP$ or $\SIVP$ at polynomial factors? Is there a DCP-type quantum reduction for LIP? The 2026 attack on Hawk \cite{SW26} reduces structured LIP to exact $\SVP$ in half the dimension; is there an analogous reduction for unstructured LIP?
\end{question}

\emph{Background.} Rotations of $\Z^n$ have no ordinary uniqueness gap, and the approximation guarantees established here do not recover their exact isometries (\cref{sec:lip}). A reduction to a promised lattice problem would give a positive composition only if its parameters matched the available oracle. A relation to hidden structure over a group such as $\mathrm{GL}_n(\Z)$ would similarly require an explicit access model and reduction.

\begin{remark}\label{q:wz}
The forward parameter audit is given in \cref{sec:wz-audit}. The remaining verification concerns the Appendix E adapter of \cite{WZ26}: propagate every search/decision conversion, unknown-noise sample amplification, Gaussian width convention and success loss, and verify that the final dimension, polynomial modulus, width and state count meet the same witnessing $\HEDCP[\ell]$ solver. This is a source-proof and composition audit, not a new mathematical open problem.
\end{remark}

The audit must include the $\Omega(\sqrt\kappa)$ input-width condition of the invoked EDCP-to-LWE reduction and the two calls to the unknown-noise routine, whose stated sample overhead is polynomial rather than sample preserving. A corrected adapter would still not remove the failure of the printed forward certificate at ML-KEM parameters proved in \cref{prop:wz-certificate}. Extending the forward range, handling actual bounded and rounded distributions, or improving the width bound are separate research questions.

% ==================================================================
%\subsection{Priority directions}\label{sec:priorities}

%We prioritise these directions by their role in resolving the classification, without claiming a quantitative ranking of practical impact: (i) improve the structured forward range and complete the adapter audit of \cref{q:wz}, keeping the obstruction in \cref{prop:wz-certificate} separate; (ii) bound Regev's finite-dimensional factor (\cref{q:constant}), one of several unresolved inputs to concrete embedding comparisons; (iii) investigate few-sample, cyclic-to-vector and robust-state compositions (\cref{q:fewsample,q:translate,q:robust}); (iv) verify or extend the NTRU routes (\cref{q:ntru}); (v) assess concrete structural attacks on Goppa and quasi-cyclic codes independently (\cref{q:goppa,q:codes}); (vi) analyse the specific assumptions and complete attack routes of non-abelian-action and isogeny constructions; and (vii) characterise algorithms for relaxed dihedral states whose full resource requirements match a cryptographic reduction (\cref{sec:frontier-relaxed}).

% ==================================================================
\subsection{Hybrid designs}\label{sec:hybrids}

The observations in this subsection and the next are scenario planning under the counterfactual. They do not provide a migration recommendation. X-Wing combines ML-KEM-768 and X25519 with a SHA3-based combiner and a security argument tailored to those components \cite{XWing24}. Against a sufficiently capable quantum adversary, the discrete-logarithm premise of X25519 is unavailable \cite{Sho97}; any remaining security argument must use the ML-KEM and combiner assumptions in the relevant model. This does not establish exact equality between the concrete security of the hybrid and of ML-KEM. ML-KEM's classification under the present hypotheses remains conditional (\cref{q:wz,q:constant}).

General KEM combiners can preserve security when at least one component remains secure, subject to their component, validation and security-model requirements \cite{GHP18,BBFGS19}. Combining a lattice KEM with a code-based KEM such as HQC is therefore one design to analyse; merely pairing two schemes does not establish the combiner theorem's premises. ML-KEM with FrodoKEM diversifies structure and parameters while retaining LWE-family dependencies, without proving independence of the two assumptions. For signatures, an analogous robust combination requires both component signatures to verify for the same bound message and identity, and an appropriate composition argument. The sizes in \cref{tab:sizes} inform bandwidth assessment but do not by themselves select a secure combination.

% ==================================================================
\subsection{Scenario-planning checklist}\label{sec:checklist}

The checklist in \cref{tab:checklist} records sufficient consequences and unresolved premises for each assumption in a system. A negative or unverified answer gives no attack certificate from that step. Its purpose is scenario planning and dependency accounting.

\begin{table}[htbp]
\centering\small
\renewcommand{\arraystretch}{1.18}
\begin{tabularx}{\textwidth}{@{}L{0.5cm}Y L{3.5cm}@{}}
\toprule
\textbf{\#} & \textbf{Check} & \textbf{Consequence if verified}\\
\midrule
1 & Is there a complete polynomial-time reduction to cyclic coset states (or general abelian states)? Count group-structure computation, public inputs, consistent encodings, coherent preparation, oracle calls, success and all remaining processing. & \vC\bR{} under the matching cyclic or abelian hypothesis.\\
2 & Count independent available samples $M$ for each secret. Separate the reduction subset $m_0$ from holdout samples $m'$. Correlated evaluation keys and signature hints do not automatically add independent samples. & An input budget for later checks.\\
3 & For uniform LWE, does some $m_0\in\mathcal M_q(n)$ satisfy $\alpha<\alpha_{\mathrm E}(n,m_0,\ell)$? For amplification also require $M\geq m_0+m'$, $m'\geq2n\lg q$ and the verification noise bound. Alternatively verify the worst-case $g_{\mathrm E}$ promise. & \vC\bR{} under $\HEDCP[\ell]$ with all conditions. Structured route: also check \cref{sec:wz-audit}; no deployed certificate.\\
4 & Does the actual embedding have its planted vector shortest and gap above $\Gamma_{f,\mathcal A}(D+1)$? Under $\HHS$ use the witness-dependent $f_{\mathcal A}$. Record heuristic comparisons separately. & \vC\bR{} if proved; \vK\bH{} for a heuristic or unspecified-bound comparison.\\
5 & Is the exact SIS or lattice proof assumption solved at its stated distribution and bound? Turning its solutions into a forgery or collision needs an additional algorithm. & Kind (B): proof assumption solvable. Kind (A) requires an explicit attack.\\
6 & Is a worst-case foundation solvable at the reduction's approximation factor? & Kind (C): that hardness guarantee is uninformative; no attack follows automatically.\\
7 & Is a lattice instance left without a certificate by steps 3--6? & Record \vK, \vO{} or \vS{} and the missing premise.\\
8 & Does a non-lattice family lack a known complete reduction to the assumed solver? & \vS{} within this analysis; no security claim.\\
9 & Assess independent classical and quantum attacks, including scheme-specific structure, evidence and implementation costs (\cref{tab:nonlattice,tab:sizes}; for Classic McEliece, \cref{sec:mceliece-status} and \cite{FV26mceliece}). & Record independent evidence; use a dagger where under revision.\\
\bottomrule
\end{tabularx}
\caption{Scenario-planning checklist. A hidden-shift subroutine alone is insufficient for step 1; a failure to certify later steps is not an impossibility theorem.}
\label{tab:checklist}
\end{table}

% ==================================================================
\subsection{Worked examples}\label{sec:worked}

The following applications of the checklist summarise the conclusions of \cref{sec:collapse,sec:lattice-survivors,sec:non-lattice}; conclusions are triples for $\HHS$\,/\,$\HReg(f)$\,/\,$\HEDCP[\ell]$, with the glyphs and marks of \cref{tab:collapse}.
\begin{itemize}
  \item \textbf{ML-KEM-768:} \vK\bH\,/\,\vK\bH\,/\,\vK\bH. \\
  Both key and ciphertext component $u$ have $m/n=1$. The normalised gap estimate is $0.3$--$0.4$; finite-dimensional and actual-gap checks remain, and the printed structured certificate has no feasible width (\cref{prop:wz-certificate}).
  \item \textbf{Regev's 2005 cryptosystem:} \vS\,/\,\vS\,/\,\vS. \\
  Its standard noise rate is outside the sufficient unstructured bounds derived here. This does not exclude another reduction or better behaviour of a solver.
  \item \textbf{FrodoKEM-976:} \vK\bH\,/\,\vK\bH\,/\,\vS. \\
  Each secret has $m/n=1$ and normalised heuristic gap about $0.6$. The sufficient embedding comparison is unresolved, and the restricted ball construction gives no certificate.
  \item \textbf{ML-DSA-44:} \vK\bH\,/\,\vK\bH\,/\,\vK\bH. \\
  Use the published-key estimate of about $0.3\sqrt D$. Full $t$ supplied gives a separate sensitivity estimate near $11\sqrt D$, with no reconstruction from signatures asserted. Module-SIS norm comparisons concern a proof assumption, not a forgery.
  \item \textbf{FN-DSA-512:} \vO\,/\,\vO\,/\,\vO. \\
  The ordinary-uSVP key promise fails by exact rotational symmetry. The displayed typical-norm calculation does not certify the required SIS bound; it proves no short-output impossibility.
  \item \textbf{CSIDH-512:} \vC\bR\,/\,\vC\bR\,/\,\vS. \\
  For the cyclic action satisfying the complete evaluation and public-input conditions (\cref{sec:isogenies}) \cite{BKV19}. The corresponding non-cyclic abelian action uses $\HHS^{\mathrm{ab}}$.
  \item \textbf{Classic McEliece:} \vS\,/\,\vS\,/\,\vS. \\
  With the independent-revision dagger. Step 9 must account for the version-specific structural analyses, challenge recoveries and cost-model conditions summarised in \cref{sec:mceliece-status} and analysed in \cite{FV26mceliece}: conditional per-run estimates of $2^{89}$--$2^{98}$ bit operations, or $2^{107}$--$2^{117}$ with a square-root charge per memory access, for the five parameter sets, with one run or $103$ to $1443$ runs per key recovery \cite{Wei26} (\cref{tab:mceliece-evidence}); the BSI advises against new deployments \cite{BSI26}. The DCP classification does not settle these questions; \cref{q:goppa} records them.
  \item \textbf{SQIsign:} \vS\,/\,\vS\,/\,\vS. \\ 
  The unchanged outer search in the composition of \cite{BJS14} retains its $\Otil(p^{1/4})$ cost if its commutative subroutine can be replaced under the stated evaluation conditions (\cref{sec:isogenies}).
  \item \textbf{SLH-DSA:} \vS\,/\,\vS\,/\,\vS. \\
  No DCP consequence is established; hash preimage and collision attacks remain independent checks.
\end{itemize}

% ==================================================================
\subsection{Conclusions}\label{sec:conclusion}

A noiseless cyclic DCP breakthrough would have lattice consequences at some polynomial factors; the strength needed for a given cryptographic instance depends on the witnessing algorithm and on the parameters and promises of the intervening reductions. We give explicit sufficient guarantees for the unstructured compositions considered and the structured state conversion; the explicit structured forward certificate supplies none for ML-KEM, and the polynomial-modulus adapter and several concrete embedding comparisons remain unresolved. A solved proof assumption is distinguished throughout from an explicit attack on a scheme, and absence of a certificate here is neither a security proof nor an impossibility result for future reductions. Chen's withdrawn 2024 algorithm for LWE and Zhang's conditional repair of it would, if correct, have reached only a many-sample, small-noise LWE regime inside the reach of $\HReg(f)$ for $f<1/2$, not standard-noise LWE or the deployed parameter sets (\cref{sec:chen-zhang}).

Free and transitive cyclic actions satisfying the complete evaluation and public-input conditions are a separate rigorous consequence, extended to abelian actions under $\HHS^{\mathrm{ab}}$. A hidden-shift subroutine does not suffice if the rest of an attack stays super-polynomial. For NTRU, the rotation bounds, module-uSVP distinction and lossiness oracle requirements prevent a blanket key-recovery claim. Outside lattices, no DCP-based certificate is recorded for binary code-based KEMs, but for Classic McEliece conditional classical estimates far below its targeted levels (\cref{sec:mceliece-status}, \cite{FV26mceliece}) mean that, in line with the BSI's advice, a post-DCP KEM layer cannot treat it as a conservative fallback before re-evaluation; HQC, not exposed to these attacks, needs its own analysis. The tables and checklist make explicit these dependencies and the checks a concrete application needs. Relative to the surveys of \cref{tab:prior-surveys}, the contribution is this joint accounting of solver domains, reduction parameters and evidential status, with the stated quantitative derivations, not new reduction mechanisms or an exhaustive security classification.

% ==================================================================
% ACKNOWLEDGEMENTS
% ==================================================================

\section*{Acknowledgements}
The author used Anthropic tools to assist with language editing of this manuscript. The scientific content, claims, proofs, and conclusions are the author's own. The resulting work, in its totality, is an accurate representation of the author's underlying work, and the author take full responsibility for the veracity and correctness of all material.

% ==================================================================
% REFERENCES
% ==================================================================
\bibliographystyle{splncs04}
\bibliography{Bibliography}

\begin{thebibliography}{100}
\providecommand{\url}[1]{\texttt{#1}}
\providecommand{\urlprefix}{URL }
\providecommand{\doi}[1]{https://doi.org/#1}

\bibitem{AACLMR26}
Abram, D., Arnon, G., Cini, V., Lou, P., Malavolta, G., Roy, L.: Decomposed
  {LWE} is equivalent to succinct {LWE}. Cryptology ePrint Archive, Paper
  2026/1829 (2026)

\bibitem{AMR25}
Abram, D., Malavolta, G., Roy, L.: Key-homomorphic computations for {RAM}:
  Fully succinct randomised encodings and more. In: Advances in Cryptology --
  CRYPTO 2025. LNCS, Springer (2025),
  \href{https://doi.org/10.1007/978-3-032-01881-6_8}{doi:10.1007/978-3-032-01881-6\_8};
  ePrint 2025/339

\bibitem{ADRS15}
Aggarwal, D., Dadush, D., Regev, O., Stephens-Davidowitz, N.: Solving the
  shortest vector problem in $2^n$ time using discrete {G}aussian sampling. In:
  Proceedings of the 47th Annual ACM Symposium on Theory of Computing (STOC
  2015). pp. 733--742 (2015)

\bibitem{HQC}
{Aguilar Melchor}, C., et~al.: Hamming quasi-cyclic ({HQC}). Submission to the
  NIST Post-Quantum Cryptography Standardization Process; selected for
  standardisation in March 2025 (2025)

\bibitem{AR05}
Aharonov, D., Regev, O.: Lattice problems in {NP} $\cap$ {coNP}. Journal of the
  ACM  \textbf{52}(5),  749--765 (2005)

\bibitem{ajtai1996}
Ajtai, M.: Generating hard instances of lattice problems (extended abstract).
  In: Proceedings of the 28th Annual ACM Symposium on Theory of Computing (STOC
  1996). pp. 99--108. ACM Press (1996). \doi{10.1145/237814.237838},
  \url{http://dx.doi.org/10.1145/237814.237838}

\bibitem{Ajt98}
Ajtai, M.: The shortest vector problem in {$L_2$} is {NP}-hard for randomized
  reductions. In: Proceedings of the 30th Annual ACM Symposium on Theory of
  Computing (STOC 1998). pp. 10--19 (1998)

\bibitem{AD97}
Ajtai, M., Dwork, C.: A public-key cryptosystem with worst-case/average-case
  equivalence. In: Proceedings of the 29th Annual ACM Symposium on Theory of
  Computing (STOC 1997). pp. 284--293 (1997)

\bibitem{ADMP20}
Alamati, N., Feo, L.D., Montgomery, H., Patranabis, S.: Cryptographic group
  actions and applications. In: Advances in Cryptology -- ASIACRYPT 2020. LNCS,
  vol. 12492, pp. 411--439. Springer (2020)

\bibitem{ABD16}
Albrecht, M.R., Bai, S., Ducas, L.: A subfield lattice attack on overstretched
  {NTRU} assumptions. In: Advances in Cryptology -- CRYPTO 2016. LNCS,
  vol.~9814, pp. 153--178. Springer (2016)

\bibitem{ClassicMcEliece}
Albrecht, M.R., Bernstein, D.J., Chou, T., Cid, C., Gilcher, J., Lange, T.,
  Maram, V., von Maurich, I., Misoczki, R., Niederhagen, R., Paterson, K.G.,
  Persichetti, E., Peters, C., Schwabe, P., Sendrier, N., Szefer, J., Tjhai,
  C.J., Tomlinson, M., Wang, W.: Classic {McEliece}: Conservative code-based
  cryptography. Submission to the NIST Post-Quantum Cryptography
  Standardization Process, Round 4 (2022), \url{https://classic.mceliece.org/}

\bibitem{AD17}
Albrecht, M.R., Deo, A.: Large modulus ring-{LWE} $\geq$ module-{LWE}. In:
  Advances in Cryptology -- ASIACRYPT 2017. LNCS, vol. 10624, pp. 267--296.
  Springer (2017)

\bibitem{ADHKPS19}
Albrecht, M.R., Ducas, L., Herold, G., Kirshanova, E., Postlethwaite, E.W.,
  Stevens, M.: The general sieve kernel and new records in lattice reduction.
  In: Advances in Cryptology -- EUROCRYPT 2019. LNCS, vol. 11477, pp. 717--746.
  Springer (2019)

\bibitem{APS15}
Albrecht, M.R., Player, R., Scott, S.: On the concrete hardness of learning
  with errors. Journal of Mathematical Cryptology  \textbf{9}(3),  169--203
  (2015)

\bibitem{ClassicMcElieceSec}
Albrecht, M.R., et~al.: Classic {McEliece}: Conservative code-based
  cryptography: Guide for security reviewers.
  \url{https://classic.mceliece.org/mceliece-security-20221023.pdf} (2022)

\bibitem{Ale03}
Alekhnovich, M.: More on average case vs approximation complexity. In: 44th
  IEEE Symposium on Foundations of Computer Science (FOCS 2003). pp. 298--307
  (2003)

\bibitem{FrodoKEM21}
Alkim, E., Bos, J.W., Ducas, L., Longa, P., Mironov, I., Naehrig, M.,
  Nikolaenko, V., Peikert, C., Raghunathan, A., Stebila, D.: {FrodoKEM}:
  Learning with errors key encapsulation. algorithm specifications and
  supporting documentation. Submission to the NIST Post-Quantum Cryptography
  Standardization Process, Round 3 (2021), \url{https://frodokem.org/}

\bibitem{AAG99}
Anshel, I., Anshel, M., Goldfeld, D.: An algebraic method for public-key
  cryptography. Mathematical Research Letters  \textbf{6}(3),  287--291 (1999)

\bibitem{Apo25}
Apon, D.: So about that quantum lattice thing: Rebuttal to ``{E}xact {C}oset
  {S}ampling for {Q}uantum {L}attice {A}lgorithms''. Cryptology ePrint Archive,
  Paper 2025/1945 (2025), received October 18, 2025; revised October 28, 2025.
  \url{https://eprint.iacr.org/2025/1945}

\bibitem{ACPS09}
Applebaum, B., Cash, D., Peikert, C., Sahai, A.: Fast cryptographic primitives
  and circular-secure encryption based on hard learning problems. In: Advances
  in Cryptology -- CRYPTO 2009. LNCS, vol.~5677, pp. 595--618. Springer (2009)

\bibitem{BIKE}
Aragon, N., et~al.: {BIKE}: Bit flipping key encapsulation. Submission to the
  NIST Post-Quantum Cryptography Standardization Process, Round 4 (2022)

\bibitem{AG11}
Arora, S., Ge, R.: New algorithms for learning in presence of errors. In:
  Automata, Languages and Programming -- ICALP 2011. LNCS, vol.~6755, pp.
  403--415. Springer (2011)

\bibitem{Bab86}
Babai, L.: On {L}ov{\'a}sz' lattice reduction and the nearest lattice point
  problem. Combinatorica  \textbf{6}(1),  1--13 (1986)

\bibitem{Bab16}
Babai, L.: Graph isomorphism in quasipolynomial time. In: Proceedings of the
  48th Annual ACM Symposium on Theory of Computing (STOC 2016). pp. 684--697
  (2016)

\bibitem{BCvD05}
Bacon, D., Childs, A.M., van Dam, W.: From optimal measurement to efficient
  quantum algorithms for the hidden subgroup problem over semidirect product
  groups. In: 46th IEEE Symposium on Foundations of Computer Science (FOCS
  2005). pp. 469--478 (2005)

\bibitem{DilSpec21}
Bai, S., Ducas, L., Kiltz, E., Lepoint, T., Lyubashevsky, V., Schwabe, P.,
  Seiler, G., Stehl{\'e}, D.: {CRYSTALS}-{D}ilithium: Algorithm specifications
  and supporting documentation (version 3.1). Submission to the NIST
  Post-Quantum Cryptography Standardization Process, Round 3 (2021), february
  8, 2021

\bibitem{BJKNY25}
Bai, S., Jangir, H., Kirshanova, E., Ngo, T., Youmans, W.: A quasi-polynomial
  time algorithm for the extrapolated dihedral coset problem over power-of-two
  moduli. In: Advances in Cryptology -- CRYPTO 2025. LNCS, Springer (2025),
  \href{https://doi.org/10.1007/978-3-032-01878-6_14}{doi:10.1007/978-3-032-01878-6\_14};
  ePrint 2025/1046

\bibitem{Ban93}
Banaszczyk, W.: New bounds in some transference theorems in the geometry of
  numbers. Mathematische Annalen  \textbf{296}(4),  625--635 (1993)

\bibitem{BPR12}
Banerjee, A., Peikert, C., Rosen, A.: Pseudorandom functions and lattices. In:
  Advances in Cryptology -- EUROCRYPT 2012. LNCS, vol.~7237, pp. 719--737.
  Springer (2012)

\bibitem{XWing24}
Barbosa, M., Connolly, D., Duarte, J.D., Kaiser, A., Schwabe, P., Varner, K.,
  Westerbaan, B.: {X-Wing}: The hybrid {KEM} you've been looking for. IACR
  Communications in Cryptology  \textbf{1}(1) (2024)

\bibitem{BLT26}
Bardet, M., Lemoine, A., Tillich, J.P.: An attack on the {CFS} scheme and on
  {TII} {McEliece} challenges. In: Advances in Cryptology -- CRYPTO 2026. LNCS,
  vol. 16803, pp. 380--412. Springer (2026), cryptology ePrint Archive, Paper
  2026/430

\bibitem{BM25}
Basso, A., Maino, L.: {POK{\'E}}: A compact and efficient {PKE} from
  higher-dimensional isogenies. In: Advances in Cryptology -- EUROCRYPT 2025.
  LNCS, Springer (2025)

\bibitem{BMP23}
Basso, A., Maino, L., Pope, G.: {FESTA}: Fast encryption from supersingular
  torsion attacks. In: Advances in Cryptology -- ASIACRYPT 2023. LNCS, vol.
  14444, pp. 98--126. Springer (2023)

\bibitem{FAEST23}
Baum, C., Braun, L., {Delpech de Saint Guilhem}, C., Kloo{\ss}, M., Orsini, E.,
  Roy, L., Scholl, P.: Publicly verifiable zero-knowledge and post-quantum
  signatures from {VOLE}-in-the-head. In: Advances in Cryptology -- CRYPTO
  2023. LNCS, vol. 14085, pp. 581--615. Springer (2023)

\bibitem{BDGL16}
Becker, A., Ducas, L., Gama, N., Laarhoven, T.: New directions in nearest
  neighbor searching with applications to lattice sieving. In: Proceedings of
  the 27th Annual ACM-SIAM Symposium on Discrete Algorithms (SODA 2016). pp.
  10--24 (2016)

\bibitem{Ben73}
Bennett, C.H.: Logical reversibility of computation. IBM Journal of Research
  and Development  \textbf{17}(6),  525--532 (1973)

\bibitem{BGPS23}
Bennett, H., Ganju, A., Peetathawatchai, P., Stephens-Davidowitz, N.: Just how
  hard are rotations of {$\mathbb{Z}^n$}? {A}lgorithms and cryptography with
  the simplest lattice. In: Advances in Cryptology -- EUROCRYPT 2023. LNCS,
  vol. 14008, pp. 252--281. Springer (2023)

\bibitem{BMvT78}
Berlekamp, E.R., McEliece, R.J., van Tilborg, H.C.A.: On the inherent
  intractability of certain coding problems. IEEE Transactions on Information
  Theory  \textbf{24}(3),  384--386 (1978)

\bibitem{Ber09}
Bernstein, D.J.: Cost analysis of hash collisions: Will quantum computers make
  {SHARCS} obsolete? In: Workshop Record of SHARCS 2009 (2009)

\bibitem{Ber10}
Bernstein, D.J.: Grover vs. {McEliece}. In: Post-Quantum Cryptography --
  PQCrypto 2010. LNCS, vol.~6061, pp. 73--80. Springer (2010)

\bibitem{BCLV17}
Bernstein, D.J., Chuengsatiansup, C., Lange, T., van Vredendaal, C.: {NTRU}
  prime: Reducing attack surface at low cost. In: Selected Areas in
  Cryptography -- SAC 2017. LNCS, vol. 10719, pp. 235--260. Springer (2017)

\bibitem{BHHLNPSW15}
Bernstein, D.J., Hopwood, D., H{\"u}lsing, A., Lange, T., Niederhagen, R.,
  Papachristodoulou, L., Schneider, M., Schwabe, P., Wilcox-O'Hearn, Z.:
  {SPHINCS}: Practical stateless hash-based signatures. In: Advances in
  Cryptology -- EUROCRYPT 2015. LNCS, vol.~9056, pp. 368--397. Springer (2015)

\bibitem{bernstein2019}
Bernstein, D.J., H{\"u}lsing, A., K{\"o}lbl, S., Niederhagen, R., Rijneveld,
  J., Schwabe, P.: The {SPHINCS}$^+$ signature framework. In: Proceedings of
  the 2019 ACM SIGSAC Conference on Computer and Communications Security (CCS
  2019). pp. 2129--2146 (2019)

\bibitem{BLMP19}
Bernstein, D.J., Lange, T., Martindale, C., Panny, L.: Quantum circuits for the
  {CSIDH}: Optimizing quantum evaluation of isogenies. In: Advances in
  Cryptology -- EUROCRYPT 2019. LNCS, vol. 11477, pp. 409--441. Springer (2019)

\bibitem{BY18}
Bernstein, D.J., Yang, B.Y.: Asymptotically faster quantum algorithms to solve
  multivariate quadratic equations. In: Post-Quantum Cryptography -- PQCrypto
  2018. LNCS, vol. 10786, pp. 487--506. Springer (2018)

\bibitem{Beu21}
Beullens, W.: {MAYO}: Practical post-quantum signatures from oil-and-vinegar
  maps. In: Selected Areas in Cryptography -- SAC 2021. LNCS, vol. 13203, pp.
  355--376. Springer (2021)

\bibitem{Beu22}
Beullens, W.: Breaking {Rainbow} takes a weekend on a laptop. In: Advances in
  Cryptology -- CRYPTO 2022. LNCS, vol. 13508, pp. 464--479. Springer (2022)

\bibitem{BKV19}
Beullens, W., Kleinjung, T., Vercauteren, F.: {CSI-FiSh}: Efficient isogeny
  based signatures through class group computations. In: Advances in Cryptology
  -- ASIACRYPT 2019. LNCS, vol. 11921, pp. 227--247. Springer (2019)

\bibitem{BBKSS23}
Biasse, J.F., Bonnetain, X., Kirshanova, E., Schrottenloher, A., Song, F.:
  Quantum algorithms for attacking hardness assumptions in classical and
  post-quantum cryptography. IET Information Security  \textbf{17}(2),
  171--209 (2023),
  \href{https://doi.org/10.1049/ise2.12081}{doi:10.1049/ise2.12081}

\bibitem{BJS14}
Biasse, J.F., Jao, D., Sankar, A.: A quantum algorithm for computing isogenies
  between supersingular elliptic curves. In: Progress in Cryptology --
  INDOCRYPT 2014. LNCS, vol.~8885, pp. 428--442. Springer (2014)

\bibitem{BMPS20}
Biasse, J.F., Micheli, G., Persichetti, E., Santini, P.: {LESS} is more:
  Code-based signatures without syndromes. In: Progress in Cryptology --
  AFRICACRYPT 2020. LNCS, vol. 12174, pp. 45--65. Springer (2020)

\bibitem{BS16}
Biasse, J.F., Song, F.: Efficient quantum algorithms for computing class groups
  and solving the principal ideal problem in arbitrary degree number fields.
  In: Proceedings of the 27th Annual ACM-SIAM Symposium on Discrete Algorithms
  (SODA 2016). pp. 893--902 (2016)

\bibitem{BBFGS19}
Bindel, N., Brendel, J., Fischlin, M., Goncalves, B., Stebila, D.: Hybrid key
  encapsulation mechanisms and authenticated key exchange. In: Post-Quantum
  Cryptography -- PQCrypto 2019. LNCS, vol. 11505, pp. 206--226. Springer
  (2019)

\bibitem{BBPS19}
Bolboceanu, M., Brakerski, Z., Perlman, R., Sharma, D.: Order-{LWE} and the
  hardness of ring-{LWE} with entropic secrets. In: Advances in Cryptology --
  ASIACRYPT 2019. LNCS, vol. 11922, pp. 91--120. Springer (2019)

\bibitem{BS20}
Bonnetain, X., Schrottenloher, A.: Quantum security analysis of {CSIDH}. In:
  Advances in Cryptology -- EUROCRYPT 2020. LNCS, vol. 12106, pp. 493--522.
  Springer (2020)

\bibitem{Frodo16}
Bos, J.W., Costello, C., Ducas, L., Mironov, I., Naehrig, M., Nikolaenko, V.,
  Raghunathan, A., Stebila, D.: Frodo: Take off the ring! {P}ractical,
  quantum-secure key exchange from {LWE}. In: Proceedings of the 2016 ACM
  SIGSAC Conference on Computer and Communications Security (CCS 2016). pp.
  1006--1018 (2016)

\bibitem{Kyber18}
Bos, J.W., Ducas, L., Kiltz, E., Lepoint, T., Lyubashevsky, V., Schanck, J.M.,
  Schwabe, P., Seiler, G., Stehl{\'e}, D.: {CRYSTALS}-{K}yber: A {CCA}-secure
  module-lattice-based {KEM}. In: IEEE European Symposium on Security and
  Privacy (EuroS\&P 2018). pp. 353--367 (2018)

\bibitem{BFS26}
Boucher, M., Fouque, P.A., Shen, Y.: Cyclotomic cosets: Hidden subgroup and
  quantum sieving algorithm for prime-power moduli. Cryptology ePrint Archive,
  Paper 2026/2249; arXiv:2609.34996 (2026)

\bibitem{BJRW23}
Boudgoust, K., Jeudy, C., Roux-Langlois, A., Wen, W.: On the hardness of module
  learning with errors with short distributions. Journal of Cryptology
  \textbf{36}(1), ~1 (2023),
  \href{https://doi.org/10.1007/s00145-022-09441-3}{doi:10.1007/s00145-022-09441-3}

\bibitem{BGV12}
Brakerski, Z., Gentry, C., Vaikuntanathan, V.: ({L}eveled) fully homomorphic
  encryption without bootstrapping. In: Proceedings of the 3rd Innovations in
  Theoretical Computer Science Conference (ITCS 2012). pp. 309--325 (2012)

\bibitem{BKSW18}
Brakerski, Z., Kirshanova, E., Stehl{\'e}, D., Wen, W.: Learning with errors
  and extrapolated dihedral cosets. In: Public-Key Cryptography -- PKC 2018.
  LNCS, vol. 10770, pp. 702--727. Springer (2018), arXiv:1710.08223

\bibitem{BLPRS13}
Brakerski, Z., Langlois, A., Peikert, C., Regev, O., Stehl{\'e}, D.: Classical
  hardness of learning with errors. In: Proceedings of the 45th Annual ACM
  Symposium on Theory of Computing (STOC 2013). pp. 575--584 (2013)

\bibitem{BLVW19}
Brakerski, Z., Lyubashevsky, V., Vaikuntanathan, V., Wichs, D.: Worst-case
  hardness for {LPN} and cryptographic hashing via code smoothing. In: Advances
  in Cryptology -- EUROCRYPT 2019. LNCS, vol. 11478, pp. 619--635. Springer
  (2019)

\bibitem{BHKKLS11}
Brassard, G., H{\o}yer, P., Kalach, K., Kaplan, M., Laplante, S., Salvail, L.:
  Merkle puzzles in a quantum world. In: Advances in Cryptology -- CRYPTO 2011.
  LNCS, vol.~6841, pp. 391--410. Springer (2011)

\bibitem{BHT98}
Brassard, G., H{\o}yer, P., Tapp, A.: Quantum cryptanalysis of hash and
  claw-free functions. In: LATIN'98: Theoretical Informatics. LNCS, vol.~1380,
  pp. 163--169. Springer (1998)

\bibitem{BLRT26}
Briaud, P., Lemoine, A., Randriambololona, H., Tillich, J.P.: A heuristic
  subexponential attack on the {McEliece} cryptosystem. Cryptology ePrint
  Archive, Paper 2026/1232 (2026), posted June 10, 2026; revised September 17,
  2026

\bibitem{BRST21}
Bruna, J., Regev, O., Song, M.J., Tang, Y.: Continuous {LWE}. In: Proceedings
  of the 53rd Annual ACM Symposium on Theory of Computing (STOC 2021). pp.
  694--707 (2021)

\bibitem{BSI26}
{BSI (Bundesamt f{\"u}r Sicherheit in der Informationstechnik)}: Notes on
  recent developments concerning {Classic McEliece}. BSI publication (2026),
  dated 29 September 2026; consulted October 11, 2026.
  \url{https://www.bsi.bund.de/SharedDocs/Downloads/EN/BSI/Crypto/Notes_Classic_McEliece.pdf}

\bibitem{CMY26}
Cartor, R., Manganiello, F., Youmans, W.: A locality-sensitive hashing
  framework for reducing bounded distance decoding to {EDCP}. In: Theory of
  Cryptography -- TCC 2026. LNCS, Springer (2026), cryptology ePrint Archive,
  Paper 2026/2055

\bibitem{CD23}
Castryck, W., Decru, T.: An efficient key recovery attack on {SIDH}. In:
  Advances in Cryptology -- EUROCRYPT 2023. LNCS, vol. 14008, pp. 423--447.
  Springer (2023)

\bibitem{CLMPR18}
Castryck, W., Lange, T., Martindale, C., Panny, L., Renes, J.: {CSIDH}: An
  efficient post-quantum commutative group action. In: Advances in Cryptology
  -- ASIACRYPT 2018. LNCS, vol. 11274, pp. 395--427. Springer (2018)

\bibitem{CH26}
Chailloux, A., Hermouet, P.: On the quantum equivalence between
  {$S|\mathsf{LWE}\rangle$} and {ISIS}. In: Advances in Cryptology -- CRYPTO
  2026. LNCS, Springer (2026),
  \href{https://doi.org/10.1007/978-3-032-35377-1_13}{doi:10.1007/978-3-032-35377-1\_13};
  ePrint 2025/1857

\bibitem{CL21}
Chailloux, A., Loyer, J.: Lattice sieving via quantum random walks. In:
  Advances in Cryptology -- ASIACRYPT 2021. LNCS, vol. 13093, pp. 63--91.
  Springer (2021)

\bibitem{CGL09}
Charles, D.X., Goren, E.Z., Lauter, K.E.: Cryptographic hash functions from
  expander graphs. Journal of Cryptology  \textbf{22}(1),  93--113 (2009)

\bibitem{Picnic17}
Chase, M., Derler, D., Goldfeder, S., Orlandi, C., Ramacher, S., Rechberger,
  C., Slamanig, D., Zaverucha, G.: Post-quantum zero-knowledge and signatures
  from symmetric-key primitives. In: Proceedings of the 2017 ACM SIGSAC
  Conference on Computer and Communications Security (CCS 2017). pp. 1825--1842
  (2017)

\bibitem{CS24}
Chen, I., Sun, D.: The dihedral hidden subgroup problem. Journal of
  Mathematical Cryptology  \textbf{18}(1),  20220029 (2024),
  \href{https://doi.org/10.1515/jmc-2022-0029}{doi:10.1515/jmc-2022-0029}

\bibitem{Che24}
Chen, Y.: Quantum algorithms for lattice problems. Cryptology ePrint Archive,
  Paper 2024/555 (2024), received April 10, 2024. On April 18, 2024 the author
  reported a bug in Step~9, found by Hongxun Wu and independently by Thomas
  Vidick, and withdrew the main claim; revised April 19, 2024.
  \url{https://eprint.iacr.org/2024/555}

\bibitem{CHLLT25}
Chen, Y., Hu, Z., Liu, Q., Luo, H., Tu, Y.: {LWE} with quantum amplitudes:
  Algorithm, hardness, and oblivious sampling. In: Advances in Cryptology --
  CRYPTO 2025. LNCS, Springer (2025), ePrint 2023/1498; arXiv:2310.00644

\bibitem{CLZ22}
Chen, Y., Liu, Q., Zhandry, M.: Quantum algorithms for variants of average-case
  lattice problems via filtering. In: Advances in Cryptology -- EUROCRYPT 2022.
  LNCS, vol. 13277, pp. 372--401. Springer (2022)

\bibitem{CvD07}
Childs, A.M., van Dam, W.: Quantum algorithm for a generalized hidden shift
  problem. In: Proceedings of the 18th Annual ACM-SIAM Symposium on Discrete
  Algorithms (SODA 2007). pp. 1225--1232 (2007)

\bibitem{CvD10}
Childs, A.M., van Dam, W.: Quantum algorithms for algebraic problems. Reviews
  of Modern Physics  \textbf{82}(1),  1--52 (2010)

\bibitem{CJS14}
Childs, A.M., Jao, D., Soukharev, V.: Constructing elliptic curve isogenies in
  quantum subexponential time. Journal of Mathematical Cryptology
  \textbf{8}(1),  1--29 (2014)

\bibitem{MEDS23}
Chou, T., Niederhagen, R., Persichetti, E., Randrianarisoa, T.H., Reijnders,
  K., Samardjiska, S., Trimoska, M.: Take your {MEDS}: Digital signatures from
  matrix code equivalence. In: Progress in Cryptology -- AFRICACRYPT 2023.
  LNCS, vol. 14064, pp. 28--52. Springer (2023)

\bibitem{ClassicMcElieceISO}
{Classic McEliece team}: Classic {McEliece}: {ISO}.
  \url{https://classic.mceliece.org/iso.html} (2026), web page, version June
  15, 2026; consulted October 11, 2026

\bibitem{SP800208}
Cooper, D.A., Apon, D.C., Dang, Q.H., Davidson, M.S., Dworkin, M.J., Miller,
  C.A.: Recommendation for stateful hash-based signature schemes. Special
  Publication 800-208, NIST (2020)

\bibitem{Cou06}
Couveignes, J.M.: Hard homogeneous spaces. Cryptology ePrint Archive, Paper
  2006/291 (2006)

\bibitem{CDPR16}
Cramer, R., Ducas, L., Peikert, C., Regev, O.: Recovering short generators of
  principal ideals in cyclotomic rings. In: Advances in Cryptology -- EUROCRYPT
  2016. LNCS, vol.~9666, pp. 559--585. Springer (2016)

\bibitem{CDW17}
Cramer, R., Ducas, L., Wesolowski, B.: Short {S}tickelberger class relations
  and application to {Ideal-SVP}. In: Advances in Cryptology -- EUROCRYPT 2017.
  LNCS, vol. 10210, pp. 324--348. Springer (2017)

\bibitem{DFS24}
Debris-Alazard, T., Fallahpour, P., Stehl{\'e}, D.: Quantum oblivious {LWE}
  sampling and insecurity of standard model lattice-based {SNARKs}. In:
  Proceedings of the 56th Annual ACM Symposium on Theory of Computing (STOC
  2024) (2024)

\bibitem{DRT21}
Debris-Alazard, T., Remaud, M., Tillich, J.P.: Quantum reduction of finding
  short code vectors to the decoding problem. arXiv:2106.02747 (2021)

\bibitem{DG16}
Delfs, C., Galbraith, S.D.: Computing isogenies between supersingular elliptic
  curves over {$\mathbb{F}_p$}. Designs, Codes and Cryptography
  \textbf{78}(2),  425--440 (2016)

\bibitem{DS05}
Ding, J., Schmidt, D.: Rainbow, a new multivariable polynomial signature
  scheme. In: Applied Cryptography and Network Security -- ACNS 2005. LNCS,
  vol.~3531, pp. 164--175. Springer (2005)

\bibitem{DMR11}
Dinh, H., Moore, C., Russell, A.: {McEliece} and {Niederreiter} cryptosystems
  that resist quantum {F}ourier sampling attacks. In: Advances in Cryptology --
  CRYPTO 2011. LNCS, vol.~6841, pp. 761--779. Springer (2011)

\bibitem{ducas2017}
Ducas, L., Kiltz, E., Lepoint, T., Lyubashevsky, V., Schwabe, P., Seiler, G.,
  Stehl{\'e}, D.: {CRYSTALS}-{D}ilithium: A lattice-based digital signature
  scheme. IACR Transactions on Cryptographic Hardware and Embedded Systems
  \textbf{2018}(1),  238--268 (2018)

\bibitem{DPPvW22}
Ducas, L., Postlethwaite, E.W., Pulles, L.N., van Woerden, W.: Hawk: Module
  {LIP} makes lattice signatures fast, compact and simple. In: Advances in
  Cryptology -- ASIACRYPT 2022. LNCS, vol. 13794, pp. 65--94. Springer (2022)

\bibitem{DvW21note}
Ducas, L., van Woerden, W.: A note on a claim of {E}ldar \& {H}allgren: {LLL}
  already solves it. Cryptology ePrint Archive, Paper 2021/1391 (2021)

\bibitem{DvW22}
Ducas, L., van Woerden, W.: On the lattice isomorphism problem, quadratic
  forms, remarkable lattices, and cryptography. In: Advances in Cryptology --
  EUROCRYPT 2022. LNCS, vol. 13277, pp. 643--673. Springer (2022)

\bibitem{DMM25}
Dutto, S., Mercuri, P., Murru, N., Romano, L.: A survey about hidden subgroup
  problem from a mathematical and cryptographic perspective. arXiv:2512.02087
  (2025)

\bibitem{EHKS14}
Eisentr{\"a}ger, K., Hallgren, S., Kitaev, A., Song, F.: A quantum algorithm
  for computing the unit group of an arbitrary degree number field. In:
  Proceedings of the 46th Annual ACM Symposium on Theory of Computing (STOC
  2014). pp. 293--302 (2014)

\bibitem{EHLMP18}
Eisentr{\"a}ger, K., Hallgren, S., Lauter, K., Morrison, T., Petit, C.:
  Supersingular isogeny graphs and endomorphism rings: Reductions and
  solutions. In: Advances in Cryptology -- EUROCRYPT 2018. LNCS, vol. 10822,
  pp. 329--368. Springer (2018)

\bibitem{EH22}
Eldar, L., Hallgren, S.: An efficient quantum algorithm for lattice problems
  achieving subexponential approximation factor. arXiv:2201.13450 (2022), see
  the response of Ducas and van Woerden, Cryptology ePrint Archive, Paper
  2021/1391

\bibitem{EH00}
Ettinger, M., H{\o}yer, P.: On quantum algorithms for noncommutative hidden
  subgroups. Advances in Applied Mathematics  \textbf{25}(3),  239--251 (2000)

\bibitem{EHK04}
Ettinger, M., H{\o}yer, P., Knill, E.: The quantum query complexity of the
  hidden subgroup problem is polynomial. Information Processing Letters
  \textbf{91}(1),  43--48 (2004)

\bibitem{FPS22}
Felderhoff, J., Pellet-Mary, A., Stehl{\'e}, D.: On {Module} {Unique-SVP} and
  {NTRU}. In: Advances in Cryptology -- ASIACRYPT 2022, Part III. pp. 709--740.
  Springer (2022), full version: \url{https://eprint.iacr.org/2022/1203}

\bibitem{FJR22}
Feneuil, T., Joux, A., Rivain, M.: Syndrome decoding in the head: Shorter
  signatures from zero-knowledge proofs. In: Advances in Cryptology -- CRYPTO
  2022. LNCS, vol. 13508, pp. 541--572. Springer (2022)

\bibitem{DKLPW20}
Feo, L.D., Kohel, D., Leroux, A., Petit, C., Wesolowski, B.: {SQISign}: Compact
  post-quantum signatures from quaternions and isogenies. In: Advances in
  Cryptology -- ASIACRYPT 2020. LNCS, vol. 12491, pp. 64--93. Springer (2020)

\bibitem{FV26mceliece}
Fernandez, R.: Structural cryptanalysis of {Classic McEliece}, 2024--2026: A
  critical assessment of claims, costs and evidence. arXiv:2610.31196 [cs.CR]
  (2026)

\bibitem{Falcon20}
Fouque, P.A., Hoffstein, J., Kirchner, P., Lyubashevsky, V., Pornin, T., Prest,
  T., Ricosset, T., Seiler, G., Whyte, W., Zhang, Z.: Falcon: Fast-{F}ourier
  lattice-based compact signatures over {NTRU}. Submission to the NIST
  Post-Quantum Cryptography Standardization Process, Round 3 (2020)

\bibitem{FIMSS03}
Friedl, K., Ivanyos, G., Magniez, F., Santha, M., Sen, P.: Hidden translation
  and orbit coset in quantum computing. In: Proceedings of the 35th Annual ACM
  Symposium on Theory of Computing (STOC 2003). pp.~1--9 (2003)

\bibitem{GJ79}
Garey, M.R., Johnson, D.S.: Computers and Intractability: A Guide to the Theory
  of {NP}-Completeness. W. H. Freeman (1979)

\bibitem{Gen09}
Gentry, C.: Fully homomorphic encryption using ideal lattices. In: Proceedings
  of the 41st Annual ACM Symposium on Theory of Computing (STOC 2009). pp.
  169--178 (2009)

\bibitem{GPV08}
Gentry, C., Peikert, C., Vaikuntanathan, V.: Trapdoors for hard lattices and
  new cryptographic constructions. In: Proceedings of the 40th Annual ACM
  Symposium on Theory of Computing (STOC 2008). pp. 197--206 (2008)

\bibitem{GSW13}
Gentry, C., Sahai, A., Waters, B.: Homomorphic encryption from learning with
  errors: Conceptually-simpler, asymptotically-faster, attribute-based. In:
  Advances in Cryptology -- CRYPTO 2013. LNCS, vol.~8042, pp. 75--92. Springer
  (2013)

\bibitem{GIJS26}
Ghoshal, A., Ishai, Y., Jain, A., Sun, N.: Quasipolynomial cryptanalysis of the
  {McEliece} cryptosystem (or: {PIR} meets {McEliece}). Cryptology ePrint
  Archive, Paper 2026/1630 (2026), posted August 7, 2026; revised August 27,
  2026

\bibitem{GHP18}
Giacon, F., Heuer, F., Poettering, B.: {KEM} combiners. In: Public-Key
  Cryptography -- PKC 2018. LNCS, vol. 10769, pp. 190--218. Springer (2018)

\bibitem{GG00}
Goldreich, O., Goldwasser, S.: On the limits of nonapproximability of lattice
  problems. Journal of Computer and System Sciences  \textbf{60}(3),  540--563
  (2000)

\bibitem{GGH97}
Goldreich, O., Goldwasser, S., Halevi, S.: Public-key cryptosystems from
  lattice reduction problems. In: Advances in Cryptology -- CRYPTO 1997. LNCS,
  vol.~1294, pp. 112--131. Springer (1997)

\bibitem{Gro96}
Grover, L.K.: A fast quantum mechanical algorithm for database search. In:
  Proceedings of the 28th Annual ACM Symposium on Theory of Computing (STOC
  1996). pp. 212--219 (1996)

\bibitem{GY26}
Guo, Y., Yang, S.: Rigorous statements and proofs of the lemmas in {S}imon's
  algorithm for the dihedral coset problem and their underlying hypothesis.
  Cryptology ePrint Archive, Paper 2026/1714; arXiv:2608.16598 (2026),
  \url{https://eprint.iacr.org/2026/1714}

\bibitem{GRZ26}
Gupte, A., Ragavan, S., Zhandry, M.: The {ePrint:2026/1591} quantum algorithm
  does not solve {DCP}. Cryptology ePrint Archive, Paper 2026/1693 (2026),
  posted August 15, 2026; revised September 1, 2026.
  \url{https://eprint.iacr.org/2026/1693}

\bibitem{GRZLean}
Gupte, A., Ragavan, S., Zhandry, M.: Lean~4 formalisation of ``the
  {ePrint:2026/1591} quantum algorithm does not solve {DCP}''. GitHub
  repository (2026), main theorem \texttt{MainTheorem.thm\_main}; refutation
  corollary in \texttt{Refutation.lean}.
  \url{https://github.com/sragavan99/lean-ePrint-2026-1591-refutation};
  consulted October 2, 2026 (commit identifier not recorded); README re-read
  October 18, 2026 (models the draft dated August 11, 2026)

\bibitem{GVV22}
Gupte, A., Vafa, N., Vaikuntanathan, V.: Continuous {LWE} is as hard as {LWE}
  \& applications to learning {G}aussian mixtures. In: 63rd IEEE Symposium on
  Foundations of Computer Science (FOCS 2022). pp. 1162--1173 (2022)

\bibitem{HMRRS10}
Hallgren, S., Moore, C., R{\"o}tteler, M., Russell, A., Sen, P.: Limitations of
  quantum coset states for graph isomorphism. Journal of the ACM
  \textbf{57}(6),  34:1--34:33 (2010)

\bibitem{HRT03}
Hallgren, S., Russell, A., Ta-Shma, A.: The hidden subgroup problem and quantum
  computation using group representations. SIAM Journal on Computing
  \textbf{32}(4),  916--934 (2003)

\bibitem{HawkWd26}
{HAWK team}: Withdrawal of {HAWK} from the third round of the {NIST} additional
  signature process. Announcement on the NIST pqc-forum mailing list, July 29,
  2026, as reported in The Hacker News,
  \url{https://thehackernews.com/2026/07/claude-ai-just-cracked-post-quantum.html}
  (2026)

\bibitem{Hem26}
Hemmert, T.: Key recovery for the {McEliece} cryptosystem using higher-order
  vanishing. Cryptology ePrint Archive, Paper 2026/1339 (2026), posted June 29,
  2026

\bibitem{HW26}
Hemmert, T., Wiemers, A.: Distinguishing {Goppa} codes using higher-order
  vanishing. In: Advances in Cryptology -- CRYPTO 2026. LNCS, vol. 16803, pp.
  191--223. Springer (2026), cryptology ePrint Archive, Paper 2025/1661

\bibitem{HPS98}
Hoffstein, J., Pipher, J., Silverman, J.H.: {NTRU}: A ring-based public key
  cryptosystem. In: Algorithmic Number Theory -- ANTS-III. LNCS, vol.~1423, pp.
  267--288. Springer (1998)

\bibitem{hulsing2013w}
H{\"u}lsing, A.: {W-OTS}$^+$ -- shorter signatures for hash-based signature
  schemes. In: Progress in Cryptology -- AFRICACRYPT 2013. LNCS, vol.~7918, pp.
  173--188. Springer (2013)

\bibitem{RFC8391}
H{\"u}lsing, A., Butin, D., Gazdag, S.L., Rijneveld, J., Mohaisen, A.: {XMSS}:
  {eXtended Merkle Signature Scheme}. RFC 8391 (2018)

\bibitem{IKOS07}
Ishai, Y., Kushilevitz, E., Ostrovsky, R., Sahai, A.: Zero-knowledge from
  secure multiparty computation. In: Proceedings of the 39th Annual ACM
  Symposium on Theory of Computing (STOC 2007). pp. 21--30 (2007)

\bibitem{ISO18033Amd2}
{ISO/IEC}: {ISO/IEC} 18033-2:2006/{Amd}~2:2026, information technology ---
  security techniques --- encryption algorithms --- part 2: Asymmetric ciphers,
  amendment 2. International Organization for Standardization (2026)

\bibitem{ISS07}
Ivanyos, G., Sanselme, L., Santha, M.: An efficient quantum algorithm for the
  hidden subgroup problem in extraspecial groups. In: STACS 2007. LNCS,
  vol.~4393, pp. 586--597. Springer (2007)

\bibitem{JDF11}
Jao, D., Feo, L.D.: Towards quantum-resistant cryptosystems from supersingular
  elliptic curve isogenies. In: Post-Quantum Cryptography -- PQCrypto 2011.
  LNCS, vol.~7071, pp. 19--34. Springer (2011)

\bibitem{JS19}
Jaques, S., Schanck, J.M.: Quantum cryptanalysis in the {RAM} model:
  Claw-finding attacks on {SIKE}. In: Advances in Cryptology -- CRYPTO 2019.
  LNCS, vol. 11692, pp. 32--61. Springer (2019)

\bibitem{KT17}
Kachigar, G., Tillich, J.P.: Quantum information set decoding algorithms. In:
  Post-Quantum Cryptography -- PQCrypto 2017. LNCS, vol. 10346, pp. 69--89.
  Springer (2017)

\bibitem{KPG99}
Kipnis, A., Patarin, J., Goubin, L.: Unbalanced oil and vinegar signature
  schemes. In: Advances in Cryptology -- EUROCRYPT 1999. LNCS, vol.~1592, pp.
  206--222. Springer (1999)

\bibitem{KS98}
Kipnis, A., Shamir, A.: Cryptanalysis of the oil and vinegar signature scheme.
  In: Advances in Cryptology -- CRYPTO 1998. LNCS, vol.~1462, pp. 257--266.
  Springer (1998)

\bibitem{KF17}
Kirchner, P., Fouque, P.A.: Revisiting lattice attacks on overstretched {NTRU}
  parameters. In: Advances in Cryptology -- EUROCRYPT 2017. LNCS, vol. 10210,
  pp. 3--26. Springer (2017)

\bibitem{KLCHKP00}
Ko, K.H., Lee, S.J., Cheon, J.H., Han, J.W., sung Kang, J., Park, C.: New
  public-key cryptosystem using braid groups. In: Advances in Cryptology --
  CRYPTO 2000. LNCS, vol.~1880, pp. 166--183. Springer (2000)

\bibitem{KLPT14}
Kohel, D., Lauter, K., Petit, C., Tignol, J.P.: On the quaternion
  $\ell$-isogeny path problem. LMS Journal of Computation and Mathematics
  \textbf{17}(A),  418--432 (2014)

\bibitem{KOW26}
Kothari, R., O'Donnell, R., Wu, K.: No exponential quantum speedup for
  {SIS}$^\infty$ anymore. In: Proceedings of the 58th Annual ACM Symposium on
  Theory of Computing (STOC 2026) (2026),
  \href{https://doi.org/10.1145/3798129.3800731}{doi:10.1145/3798129.3800731};
  arXiv:2510.07515

\bibitem{Kup05}
Kuperberg, G.: A subexponential-time quantum algorithm for the dihedral hidden
  subgroup problem. SIAM Journal on Computing  \textbf{35}(1),  170--188 (2005)

\bibitem{Kup13}
Kuperberg, G.: Another subexponential-time quantum algorithm for the dihedral
  hidden subgroup problem. In: 8th Conference on the Theory of Quantum
  Computation, Communication and Cryptography (TQC 2013). LIPIcs, vol.~22, pp.
  20--34 (2013)

\bibitem{Kup25}
Kuperberg, G.: The hidden subgroup problem for infinite groups.
  arXiv:2507.18499 (2025)

\bibitem{LMvdP13}
Laarhoven, T., Mosca, M., van~de Pol, J.: Solving the shortest vector problem
  in lattices faster using quantum search. In: Post-Quantum Cryptography --
  PQCrypto 2013. LNCS, vol.~7932, pp. 83--101. Springer (2013)

\bibitem{LvdPdW12}
Laarhoven, T., van~de Pol, J., de~Weger, B.: Solving hard lattice problems and
  the security of lattice-based cryptosystems. Cryptology ePrint Archive, Paper
  2012/533 (2012)

\bibitem{LO85}
Lagarias, J.C., Odlyzko, A.M.: Solving low-density subset sum problems. Journal
  of the ACM  \textbf{32}(1),  229--246 (1985)

\bibitem{lamport1979}
Lamport, L.: Constructing digital signatures from a one way function. Tech.
  Rep. CSL-98, SRI International (1979)

\bibitem{LS15}
Langlois, A., Stehl{\'e}, D.: Worst-case to average-case reductions for module
  lattices. Designs, Codes and Cryptography  \textbf{75}(3),  565--599 (2015)

\bibitem{LLL82}
Lenstra, A.K., Lenstra, H.W., Lov{\'a}sz, L.: Factoring polynomials with
  rational coefficients. Mathematische Annalen  \textbf{261},  515--534 (1982)

\bibitem{LP11}
Lindner, R., Peikert, C.: Better key sizes (and attacks) for {LWE}-based
  encryption. In: Topics in Cryptology -- CT-RSA 2011. LNCS, vol.~6558, pp.
  319--339. Springer (2011)

\bibitem{Lyu12}
Lyubashevsky, V.: Lattice signatures without trapdoors. In: Advances in
  Cryptology -- EUROCRYPT 2012. LNCS, vol.~7237, pp. 738--755. Springer (2012)

\bibitem{LM09}
Lyubashevsky, V., Micciancio, D.: On bounded distance decoding, unique shortest
  vectors, and the minimum distance problem. In: Advances in Cryptology --
  CRYPTO 2009. LNCS, vol.~5677, pp. 577--594. Springer (2009)

\bibitem{LMPR08}
Lyubashevsky, V., Micciancio, D., Peikert, C., Rosen, A.: {SWIFFT}: A modest
  proposal for {FFT} hashing. In: Fast Software Encryption -- FSE 2008. LNCS,
  vol.~5086, pp. 54--72. Springer (2008)

\bibitem{LPS10}
Lyubashevsky, V., Palacio, A., Segev, G.: Public-key cryptographic primitives
  provably as secure as subset sum. In: Theory of Cryptography -- TCC 2010.
  LNCS, vol.~5978, pp. 382--400. Springer (2010)

\bibitem{LPR13}
Lyubashevsky, V., Peikert, C., Regev, O.: On ideal lattices and learning with
  errors over rings. Journal of the ACM  \textbf{60}(6),  43:1--43:35 (2013),
  preliminary version in EUROCRYPT 2010

\bibitem{MMPPW23}
Maino, L., Martindale, C., Panny, L., Pope, G., Wesolowski, B.: A direct key
  recovery attack on {SIDH}. In: Advances in Cryptology -- EUROCRYPT 2023.
  LNCS, vol. 14008, pp. 448--471. Springer (2023)

\bibitem{McE78}
McEliece, R.J.: A public-key cryptosystem based on algebraic coding theory. DSN
  Progress Report 42--44, Jet Propulsion Laboratory (1978)

\bibitem{RFC8554}
McGrew, D., Curcio, M., Fluhrer, S.: {Leighton-Micali} hash-based signatures.
  RFC 8554 (2019)

\bibitem{Mer89}
Merkle, R.C.: A certified digital signature. In: Advances in Cryptology --
  CRYPTO '89. LNCS, vol.~435, pp. 218--238. Springer (1990)

\bibitem{MG02}
Micciancio, D., Goldwasser, S.: Complexity of Lattice Problems: A Cryptographic
  Perspective, The Kluwer International Series in Engineering and Computer
  Science, vol.~671. Kluwer Academic Publishers (2002)

\bibitem{MM11}
Micciancio, D., Mol, P.: Pseudorandom knapsacks and the sample complexity of
  {LWE} search-to-decision reductions. In: Advances in Cryptology -- CRYPTO
  2011. LNCS, vol.~6841, pp. 465--484. Springer (2011)

\bibitem{MP12}
Micciancio, D., Peikert, C.: Trapdoors for lattices: Simpler, tighter, faster,
  smaller. In: Advances in Cryptology -- EUROCRYPT 2012. LNCS, vol.~7237, pp.
  700--718. Springer (2012)

\bibitem{MP13}
Micciancio, D., Peikert, C.: Hardness of {SIS} and {LWE} with small parameters.
  In: Advances in Cryptology -- CRYPTO 2013. LNCS, vol.~8042, pp. 21--39.
  Springer (2013)

\bibitem{MR07}
Micciancio, D., Regev, O.: Worst-case to average-case reductions based on
  {G}aussian measures. SIAM Journal on Computing  \textbf{37}(1),  267--302
  (2007)

\bibitem{MRRS07}
Moore, C., Rockmore, D., Russell, A., Schulman, L.J.: The power of strong
  {F}ourier sampling: Quantum algorithms for affine groups and hidden shifts.
  SIAM Journal on Computing  \textbf{37}(3),  938--958 (2007)

\bibitem{MRS08}
Moore, C., Russell, A., Schulman, L.J.: The symmetric group defies strong
  {F}ourier sampling. SIAM Journal on Computing  \textbf{37}(6),  1842--1864
  (2008)

\bibitem{Mor26}
Morales, M.E.S.: The hidden subgroup problem in semidirect products and
  quasi-{H}amiltonian groups. arXiv:2608.05321 (2026)

\bibitem{MSU08}
Myasnikov, A., Shpilrain, V., Ushakov, A.: Group-Based Cryptography. Advanced
  Courses in Mathematics -- CRM Barcelona, Birkh{\"a}user (2008)

\bibitem{FIPS204}
{National Institute of Standards and Technology}: Module-lattice-based digital
  signature standard. FIPS~204, NIST (2024)

\bibitem{FIPS203}
{National Institute of Standards and Technology}: Module-lattice-based
  key-encapsulation mechanism standard. FIPS~203, NIST (2024)

\bibitem{FIPS205}
{National Institute of Standards and Technology}: Stateless hash-based digital
  signature standard. FIPS~205, NIST (2024)

\bibitem{NISTIR8545}
{National Institute of Standards and Technology}: Status report on the fourth
  round of the {NIST} post-quantum cryptography standardization process. NIST
  IR~8545, NIST (2025)

\bibitem{NIST26sig}
{National Institute of Standards and Technology}: Nine candidates advance to
  the third round of the additional digital signatures for the {PQC}
  standardization process. NIST News, May 14, 2026 (accompanying NIST IR 8610)
  (2026)

\bibitem{NIST26r3}
{National Institute of Standards and Technology}: Post-quantum cryptography:
  Additional digital signature schemes -- round 3 additional signatures. CSRC
  project page, last updated September 28, 2026; records that the HAWK
  submission team has withdrawn HAWK and links its announcement on the
  pqc-forum.
  \url{https://csrc.nist.gov/projects/pqc-dig-sig/round-3-additional-signatures}
  (2026), accessed October 8, 2026

\bibitem{Ngu99}
Nguyen, P.Q.: Cryptanalysis of the {G}oldreich--{G}oldwasser--{H}alevi
  cryptosystem from {Crypto} '97. In: Advances in Cryptology -- CRYPTO 1999.
  LNCS, vol.~1666, pp. 288--304. Springer (1999)

\bibitem{NR06}
Nguyen, P.Q., Regev, O.: Learning a parallelepiped: Cryptanalysis of {GGH} and
  {NTRU} signatures. In: Advances in Cryptology -- EUROCRYPT 2006. LNCS,
  vol.~4004, pp. 271--288. Springer (2006)

\bibitem{Nie86}
Niederreiter, H.: Knapsack-type cryptosystems and algebraic coding theory.
  Problems of Control and Information Theory  \textbf{15}(2),  159--166 (1986)

\bibitem{NP26}
Njah~Nchiwo, R.Y., Pellet-Mary, A.: Improved reduction from {RLWE} to {MP-LWE}.
  In: Public-Key Cryptography -- PKC 2026. LNCS, Springer (2026),
  \href{https://doi.org/10.1007/978-3-032-26731-3_6}{doi:10.1007/978-3-032-26731-3\_6};
  ePrint 2026/342

\bibitem{OQS}
{Open Quantum Safe Project}: liboqs algorithm datasheets: {Classic McEliece},
  {HQC}, {UOV}, {MAYO}. \url{https://openquantumsafe.org/liboqs/algorithms/}
  (2026), accessed October 1, 2026

\bibitem{PR23}
Page, A., Robert, D.: Introducing {Clapoti(s)}: Evaluating the isogeny class
  group action in polynomial time. Cryptology ePrint Archive, Paper 2023/1766
  (2023)

\bibitem{Pat96}
Patarin, J.: Hidden fields equations ({HFE}) and isomorphisms of polynomials
  ({IP}): Two new families of asymmetric algorithms. In: Advances in Cryptology
  -- EUROCRYPT 1996. LNCS, vol.~1070, pp. 33--48. Springer (1996)

\bibitem{Pat97}
Patarin, J.: The oil and vinegar signature scheme. Presented at the Dagstuhl
  Workshop on Cryptography (1997)

\bibitem{Pei09}
Peikert, C.: Public-key cryptosystems from the worst-case shortest vector
  problem. In: Proceedings of the 41st Annual ACM Symposium on Theory of
  Computing (STOC 2009). pp. 333--342 (2009)

\bibitem{Pei10}
Peikert, C.: An efficient and parallel {G}aussian sampler for lattices. In:
  Advances in Cryptology -- CRYPTO 2010. LNCS, vol.~6223, pp. 80--97. Springer
  (2010)

\bibitem{Pei16}
Peikert, C.: A decade of lattice cryptography. Foundations and Trends in
  Theoretical Computer Science  \textbf{10}(4),  283--424 (2016)

\bibitem{Pei20}
Peikert, C.: He gives {C}-sieves on the {CSIDH}. In: Advances in Cryptology --
  EUROCRYPT 2020. LNCS, vol. 12106, pp. 463--492. Springer (2020)

\bibitem{PP24}
Peikert, C., Pepin, Z.: Algebraically structured {LWE}, revisited. Journal of
  Cryptology  \textbf{37}, ~28 (2024),
  \href{https://doi.org/10.1007/s00145-024-09508-3}{doi:10.1007/s00145-024-09508-3}

\bibitem{PRS17}
Peikert, C., Regev, O., Stephens-Davidowitz, N.: Pseudorandomness of ring-{LWE}
  for any ring and modulus. In: Proceedings of the 49th Annual ACM Symposium on
  Theory of Computing (STOC 2017). pp. 461--473 (2017)

\bibitem{PHS19}
Pellet-Mary, A., Hanrot, G., Stehl{\'e}, D.: Approx-{SVP} in ideal lattices
  with pre-processing. In: Advances in Cryptology -- EUROCRYPT 2019. LNCS, vol.
  11477, pp. 685--716. Springer (2019)

\bibitem{PS21}
Pellet-Mary, A., Stehl{\'e}, D.: On the hardness of the {NTRU} problem. In:
  Advances in Cryptology -- ASIACRYPT 2021. LNCS, vol. 13090, pp. 3--35.
  Springer (2021)

\bibitem{PX26}
Pellet-Mary, A., Xia, W.: A unified reduction from {RLWE} to {MP-LWE}. In:
  Advances in Cryptology -- ASIACRYPT 2026. LNCS, Springer (2026), cryptology
  ePrint Archive, Paper 2026/2026. \url{https://eprint.iacr.org/2026/2026}

\bibitem{PQC26news}
{PostQuantum.com}: {S}imon claims polynomial-time {DCP} quantum algorithm. News
  article (2026), timeline of the claim, D.~Apon's analysis (August 9--11,
  2026) and the refutation; update of October 11, 2026 on the revision posted
  October 8.
  \url{https://postquantum.com/security-pqc/simon-quantum-algorithm-lattice-pqc/}

\bibitem{Prange62}
Prange, E.: The use of information sets in decoding cyclic codes. IRE
  Transactions on Information Theory  \textbf{8}(5), ~5--9 (1962)

\bibitem{Quentir26}
{Quentir}: A lattice attack was claimed on 3 august 2026 and answered on 15
  august: What {ePrint} 2026/1591 and 2026/1693 say about {ML-KEM}. Blog post
  (2026),
  \url{https://www.quentir.ai/blog/2026/simon-dcp-lattice-claim-refuted-eprint-1591-1693-ml-kem}

\bibitem{Ran25}
Randriambololona, H.: The syzygy distinguisher. In: Advances in Cryptology --
  EUROCRYPT 2025. LNCS, vol. 15606, pp. 324--354. Springer (2025), cryptology
  ePrint Archive, Paper 2024/1193

\bibitem{Reg04c}
Regev, O.: New lattice-based cryptographic constructions. Journal of the ACM
  \textbf{51}(6),  899--942 (2004)

\bibitem{Reg04}
Regev, O.: Quantum computation and lattice problems. SIAM Journal on Computing
  \textbf{33}(3),  738--760 (2004), preliminary version in FOCS 2002

\bibitem{Reg04b}
Regev, O.: A subexponential time algorithm for the dihedral hidden subgroup
  problem with polynomial space. arXiv:quant-ph/0406151 (2004)

\bibitem{Reg09}
Regev, O.: On lattices, learning with errors, random linear codes, and
  cryptography. Journal of the ACM  \textbf{56}(6),  34:1--34:40 (2009),
  preliminary version in STOC 2005

\bibitem{Rob23}
Robert, D.: Breaking {SIDH} in polynomial time. In: Advances in Cryptology --
  EUROCRYPT 2023. LNCS, vol. 14008, pp. 472--503. Springer (2023)

\bibitem{RSSS17}
Ro{\c{s}}ca, M., Sakzad, A., Stehl{\'e}, D., Steinfeld, R.: Middle-product
  learning with errors. In: Advances in Cryptology -- CRYPTO 2017. LNCS, vol.
  10403, pp. 283--297. Springer (2017)

\bibitem{RS06}
Rostovtsev, A., Stolbunov, A.: Public-key cryptosystem based on isogenies.
  Cryptology ePrint Archive, Paper 2006/145 (2006)

\bibitem{RoyessAudit}
{royess}: Lean~4 audit of the dihedral coset algorithm of {ePrint} 2026/1591.
  GitHub repository (2026),
  \url{https://github.com/royess/dihedral-coset-lean-audit}; consulted October
  2, 2026 (commit identifier not recorded); README re-read October 18, 2026

\bibitem{Saa26cost}
Saarinen, M.J.O.: Bit operation cost of ``holdout'' key-recovery attacks
  against {Classic McEliece}. Cryptology ePrint Archive, Paper 2026/1786
  (2026), posted August 24, 2026; revised September 15, 2026.
  \url{https://eprint.iacr.org/2026/1786}

\bibitem{Saa26a}
Saarinen, M.J.O.: {HOVER}: Higher-order vanishing endomorphism recovery.
  Cryptology ePrint Archive, Paper 2026/1778 (2026), posted August 22, 2026;
  revised August 24 and October 8, 2026. NORDSEC 2026, LNCS to appear; ePrint
  revision consulted October 11, 2026. Technical details cited from this
  preprint follow the revision of August 24, 2026.
  \url{https://eprint.iacr.org/2026/1778}

\bibitem{Saa26b}
Saarinen, M.J.O.: Two-anchor holdout/hermite: Solving the {TII}-254 {McEliece}
  key recovery challenge. Cryptology ePrint Archive, Paper 2026/1986 (2026),
  posted September 11, 2026; revised September 24, 2026.
  \url{https://eprint.iacr.org/2026/1986}

\bibitem{SE94}
Schnorr, C.P., Euchner, M.: Lattice basis reduction: Improved practical
  algorithms and solving subset sum problems. Mathematical Programming
  \textbf{66},  181--199 (1994)

\bibitem{Sen00}
Sendrier, N.: Finding the permutation between equivalent linear codes: The
  support splitting algorithm. IEEE Transactions on Information Theory
  \textbf{46}(4),  1193--1203 (2000)

\bibitem{SL26}
Shen, T., Li, X.: A survey on classical lattice algorithms. Cryptography
  \textbf{10}(2), ~17 (2026),
  \href{https://doi.org/10.3390/cryptography10020017}{doi:10.3390/cryptography10020017}

\bibitem{Sho97}
Shor, P.W.: Polynomial-time algorithms for prime factorization and discrete
  logarithms on a quantum computer. SIAM Journal on Computing  \textbf{26}(5),
  1484--1509 (1997)

\bibitem{Sim97}
Simon, D.R.: On the power of quantum computation. SIAM Journal on Computing
  \textbf{26}(5),  1474--1483 (1997)

\bibitem{Sim26}
Simon, D.R.: A polynomial-time quantum algorithm for the dihedral coset
  problem. Cryptology ePrint Archive, Paper 2026/1591 (2026), received August
  3, 2026; made public August 6; revised six times, last on October 8, 2026
  (PDF dated October 7, 2026, still labelled a preliminary draft), adding
  Section 3, ``A Note on [6]'', on ePrint 2026/1693.
  \url{https://eprint.iacr.org/2026/1591}

\bibitem{SS11}
Stehl{\'e}, D., Steinfeld, R.: Making {NTRU} as secure as worst-case problems
  over ideal lattices. In: Advances in Cryptology -- EUROCRYPT 2011. LNCS,
  vol.~6632, pp. 27--47. Springer (2011)

\bibitem{SSTX09}
Stehl{\'e}, D., Steinfeld, R., Tanaka, K., Xagawa, K.: Efficient public key
  encryption based on ideal lattices. In: Advances in Cryptology -- ASIACRYPT
  2009. LNCS, vol.~5912, pp. 617--635. Springer (2009)

\bibitem{SW26}
Straznickas, Z., Weis, S.A.: {HAWK}-$n$ key recovery reduces to {SVP} in
  dimension $n/2+1$. Cryptology ePrint Archive, Paper 2026/1593 (2026),
  \url{https://eprint.iacr.org/2026/1593}

\bibitem{STHY23}
Sudo, K., Tezuka, M., Hara, K., Yoshida, Y.: Quantum search-to-decision
  reduction for the {LWE} problem. In: Progress in Cryptology -- AFRICACRYPT
  2023. LNCS, vol. 14064. Springer (2023), journal version in IEICE Trans.
  Fundamentals E108.A(2), 2025; ePrint 2023/344

\bibitem{ALTEQ22}
Tang, G., Duong, D.H., Joux, A., Plantard, T., Qiao, Y., Susilo, W.: Practical
  post-quantum signature schemes from isomorphism problems of trilinear forms.
  In: Advances in Cryptology -- EUROCRYPT 2022. LNCS, vol. 13277, pp. 582--612.
  Springer (2022)

\bibitem{TPD21}
Tao, C., Petzoldt, A., Ding, J.: Efficient key recovery for all {HFE} signature
  variants. In: Advances in Cryptology -- CRYPTO 2021. LNCS, vol. 12825, pp.
  70--93. Springer (2021)

\bibitem{SQIsignSpec}
{The SQIsign team}: {SQIsign}: Algorithm specifications and supporting
  documentation, version 3.0. \url{https://sqisign.org/} (2026), round-3
  submission to the NIST additional signatures process

\bibitem{Ved26}
Vedenev, K.: Extending distinguishing to key recovery for subfield subcodes of
  {GRS} codes. Cryptology ePrint Archive, Paper 2026/1747 (2026), posted August
  20, 2026; revised October 7, 2026. \url{https://eprint.iacr.org/2026/1747}.
  Earlier versions, in particular the working draft of September 6, 2026, are
  identified explicitly where they are discussed

\bibitem{Wee24}
Wee, H.: Circuit {ABE} with poly(depth, $\lambda$)-sized ciphertexts and keys
  from lattices. In: Advances in Cryptology -- CRYPTO 2024. LNCS, Springer
  (2024), ePrint 2024/1416

\bibitem{Wei26}
Weis, S.A.: Improving {GIJS} key recovery for {Classic McEliece}. Cryptology
  ePrint Archive, Paper 2026/1984 (2026), version 0.5, October 6, 2026; version
  0.1 posted September 11, 2026. \url{https://eprint.iacr.org/2026/1984}

\bibitem{WZ26}
Wen, W., Zheng, J.: Module learning with errors and structured extrapolated
  dihedral cosets. In: Advances in Cryptology -- CRYPTO 2026. LNCS, Springer
  (2026),
  \href{https://doi.org/10.1007/978-3-032-35377-1_9}{doi:10.1007/978-3-032-35377-1\_9};
  ePrint 2026/155

\bibitem{Wes21}
Wesolowski, B.: The supersingular isogeny path and endomorphism ring problems
  are equivalent. In: 62nd IEEE Symposium on Foundations of Computer Science
  (FOCS 2021). pp. 1100--1111 (2021)

\bibitem{Zha25}
Zhang, Y.: Exact coset sampling for quantum lattice algorithms.
  arXiv:2509.12341 [quant-ph] (2025), first posted September 15, 2025;
  version~8 of May 14, 2026 used here. \url{https://arxiv.org/abs/2509.12341}

\bibitem{Zho26}
Zhou, J.: On the sample complexity of the generalized hidden shift problem over
  arbitrary finite groups. arXiv:2609.05828 (2026)

\end{thebibliography}

\end{document}